\documentclass{Dissertate}
\usepackage{amsmath,amssymb,xspace,amsthm}
\usepackage{amsfonts}
\usepackage{array, booktabs}
\usepackage{xr}
\usepackage{zref-xr}
\usepackage{hyperref}
\usepackage{algorithm}% http://ctan.org/pkg/algorithms
\usepackage{algpseudocode}% http://ctan.org/pkg/algorithmicx

\usepackage{mathtools}
\usepackage{float}
\usepackage{physics}
\usepackage{dsfont}
\usepackage{makecell}

\usepackage[breakable]{tcolorbox}

\newtcolorbox{examplebox}[2][]
{
  breakable,
  colframe = gray!50,
  colback  = gray!10,
  coltitle = gray!10!black,
  before skip = 10pt,
  after skip = 10pt,
  title    = \textbf{#2},
  #1,
}

\usepackage{blindtext}
\usepackage{enumitem}
\usepackage{sectsty}
\allsectionsfont{\normalfont\sffamily\bfseries}
\usepackage{amsthm}
\usepackage{thmtools,thm-restate}

\theoremstyle{plain}
\newtheorem{theorem}{Theorem}[chapter]

\theoremstyle{plain}
\newtheorem{lemma}{Lemma}[chapter]

\theoremstyle{definition}
\newtheorem{definition}{Definition}[chapter]

\theoremstyle{plain}
\newtheorem{conjecture}{Conjecture}[chapter]

\theoremstyle{plain}
\newtheorem{proposition}{Proposition}[chapter]

\theoremstyle{plain}
\newtheorem{claim}{Claim}[chapter]

\theoremstyle{plain}
\newtheorem{numericalresult}{Numerical Result}[chapter]

\theoremstyle{plain}
\newtheorem{rigorousresult}{Rigorous Result}[chapter]

\newcommand{\N}{\mathbb{N}}

\newcommand{\Real}{\mathbb{R}}
\newcommand{\Exp}{\mathbb{E}}

\DeclareMathOperator*{\argmin}{argmin}

\newcommand{\Dt}{\Delta t}

\newcommand{\br}{\mathbf{r}}
\newcommand{\bs}{\mathbf{s}}

\newcommand{\bu}{\mathbf{u}}
\newcommand{\bv}{\mathbf{v}}
\newcommand{\bx}{\mathbf{x}}
\newcommand{\by}{\mathbf{y}}
\newcommand{\bz}{\mathbf{z}}

\newcommand{\blambda}{\mathbf{\lambda}}
\newcommand{\bnu}{\mathbf{\nu}}

\newcommand{\minipagebeforesep}{3mm}
\newcommand{\minipageaftersep}{1mm}

\newcommand{\Circ}{\mathcal{C}}
\newcommand{\OPT}{\textsf{OPT}}
\newcommand{\poly}{\textsf{poly}}

\newcommand{\innerp}[2]{\langle#1|#2\rangle}
\newcommand{\ex}[2]{\left\langle#1\right\rangle_#2}
\newcommand{\kket}[1]{|#1\rangle\rangle}
\newcommand{\bbra}[1]{\langle \langle #1|}

\usepackage{titlesec}
\titleformat{\chapter}[display]   
{\normalfont\sffamily\huge\bfseries}{\chaptertitlename\ \thechapter}{20pt}{\huge}   
\titlespacing*{\chapter}{0pt}{150pt}{60pt}

\begin{document}

% the front matter
%\input{frontmatter/personalize}
% \maketitle
\frontmatter
\setstretch{\dnormalspacing}
%\abstractpage
%\tableofcontents
% \authorlist
% \listoffigures
% \dedicationpage
% \acknowledgments 

% \doublespacing
\renewcommand{\thefootnote}{\arabic{footnote}}

% include each chapter...
%\setcounter{chapter}{0}  % start chapter numbering at 0
\begin{savequote}[75mm]
The method is more important than the discovery, because the correct method of research will lead to new, even more valuable discoveries.
\qauthor{Lev D.~Landau}
\end{savequote}

\chapter{Introduction}\label{ch:introduction}

We find ourselves in an era defined by computation. From everyday digital technologies to the rapid advancements in artificial intelligence, to the influential computational methodologies employed across various scientific fields, computation is everywhere. 
Concurrently, the study of computer science has evolved from the abstract modeling of the Turing machine to a diverse list of subareas, including machine learning, robotics, human-computer interaction, and more. 
As we navigate this swiftly changing period in human history, one cannot help but wonder~—~where will computation lead us in the upcoming decades in terms of science, technology, and society?

Before we speculate on such a grand question, it is imperative to first reconsider the fundamental nature of computation and its interplay with other areas of interest. 
This thesis argues for the advantage of employing a conceptual paradigm and a quantitative language to understand our world through the lens of computation, aptly termed, ``the computational lens''. 
The remainder of this introductory chapter aims to motivate and elucidate this proposal by addressing the integral questions of \textit{what} and \textit{why}. The bulk of this thesis will then center around the author's own research in quantum computation and neuroscience. Finally, in the concluding chapter, I will address the question of \textit{how} to properly apply the computational lens in science.

\section{What is the computational lens?}
In Oxford English Dictionary, the word ``computation'' is defined as
\begin{quote}
Computation, noun.\\
The action or process of computing, reckoning, or counting; arithmetical or mathematical calculation; an instance of this.Obsolete.
\end{quote}
This definition indeed aligns quite well with most people's initial instinct upon hearing the word computation. For instance, our laptops and smartphones continuously execute countless calculations in our day-to-day lives. Or, when faced with an important decision, we thoroughly weigh potential solutions, considering consequences and pros and cons, to pinpoint the optimal one.
Meanwhile, the meaning of the term has inevitably evolved along with our changing world. In this dissertation, I propose the following definition of computation:
\begin{definition}[Informal]\label{def:computation}
Computation is a method of reasoning that emphasizes the analysis and transformation of inputs to outputs in an information processing system, conducted through precise mechanistic procedures and informed by mathematical objectives.
\end{definition}

As this is a dissertation for a degree in computer science, we would not dwell excessively on etymology and philosophy. Nevertheless, to elucidate and motivate the above proposed definition of computation, I will present multiple examples of applying the \textit{computational lens} across various scientific fields.

But first, what is the computational lens? The idea is actually very simple and has been around for decades~\cite{karp2011understanding}: studying any subject of interest from the angle of computation. Namely, characterizing the underlying functionality by deconstructing a system with both bottom-up mechanistic implementations as well as top-down computational principles. Recognizing that a thousand pages of high-level reasoning may not elucidate as much as a single concrete example, I will provide numerous relevant examples across a variety of fields, including physics, biology, economics, and mathematics.

Before embarking on this journey, I want to make three remarks. First, it is important to recognize that not all examples below may precisely align with the idealized definition of the computational lens as promoted throughout this thesis. However, for the sake of providing a comprehensive literature review, I have attempted to incorporate a broad range of relevant examples. Secondly, the notion of computation is broad, and as such, some well-known theories or methodologies in other fields such as physics and biology could potentially be encompassed by the computational lens. That is to say, some readers might feel awkward to see concepts they are familiar with being rephrased in terms of the language of computation. I invite the reader to momentarily set aside their previous definitions of computation, enjoy the following examples, and contemplate why I have chosen such a narrative.
Finally, it is important to note that it requires expertise and time for a new application of the computational lens to make real impact in the other field. As such, many of the examples outlined below still require extensive future work. This includes fostering a deeper dialogue between the involved communities, learning each other's languages, adopting and adapting various tools, and developing a shared appreciation for the scientific questions and methods in question. 

\subsection{Computational lens in physics}
Physics is a discipline devoted to studying the natural world, focusing on matter, motion, the concept of time and space, and beyond. A key methodology in physics involves designing mechanical models or theories for the physical systems under investigation, and then using these models to make predictions or test hypotheses. This approach naturally connects to computation as the simulation of physical models and the calculation within physical theories both involve substantial computational effort. However, by ``the computational lens in physics'', I mean something more than mere simulations. I will illustrate this concept using examples from three different sub-fields in physics.

\medskip \noindent \textbf{Classical mechanics.} 
In essence, a mechanical model is an algorithm and hence is inherently computational~\footnote{From the Newton's laws to the Lagrangian and/or Hamiltonian formalism, all these theories in classical mechanics provide mathematical frameworks in which given the configuration of a system (e.g., the initial position and momentum of each particles, and the potential function of the system), the evolution of the system will be governed by some well-specified rules (e.g., the least action principle)}. 
As long as we are using such a mechanistic view to understand the physical world, then our understanding should be constrained by the algorithmic limit from the theory of computation. This suggests that certain physical systems may never be fully understood through detailed mechanistic models, and these algorithmic limits could shed light on why this is so.

For example, although Newton famously solved the two-body problem, i.e., analytically characterizing all possible physical evolution when there are only two objects in the system, the problem becomes intractable when the number of objects increases to three. 
Turing Award laureate Andrew Yao proposed the use of the \emph{extended Church-Turing thesis}\footnote{The extended Church-Turing thesis postulates that every feasible computation in the physical world can be efficiently simulated by a Turing machine.} to elucidate the challenges inherent in solving the three-body problem\cite{yao2003classical}. Another example beyond the realm of classical mechanics is the diffraction limit, i.e., the point at which two light sources become indistinguishable. A recent study by Chen and Moitra proposed an algorithmic foundation to reassess this limit through the lens of sample complexity~\cite{chen2021algorithmic}. 
Indeed, these previous works from the theoretical computer science community had brought fresh perspectives into ancient physics problems. Meanwhile, it would require further time and effort to translate and adjust the algorithmic reasoning into the language and appreciation of physics to yield illuminating insights.

\medskip \noindent \textbf{Statistical physics} concerns the macroscopic properties of large-scale physical systems. The key difference to classical mechanics is that, in statistical physics we now use probability distributions (or more precisely, ensembles) over all possible states to model a system. In this framework, the value of an observable is treated as the expected value over such a probability distribution. Furthermore, the chosen distributions adopt the form of a softmax function over the energy of a state. Therefore, the way statistical physicists approach to understand the subjects of interest is essentially through a certain optimization paradigm.

% While calculating the expectation of a probability distribution is nothing more than computing a big summation (with the number of summands being exponential in the system size), 
This paradigm naturally connects to more sub-fields in the realm of computation, e.g., counting problem (related to entropy) and sampling problem (related to numerical simulation). 
For example, the important Ising models in statistical mechanics is dual to the constraint satisfaction problems (CSPs) and hence enable the application of computational complexity theory. Additionally, numerous heuristics or algorithms utilized in numerical simulations, such as belief propagation, have profound mathematical connections to optimization hierarchies and spectral algorithms.
Through these bridges, lots of physical questions (e.g., phase transition) can be studied algorithmically, and vice versa. 
I recommend a textbook by Mezard and Montanari~\cite{mezard2009information} for readers who want to learn more about the computations in statistical physics.

\medskip \noindent \textbf{Quantum physics} is the study of physics at the scale of atoms and subatomic particles. Its glorious development in the past century has not only reshaped the landscape of physics but also brought up exciting technological advancements. For readers who want to learn more, see~\autoref{ch:prelim quantum} for a preliminary introduction for non-physicists. Here the reader can think of quantum physics as another mathematical framework in explaining the microscopic dynamics, in which given the configuration of a system, the evolution of the system will be governed by some well-specified rules. However, while the quantum theory provides much more accurate predictions to the real world, its physical interpretation as well as its full mathematical underpinning remain elusive.

Previously we saw that the computational aspect in physics is mainly about understanding a given system. Given the unintuitive framework of quantum physics, it is the first time when people started to ask: could there be any non-trivial computation done by a quantum system? Indeed, the seminal Shor's algorithm~\cite{shor1994algorithms} answered this quest affirmatively and opened up this exploding field now known as \emph{quantum computation}. In short, Shor cleverly utilized the new mathematical structure in the quantum paradigm and designed an efficient quantum algorithm for the integer factoring problem, which is believed to be extremely difficult to any practical computer as many real-world cryptographic applications are based on its computational hardness.

In addition to Shor's algorithm and other quantum algorithms designed to solve computational problems more rapidly, the computational way of thinking also enriches our understanding of physics. For example, the concept of \textit{entanglement} (see~\autoref{ch:prelim quantum} for more details) has long mystified physicists, and continues to do so even today. However, the design of quantum algorithms and the ongoing research program exploring the computational advantage of quantum systems have fostered the perspective of viewing entanglement as a particular \emph{computational resource}. This viewpoint has since facilitated a accessible and intuitive approach to reasoning about the bizarre concepts emerging from the quantum formalism.

\medskip \noindent \textbf{Summary of the computational lens in physics.}
The above examples of the computational lens applied to physics can be grouped into three categories:
(i) Usage of algorithmic perspectives to elucidate the limits of mechanical models in physics, such as the challenges inherent to the three-body problem and the unsolvability of 3D Ising models.
(ii) Algorithms serving as new models for comprehending physical phenomena. For instance, the belief propagation algorithm in statistical physics correspond to the concept of Bethe free energy, and the utilization of variational algorithms as an ansatz~\footnote{Often time certain physical models of interest may not immediately yield analytical solutions, as in the case of 3D Ising models mentioned previously. In response, physicists often propose hypotheses regarding the nature of a solution (or approximate solution), effectively creating a model for this solution. For instance, they might suggest that the solution could be generated by a small neural network. This model of a potential solution is referred to as an \textit{ansatz}, and the task then becomes solving the physical model with respect to this proposed ansatz.} for solutions to a physical problem.
(iii) The discovery of novel physical features or insights via the examination of computational perspectives. For example, facilitated by quantum computation, we have enhanced our understanding of the role of quantum entanglement.

\subsection{Computational lens in biology}
Biology is the study of the living world. Over the past two centuries, scientists have accumulated an unprecedented understanding at various scales, but it remains relatively rare to find a theory that seamlessly connects high-level macroscopic phenomena to microscopic mechanisms. In the author's opinion, the theory of evolution is probably the most successful example to date. Specifically, while the concept of natural selection, as proposed by Charles Darwin and Alfred Russell Wallace in 1858~\cite{darwin-wallace}, serves as a computational principle for the emergence of species diversity, genetics provide concretes blueprints for the detailed implementations. Nonetheless, a huge gap persists between the top-down and bottom-up understanding in most areas in biology. For instance, we are far from fully comprehending how DNA sequences encode protein structure or how proteins interact with each other. 

As moving between different levels of organisms  and different layers of abstraction naturally involves computing data from one side to the other, the lens of computation hence naturally serve as an interface to build up a hierarchical understanding with concrete algorithmic mechanisms as well as interpretable computational objectives. In the rest of this subsection, we are going to see several examples of the applications of computational lens to the realm of the living world.

\medskip \noindent \textbf{Evolution.} Natural selection and its modern synthesis together are already a paradigm to reason about biological diversity through the computational lens. 
By thinking about how variations incur the changes in fitness and how it further affect the size of population and so on, such a reasoning template provides a causal model connecting detailed mechanism with high-level functionality.
To further apply the evolutionary thinking to even more complicated problems, there have been many efforts in enriching the paradigm of natural selection. 

To name a few, evolutionary dynamics~\cite{nowak2006evolutionary} utilize mathematical models such as dynamical equations and game theory to further concretize both the mechanism at the genetic level as well as the objective at the population level. On the other hand, there have also been efforts from the theoretical computer science side to explore questions such as the role of sex through an algorithmic perspective~\cite{chastain2014algorithms} or evolutionary timescale through the learning perspective~\cite{valiant2009evolvability}.
Finally, the theory of evolution also inspires back to computer science and open up the fascinating field of evolutionary computing~\cite{eiben2015introduction}.

\medskip \noindent \textbf{Neuroscience} is the study of brain with focuses both on the fundamental biological substrates as well as higher level computations. A more comprehensive introduction to neuroscience and numerous examples of computations in neuroscience from a more biological perspective will be deferred to~\autoref{ch:prelim neuro}. Here, we will only focus on the examples of some theoretical computer science attempts in exploring various topics neuroscience. 

One main challenge in neuroscience is to bridge the understanding across different scales in the brains: from single neuron's behaviors to the dynamics of a population of closely connected neurons, to multiple brain regions, or even to the whole cognitive level. 
In the past few decades, theoretical computer scientists have tried to bring in new perspectives through the computational lens. Theses prior works can be roughly categorized into three types. 
(i) Defining a computational model and studying its computational power as well as biological relevance, e.g., neuroidal model~\cite{valiant2000neuroidal} and assembly calculus~\cite{papadimitriou2020brain}. 
(ii) Viewing neural dynamics as algorithms and examine the underlying computational properties, e.g., distributed algorithms~\cite{lynch2017spiking,su2019spike,hitron2020random}. 
(iii) Studying computational tasks inspired by neuroscience and designing biologically plausible models or algorithms to realize them, e.g., memorization and association~\cite{valiant2006quantitative,valiant2005memorization}
I also recommend a survey by Maass et al.~\cite{maass2019brain}.

\medskip \noindent \textbf{Other examples.} 
In the study of pattern formation in living organisms, Turing proposed a diffusion-reaction model to explain how interactions among two chemical substances could generate stable yet diverse patterns~\cite{turing1952chemical}.
The question of how a group of simple agents can collectively produce non-trivial behaviors is central to the study of swarm intelligence. This subject is explored through the examination of emergent computation (e.g., slime molds solving a maze~\cite{nakagaki2000maze}), the analysis of convergent properties (e.g., bird flocking~\cite{chazelle2014convergence}), or the investigation of underlying algorithmic ideas (e.g., ant colony~\cite{musco2016ant,zhao2022power}). 
In molecular biology, researchers are concerned with fundamental questions such as the processes that map DNA sequences to protein structures. Efforts have been made to apply a computational lens to algorithmic descriptions in gene regulation~\cite{xing2004motifprototyper,ben2002discovering}.

\medskip \noindent \textbf{Summary of the computational lens in biology.} 
The previously discussed examples of the computational lens applied to biology can be classified into three categories:
(i) The use of evolutionary thinking as a computational framework, which serves to elucidate biological diversity and extends even beyond this scope.
(ii) The application of algorithmic perspectives to investigate emergence~\footnote{Emergence refers to the process or phenomenon where novel properties, patterns, or behaviors arise from the interactions and collective behavior of simpler components within a complex system.} in biology, such as pattern formation and bird flocking.
(iii) The employment of computational models to investigate the intricate hierarchical systems in biology, as demonstrated by numerous examples in the field of neuroscience.

\subsection{Computational lens in economics}
There is a huge literature in the intersection of computer science and economics and here I only provide an incomplete list of examples to highlight the new perspective the computational lens could bring. In game theory, equilibrium is the key property to investigate while it becomes computational hard to calculate find the equilibrium of a give game. The seminal work of Daskalakis et al.~\cite{daskalakis2009complexity} showed that finding the Nash equilibrium of a game is indeed computationally intractable in the worst-case assuming some mathematical conjecture. This brought up reexamination on the notion of equilibrium as if the equilibrium is not computationally efficient to find, how would a market converge to it so fast? Or maybe the market does not converge to an equilibrium? Beyond equilibrium, the computational lens also influence other sub-areas of game theory such as mechanism design by considering the algorithmic aspect. See the textbook by Roughgarden~\cite{roughgarden2010algorithmic} for a comprehensive review.

Another important application of the computational lens in economics is modeling an agent as a computationally bounded computer instead of an all-powerful function. For example, this has be adopted to the study of bounded rationality~\cite{kalai1990bounded} where the goal is to better capture how human's decisions deviate from that from the perfect economic rationality.

Bitcoin~\cite{nakamoto2008bitcoin} utilizes the concept of proof of work, which has its root from theoretical cryptography in combating spam mail~\cite{dwork1992pricing}, to base consensus among decentralized agents on the computational hardness in cryptographic primitives. The further growing study of blockchain has brought to decentralized consensus, see a recent textbook~\cite{shi2020foundations} for more coverage.

Last but not least, the concept of fairness has been heatedly discussed in the past few years due to the surging development of machine learning. These advancements immediately call for a quantitative notion of fairness for future deployment of machine learning algorithms. Hence, the area of algorithmic fairness had born~\cite{kleinberg2018algorithmic}.

\medskip \noindent \textbf{Summary of the computational lens in economics.} We can organize the aforementioned examples of the computational lens in economics into three categories: (i) Incorporating computational efficiency considerations into economic problems, e.g., computational hardness of Nash equilibrium and bounded rationality. (ii) Introducing new constructions to economic concepts from computation, e.g., Bitcoin. (iii) Formulating algorithmic definitions for economic concepts, as showcased by algorithmic fairness. 
% Lastly ,given my limited formal training in economics, I encourage readers to interpret these examples with appropriate discretion.

\subsection{Computational lens in pure mathematics}
While the theory of computation brings up important mathematical problems such as the millennium ``P versus NP'' problem~\cite{cook1971complexity,karp1972reducibility}, the computational way of think also has led to several exciting progresses in some sub-areas in pure mathematics. For example, the recent improvement of sunflower lemma~\cite{alweiss2020improved,rao2019coding}, a big progress on upper bounds for 3-arithmetic progressions~\cite{kelley2023strong}, and the breakthrough quantum computational complexity results ($\text{MIP}^*=\text{RE}$) leading to the resolution of a long-standing open problem in operator algebra~\cite{ji2021mip}.

The examples mentioned above all share a similar vein: they utilize the language of theoretical computer science to traverse the narrow alley of precise logical reasoning. In the Sunflower Lemma result, the author leveraged the angles of coding theory and pseudorandomness to iteratively extricate mathematical structures that are not easily described in one shot. In the 3-arithmetic progressions result, the authors again capitalized on the computational perspective of pseudorandomness to facilitate a computational trade-off among mathematical notions of \textit{density increment}, \textit{spreadness}, and \textit{regularity}. Lastly, in the $\text{MIP}^*=\text{RE}$ result, the author heavily employed the language of nonlocal games as well as concepts from coding theory and proof systems, transforming a matrix approximation problems in infinite-dimensional space into intuitive games of players sending messages to each other.

\medskip \noindent \textbf{Summary of the computational lens in mathematics.}
From a bird's-eye view, the common theme in these works is the ``method of reduction'' (see~\autoref{fig:reduction} for an example). By this, I mean the technique of transforming (i.e., reducing) mathematical or computational problems or structures into different forms that may be more solvable or analyzable. 
Notably, due to its flexibility, the method of reduction allows for creative freedom within the realm of logical reasoning, while still preserving mathematical integrity.
Furthermore, the computational lens facilitates intermediate language to structure mathematical thinking, e.g., many concepts from pseudorandomness and communication complexity theory.
It is exhilarating to anticipate future mathematical breakthroughs achieved through the computational lens once more.

\begin{figure}[h]
    \centering
    \includegraphics[width=13cm]{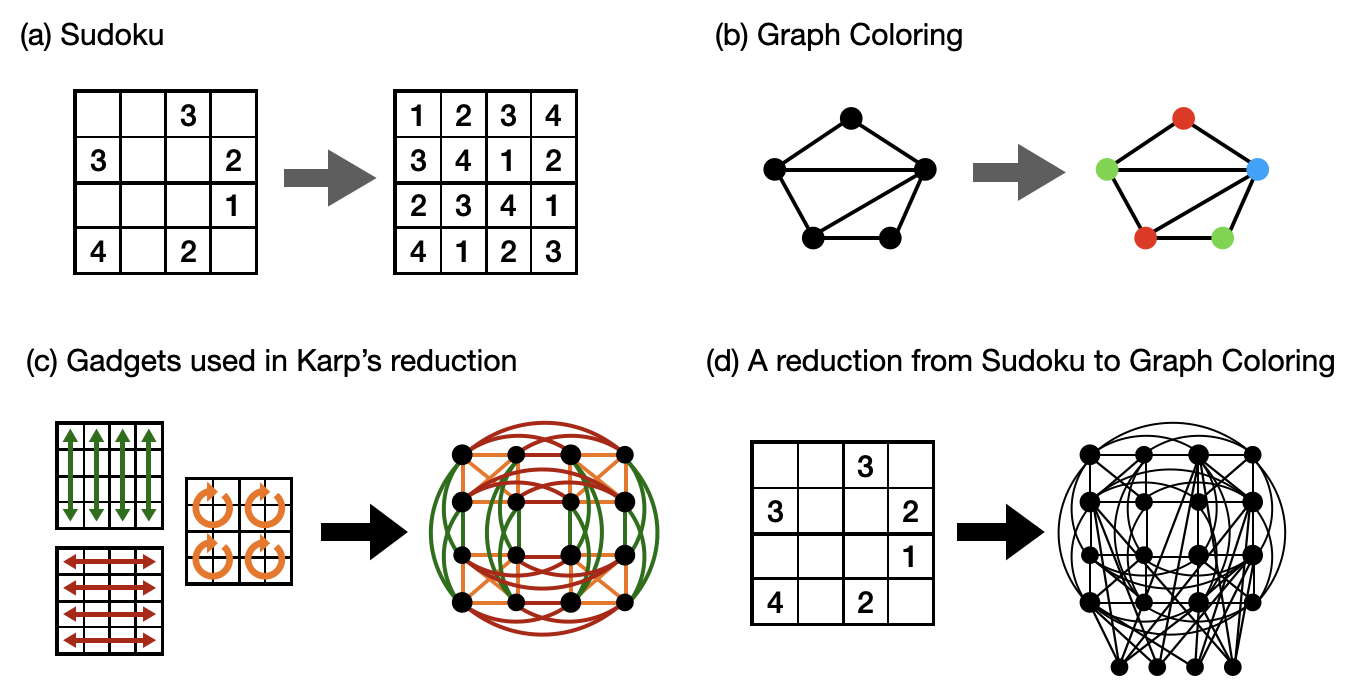}
    \captionof{figure}{An example of a reduction from Sudoku to Graph Coloring. (a) Consider the Sudoku problem where, for simplicity, we use a $4\mathord{\times}4$ grid instead of the common $9\mathord{\times}9$ version. In Sudoku, we need to fill the squares with numbers from 1 to 4, ensuring that each row, each column, and each $2\mathord{\times}2$ region does not contain duplicate numbers. (b) Consider the graph coloring problem commonly seen in theoretical computer science: the problem input is a bunch of vertices connected by edges and a number telling you how many different colors you can use. A valid graph coloring needs to color each vertex, and make sure that the vertices connected by edges are not colored the same. For example, in this figure we can find a valid graph coloring with three colors, but there is no way to achieve a valid graph coloring with only two colors. (c) Gadgets used in a Karp's reduction from Sudoku to graph coloring. (d) we add four more vertices (at the bottom), each potentially representing four different numbers (from left to right, potentially representing 1 to 4), and connect the vertices corresponding to the squares that already have numbers to some of the vertices below. For example, the lower left square is already filled with 4, so we need to connect the corresponding vertex to the first three of the four vertices at the bottom to ensure that it cannot be colored the color corresponding to 4.}\label{fig:reduction}
\end{figure}

% Due to its flexibility, people often apply multiple reductions in a recursive or iterative manner. 
% Notably, the method of reduction allows for creative freedom within the realm of logical reasoning, while still preserving mathematical integrity.

\section{Why should we consider the computational lens?}
Now that we have reviewed several examples of the computational lens in physics, biology, economics, and pure mathematics, I hope the proposed definition of computation in~\autoref{def:computation} is much clearer. Moreover, these examples should convince readers of the ubiquity of the computational lens across various research fields. Having established that, the next question would be, why consider the computational lens? What can scientists gain from its application? 

\medskip \noindent \textbf{The challenge of studying complex systems in science.}
Complex systems refer to systems composed of many components that interact with each other, resulting in intricate joint behaviors. Examples of complex systems in science range from cells and protein networks to our brains, from economic markets to climate, and even galaxies. It is not an exaggeration to say that understanding complex systems represents one of the most vital scientific inquiries of our time.

Since the birth of science in the 17th century, mathematics and physics have played a central role in shaping the ways in which we explore and understand the world. However, due to their rigid and analytical nature, these fields have often struggled to handle the inherent complexities of certain systems. A key reason why complex systems present such a challenge to study is the difficulty involved in reconciling multi-level interactions through bottom-up mechanical methods, top-down phenomenological models, and high-level understanding.

\medskip \noindent \textbf{The dawn to a new scientific methodology.}
Recently, we have begun to see a paradigm shift towards the use of computational methodologies, offering a new lens through which we can investigate and comprehend complex systems.
In particular, the two waves of computing~-~the invention of digital computers and the recent explosion of artificial intelligence (AI)~-~have fundamentally altered the way we approach science at the \textit{instrumental level}. Today, scientists often use computers to simulate models for inference or employ AI models for making predictions. These new computational tools have greatly expanded the capabilities of scientific research, offering innovative ways to explore and understand complex systems.

While these technological advancements have provided hope for decoding the intricacies of complex systems, they also necessitate a conceptual framework to guide our scientific inferences and theory-building. As we navigate the two waves of computing, the twilight at the horizon of our knowledge suggests the dawn of a new scientific methodology.

\medskip \noindent \textbf{The computational lens as a new scientific language.}
The goal of this dissertation is to propose a sketch of how the computational lens can bring forth a new methodology at the \textit{conceptual level}, and provide us with a scientific language to understand and study complex systems. The central thesis of this dissertation can be summarized as follows.

\begin{claim}\label{claim:computation}
Computation is a convenient and mechanistic language to understand and analyze information processing systems with the advantage of its composability and modularity.
\end{claim}

In the remainder of this section, I will provide high-level justifications for the aforementioned claim. While recognizing the boldness of this thesis, it is important to emphasize that the development of a new scientific methodology is a process that necessitates time and extensive research. The following chapters of this dissertation are intended to serve as a precursor to a long-term research agenda, which I will elaborate further in~\autoref{ch:conclusion}. More detailed blueprints for the application of the computational lens in science will emerge in the course of my future research.

\medskip \noindent \textbf{The advantage of composability.}
The term ``compose'' has two meanings: (i) to synthesize and (ii) to create. Interestingly, what I want to emphasize here is that computing possesses both meanings of the term compose.

Most physical models use partial differential equations or more generally dynamical systems for describing the systems of interest. These mathematical tools not only have rich descriptive power but also have a strong and solid theoretical foundation, thus leading physicists to go very far. However, these mathematical models are often too detailed, sometimes delicate, so they are not easy to assemble. In contrast, computer science, starting with Turing machines, has the characteristic of easily synthesizing and assembling different functions/algorithms. This flexible and free way of thinking is more convenient for us in terms of \textit{creation}.

At first glance, creation seems to be going in the opposite direction to the scientific method of trying to understand nature: the former is an \textit{artificial} process, while the latter is a \textit{natural process}. If too many artificial elements are added to the scientific model, can we still accurately understand nature? When facing the real world, we will inevitably create a set of understandings through the mental world and abstract formal language. Since this scientific process is inevitably a creative artificial process, can computing, as an interface that makes it easy for people to think about how to create, be used to establish our scientific understanding and even provide a common language between different fields?

\medskip \noindent \textbf{The advantage of modularity.}
Modularity refers to the ability of being separated and recombined. In software engineering, the concept of \emph{Design Patterns}~\cite{gamma1995design} advocates to abstract out common functionalities and problems into templates with well-specified classes, objects, functions, and so on. For example, imagine we want to open a Taiwanese restaurant, with signature dishes by the head chef such as sweet and sour pork ribs, scrambled eggs with tomatoes, Kung Pao chicken, etc. Although each dish is prepared differently, they all have the following steps in their recipe: (1) preparing and cutting ingredients, (2) pre-processing of meat/eggs, (3) sautéing spices, (4) mixing and stir-frying ingredients, (5) final seasoning. Therefore, when writing recipes for dishes, we can design each of these five stages separately, which makes the coordination of teamwork in the kitchen easier in the future. This concept is also known as the \textit{factory pattern}.

Now, when planning the menu, besides the main dishes, appetizers, desserts, drinks, etc., are indispensable. Moreover, depending on the time and special days, we might want to offer special meal deals, providing customers with various combinations of dishes. If we list down all the possible combinations of dishes, the menu would probably need dozens of pages. So, generally, we would use a tree-like structure to flexibly express the meal options. This concept is also known as the \textit{composite pattern}.

Finally, when formulating the restaurant's operation procedures and staff movements, there will be many \textit{command} actions. For example, Head Chef A orders Sous Chef B to chop vegetables, Sous Chef C asks Assistant D to wash dishes, and Manager E instructs Waiter F to ask the customer about food allergies, etc. The participants in these commands may change as the restaurant staff come and go, while the contents of the commands or the connections between them might also update. Therefore, rather than meticulously describing all the interactions between individuals, a convenient abstraction is to view these inter-role \textit{commands} as an \textit{object}. This not only allows the contents of the command to be flexibly modified, but it also makes it easy to change the senders and receivers of the command. This concept is also known as the \textit{command pattern}.

The above examples only cover the tip of the iceberg of design patterns, but it is hoped that they already give readers a rough feeling for the concept of modularity: by abstracting things and functionalities into modules, operation and understanding become more convenient and clear. The perspective and language of computation lens can naturally accommodate modularity, and such advantage may allow computation to extend quantitative/mechanistic analysis methods to originally more \textit{descriptive} research topics.

\medskip \noindent \textbf{Computation as a convenient and mechanistic language.}
Our understanding is built on the foundation of languages. This can take the form of a descriptive language, as used in storytelling or intuition building. It can also be a quantitative, analytical, or mechanistic language, among others. Moreover, a school of thinking can potentially lie at the intersection of several of these categories. For instance, mathematics is quantitative, analytical, mechanical, and beyond. Different languages each possess their unique strengths and weaknesses in terms of expressive power, convenience of manipulation, and appropriate subjects of coverage.

Computation and the application of the computational lens can also be viewed as a language. Beyond the advantages of composability and modularity, computation is an exceptionally convenient and mechanistic language. It is convenient because it revolves around well-defined computational goals and input-output relations, which facilitates flexible reasoning both in our minds and in our communications with others. Its mechanistic nature allows efficient simulations on computers, implying that the computational lens could serve as a language for communication between humans and machines. Moreover, when we appropriately apply the computational lens to other areas in science, it can serve as a language facilitating dialogue across various fields. Mathematics has been performing such a role since the dawn of science. Perhaps now is the time to incorporate computation as a second unifying language?

\section{Outline of the thesis}
The remainder of this thesis is structured in three parts. In~\autoref{ch:prelim quantum} and~\autoref{ch:example RCS} I will elucidate the computational lens in quantum physics and my own work in quantum computational advantage. In~\autoref{ch:prelim neuro} and~\autoref{ch:example EI} I will discuss the computational lens in neuroscience and my own work on spiking neural networks. The thesis will reach its conclusion in~\autoref{ch:conclusion}, where I will present my perspectives on how to properly apply the computational lens in scientific inquiry, along with contemplation on potential future directions. 

The relations of each chapter to published work or manuscript are as follow.

\medskip \noindent \textbf{\autoref{ch:prelim quantum} (Preliminary Introduction to Quantum Physics)} borrows figures and materials from the author's unpublished book~\cite{chou2023}.

\vspace{1mm}
\medskip \noindent \textbf{\noindent\autoref{ch:example RCS} (Quantum Computational Advantage)} is mainly based on the following two papers.

\vspace{2mm}
\noindent Xun Gao, Marcin Kalinowski, Chi-Ning Chou, Mikhail Lukin, Boaz Barak, and Soonwon Choi. Limitations of linear cross-entropy as a measure for quantum advantage. arXiv preprint: 2112.01657, 2021

\vspace{3mm}
\noindent Boaz Barak, Chi-Ning Chou, and Xun Gao. Spoofing linear cross-entropy benchmarking in shallow quantum circuits. In 12th Innovations in Theoretical Computer Science Conference (ITCS 2021), 2021

\vspace{5mm}
\medskip \noindent \textbf{\autoref{ch:prelim neuro} (Preliminary Introduction to Neuroscience)} borrows figures and materials from the author's unpublished book~\cite{chou2023}.

\vspace{5mm}
\medskip \noindent \textbf{\autoref{ch:example EI} (Algorithmic Neuroscience and Emergent Computations)} is mainly based on the following two papers.

\vspace{3mm}
\noindent Chi-Ning Chou, Kai-Min Chung, and Chi-Jen Lu. On the Algorithmic Power of Spiking Neural Networks. In 10th Innovations in Theoretical Computer Science Conference (ITCS 2019), 2019

\vspace{3mm}
\noindent Andres E. Lombo, Jesus E. Lares, Matteo Castellani, Chi-Ning Chou, Nancy Lynch, and Karl K. Berggren. A superconducting nanowire-based architecture for neuromorphic computing. Neuromorphic Computing and Engineering, 2022

\vspace{5mm}
For the sake of completeness, other published works or manuscripts that are not listed above will be cataloged in~\autoref{app:papers}.

\begin{savequote}[75mm]
I think I can safely say that nobody understands quantum mechanics. So do not take the lecture too seriously, feeling that you really have to understand in terms of some model what I am going to describe, but just relax and enjoy it.
\qauthor{Richard P.~Feynman}
\end{savequote}

\chapter[Preliminary Introduction to Quantum Physics]{Preliminary Introduction to\\Quantum Physics}\label{ch:prelim quantum}

When physicists entered the atomic scale, classical mechanics began to fail to explain the phenomena they were observing. These \textit{quantum phenomena} not only gave birth to the skyscraper of quantum physics, but the emerging sophisticated mathematical structures also enriched the development of modern theoretical computer science. Despite being a gem of physics since the 20th century, quantum physics still holds many mysteries. With the popularization of science, an increasing number of terms prefixed with ``quantum'' have emerged in everyday life. What exactly is ``quantum'', and what does it have to do with computation?

The primary aim of this chapter is to provide readers with limited background in physics and mathematics a glimpse into the world of quantum physics and quantum computing. To appeal to a wider audience, I have made the conscious decision to favor simplicity over depth. I encourage readers to concentrate more on the overarching concepts and narratives and to consult the cited references for a more in-depth exploration of the subject matter. For those readers who are keen on delving more seriously into these topics, I highly recommend the textbook by Sakurai~\cite{sakurai1995modern} for quantum mechanics, and the textbook by Nielsen and Chung~\cite{nielsen2002quantum} for quantum computing and quantum information.

\section{Quantum  phenomena and the mathematical formalism}
In the early 20th century, physicists began to encounter a host of phenomena in microscopic experiments that couldn't be explained by classical mechanics. From the discretization of physical quantities to the wave-particle duality of light, these \textit{quantum phenomena} compelled theoreticians to construct an entirely new physical framework to understand the microscopic world. Thus, under the collective brainstorming of numerous brilliant minds (including Einstein, Heisenberg, and Bohr), a completely new mathematical theory was established. These theories not only accounted for the diverse quantum phenomena but also predicted many novel physical properties, even further leading to innovative mathematical and computational theories. However, the quantum theory's worldview, which is inherently quite different from classical mechanics, still perplexes physicists and philosophers as to how to interpret the relationship between these mathematical structures and the real world. In this section, we will focus on introducing several classic quantum phenomena, the basics of quantum formalism, and the entanglement between quantum physics and computation.

\subsection{Quantum phenomena}
In the framework of classical mechanics, we describe many physical quantities using ``continuous'' variables. For instance, the speed $v$ of a ball could be 10 meters/second, 0.1 meters/second, or 0.0001 meters/second, meaning that the magnitude of $v$ can be any number arbitrarily close to 0. However, with the advancement of experimental techniques, physicists gradually discovered that many physical quantities surprisingly could not take on just any values. These physical quantities, which can only have discrete values, are also known as ``quantized'' values. For example, the seminal Stern–Gerlach experiment demonstrated that angular momentum is quantized. In~\autoref{fig:quantum sg 1}, the reader can sense and understand this without needing to possess much related physics knowledge.

\begin{figure}[ht]
    \centering
    \includegraphics[width=12cm]{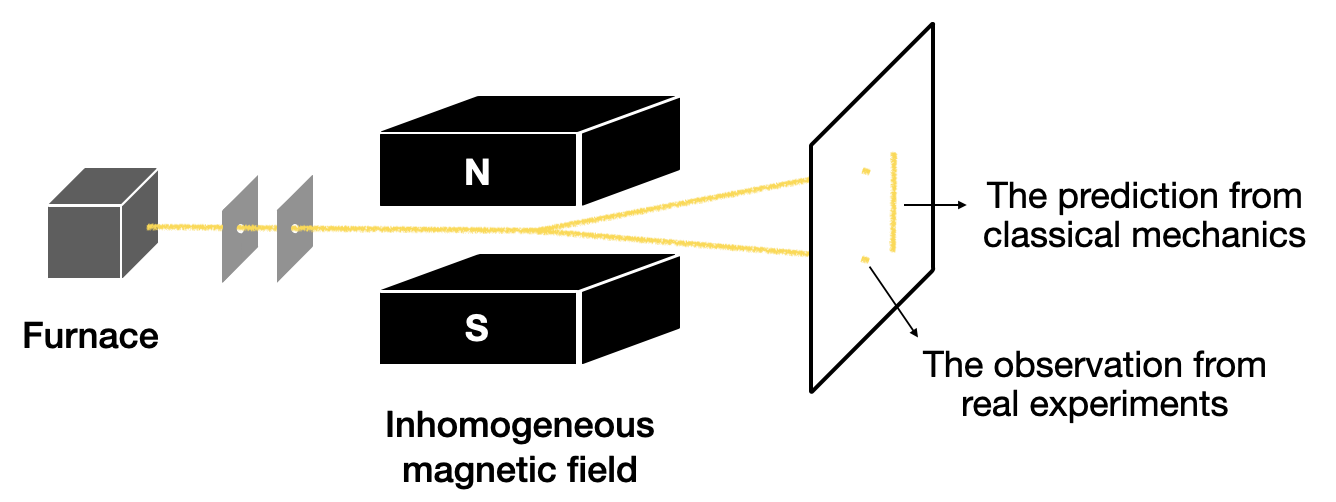}
    \caption{Stern–Gerlach experiment showed that angular momentum is quantized. In the experiment, first, a furnace is used to heat many silver ions to high temperatures, causing each silver ion to start having nonzero angular momentum in various directions. Next, these silver ions are shot into a non-uniform magnetic field, which runs from top to bottom, parallel to the ground. This magnetic field pushes each silver ion up or down based on the direction and magnitude of its angular momentum. If the value of angular momentum could be arbitrarily close to 0, the distribution of silver ions on the final screen should be a continuous line. However, in the Stern-Gerlach experiment, the distribution of silver ions on the screen is discrete, thus demonstrating that angular momentum is quantized.}
    \label{fig:quantum sg 1}
\end{figure}

\subsection{Wave-Particle Duality}
When we think of waves and particles, the images that might come to our minds are of a vibrating line and a tiny ball. In physics, waves and particles have been generalized and abstracted as fundamental characteristics of physical properties. A ``wave'' depicts physical properties that exhibit oscillation and can superpose on each other, such as when two sound waves of different frequencies meet, they can superpose to create a sound wave of a new shape. A ``particle'' corresponds to physical properties that have definite values, such as the mass of an object. In the world of classical physics, a physical property can be described either as a wave or as a particle. However, in some experiments, physicists gradually discovered that many physical properties surprisingly have both wave and particle characteristics. This is known as wave-particle duality. When we extend the setup of the previously mentioned Stern-Gerlach experiment, it can demonstrate that angular momentum possesses wave-particle duality.

\begin{figure}[ht]
    \centering
    \includegraphics[width=10cm]{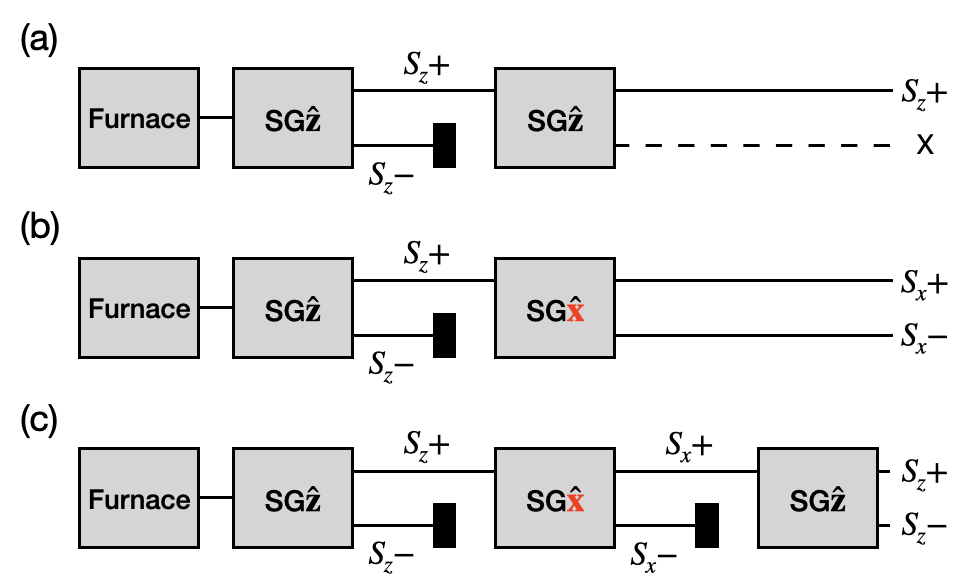}
    \caption{The sequential version of Stern-Gerlach experiment showed that angular momentum possesses wave-particle duality. For the sake of notation, we use $SG\hat{\mathbf{z}}$ to denote a Stern-Gerlach experiment conducted with the magnetic field in the z-direction (i.e., up and down), and the silver ions are classified into two groups, $S_z^+$ and $S_z^-$, according to their angular momentum in the z-direction. Additionally, we use $SG\hat{\mathbf{x}}$ to denote an experiment conducted with the magnetic field in the x-direction (i.e., in and out of the picture), and the silver ions are divided into two groups, $S_x^+$ and $S_x^-$, based on their angular momentum in the x-direction. Here are the observed results of three different sequences of SG experiments. (a) After conducting an $SG\hat{\mathbf{z}}$ experiment once, if the $S_z^-$ ions are removed and another $SG\hat{\mathbf{z}}$ experiment is conducted, only $S_z^+$ ions will remain. (b) If an $SG\hat{\mathbf{z}}$ experiment is conducted and the $S_z^-$ ions are removed, then an $SG\hat{\mathbf{x}}$ experiment is performed, both $S_x^+$ and $S_x^-$ ions will be observed. (c) If in experiment b the $S_x^-$ ions are removed and another $SG\hat{\mathbf{z}}$ experiment is conducted, both $S_z^+$ and $S_z^-$ ions will be observed. This means that although in experiment a we saw that after the $S_z^-$ ions were removed, conducting $SG\hat{\mathbf{z}}$ again would not yield any $S_z^-$ ions. However, if an $SG\hat{\mathbf{x}}$ experiment is conducted followed by an $SG\hat{\mathbf{z}}$ experiment, $S_z^-$ ions will indeed be observed. This physical property is a wave-like characteristic, so this series of experiments tells us that light exhibits wave behavior.}
    \label{fig:quantum sg 2}
\end{figure}

\subsection{How to mathematically model quantum phenomena?}
From the above discussions, we can think of the observed quantum phenomena as giving physicists a new set of \emph{constraints} and \emph{relations} on how atoms work at the microscopic scale. To be a little more precise, there are two things we care about an object of interest: (i) its state and (ii) the operations acting on the object and how it changes its state. It is very tempting to simultaneously model these two aspects and incorporate the mathematical model with our own physical intuitions. For example, for a particle (in the classical physics sense), the possible states are its location and the possible operations are shifting it in various directions. This immediately gives the mathematical picture of a point moving in the (3-dimensional) Euclidean space. The same modeling approach can be done for waves (in the Fourier space though) too. How can we identify the right mathematical space for the states of a quantum object and the associated operations?

One brilliant perspective in quantum physics is to shift our focus from the state space to the space of all possible operators~\footnote{The operator view had also been taken in some classical theories, but in the author's opinion, it brought much more impacts in the quantum realm.}. Notice that once we fully understand how all the possible operators work, we don't even need to describe the underlying state anymore. Here, by ``fully understand how all the possible operators work'', I mean in a very mathematically precise sense: understanding all the \textit{commutation relations} among operators. Namely, We need to answer questions such as which operation will always keep the object the same (identity), what's the difference between applying operator $A$ before operator $B$ and applying $B$ before $A$. Rigorously, we have to define an \emph{algebraic structure}~\footnote{An algebra is a mathematical structure consisting of a set together with multiplication, addition, and scalar multiplication by elements of a field (e.g., complex numbers) and satisfying the axioms of vector space and bilinearity. In the context of this chapter, one should think of the element in an algebra as an operator we can apply to the quantum object we are studying. Meanwhile, most of the time we don't necessarily need the full generality of algebra to study a physical object. For example, group or ring structure could already be sufficient.} for the collection of operators of interest. For concreteness, in the following we give an example on \textit{angular momentum}.

\begin{examplebox}{An example: angular momentum.}
In high school, we learn in the physics class that the (orbital) angular momentum of a rotating (rigid) object is defined as $L\mathord{=}rmv$ where $r$ is the orbit's radius, $m$ is the mass, and $v$ is the tangential speed. Or more generally, in 3-dimensional space angular momentum becomes a 3-dimensional vector $\mathbf{L}\mathord{=}\mathbf{r}\mathord{\times}\mathbf{p}$ where $\mathbf{r}$ and $\mathbf{p}$ are position and momentum vector respectively, and $\times$ is the cross product. Furthermore, angular momentum is important because it is a conserved quantity of a closed system.

To define angular momentum in terms of operators, we instead specify the three operators $\hat{\mathbf{L}}\mathord{=}(\hat{L}_x,\hat{L}_y,\hat{L}_z)$ that measures the angular momentum along the $x,y,z$ axis respectively. Namely, when applying operator $\hat{L}_x$ on the object of interest, it will tell us the the angular momentum of this object along the $x$ axis, and also potentially change the state of the system~\footnote{This way we model how a physical experiment would affect the underlying state of the object under investigation.}. To fully understand $\hat{\mathbf{L}}$, we have to understand all the commutation relations among $\hat{L}_x,\hat{L}_y,\hat{L}_z$. Beautifully, angular momentum exhibits the following elegant commutation relations~\footnote{Here we omit the reduced Planck constant $\hbar$.}:
\begin{equation}\label{eq:angular momentum}
[\hat{L}_x,\hat{L}_y]=i\hat{L}_z,\ [\hat{L}_y,\hat{L}_z]=i\hat{L}_x,\ [\hat{L}_z,\hat{L}_x]=i\hat{L}_y
\end{equation}
where $[A,B]=AB-BA$ and $i$ is the imaginary number.

Finally, through this formalism of operator algebra, as long as some triplet of operators satisfies the above commutation relations, we can view them as a certain kind of angular momentum. Moreover, this also allows us to connect angular momentum to mathematical object known as \textit{spinor}, and hence opens up a rich abstraction for exploring the underlying physics and generalization.
\end{examplebox}

Note that here the collection of operators can also contain operations that give us information about the object under investigation. In physics, we also call such an operator an \emph{observable}. For example, there could be position operators giving us the position of the object or momentum operators telling us the momentum of the object. Once again, notice that the operator picture is indeed complete in the sense that it captures all the information we could study on a physical object through experiments.

\medskip \noindent \textbf{An important example: spin.}
Interestingly, in the above-mentioned Stern-Gerlach experiment, while the silver atom used in the experiment is a neutral particle which does not possess orbital angular momentum~\footnote{Here orbital angular momentum refers to the angular momentum from electron revolving in an orbit.}, the Stern-Gerlach experiment suggests that the silver atom contains an intrinsic quantity that assembles angular momentum. Specifically, the experiment hints that there are angular momentum operators (i.e., a triplet of operators satisfying~\autoref{eq:angular momentum}) associated with this intrinsic quantity. Physicists call such an intrinsic property of a particle as \textit{spin}.

\begin{examplebox}{More on spin.}
Physically, spin is an intrinsic property of a particle that exhibits wave-like properties. Mathematically, spin corresponds to a triplet of operators $(\hat{S}_x,\hat{S}_y,\hat{S}_z)$ that satisfies the commutation relations of angular momentum (i.e.,~\autoref{eq:angular momentum}). In addition to the three spin operators, it is also natural to define the total spin operator as $\hat{S^2}\mathord{=}\hat{S_x^2}\mathord{+}\hat{S_y^2}\mathord{+}\hat{S_z^2}$, which is a conserved quantity of a closed system. With some algebraic manipulations, one can check that $\hat{S^2}$ can only take values half-integer values, i.e., $0,\frac{1}{2},1,\frac{3}{2}$ and so on. Namely, the simplest non-trivial case would be a closed system with total spin number being $\frac{1}{2}$. We call particle with total spin number $\frac{1}{2}$ a spin-$\frac{1}{2}$ particle. For example, proton, neutron, and electron are all spin-$\frac{1}{2}$ particles.
\end{examplebox}

\section{On the theoretical side: quantum computational speedup}

Beginning in the 1960s, some physicists started to ponder the construction of computing systems that incorporated quantum physics. It wasn't until the 1980s, after pioneering research by several physicists, mathematicians, and computer scientists, that the theoretical definitions of a ``quantum Turing machine'' and ``quantum circuit'' gradually took shape. Building on these mathematical models, more and more people began to design various ``quantum algorithms'', trying to understand if the special properties in quantum physics could provide speedups in certain computational problems.

A line of exciting research developments culminated in the Shor's algorithm, proposed by Peter Shor in 1994~\cite{shor1994algorithms}, which drew worldwide attention to quantum computing. Shor discovered that quantum computers could theoretically crack a core computational problems in cryptography at high speeds. This means that if large-scale quantum computers could be implemented, many financial encryption systems used in real life would no longer be secure.

In this section, we will see how do physicists and computer scientists use the spooky quantum phenomena to encode information and further accommodate non-trivial computations.

\subsection{Spin-$\frac{1}{2}$ as the quantum bit}\label{sec:quantum spin qubit}
In computer science, we use strings of bits (i.e., 0 or 1) to encode information and perform computation on them. Physically, a bit is realized by a transistor which can either be in an ``on'' or ``off'' state. We can use other materials to implement the abstract concept of bits as well. For example, prior to transistor-based computers, vacuum-tube computers with magnetic-core memory used the direction of magnetization to represent 0 and 1. The key is to map the physical state of a system to the abstraction of bit strings and correspond the manipulation of bit operations to changes in the system's state. Once this is achieved, we have a digital computer.
That's being said, on one side we have the abstraction of bits and some available logical operations to perform computation. On the other side we have the physical implementation of bits and some available physical operations on top of them to fully simulate the logical bits. The development of computer science had long been focusing on the abstraction side while engineers worked on building physical systems that can accommodate the logical abstraction in use.

The landscape changed when quantum physics entered the scene.  Pioneering scientists noticed that a quantum system can cater to much richer operations than the tradition logical ones. By choosing (quantum) spin, the simplest non-trivial quantum object, as the fundamental information carrier, the concept of a quantum bit, or qubit, was born. This gave birth to the field of quantum computation, which aims to explore this innovative abstraction of the qubit and the potential computations associated with it.

\begin{examplebox}{A short technical introduction to qubit.}
A quantum bit, or qubit, is represented mathematically by a two-dimensional vector in a complex Hilbert space, for example, $\alpha\ket{0}\mathord{+}\beta\ket{1}$, where $\alpha$ and $\beta$ are referred to as the amplitudes of the qubit. Both of them are complex numbers, and $|\alpha|^2\mathord{+}|\beta|^2\mathord{=}1$. This means a qubit is a unit-length two-dimensional complex vector, where $\ket{0}$ and $\ket{1}$ form a coordinate (orthogonal basis) in the Hilbert space the qubit resides in.

\vspace{3mm}
\includegraphics[width=13.5cm]{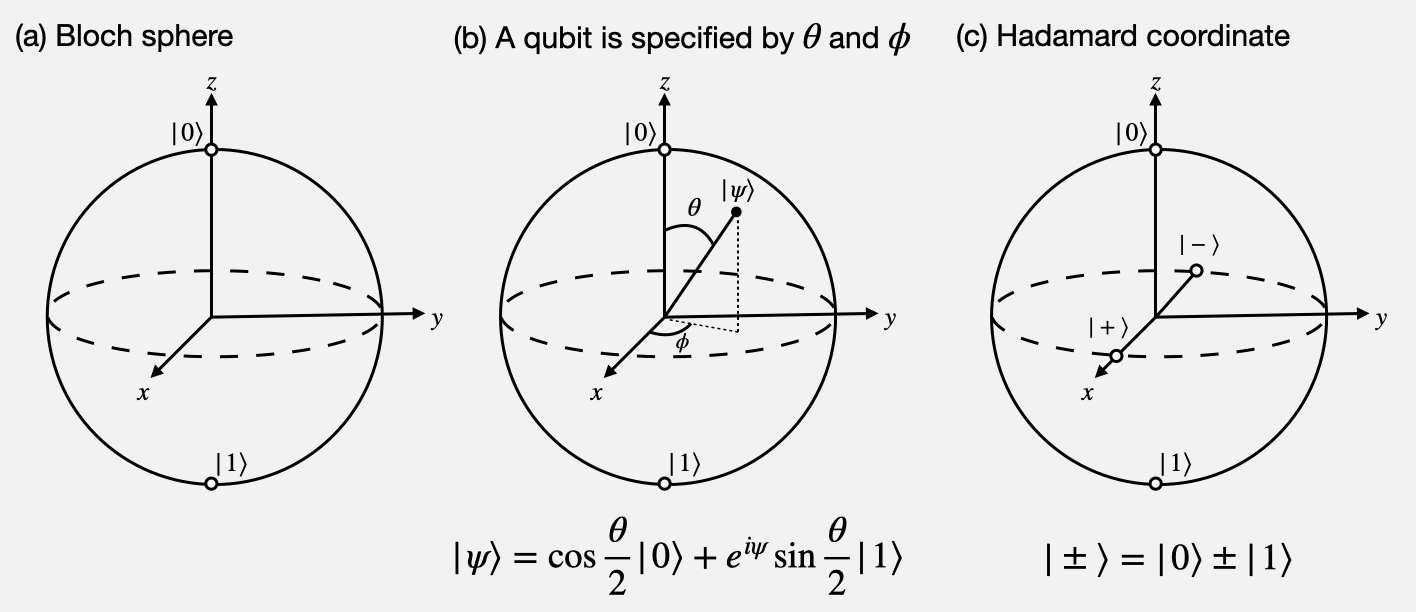}
\captionof{figure}{(a) A quantum bit, or qubit, can be represented geometrically on the surface of a Bloch sphere, where the north and south poles correspond to the states $\ket{0}$ and $\ket{1}$ respectively. It's important to note that the Bloch sphere should not be directly associated with the three-dimensional space in the real world, it is merely a graphical aid to help us understand qubits. (b) Any quantum state can be described by two parameters. The first parameter, $\theta$, represents the difference in amplitude between $\ket{0}$ and $\ket{1}$, and the second parameter, $\phi$, represents the phase difference between the amplitudes of $\ket{0}$ and $\ket{1}$. (c) The Hadamard coordinate system, which we have been discussing, includes states $\ket{+}$ and $\ket{-}$. Their locations on the Bloch sphere are as follows: $\ket{+}$ is on the equator of the Bloch sphere halfway between $\ket{0}$ and $\ket{1}$, while $\ket{-}$ is also on the equator but on the opposite side of the sphere from $\ket{+}$.}\label{fig:qubit}

However, we cannot directly read the amplitudes (i.e., $\alpha$ and $\beta$) of a qubit. The only way to acquire information about the qubit is through a ``measurement''. Similar to the concept of measurement in quantum physics, we first need to select a coordinate system and then project the qubit onto this system. The result obtained after measurement will depend on the length of the projection. For example, if the qubit $(\ket{0}\mathord{+}\ket{1})/\sqrt{2}$ is measured in the coordinate system formed by $\ket{0}$ and $\ket{1}$, there's a 50\% chance to get $\ket{0}$ and a 50\% chance to get $\ket{1}$.

Another common measurement system is the Hadamard coordinate system, where the basis states are $\ket{+}\mathord{=}(\ket{0}\mathord{+}\ket{1})/\sqrt{2}$ and $\ket{-}\mathord{=}(\ket{0}\mathord{-}\ket{1})/\sqrt{2}$. For the previously mentioned qubit $(\ket{0}\mathord{+}\ket{1})/\sqrt{2}$, if measured in the Hadamard coordinate system formed by $\ket{+}$ and $\ket{-}$, there is a 100\% chance to get $\ket{+}$.

Finally, note that the mathematical formulation above for qubit works for a spin-$\frac{1}{2}$ particle as qubit is simply an alias of spin-$\frac{1}{2}$ particle in the context of computing.
\end{examplebox}

\subsection{Quantum computational models}
Algorithms are mechanistic recipes designed to solve specific computational problems. Similar to how a baker's recipe is constrained by the available kitchenware, algorithms for a computational model are also tied to the available atomic operations. Hence, before talking about design quantum algorithms, we have to first figure out the quantum computational model we are going to use.

In the world of quantum computing, we can naturally define analogous computational models of circuits and Turing machines from the kingdom of classical computing. However, since the operations that can be performed on a quantum state without measurement are very constrained~\footnote{Mathematically speaking, only so-called unitary transformations can be performed.}, the definition of a quantum Turing machine is far more complex than a classical Turing machine, and it is also less intuitive to use. Therefore, people mainly focus on quantum circuits.

\begin{figure}[ht]
    \centering
    \includegraphics[width=13cm]{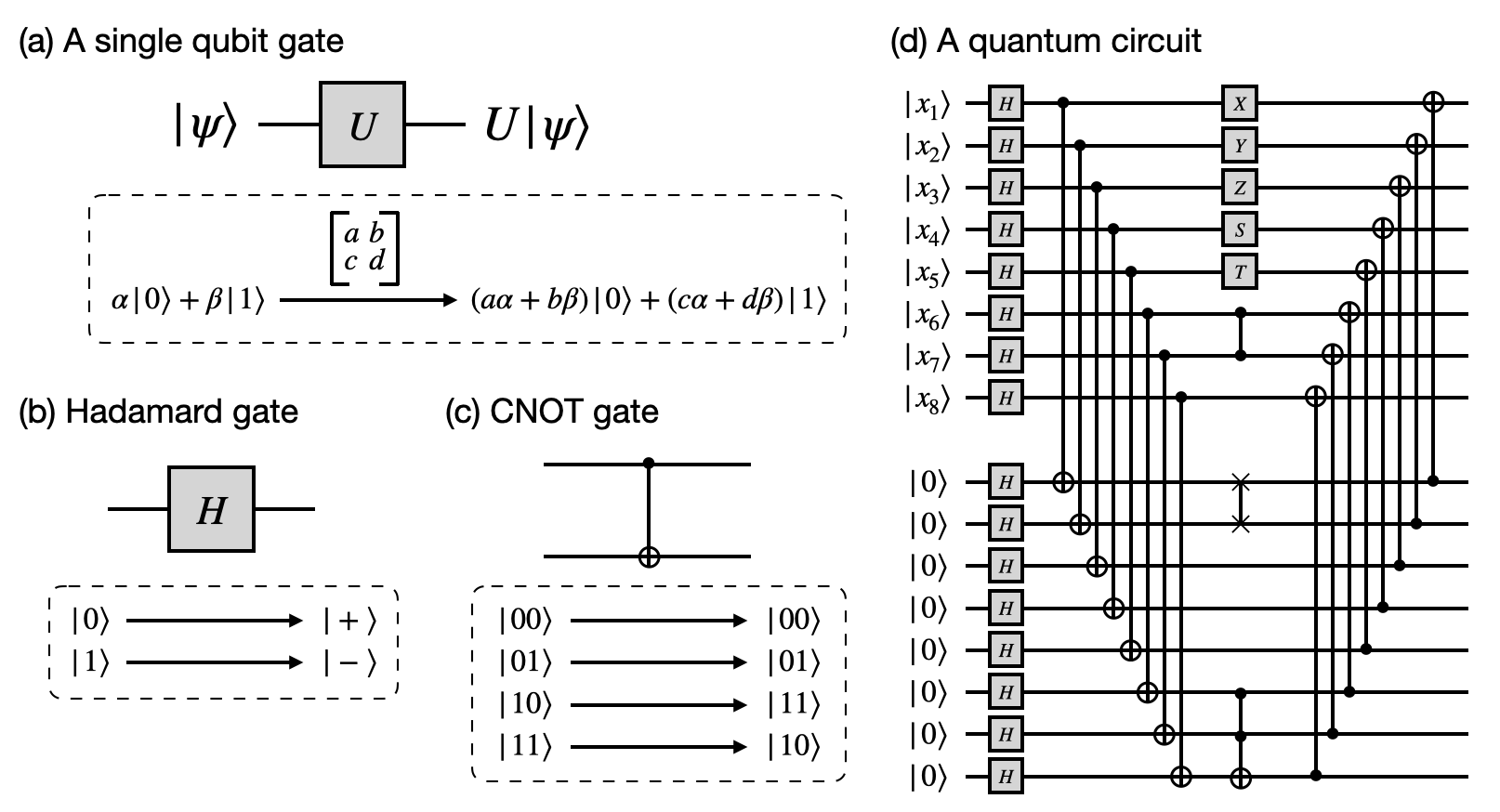}
    \caption{(a) A single qubit gate acts on a single qubit, and its action can be described by a $2\times 2$ matrix. (b) One common single-qubit gate is the Hadamard gate, which transforms the basis of a qubit from the 0/1 basis to the +/- basis. Mathematically, the Hadamard gate is represented by a $2\times 2$ matrix, and when it acts on a qubit, it creates a superposition of states, which is one of the fundamental features of quantum computation. (c) The controlled NOT gate (CNOT) acts on two qubits. The CNOT gate flips the second qubit (target qubit) if and only if the first qubit (control qubit) is in state $\ket{1}$. (d) A quantum circuit is typically depicted with horizontal lines, each representing a qubit, and boxes on these lines, each representing a quantum gate. The input is usually placed on the left, and there might be some auxiliary qubits (ancilla qubits) set to 0. The gates in the middle include the Pauli $X, Y, Z$ gates, phase gate, $T$ gate, controlled $Z$ gate, swap gate, and Toffoli gate. The circuit is read from left to right, representing the sequence of operations applied to the qubits.}
    \label{fig:quantum circuit}
\end{figure}

In quantum circuits, the basic computational units are known as quantum gates, which can perform transformations on a few quantum bits (qubits). As mentioned earlier, due to the constraints of quantum mechanics, quantum gates must conform to a specific mathematical form and cannot be any arbitrary function like classical logic gates. In simple terms, intuitively, quantum gates must preserve all information, and mathematically, this corresponds to the so-called \textit{invertibility} (or \textit{unitarity}). In~\autoref{fig:quantum circuit}, we will see a few common quantum gates (classified according to the number of qubits they act on).

Quantum circuits are computational models that sequentially apply quantum gates on different sets of qubits. So how do we compute with quantum circuits? It can basically be described in the following three steps:
\begin{enumerate}
\item Design a good quantum circuit. Note that the description of a quantum circuit itself is classical information, meaning it can be designed and represented on a classical computer.
\item Prepare the input of the computational problem instance into a quantum state and connect it to the corresponding input of the quantum circuit. Note that we usually add some \textit{ancilla qubits}, set to all zeros, as extra input qubits.
\item Execute the quantum circuit and measure the output to get an output quantum state. If it is a decision-type computational problem, we usually just look at the first output result.
\end{enumerate}

Two remarks on quantum circuits are made here. First, this is a universal computational model, meaning that in terms of computational power it is equivalent to Turing machine. Secondly, it is widely believed that for sufficiently complex quantum circuits, a direct simulation from classical computers would require exponential time~\cite{aaronson2017complexity}

\begin{examplebox}{Another computational model: quantum adiabatic computing.}
Quantum circuit is a computational model that quantizes classical circuits. Can we use the time evolution in quantum physics itself for computation? Perhaps this is easier to realize in practice? Quantum adiabatic evolution is just such a computational model.

In Schr\"{o}dinger equation $i\hbar\frac{\partial}{\partial t}\ket{\psi}\mathord{=}\hat{H}\ket{\psi}$, the Hamiltonian operator $\hat{H}$ plays an important role. Now let's consider the eigenstates of the Hamiltonian operator, and sort them according to the size of their eigenvalues. If we start from the ground state of a Hamiltonian operator $\hat{H_0}$, and slowly adjust $\hat{H_0}$ to another Hamiltonian operator $\hat{H_1}$, can we make this quantum state evolve to always stay in the ground state?

If possible, then we can perform computation in the following way: Take $\hat{H_0}$ as a simple Hamiltonian operator, so that we clearly know its ground state. Then, according to the computational problem to be solved, set $\hat{H_1}$ so that its ground state is the solution to this computational problem. Similar to quantum circuits, this computational model is universal.

However, there is a crucial point here: How slow must the transition from $\hat{H_0}$ to $\hat{H_1}$ be to ensure that the evolution of the quantum state always remains in the ground state? The quantum adiabatic theorem tells us that the speed of this transition will be related to the energy gap between the ground state and the first excited state - the smaller the energy gap, the slower the required speed. Intuitively speaking, if the movement is too fast, it will inject some extra energy into the system, causing the ground state to possibly be excited to other eigenstates. This is why such a method of computation is called adiabatic evolution.

\vspace{3mm}
\begin{center}
\includegraphics[width=8cm]{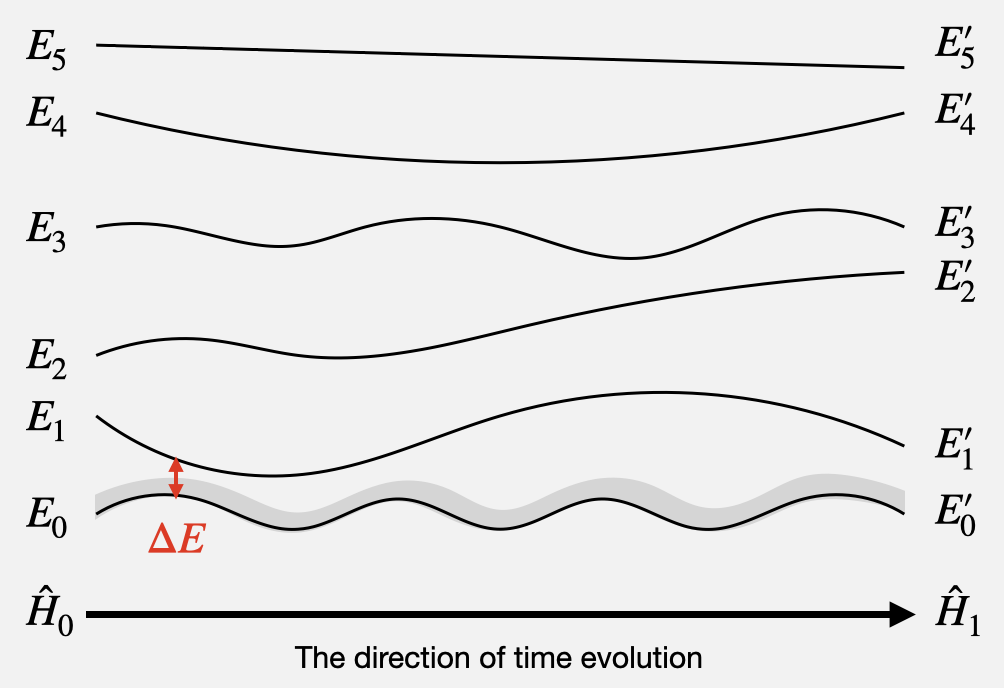}
\captionof{figure}{Quantum adiabatic computing. Hamiltonian operators $\hat{H_0}$ and $\hat{H_1}$ have their respective energy levels/eigenvalue spectra (as in the far left and far right of the figure), where the goal of the calculation is to find the ground state of $\hat{H_1}$ (i.e., the eigenstate corresponding to $E_0'$). First, we would intentionally choose an H0 such that its ground state is well prepared, and then start from it, gradually transforming $\hat{H_0}$ into $\hat{H_1}$. The vibration/instability of the ground state (as shown in the gray area in the figure) will be proportional to the size of the energy gap. Therefore, if the energy gap remains large during the evolution process, it can ensure that the quantum state always remains in the ground state (at each respective time).}\label{fig:qac}
\end{center}
\end{examplebox}

\subsection{Quantum algorithms}\label{sec:quantum prelim algorithm}
Interestingly, with just a few quantum operations, we can already perform tasks that seem infeasible for a classical computer. Thanks to the ability to superpose across multiple quantum states, we can superpose various inputs and evaluate them on a given function \textit{simultaneously}. However, this comes with a significant caveat: the resulting evaluation remains a superposition of results. That is, if we directly measure the outcome, we only get the result from a random input, rendering it no different from probabilistic computing.

So, it is no exaggeration to say that the crux of quantum algorithms is figuring out how to efficiently extract useful information from a superposition. Lov Grover discovered that, generally, one could always achieve a quadratic speed-up over a classical computer. More specifically, let $f:\{0,1\}^n\to\{0,1\}$ be a Boolean function that one can access freely. In the worst-case scenario - with respect to the unknown function $f$ (imagine $f$ as having exactly one input $x$ evaluating to $1$) - it would require a classical computer $\Omega(2^n)$ function calls to find an input $x$ such that $f(x)=1$. However, the seminal Grover's quantum algorithm can find such an $x$ with high probability using only $O(\sqrt{2^n})$ function calls\cite{grover1996fast}. Remarkably, this quadratic speedup for quantum algorithms is optimal, meaning for a generic $f$, we can only achieve a quadratic quantum speedup over classical algorithms\cite{bennett1997strengths}.

While the improvement from $2^n$ to $\sqrt{2^n}$ is indeed impressive, it's not astoundingly exciting, because the running time of Grover's quantum algorithm is still not polynomial-time. The real excitement came when a line of works showed that there are quantum algorithms achieving \textit{exponential} speed-up on certain computational problems with an inherent structure. Notably, Peter Shor, in 1994, demonstrated a polynomial-time quantum algorithm for the integer factoring problem~\cite{shor1994algorithms}, which is not believed to have polynomial-time classical algorithms. More discussion on the comparison of polynomial-time and exponential time will be provided in the next subsection.

From a bird's eye view, all these non-trivial quantum algorithms leverage the special structure underlying the computational problems and use quantum algorithmic techniques such as quantum Fourier transform~\cite{coppersmith2002approximate}, quantum phase estimation~\cite{kitaev1995quantum}, etc., to extract useful information from quantum superposition. For readers interested in a deeper understanding of quantum computational models and quantum algorithms, I recommend the textbook by Nielsen and Chung~\cite{nielsen2002quantum}.

\subsection{Quantum computational complexity}
It's critical to emphasize upfront that in terms of computability, quantum computational models are \textbf{not} inherently superior to classical computers (those using bits rather than qubits). Instead, the main focus is computational efficiency, the key battleground where researchers hope to demonstrate a quantum advantage over classical computers.

How to demonstrate one computational model being much faster than another in solving a certain computational problem? It would require some benchmark or quantitative measure for computational efficiency. In practice, it is always easier to see concrete numbers, e.g., in the machine learning community, there are countless benchmarks out on the market for different learning algorithms to compete and compare with each other. However, the situation is drastically different in the realm of quantum computing. Here, we need an alternative evaluation methodology for two reasons. First, as the technology of quantum computing is still in its early stages~\footnote{Or, depending on your perspective, quantum computing could be considered in its childhood or adolescence.}, it is not yet feasible to implement and assess the performance of theoretically interesting quantum algorithms. Second, the ultimate aim in quantum computation is to demonstrate significant computational speedup over \textit{any} classical computer. Thus, evaluation criteria should provide convincing evidence of the infeasibility for classical computers, encompassing not only current but also future algorithms and computing architectures yet to be born.

Given this historical context, the (theoretical) quantum computing community adopted \textit{asymptotic analysis} from computational complexity theory as the main evaluation criteria. In brief, instead of assessing the computation speed on a specific problem instance, theorists consider a family of computational problems with increasing input sizes. 
Computational efficiency is then defined as the time an algorithm spends in relation to the input size - this function is also referred to as the \textit{time complexity} of an algorithm.
Furthermore, the performance of an algorithm is evaluated based on its running time with large input sizes~\footnote{More precisely, theorists consider scenarios where the input size tends towards infinity.}. As such, they focus only on the leading order in the time complexity, grouping those algorithms or models with the same order into the same complexity class.
In summary, quantum computational complexity involves classifying the quantum complexity of various computational problems and identifying those whose quantum complexity is asymptotically smaller than the classical complexity.

\begin{examplebox}{More on asymptotic analysis and complexity classes.}
For algorithms, since the speed of modern computers is so fast, we often can't even perceive the difference between time complexities of $10n$, $100n$, and $n\mathord{+}1000$. Therefore, for computer scientists, the time complexities of $10n$ and $100n$ are essentially the same. They belong to the same order of complexity and are denoted as $O(n)$. The $O(\cdot)$ here is also known as the \textit{Big O notation}, which conceptually tells us to ignore the impact of lower order terms and the different constant factors in front of $n$.

According to the amount of resources used, we can further distinguish different categories of computational complexity. The most common two are the so-called \textit{polynomial complexity} and \textit{exponential complexity}. Conceptually, this is closely related to \textit{Moore's Law}. Moore's Law is an observation made by Gordon Moore, the co-founder of the famous semiconductor company Intel, in 1965. He speculated that the number of transistors that can be accommodated per unit of integrated circuits could double every two years. In layman's terms, Moore predicted that computer performance would roughly double every two years. That is, after $N$ years, the performance of the computer could be $2^{N/2}$ times what it is now. This rate of growth is also known as exponential growth.

If we believe in Moore's Law, then for computational problems with a computational complexity far less than exponential growth, there will always be a day when we have a fast enough computer to easily solve them. Polynomial complexity is one of these types of computational problems. Specifically, if a complexity function $f(n)$ can be proven to have a constant $k\mathord{>}0$ such that $f(n)\mathord{=}O(n^k)$, then we would say its complexity is polynomial. Common complexity scaling such as linear complexity ($O(n)$) and quadratic complexity ($O(n^2)$) both belongs to polynomial complexity. On the other hand, for computational problems with exponential complexity, even if Moore's Law is true, we may still not be able to quickly compute them with computers.

\vspace{3mm}
\begin{center}
\includegraphics[width=13cm]{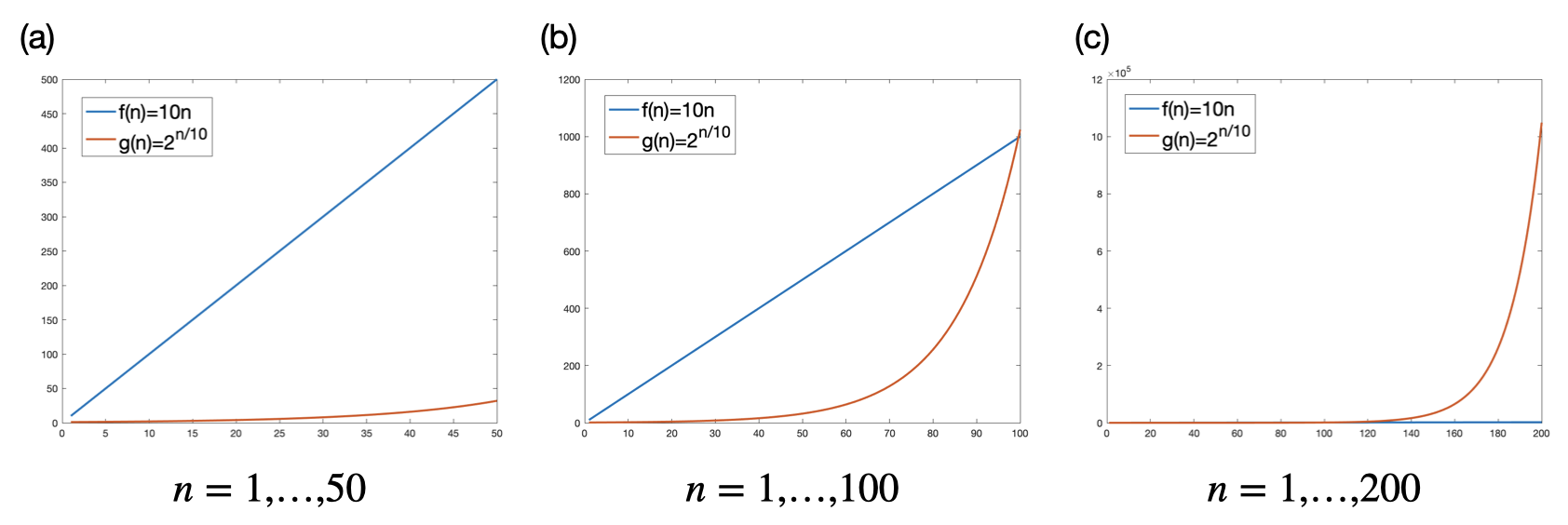}
\captionof{figure}{Asymptotic analysis, polynomial complexity, and exponential complexity. In (a) to (c), the blue line corresponds to a function with polynomial complexity, and the orange line corresponds to a function with exponential complexity. We can see that as n becomes larger, the size of the exponential complexity will very quickly surpass the size of the polynomial complexity.}\label{fig:polynomial exponential}
\end{center}
\end{examplebox}

\section{On the practical side: near-term quantum devices}
Now that we have explored the theoretical blueprint of quantum computing, let's shift our attention to the practical side. In this section, I will start by discussing DiVincenzo's criteria~\cite{divincenzo2000physical} for constructing a quantum computer to orient the reader towards practical considerations. Next, I will introduce a few of the current approaches to implementing quantum computers. Finally, I will address the crucial topics of quantum error correction and fault-tolerant quantum computing.

\subsection{DiVincenzo's criteria}
In 2000, theoretical physicists David  DiVincenzo put out a seminal paper, titled ``The Physical Implementation of Quantum Computation''~\cite{divincenzo2000physical}, in which he delineated seven conditions for constructing a quantum computer. The first five are necessary for quantum computation:

\begin{enumerate}
\item A scalable physical system with well-characterized qubit.
\item The ability to initialize the state of the qubits to a simple fiducial state, such as $\ket{000\cdots}$.
\item Long relevant decoherence times.
\item A ``universal'' set of quantum gates.
\item A qubit-specific measurement capability.
\end{enumerate}
The remaining two are necessary for quantum communication:
\begin{enumerate}
\item The ability to interconvert stationary and flying qubits.
\item The ability to faithfully transmit flying qubits between specified locations.
\end{enumerate}

In the interest of brevity, I will not delve into these criteria here. However, I encourage interested readers to take a look at DiVincenzo's paper, which is only nine pages long. Here, I merely wish to provide two comments as food for thought.

First, note that DiVincenzo's criteria are relatively high-level, allowing for the possibility of multiple types of implementation. The outlined conditions guide practitioners and engineers in tackling the formidable task of quantum computing development step by step. This approach differs somewhat from the usual theoretical strategy, and in my opinion, it is vital for future theorists to propose more of such ``criteria'' to steer the experimental advancement of quantum computing.

Secondly, as DiVincenzo himself noted in his paper, these criteria pertain to \textit{general-purpose} quantum computing - that is, a computing architecture that is universal in terms of its computability. It's entirely plausible that other forms of quantum computing could be more attainable. Therefore, while DiVincenzo's criteria serve as a lighthouse in the voyage of constructing quantum computers, we still need researchers to continuously explore potential alternatives until we reach our destination.

\subsection{Current approaches of implementing quantum computation}
Previously, as mentioned in~\autoref{sec:quantum spin qubit}, the concept of a qubit was inspired by spin in quantum physics. However, practically realizing a qubit isn't as straightforward as directly employing a physical spin (e.g., the spin of an electron). Instead, experimentalists typically construct artificial systems that mimic the fundamental properties of a spin. Broadly speaking, a qubit can be realized as a two-level quantum system, where a high and a low energy level correspond to $\ket{1}$ and $\ket{0}$ respectively. The main task for experimentalists is to ensure that such a quantum system can exhibit quantum phenomena (like superposition), satisfy DiVincenzo's criteria, and demonstrate solving some computational problems.

There are several different practical approaches to implementing quantum computation. In the remainder of this subsection, we will examine three of them: superconducting qubits, trapped ions, and neutral atoms.

\medskip \noindent \textbf{Superconducting qubits.}
Superconductivity is a phenomenon where electrons can travel through a material with no energy cost when the temperature is low enough - this is a quantum effect. The concept of quantum computing with superconducting qubits revolves around the use of superconducting materials as wires to construct an electrical circuit. When the temperature is sufficiently low, the superconducting circuit begins to exhibit quantum phenomena and contains discrete energy levels. By properly spacing the energy levels and associating the lowest and the second lowest level to $\ket{0}$ and $\ket{1}$ respectively, we arrive a candidate qubit construction.

Even though the two most prominent demonstrations of quantum computing to date - by Google~\cite{arute2019quantum} and IBM~\cite{kim2023evidence} - both use superconducting qubits, this approach comes with significant challenges. Firstly, superconductivity requires extremely low temperatures (around 10 mK), which necessitates considerable energy and space for refrigeration, and complicates the scaling up of quantum computers. Secondly, as a superconducting circuit is relatively large, it is more likely to interact with the environment, leading to decoherence. For more details on quantum computing using superconducting qubits, refer to a recent perspective article~\cite{bravyi2022future}.

\medskip \noindent \textbf{Trapped ions.}
Ions are atoms that carry an electric charge and hence they respond to electric fields. A wealth of technology, dating back to the 1950s, has centered on leveraging this fact to trap ions with extraordinary precision. By using these existing techniques to control the location of ions, and identifying the ground state and an excited state of an electron (in the ion) as $\ket{0}$ and $\ket{1}$, a natural candidate of qubit comes up.

Given that the necessary technology is already established, trapped ion quantum computing is perhaps the most widespread implementation, both in industry (e.g., IonQ, Quantinuum) and academia (e.g., University of Waterloo, University of Maryland). However, the primary challenge in this approach remains scalability. For more details, refer to a recent survey~\cite{bruzewicz2019trapped}.

\medskip \noindent \textbf{Neutral atoms.}
With the aid of optical tweezers, it is possible to trap neutral atoms~—~atoms without a charge—with high precision. This paves the way for another method of implementing qubits, using the electronic states of certain neutral atoms, such as Rydberg atoms. Qubits constructed in this way can be stably controlled via laser techniques, and the unique Rydberg state allows for atom-to-atom interaction, enabling the construction of two-qubit gates.

While neutral-atom quantum devices face similar hurdles as the other methods we've discussed, this approach has recently seen significant advancements, with companies such as QuEra leading the charge~\cite{quera}. For more details on neutral-atom quantum computing, refer to a recent review~\cite{morgado2021quantum}.

\subsection{Quantum error correction and fault tolerant quantum computing}
At we saw in the previous subsection, there are many quantum computer implementations based on different physical/engineering technologies, each with its own advantages. However, if compared to the historical development of classical computers, quantum computers are probably still in the stage of vacuum tube computers of the 40s to 50s, or even earlier. Although many laboratories can now realize dozens, hundreds, or even over a thousand quantum bits, these quantum bits are all so-called \textit{physical qubits}, which are affected by noise from experimental instruments or unknown sources, so they may exhibit unexpected behaviors. The algorithms we saw in~\autoref{sec:quantum prelim algorithm} are all based on so-called \textit{logical qubits}, which are perfect and flawless qubits that can perform various mathematical operations freely. Consequently, it's a pressing challenge for both theorists and experimentalists to figure out how to efficiently convert physical qubits into logical qubits, enabling the systematic scaling up of quantum computing devices.

The \textit{quantum fault-tolerance theorem}, also known as the \textit{threshold theorem}, presents a significant breakthrough in tackling this issue. Several groups of theorists~\cite{aharonov1997fault,knill1998resilient,kitaev2003fault} independently found that once the noise of physical qubits is reduced to a certain level, quantum error correction technology can systematically transform physical qubits into logical qubits where the transformation only requires a manageable increase in the number of qubits. The quantum fault-tolerant theorem provides considerable reassurance to the community. However, it's also important to note that determining the precise value of the error threshold and the overhead of the number of qubits requires further research. As it stands, we are still a long way from reaching this boundary~\cite{campbell2017roads}.

\section{Concluding remarks}
The discovery of the quantum world stems from physicists' endless curiosity about microscopic phenomena. The subsequent physical models that emerged have brought about rich mathematical structures, and have given people anticipations for new computational models. Looking back in fifty years, will we marvel at the prophetic insights of these theoretical algorithms, or will we have a completely different understanding of quantum computing?
\begin{savequote}[75mm]
Quantum computation is as much about testing Quantum Physics as it is about building powerful computers.
\qauthor{Umesh Vazirani}
\end{savequote}

\chapter{Quantum Computational Advantage}\label{ch:example RCS}

\emph{Quantum computational advantage} refers to the experimental demonstration of the computational power of a quantum device far beyond that of any existing classical devices.
Such demonstration is important because it not only constitutes a milestone of quantum technology, but also challenges the so called \emph{extended Church-Turing thesis}~\cite{arora2009computational,aaronson2011computational}, which has been central to computational complexity theory. A straightforward way to demonstrate quantum advantage would be to explicitly run a quantum algorithm, such as the Shor's integer factoring~\cite{shor1994algorithms}, for problems whose size is too large (e.g. $2048$-bit integers) to be solved by any known algorithm running on classical computers.
However, this would require a quantum device with a large number of near-perfect qubits, which is well beyond the capabilities of the existing technology. State-of-the-art quantum devices consist of several dozens of \emph{imperfect} qubits~\cite{zhang2017observation,arute2019quantum,USTC,zhu2021quantum,ebadi2020quantum,Scholl.2021}.
Even the exploration of a potential scaling advantage requires larger systems, consisting of at least several hundred coherent qubits.

\section{Quantum computational advantage proposals based on RCS}

Instead of implementing such quantum algorithms, most of the current efforts towards demonstrating quantum advantage have focused on \emph{sampling problems}~\cite{preskill2012quantum,harrow2017quantum,lund2017quantum}, which are well suited for near-term quantum devices~\cite{arute2019quantum,USTC,zhu2021quantum,zhong2020quantum,zhong2021phase}.
In these problems, one is asked to produce a sequence of random bitstrings drawn from a certain probability distribution. A natural choice of a distribution that would be challenging for a classical computer to reproduce is one based on a highly entangled many-body wavefunction.
Indeed, it has been shown~\cite{bremner2011classical,aaronson2011computational,bremner2016average,farhi2016quantum,boixo2018characterizing,bremner2017achieving,gao2017quantum,bermejo2017architectures,terhal2002adaptive} that, for a wide class of quantum states, exact sampling by classical computers  is intractable under plausible assumptions~\cite{terhal2002adaptive,bremner2011classical,aaronson2011computational,aaronson2016complexity,bremner2016average,aaronson2016complexity,bouland2019complexity,movassagh2018efficient,bouland2021noise,kondo2021fine}.

To demonstrate quantum advantage using an actual sampling experiment, one needs to introduce a \emph{benchmark} that measures how close the sampled distribution $q(x)$ of a quantum device is to the (ideal) target distribution $p(x)$.
The idea is that on one hand, one shows that the samples from the quantum device achieve high values (indicating good correlation with the ideal distribution), while on the other hand, one presents evidence that there \emph{does not exist} an efficient classical algorithm that can produce samples achieving comparable values.
If the difference between the classical and quantum resources needed to achieve a certain value of the benchmark scales \emph{exponentially} with the system size, this demonstrates that quantum devices have an exponential computational advantage even in the regime where the gates are too noisy to allow for quantum error correction. Quantum computational advantage proposals of this kind is known as the \emph{random circuit sampling (RCS)} approach.

\subsection{Benchmarks for demonstrating computational advantage via RCS}
Now, let us be more mathematically precise on the setup of RCS-based quantum advantage. Let $\Circ$ be an $n$-qubit random quantum circuit and $U_\Circ$ be the underlying unitary. In RCS, we are interested in the output distribution $p_\Circ(x)$ induced by the quantum state $U_\Circ\ket{0^n}$. In particular, there are two common quantitative measures: fidelity and cross-entropy benchmark (XEB). The former captures the distance between two quantum states while the latter captures the distance between two classical distributions.

\begin{definition}[Fidelity]\label{def:fidelity}
Let $\Circ$ be a quantum circuit and $\rho$ be a (mixed) quantum state. The fidelity between the ideal state $U_C\ket{0^n}$ and $\rho$ is defined as follows.
\begin{equation}\label{eq:fidelity}
F_\Circ(\rho) := \tr(U_\Circ\ketbra{0^n}U_\Circ^\dagger~\rho) \, .
\end{equation}    
\end{definition}

While fidelity is a fundamental measure in quantum information, it is inherently ``of many-body quantum nature''. Thus, it is in principle computationally expansive to estimate the fidelity between two quantum states. Consequently, in practice people came up with proxies for fidelity that are (i) empirically estimable and (ii) revealing the closeness of two quantum states. The measure that has been widely adopted in RCS-based quantum advantage experiments is the cross entropy benchmark (XEB) defined as follows.

\begin{definition}[Log XEB]\label{def:log XEB}
Let $\Circ$ be a quantum circuit and $q$ be a probability distribution over $\{0,1\}^n$. Denote $p_\Circ$ as the ideal distribution of $\Circ$ and is defined as $p_\Circ(x)=|\bra{x}U_\Circ\ket{0^n}|^2$. The logarithmic cross-entropy (log XEB) between $p_\Circ$ and $q$ is defined as follows.
\begin{equation}\label{eq:log XEB}
\chi_\Circ^\text{logarithmic}(q) := \sum_{x\in\{0,1\}^n} p_\Circ(x)\log\frac{p_\Circ(x)}{q(x)}
\end{equation}
\end{definition}

For statistical efficiency, in practice people use the linearized version of cross-entropy defined as follows.
\begin{definition}[Linear XEB]\label{def:linear XEB}
Let $\Circ$ be a quantum circuit and $q$ be a probability distribution over $\{0,1\}^n$. Denote $p_\Circ$ as the ideal distribution of $\Circ$ and is defined as $p_\Circ(x)=|\bra{x}U_\Circ\ket{0^n}|^2$. The linear cross-entropy (linear XEB) between $p_\Circ$ and $q$ is defined as follows.
\begin{equation}\label{eq:linear XEB}
\chi_\Circ^\text{linear}(q) := \left(2^n\sum_{x\in\{0,1\}^n}p_\Circ(x)q(x)\right)-1
\end{equation}
\end{definition}

As we will only discuss linear XEB in the rest of the thesis, from now on we denote linear XEB as $\chi_\Circ$ for simplicity. See Rinott et al.~\cite{rinott2020statistical} and the supplementary materials of Arute et al.~\cite{arute2019quantum} for more discussions on why empirically people prefer linear XEB (\autoref{eq:linear XEB}) over the logarithmic version (\autoref{eq:log XEB}).

\medskip \noindent \textbf{On the completeness side}, i.e., when an experiment perfectly simulates the circuit $\Circ$, we have $\rho=U_\Circ\ketbra{0^n}U_\Circ^\dagger$ and $q(x)=|\bra{0^n}U_\Circ^\dagger\ket{x}|^2=p_\Circ(x)$. Thus,
\[
F_\Circ(\rho) = \tr(U_\Circ\ketbra{0^n}U_\Circ^\dagger~\rho)  = \tr(U_\Circ\ketbra{0^n}U_\Circ^\dagger U_\Circ\ketbra{0^n}U_\Circ^\dagger) = 1 \, .
\]
On the other hand, it is well-known that when a random quantum circuit is scrambling enough, the ideal distribution $p_\Circ$ follows the Porter-Thomas distribution. Hence, for each fixed $x\in\{0,1\}^n$ and every $a>0$, we have $\Pr_{\Circ}[p_\Circ(x)=a]\propto e^{-Da}$ where the randomness is over the choice of random circuits $\Circ$ and $D>0$ is some global normalizing factor to ensure $\sum_xp_\Circ(x)=1$. With this, one can derive that 
\[
\chi_\Circ(q) = \left(2^n\sum_{x\in\{0,1\}^n}p_\Circ(x)q(x)\right)-1 = \left(2^n\sum_{x\in\{0,1\}^n}p_\Circ(x)^2\right)-1\approx1
\]
with high probability over the choice of random circuits $\Circ$.

\medskip \noindent \textbf{On the soundness side}, i.e., if the quantum simulation is not perfect, we would expect both the fidelity and XEB would reflect such discrepancy. Indeed, when the quantum simulation is completely random, i.e., $\rho=\frac{1}{2^n}\sum_{x\in\{0,1\}^n}\ketbra{x}$ being the maximally mixed state, and $q(x)=\frac{1}{2^n}$ being the uniform distribution, we have
\[
\Exp\left[F_\Circ(\rho) = \tr(U_\Circ\ketbra{0^n}U_\Circ^\dagger~\rho)\right] = \frac{1}{2^n}\sum_{x\in\{0,1\}^n}|\bra{0^n}U_\Circ^\dagger\ket{x}|^2 = \frac{1}{2^n}\sum_{x\in\{0,1\}^n}p_\Circ(x) = \frac{1}{2^n}
\]
and
\[
\chi_\Circ(q) = \left(2^n\sum_{x\in\{0,1\}^n}p_\Circ(x)q(x)\right)-1 = \left(\sum_{x\in\{0,1\}^n}p_\Circ(x)\right)-1 = 0 \, .
\]

To sum up, we now know that both fidelity and XEB are measures taking value (roughly) within $[0,1]$ such that a perfect simulation would get a \textit{score} close to $1$ and a completely random simulation would get a score close to $0$. The whole business of the RCS-based quantum computational advantage using XEB is then centering around the following two conjectures.

\begin{conjecture}[Informal]\label{conj: XEB quantum}
For a quantum simulation with high fidelity, it also has high XEB value.
\end{conjecture}

\begin{conjecture}[Informal]\label{conj: XEB classical}
Every classical algorithm requires a substantial amount of time to achieve high XEB value.
\end{conjecture}

The above two conjectures are mathematically ill-defined and indeed that's the situation prior to our work~\cite{GKCLBC21}. In the rest of this chapter, we are going to give a quantitative treatment on the above two conjectures and clarify the landscape.

\subsection{Prior and concurrent works}
The XEB measure has been used in recent  experiments~\cite{arute2019quantum,USTC}, where sampling from random unitary circuits was performed.
Specifically, Google~\cite{arute2019quantum}  achieved an  XEB value of $\chi_p \approx 0.002$ on a two-dimensional, 53-qubit quantum device (Sycamore) implementing circuits up to depth 20 and estimating XEB under reasonable assumptions~\footnote{Since directly calculating the XEB value of 53-qubit Sycamore at depth 20 is computationally intractable, Google extrapolated their smaller-system results to estimate the final XEB value.}.
{Recently, the USTC group~\cite{USTC,zhu2021quantum} extended the number of qubits and claimed the XEB value of $6.62\times 10^{-4}$ and $3.66\times 10^{-4}$, for system sizes up to 56 qubits and 60 qubits, respectively.} In both cases, it has been  conjectured that such values are challenging to achieve using state-of-the-art classical computing devices on a realistic time scale.

Prior works challenging  quantum advantage~\cite{gray2020hyper,huang2020classical,pan,fu2021closing,refute1,refute2} obtained comparable or higher XEB values using heavy computational resources. 
While these classical methods are tailored to challenge Google's current setup (53 qubits, depth 20), up to now it was unclear if and how they could be extended to larger systems. In fact, it has been argued that by simply increasing the system size to about 60$\sim$70 qubits, one could defeat such classical spoofing algorithms~\cite{blog}.
Indeed, in  more recent experiments~\cite{zhu2021quantum} 
(60 qubits, depth 24), it has been suggested that the new device bypasses the challenge of these algorithms.

\section{Limitations of XEB as a measure for RCS-based quantum advantage}

In a joint work with Barak and Gao~\cite{BCG21}, and a subsequent work with Gao, Kalinowski, Lukin, Barak, and Choi~\cite{GKCLBC21}, we reexamine the above-mentioned quantum computational advantage proposal based on XEB, adopted by Google and USTC. Through the lens of a classical algorithm for spoofing XEB, we simultaneously examine three distinct perspectives: (i) the experimental regime, (ii) the complexity-theoretic scaling, and (iii) the physical picture. For the sake of clarity and cohesion in the presentation, we defer the introduction of our classical spoofing algorithm to~\autoref{sec:classical spoofing algorithm}. In the remainder of this section, we will provide an overview of our contributions from each perspective. Further details will be elucidated in the subsequent sections and corresponding appendices.

\medskip \noindent \textbf{The experimental regime.}
\emph{A highly efficient classical algorithm (1 GPU around 1s), whose performance is within around one order magnitude} with current experimental devices.
We consider a random circuit ensemble modelled after the one used in Ref.~\cite{arute2019quantum,USTC,zhu2021quantum} (see~\autoref{app:circuit architecture} for detailed information).
Our algorithm achieves a mean XEB value that is about 8\% of Google's experiment (53 qubits, depth 20), and 12\% and 2\% of USTC's experiments (56 qubits, depth 20 and 60 qubits, depth 24) respectively, with the running time  $\approx$1s using 1 GPU (32GB NVIDIA Tesla V100). We can get higher XEB value by taking more running time. E.g., 12.3\% of Google's experiment with 50s and 5\% of USTC's second experiment with 4s.

Remarkably, the XEB value of our algorithm generally \emph{improves} for larger quantum circuits, whereas that of noisy quantum devices quickly deteriorates. Such scaling continues to hold when the number of qubits is increased while the depth of the circuit and the error-per-gate are fixed, as explicitly confirmed from numerical simulations for 1D and 2D square and the extended Sycamore architecture in~\autoref{fig:intro_scaling}(b-d).

\medskip \noindent \textbf{The complexity-theoretic scaling.}
\emph{A linear-time classical algorithm that outperforms any noisy 1D quantum circuit.} For one-dimensional quantum circuits consisting of Haar random unitary gates, we present a linear-time classical algorithm which achieves higher XEB values than noisy quantum devices. Concretely, for every uncorrelated error rate $\epsilon>0$ per gate, our algorithm can spoof the XEB measure when the number of qubits is sufficiently large. Here, uncorrelated error refers to errors from different locations being uncorrelated, i.e., the error channel is a tensor product of error channels of each location.

For general circuit architecture, Aaronson and Gunn~\cite{aaronson2019classical} reduced the classical hardness of spoofing the XEB measure to the Linear Cross-Entropy Quantum Threshold Assumption (XQUATH), which is a stronger version of the Quantum Threshold Assumption (QUATH)~\cite{aaronson2016complexity}.Our results refute XQUATH assuming the single qubit gates are Haar random.

\medskip \noindent \textbf{The physic picture.}
We present a way to analyze quantum circuit dynamics using classical statistical physics. Specifically, for a wide class of random circuit ensembles involving single qubit Haar random gates, we show that the dynamics of both noisy quantum circuits and our classical algorithm can be understood in terms of an effective diffusion-reaction process which was originally used to study the scrambling of circuits~\cite{mi2021information}. In this effective description, the application of each layer of a quantum circuit translates to particles undergoing a random walk (diffusion) for a single time step on a graph representing the  circuit architecture. Furthermore, each particle can duplicate itself, and a pair of particles may recombine into a single particle at a certain rate (reaction). The rates of particle diffusion and reaction are determined by the properties of two-qubit quantum gates, such as the average amount of entanglement they generate. The XEB and the fidelity of ideal circuits are given by different aspects of particle distribution at the last circuit layer, as we elaborate in~\autoref{app:stat_DR_dynamics}.

The XEB value in a noisy circuit and our algorithm will decrease from the ideal value when a particle hits a defective (omitted or noisy) gate.
In the case of noisy quantum circuits, every gate is noisy, so the decrease in the XEB value is proportional to the total number of particles in the diffusion-reaction process. Intuitively, when the system size grows, there are more particles hitting noisy gates and thus the XEB value becomes smaller. In our algorithm, the XEB decreases whenever a particle hits an omitted gate at the boundaries of disconnected sub-regions. Intuitively, when the system size grows, there is more space for particles to diffuse away from the boundary and thus, in general, the XEB value can become larger. This qualitatively explains the asymptotic scaling of XEB in~\autoref{fig:intro_scaling}.

\subsection{Summary of our results}

\vspace{2mm}
\noindent\textbf{Numerical result 3.1:}
Our algorithm exhibits favorable scaling behavior in increasing system sizes (while fixing the strength of noise and circuit depth) for both conventional gatesets and Google's gatesets. 

\vspace{2mm}
\noindent\textbf{Numerical result 3.2:} 
Our algorithm outperforms the experiments of Google and USTC for conventional random unitary circuits, that is, the single qubit gate is Haar random.

\vspace{2mm}
\noindent\textbf{Numerical result 3.3:}
Our algorithms achieve XEB value within around one order of magnitude of the experiments of Google and USTC for a slight modification of their gatesets. 

\vspace{2mm}
\noindent\textbf{Rigorous result 3.1:}
The average XEB is additive and the average fidelity is multiplicative between our algorithm and distribution from ideal circuit. The assumption is that the average is over unitary 1-design random gates.

\vspace{2mm}
\noindent\textbf{Rigorous result 3.2:}
Our algorithm achieves XEB value $2^{-O(d)}$ in linear time of system size and hence refutes the theoretical guarantee of XEB-based quantum computational advantage, i.e., XQUATH, in sublinear depth. The assumption is that the single qubit gate ensemble is unitary 2-design. No assumption on two qubit gate and circuit architecture.

\vspace{2mm}
\noindent\textbf{Claim:} For 1D circuits with Haar random two qubit gate, the XEB of our algorithm is greater than noisy circuit with arbitrarily small constant noise when $d=\Omega(\log n)$. 

This claim is made by combination of numerical result and analytical formula, and further supported from statistical physics argument.

\subsection{Our classical spoofing algorithms}\label{sec:classical spoofing algorithm}

We now describe an efficient classical algorithm $C$ that, in a wide range of physically relevant situations, produces a probability distribution with XEB values larger or comparable to that of an $\epsilon$-noisy circuit, at least on average. In such situations, the existence of our algorithm suggests XEB on its own is not a good benchmark for certifying quantum advantage.

\vspace{2mm}
\noindent\textbf{The underlying intuition of our algorithms.} 
Our algorithm is inspired by the observation that entanglement growth in a noisy quantum circuit is reduced by errors spread over the entire circuit in both space and time [\autoref{fig:intro_scaling}(a)]. These  effectively truncate entanglement and correlations among different subsystems. In our algorithm, we introduce similar amount of effective errors, but they occur only at specific locations such that the quantum circuit becomes easier to simulate. As an example,~\autoref{fig:intro_scaling}(a) shows how omitting a few specific gates at certain locations (which amounts to particular types of error, i.e. \emph{gate defects}) can split a circuit into multiple disconnected sub-circuits. Alternatively, one can apply completely depolarizing channel before and after an entangling gates. These approaches explicitly remove correlations between subsystems. Intuitively, when the amount of ``effective noise'' in a noisy quantum simulation is comparable to the ``effective error'' in our algorithm (proportional to the number of omitted gates), the XEB of the latter is larger due to the stronger correlation among errors [see~\autoref{fig:oneerror}(b)].

Since the size of each sub-circuit is much smaller than that of the original circuit, the algorithm can be significantly faster than a direct simulation of the global circuit. For example, when ran on 53-qubit circuits, such as Google's, it takes a few seconds using a single GPU (32GB NVIDIA Tesla V100).

\vspace{2mm}
\noindent\textbf{The skeleton of our algorithms.} 
We now describe our classical algorithm. For concreteness, we illustrate our algorithm using 1D quantum circuits although it is straightforward to generalize it to other circuit architectures. Let $N$ be the total number of qubits, $d$ be the depth, and $ l$ be
the maximum size of subsystems [see~\autoref{fig:intro_alg}(a) for an example with $N=12$, $d=7$, and $ l=4$]. We start by partitioning the $N$ qubits into subsystems of size at most $ l$ by
omitting any gates acting across two different subsystems [see~\autoref{fig:intro_alg}(c)].
We then simulate each subsystem separately. 
Using brute-force methods, simulating a subsystem of $ l$ qubits takes at most $2^{O(l)}d$ time. There are $\lceil N/ l\rceil$ subsystems and hence the total running time of our algorithm is at most $\frac{2^{O( l)}}{ l} Nd.$ In particular, if $ l$ is fixed and does not scale with the total system size $N$ or depth $d$, the time complexity is linear in the circuit size $Nd$.
We claim that the bitstring  distribution induced by the factorizable wavefunction obtained from our algorithm achieves relatively high XEB values.
We also improve our basic algorithm using several statistical tricks as elucidated in~\autoref{app:algorithm improve}.

\begin{figure}[h]
\includegraphics[width=7cm]{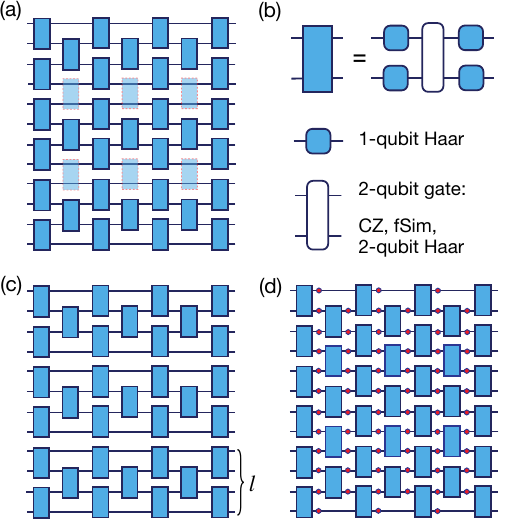}
\centering
\caption{Illustration of our algorithms.
(a) The target (ideal) circuit to simulate. The light blue gates correspond to the ones omitted in (c).
(b) Each random two-qubit gate in our circuit consists of any (potentially fixed) two-qubit gate surrounded by 4 single-qubit Haar random gates. When compared to experimental data, the single-qubit random gates are chosen to be a slight modification of those used in Ref.~\cite{arute2019quantum,USTC,zhu2021quantum},
(c) Our algorithm: one can approximately simulate the ideal circuit by simply omitting a certain subset of gates (in light blue color with red dashed boxes) in the ideal circuit (a). Then, the circuit separates into isolated subsystems. We denote the maximal size of a subsystem as $ l$.
(d) Noisy circuit: we model the dynamics of noisy quantum circuits by applying probabilistic single-qubit noise (e.g. depolarizing or amplitude damping) channels  to all qubits, after each layer of unitary evolution.
}
\label{fig:intro_alg}
\end{figure}

\subsection{Experimental regime}\label{sec:quantum experimental regime}
For the experimental regime, we consider quantum circuit architectures that are of experimentally relevance. In particular, we consider 2D quantum circuits in the Sycamore and Zuchongzhi architectures in two different settings. First, we focus on the role of the two-qubit gate, and we analyze the performance of our algorithm for three different two-qubit gate ensembles: Haar, CZ, and fSim. For the single-qubit gate we choose either independent Haar-random gates which allows for efficient analysis using the diffusion-reaction model or the more experimentally-relevant discrete gate set. Second, we compare our algorithm against the experimental results of Refs.~\cite{arute2019quantum,USTC,zhu2021quantum}. There, we focus on the fSim gate, and we assume the experimentally relevant discrete single-qubit gate set.
These analyses lead to two main results, summarized in~\autoref{fig:intro_advantage_regime} and~\autoref{tab:running_time}.
For numerical calculations, we used a single GPU machine (32GB NVIDIA Tesla V100). 

\begin{numericalresult}{(Favorable scaling in increasing system sizes)}\label{res:numerical 1}
In the Sycamore architecture with Haar-random single-qubit gates or the experimentally relevant gate set (fSim + discrete single-qubit gates), 
Our algorithm has the following properties:
\begin{itemize} 
\item for 1D circuits with the Haar random two‐qubit gate ensemble, the algorithm achieves higher average XEB value than quantum simulation with noise level $\epsilon=1\%$ as the system size grows (with fixed depth to $16$). Results summarized in~\autoref{fig:intro_scaling}(b);
\item for 2D circuits with the Haar random two‐qubit gate ensemble, the algorithm achieves higher average XEB value than quantum simulation with noise level $\epsilon=2\%$ or $2\%$ as the system size grows (with fixed depth to $16$). Results summarized in~\autoref{fig:intro_scaling}(c);
\item for the extended version of google's Sycamore circuits, the algorithm achieves higher average XEB value than quantum simulation as the system size grows by extrapolation. Results summarized in~\autoref{fig:intro_scaling}(d);
\end{itemize}
\end{numericalresult}

\begin{figure}[h]
\includegraphics[width=14cm]{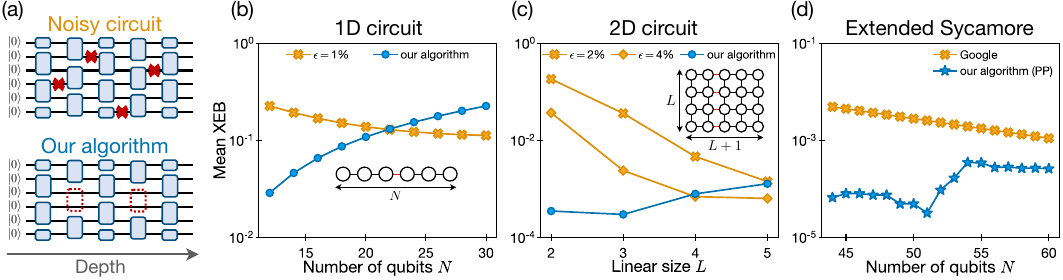}
\centering
\caption{
Classical algorithms spoofing XEB for quantum circuits in various  architectures.
(a) Schematic diagrams illustrating the key idea of our algorithm.
In noisy quantum circuits, errors (red crosses) randomly occur at a rate $\epsilon>0$, spread over the entire circuit.
In our algorithm, we introduce effective, highly localized errors by omitting or modifying a few entangling quantum gates (red dotted boxes) such that the circuit splits into smaller segments and becomes easier to simulate classically.
(b-d) Performance of our algorithm. We obtain high XEB values (blue circles and stars)  compared to noisy circuits (yellow crosses and diamonds) for 1D, 2D, and the extended Sycamore circuit architectures [see~\autoref{fig:intro_partition}].
(b) 1D circuits of depth $d=16$ in the brick-work layout, with the Haar random two-qubit gate ensemble.
(c) 2D circuits of depth $d=16$ in a $L\times (L+1)$ square lattice, with the Haar random two-qubit gate ensemble.
Our algorithm outperforms noisy quantum circuits (here with error rates $\epsilon=0.02$,  $0.04$) for sufficiently large system sizes.
Insets in (b-c) show the circuit architecture and the position of omitted gates (red lines).
(d) Comparison of the mean XEB value obtained by our improved algorithm (light blue circles) to Google's Sycamore in which case we extrapolated experimental results using the ansatz ${\rm XEB}\sim \exp(-c_1 N-c_2 Nd )$. We extended the Sycamore architecture horizontally up to 60 qubits; see~\autoref{fig:intro_partition} for more details.
For this simulation, we assumed a quantum circuit ensemble with random single-qubit gates similar to (but slightly modified) those used in Ref.~\cite{arute2019quantum,USTC,zhu2021quantum}.
}
\label{fig:intro_scaling}
\end{figure}

\begin{numericalresult}{(Different gate ensembles)}\label{res:numerical 2}
In the Sycamore architecture with $N=53$, $d=20$ with Haar-random single-qubit gates, our algorithm (using the partition in~\autoref{fig:intro_partition}) has the following properties:
\begin{itemize} 
\item the algorithm achieves significant average XEB value for all depths shown in~\autoref{fig:intro_advantage_regime}. As a reference, the expected XEB value of a noisy quantum device with depth 20 and error rate $\epsilon\approx0.5\%$ is $\approx0.002$;
\item the choice of the two-qubit gate affects the value of XEB, which can be understood in terms of the diffusion-reaction model~\autoref{sec:stat_DR_mapping};
\item the discrete single-qubit ensemble results in much lower XEB values (green crosses in~\autoref{fig:intro_advantage_regime}), which is caused by the faster scrambling time;

\item the running time (computing the vector of output probabilities) is only 4-8 seconds;
\end{itemize}
\end{numericalresult}

\begin{figure}
\includegraphics[width=10cm]{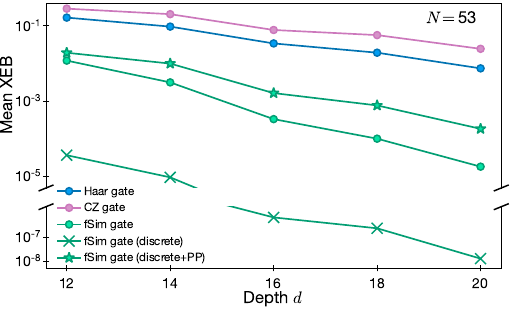}
\centering
\caption{Mean XEB obtained by our algorithm for different two-qubit gate ensembles, on Google's circuit geometry. Circles denote the Haar single-qubit gate set, while the green crosses (stars) correspond to the more experimentally relevant discrete set (with amplification using the \emph{top-$k$} method).}
\label{fig:intro_advantage_regime}
\end{figure}

\begin{table*}
\begin{center}
\begin{tabular}{|p{3.5cm}||c|c|c|}
\hline
   & Google~\cite{arute2019quantum} & USTC-1~\cite{USTC} & USTC-2~\cite{zhu2021quantum}
\\
\hline
\hline
System size & 53 qubits, 20 depth & 56 qubits, 20 depth & 60 qubits, 24 depth
\\
\hline
Claimed running time on a supercomputer~\cite{zhu2021quantum} & 15.9d & 8.2yr & $4.8\times10^{4}$yr
\\
\hline
Running time on quantum processor & 600s & 1.2h & 4.2h
\\
\hline
Experimental XEB   & $2.24\times10^{-3}$  & $6.62\times10^{-4}$  & $3.66\times10^{-4}$
\\
\hline
\hline
Running time of our algorithm (1 GPU$^{(a,b)}$) & 0.6s & 0.6s &  1.5s
\\
\hline
XEB of our algorithm$^{(b)}$ & $1.85\times10^{-4}$ & $8.18\times10^{-5}$ & $7.75\times10^{-6}$  \\%&
\hline
Ratio of ours to experimental XEB  & $8.26\%$ & $12.4\%$ & $2.12\%$\\ 
\hline
\hline
Running time of our algorithm (a different partition) & 50s &  &  4s
\\
\hline
XEB of our algorithm (a different partition) & $2.7\times10^{-4}$ &  & $2.05\times10^{-5}$ \\%&
\hline
Ratio of ours to experimental XEB (a different partition)  & $12.3\%$ &  & $5.6\%$\\
\hline
\end{tabular}
\end{center}
\caption{
The comparison of XEB values (using the top-$k$ post-processing) and running times in the quantum advantage regime.
We find that the average XEB values from our algorithm is largely independent of the choice $k\lesssim 10^4$ (corresponding to more than $k^2\sim 10^8$ distinct bitstrings for two subsystems), above which they slowly decrease. See~\autoref{sapp:improved} for the $k$-dependence as well as the estimated STD of XEB values.
(a) The running time is measured on a device using 1 GPU (NVIDIA Tesla V100).
 (b) The performance of our algorithm (XEB value and running time) listed here are measured for the partitions in~\autoref{sapp:improved} which are not optimized and are chosen for 1 GPU simulation with bounded memory (32GB for our device).
 The tensor network algorithm is based on Ref.~\cite{kalachev2021recursive} and implemented by a Julia package OMEinsum.jl~\cite{julia2}.
 }
\label{tab:running_time}
\end{table*}

\begin{numericalresult}{(Comparison with experimental results)}\label{res:comparisons}
For the experimentally relevant gate set (fSim + discrete single-qubit gates) the performance of our algorithm can be summarized (see also~\autoref{tab:running_time}) as follows
\begin{itemize} 
\item using the top-$k$ post-processing method, the algorithm achieves average XEB values within around one order magnitude ($\approx 2\%\sim 12\%$) to recent experiments up to depth 20 and 24, respectively.;
\item the running times (computing the vector of output probabilities and choosing the top-$k$ bitstrings) are on the order of one second.
\item the STD is conjectured to be comparable to the mean value for large enough $k$ but without decreasing XEB too much; this is supported numerically for Google's Sycamore architecture [see~\autoref{sapp:improved}]. 
\end{itemize}
\end{numericalresult}
In summary, our numerical simulations show that our algorithm achieves XEB values within around one order magnitude to Google's and USTC's circuits in the quantum advantage regime with the experimentally-relevant gate set.
While our basic algorithm is simple and efficient, there are ways to achieve higher XEB values by adding more sophisticated algorithmic ingredients. For example, we show that after adding a simple post-processing step (the top-$k$ method), our algorithm can achieve much higher XEB values; e.g., compare green crosses and stars in ~\autoref{fig:intro_advantage_regime}. In fact, we only considered here the most straightforward way to determine the locations of omitted gates (or maximal depolarization noise), which may not be optimal. As we see from Table \ref{tab:running_time}, different partitions with roughly the same number of qubits in each part has different XEB values and running time. By generalizing our method, e.g., making the locations of omitted gates (maximal depolarization noise) time/depth dependent, we expect an improved version of our algorithm may produce higher XEB without substantially increasing the computational resources. In this work, we mainly focus on using 1 GPU, which limits the possibilities of partitions. It is interesting to explore Using multiple GPUs with better XEB values.
In addition, it is an interesting future direction to explore further algorithmic improvements (e.g., adding a modest amount of entanglement).

\subsection{Complexity-theoretical scaling}\label{sec:quantum complexity results}
XEB and fidelity exhibit different scaling behaviors when a system size is increased with a fixed error rate, implying that two quantities cannot agree in a certain scaling limit. While a rigorous analysis can be made using the framework presented in~\autoref{sec:stat_DR_mapping}, here we consider a toy model illustrating the origin of the different scaling  behaviors.

\begin{rigorousresult}[Scaling of XEB and fidelity]
Consider $k$ disjoint $N$-qubit systems, each undergoing noisy circuit evolution with corresponding XEB values $\chi_i = 2^N \sum_x p_i(x) q_i(x) - 1$ and fidelities $F_i$ with $i=1,2,\dots, k$.
If we consider the $k$ disjoint systems as a single composite system of $kN$ qubits, then the fidelity scales multiplicatively, i.e. $F_\textrm{total} = \prod_i F_i$, while the XEB scales additively: $\chi_\textrm{total}\approx\sum_i \chi_i$.
\end{rigorousresult}

The multiplicity of fidelity can be directly checked from its definition in~\autoref{eq:fidelity}. For the additivity of XEB, we have
\begin{align}
    \chi_\textrm{total} &= 2^{kN} \sum_{\{x_i\}} \prod_i p_i(x_i) q_i (x_i) - 1\\
   \nonumber &= \prod_i\left(2^N\sum_{x_i}p_i(x_i)q_i(x_i)\right)-1\\
    &=\prod_i (\chi_i + 1) - 1 \approx \sum_i \chi_i,
\label{eqn:additive_nature}
\end{align}
where the second equality is due to the product structure across the subsystems and we assumed that $\chi_i \ll 1$ in the last line, relevant for the regime of our interest. While this example may seem contrived as each subsystem is perfectly isolated, one can also devise an example, where all subsystems are strongly coupled by unitary gates and result in fully globally scrambled quantum states. 

This discrepancy in scaling stems fundamentally from the structure of the XEB formula in~\autoref{eq:linear XEB}: as two distributions $p(x)$ and $q(x)$ become uncorrelated~\footnote{Mathematically, we say $p(x)$ and $q(x)$ are uncorrelated if  $\Sigma_x p(x)q(x)/2^N = (\Sigma_{x} p(x)/2^N)(\Sigma_{x} q(x)/2^N)$.} from one another, the first term in~\autoref{eq:linear XEB} tends to a finite value, $1$, rather than approaching zero. This offset is explicitly subtracted in order to obtain a value within an interval $[0,1]$, but it also leads to distinct scaling behavior for large composite systems.

Next, we discuss the performance of our algorithm on 1D circuits with gates drawn from the Haar ensemble. For the purpose of this section, $C$ denotes either the algorithm introduced in~\autoref{sec:classical spoofing algorithm} or its self-averaging version described in detail in~\autoref{app:algorithm improve}. The self-averaging version has the same average XEB but a smaller STD, at the cost of requiring more computational power. However, we consider constant subsystem size $l=O(1)$; thus, even the self-averaging algorithm runs in the time linear in $Nd$.

\begin{rigorousresult}{(1D circuits with Haar ensemble)}\label{thm:xeb 1D}
For 1D random quantum circuits with gates drawn from the Haar ensemble and depth at least $d> c\cdot\log N$ for some constant $c>0$, 
\begin{itemize}
\item
for any constant $\epsilon>0$ and large enough $N$ (roughly $N\epsilon>1$), we have
\begin{equation}\label{eq:1D exp}
\Exp_{U}[\chi_U(C)] \ge \Exp_{U}[\chi_U(\mathcal{N}_\epsilon)]
\end{equation}
for both the basic and the self-averaging algorithms.
\item
we conjecture that
\begin{equation}\label{eq:1D var}
\sqrt{\textsf{Var}(\chi_U(C))_U} \approx \Exp_{U}[\chi_U(\mathcal{N}_\epsilon)] \,
\end{equation}
for the self-averaging algorithm (see~\autoref{app:algorithm improve}), which is suggested by numerical simulations. Namely, the standard deviation of $\chi_U(C)$ is comparable to its expectation value $\Exp_{U}\chi_U(C)]$.
\end{itemize}
Combined, this yields a linear-time classical algorithm that spoofs XEB for any noisy quantum simulation of 1D circuits with the Haar gate ensemble, when the number of qubits is large enough.
\end{rigorousresult}

\autoref{eq:1D exp} states that the average XEB of our algorithm is at least as large as that of any noisy circuit with a constant noise level $\epsilon>0$. As mentioned previously, in practice, we would like the conclusion of~\autoref{eq:1D exp} to generalize to typical circuits $U$ (not only on average) --- this can be guaranteed by showing that the variance of the XEB value is small. This notion is expressed in~\autoref{eq:1D var}, which says that the variance is comparable to the expectation value, and hence our algorithm works for typical instances with large probability. Notice that, in the large depth limit, we expect this to hold only for the self-averaging algorithm. When discussing 1D circuits, where the purpose is to provide complexity-theoretic implications, the analysis of the STD concerns only the self-averaging algorithm.

From a technical point of view, our results are derived by showing that the following quantities decay exponentially with the depth of the circuit
\begin{eqnarray}
\nonumber \Exp_{U}[\chi_U(C)]&=&O(e^{-\Delta_1 d}), \\
\nonumber  \Exp_{U}[\chi_U(\mathcal{N}_\epsilon)]|_{\epsilon\rightarrow0\text{ while } N\epsilon>1} &=&O(e^{-\Delta_3 d}).
\end{eqnarray}
Additionally, numerical simulations support the scaling of the STD as
\begin{equation*}
   \sqrt{\Exp_{U}[\chi^2_U(C)}{U}-\Exp_{U}[\chi_U(C)]^2]=O(e^{-\Delta_2 d}) 
\end{equation*}
for some constants $\Delta_1,\Delta_2>0$ that depend on the subsystem size $l$ and $\Delta_3>0$ that depends on the noise level $\epsilon$. 

\begin{figure}
\includegraphics[width=10cm]{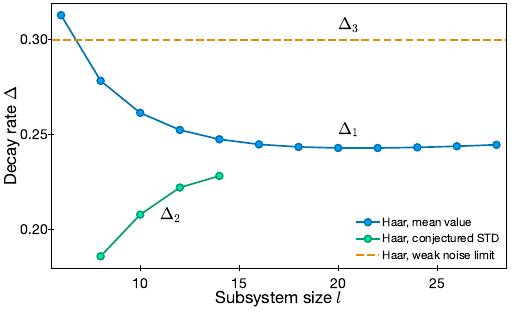}
\centering
\caption{Exponential decay rates in 1D circuits with the Haar gate ensemble. The mean value (blue) and the standard deviation (green) of the XEB obtained by our algorithm. The horizontal (dashed orange) line is the mean XEB value of the noisy circuit in the weak-noise limit. Intuitively, a smaller $\Delta$ corresponds to a larger XEB value. The STD is estimated by an approximate method~\autoref{sapp:Ising_model}, since the direct calculation is not practical. In~\autoref{sapp:Ising_model}, 
 we give a strong numerical evidence  that this approximation is in fact a conservative estimation, i.e., the true STD should be even smaller ($\Delta_2$ should be larger).
}
\label{fig:intro_1D_gaps}
\end{figure}

We emphasize that this scaling is unexpected: the decay rate of the expected XEB value achieved by our algorithm does not depend on the system size but only depends on the depth of the circuit.
We derive $\Delta_1,\Delta_3$ as constants in~\autoref{sapp:DR_stationary}.
Numerically, we show in~\autoref{fig:intro_1D_gaps} an estimate on $\Delta_1,\Delta_2$ and $\Delta_3$ {(of 1D circuits with the Haar gate ensemble) with $\epsilon\to0$ while keeping the system size large enough; i.e., $N\epsilon>1$. For the Haar ensemble,
our numerical results show $\Delta_1<\Delta_3$, where a larger $\Delta$ implies a smaller corresponding quantity in the deep-circuit limit. The numerical calculations suggest that $\Delta_1\approx \Delta_2$: around $l=14$, the gap between the two is very small and $\Delta_2$ (green curve) seems to increase continuously. The green curve is expected to be only a conservative estimation, as explained in~\autoref{sapp:DR_stationary}.

\subsection{The physical picture}

In this section, we assume the single qubit gate is haar random and present an analytic framework to understand the relation between the XEB and the fidelity under various conditions, including different quantum circuit architectures and the presence of noise or omitted gates.
We will find that, in these settings, both the XEB and the fidelity, averaged over an ensemble of unitary circuits, can be efficiently estimated by mapping the quantum dynamics to classical statistical mechanics models, such as the diffusion-reaction model.
This mapping to the diffusion-reaction model was previously developed in Ref.~\cite{mi2021information} for the purpose of studying quantum information scrambling under random circuit dynamics.
Here we use a similar method to study behavior of the XEB and fidelity in random circuits with various entangling gates.
In the special case of 1D circuits, the effective model can be further simplified to a ferromagnetic Ising spin model in two dimensions, allowing us to obtain the scaling behavior analytically. 

\begin{figure}
    \centering
    \includegraphics[width=0.5\linewidth]{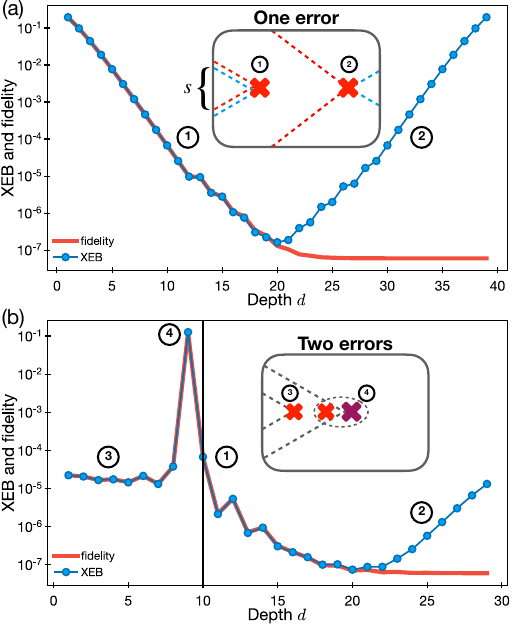}
    \caption{Effects of a single or double error at various locations on the XEB and fidelity. 
    (a) In the presence of a single error, the XEB and fidelity is reduced to an exponentially small but nonzero value that depends on the location of the error. The scaling of the XEB or fidelity can be understood in terms of the size $|s|$ of the error operator propagated to boundaries in the Heisenberg picture (inset). 
    (b) In the presence of two errors, the XEB and fidelity significantly depend on their relative location: the effect of one error can be masked (marked 3) or even cancelled (marked 4) by that of another error.}
    \label{fig:oneerror}
\end{figure}

\vspace{2mm}
\noindent\textbf{Overall methodology.}
We first outline how quantum dynamics can be mapped to a classical statistical mechanics model. The XEB and the fidelity can be written as
\begin{align}
    \label{eqn:chi_dup_hilbert}
    \chi_U+1 &= \sum_{x} \bra{x} U \rho_0 U^\dagger\ket{x} \bra{x} \mathcal{M}_U^{(a)} [\rho_0 ] \ket{x} 2^N,\\
    \label{eqn:f_dup_hilbert}
    F_U &= \sum_{x,x'} \bra{x} U\rho_0U^\dagger\ket{x'} \bra{x'} \mathcal{M}_U^{(a)} [\rho_0 ] \ket{x},
\end{align}
where $\rho_0 = \ket{0^N}\bra{0^N}$ is the initial state of the system, and $\mathcal{M}_U^{(a)}[\cdot]$ is a quantum channel associated with the ideal unitary evolution ($a$=ideal), noisy quantum dynamics ($a$=noisy), or our classical algorithm with omitted gates ($a$=algo).
For a different choice of $a=\{\textrm{ideal}, \textrm{noisy},\textrm{algo}\}$,~\autoref{eqn:chi_dup_hilbert} and~\autoref{eqn:f_dup_hilbert} become the XEB and the fidelity of the corresponding case, respectively. The sum over $x,x'$ represents the summation over all possible $N$-qubit configurations (bitstrings). 

The key idea is to realize that both the XEB and the fidelity can be expressed as the expectation values of observables in an extended Hilbert space. 
More explicitly, we envision having two identical copies of the  Hilbert space: one representing the ideal circuit dynamics, and the other representing the dynamics in either the ideal circuit, noisy circuit, or our algorithm [see~\autoref{fig:stat_mapping_outline}(a)].
Then, we have 
\begin{align}
    \chi_U+1 &= \Tr{ \mathcal{B}_\textrm{XEB} \left( U\rho_0U^\dagger \otimes\mathcal{M}_U^{(a)}[\rho_0] \right) },\label{eq:chi_Bchi}\\ %\mathcal{M}_a[\rho_0] \right) }\\
    F_U &= \Tr{ \mathcal{B}_F \left( U\rho_0U^\dagger \otimes \mathcal{M}_U^{(a)}[\rho_0] \right) },\label{eq:f_Bf}
    %\mathcal{M}_a[\rho_0] \right) },
\end{align}
where $\mathcal{B}_\textrm{XEB} = 2^N \sum_x \ket{x}\bra{x}\otimes \ket{x}\bra{x} $ and $\mathcal{B}_F = \sum_{x,x'} \ket{x}\bra{x'}\otimes \ket{x'}\bra{x} $ are Hermitian observables defined in the enlarged space. 
In the following, we simply use $\mathcal{B}_b$ with $b\in\{{\rm XEB},F\}$.

A convenient way to study the type of operators in Eqs.~\eqref{eq:chi_Bchi}-\eqref{eq:f_Bf} is to represent them as tensor networks whose contraction results in $\chi_U+1$ or $F_U$,  as shown in~\autoref{fig:stat_mapping_outline}(a,b).
In general, the contraction of these tensor network diagrams for any given $U$ would be computationally difficult, as it is equivalent to evaluating the corresponding quantum circuit.
However, we are mostly interested in the average-case behavior of a class of random quantum circuits with gates drawn from specific gate ensembles. In this case, we can perform the averaging over the gate ensemble before contracting the network. Crucially, we find that the averaging process allows us to re-express the tensor network as a summation over exponentially many simple diagrams enumerated by different configurations of classical variables $s$ [see~\autoref{fig:stat_mapping_outline}(b)].

This emergent mathematical structure---namely the summation over all possible configurations of classical variables---is similar to the path integral formulation of a classical Markov process, or a partition function in statistical mechanics models~\cite{huang1963statistical}.
Indeed, we will show that $\chi_U+1$ and $F_U$, averaged over an ensemble of unitary gates, are \emph{exactly} described by a diffusion-reaction model or a classical Ising spin model.

\subsubsection{An emergent diffusion-reaction model}\label{sec:stat_DR_mapping}

\begin{figure*}[h]
\includegraphics[width=14cm]{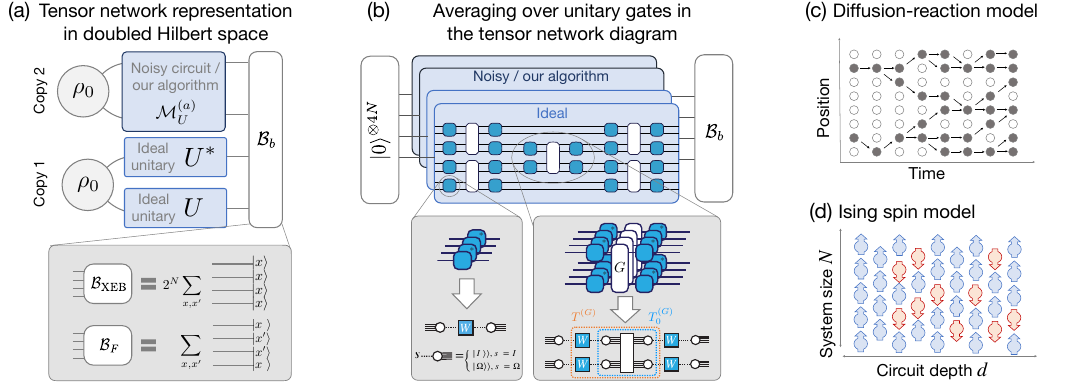}
\centering
\caption{
Mapping quantum circuits to statistical mechanics models.
(a) Both XEB and fidelity can be written as observables $\mathcal{B}_b$ with $b=\textrm{XEB},F$ in a duplicated Hilbert space by using tensor network representations.
The duplicated Hilbert space consists of the tensor product of Copy 1, representing an ideal circuit evolution, and Copy 2, representing the dynamics of either noisy circuit or our algorithm with omitted gates.
(b) For the tensor network diagrams representing XEB or fidelity, each random unitary gate (blue boxes) and its complex conjugate (blue boxes with asterisks) appear twice: in Copy 1 and in Copy 2.  One can perform averaging over an ensemble of unitary gates without explicitly evaluating the tensor network diagram, which gives rise to a simpler tensor network diagram with new classical variables, $s$, associated with each averaged single-qubit unitary gate (bottom left).
Entangling unitary gates $G$ dictate the dynamics of variables $s$, which is encapsulated in the transfer matrices of the classical statistical mechanics model (bottom right).
(c) Schematic diagram for the diffusion-reaction model. Each site can be occupied by a particle (filled) or remain unoccupied (empty). In every discrete time step, each particle may either stay on the same site, move to a neighboring site (diffusion), or duplicate itself to a neighboring site (reaction). Finally, a pair of particles located on neighboring sites may recombine into a single particle (reaction). Each of these processes has a specific probability that depends on the underlying gate ensemble. 
(d) Quantum circuits in 1D can be mapped to the classical Ising spin model in 2D.
}
\label{fig:stat_mapping_outline}
\end{figure*} 

We now describe the exact mapping from random unitary circuits to the diffusion-reaction model.
To derive this mapping, we will first consider the bulk of the tensor network in the absence of any noise or omitted gates, i.e., $\mathcal{M}_U^{\textrm{ideal}}[\rho_0] = U \rho_0 U^\dagger$. We will follow with the analysis of the boundaries at $t=0$ (initial state) and at $t=d$ (contraction with the observable $\mathcal{B}_b$).
Finally, we will consider how the presence of noise or omitted gates influences the system.

\emph{Bulk of the ideal circuit.}---
The central ingredient of the mapping to statistical mechanics models is the averaging over an ensemble of unitary gates~\cite{dankert2009exact}. In our case, we consider a single-qubit unitary $u \in SU(2)$ averaged over the Haar ensemble (or any other ensemble that forms a unitary 2-design). As depicted in~\autoref{fig:stat_mapping_outline}(b), every random unitary $u$ appears exactly 4 times: a pair of $u$ and $u^\dagger$ for the ideal dynamics and another pair for the quantum channel $\mathcal{M}_{U}^{(a)}$. Since these sets of 4 random gates are independent, we can average them locally within the circuit using the 2-design property~\cite{dankert2009exact},
\begin{align}\label{eq:2design}
    \mathbb E_{u}[ u\otimes u^* \otimes u \otimes u^*] = \kket{I}\bbra{I} + \frac{1}{3} \kket{\Omega} \bbra \Omega, 
\end{align}
where $\kket{I}$ and $\kket{\Omega}$ are mutually orthogonal operators in the duplicated Hilbert space defined as
\begin{align}\label{eq:I_Omega}
  \nonumber  \bbra{a,b,c,d} I\rangle \rangle &= \frac{1}{2} \delta_{ab} \delta_{cd},\\
    \bbra{a,b,c,d} \Omega \rangle \rangle &= \frac{1}{2} \sum_{\mu=\text{x,y,z}} \sigma^{\mu}_{ab} \sigma^{\mu}_{cd},
\end{align}
with Pauli matrices $\sigma^\mu$, and $a,b,c,d\in\{0,1\}$.
We note that by using this notation, we are implicitly utilizing the channel-state duality (also known as  the Choi–Jamiołkowski isomorphism~\cite{choi1975completely}), where operators such as density matrices are vectorized: $\rho = \sum_{ij} \rho_{ij} \ket{i}\bra{j}\rightarrow \kket\rho =\sum_{ij}\rho_{ij}\ket i\ket j$.
Intuitively, $\bbra{I}$ and $\bbra{\Omega}$ represent the normalization and the total polarization correlation between the two copies, respectively; see~\autoref{sec:stat_DR_mapping} for the detailed derivation of these properties.

Notice that~\autoref{eq:2design} is a sum of two projectors, up to normalization factors.
Therefore, by applying~\autoref{eq:2design} to every quadruple of single-qubit unitary gates, the tensor network diagram factorizes into smaller parts, which are enumerated by different assignments of classical variables $s\in\{I,\Omega\}$ associated with every independent single-qubit unitary gate.
We interpret the classical variable $s$ at a certain site in space-time as if that site is in a vacuum state ($s=I$) or occupied by a particle ($s=\Omega$).
In this picture, the particle configuration at a specific time step is given by the assignment of $I$ or $\Omega$ values to $s$ variables within that time slice. Then, the tensor network describes how the particle configuration is advanced in every time step, which is captured by the transfer matrix $\mathcal{T}$.

The transfer matrix between two time steps is determined by the product of local transfer matrices $\mathcal T = \prod_G T^{(G)}$.
In turn, a local transfer matrix $T^{(G)}$ is given by the combination of the prefactor $1/3$, originating from~\autoref{eq:2design}, and a non-trivial contribution $T_0^{(G)}$  associated with a single two-qubit gate $G$, as shown in~\autoref{fig:stat_mapping_outline}(b). 
We evaluate $T_0^{(G)}$ explicitly by contracting (four copies of) a two-qubit gate $G$ with four vectors $\kket{s}$, where $s=I,\Omega$, arising from four single-qubit random gates before and after $G$ [see~\autoref{fig:intro_alg}(b)]:
\begin{align}\label{eq:transfer G}
    T^{(G)}_{0;s_1 s_2 s_3 s_4} = \bbra{s_1}\bbra{s_2} G\otimes G^* \otimes G \otimes G^* \kket{s_3}\kket{s_4}.
\end{align}
Explicit calculations lead to the general form of the $T$-matrix
\begin{equation}\label{eq:transfer_general}
T^{(G)}=\begin{pmatrix}
1 & 0 & 0 & 0\\
0 & 1-D & D-R & R/\eta\\
0 & D-R & 1-D & R/\eta\\
0 & R & R & 1-2R/\eta
\end{pmatrix},
\end{equation}
written in the basis $\{II,I\Omega,\Omega I,\Omega\Omega\}$. 
This formula has been derived in Ref.~\cite{mi2021information} for studying quantum scrambling. In this work, we apply it to study  vulnerabilities of the XEB.
Here, $D\ge0$ and $R\ge0$ are parameters that depend on the specific choice of the entangling unitary gate $G$ (the gate ensemble), while $\eta=3$ for any 2-qubit gate. We call $D$, $R$ and $\eta$, the diffusion rate, reaction rate, and reaction ratio, respectively, and summarize their values for a few common entangling gates in~\autoref{tab:DR}. 

\begin{table}
\begin{center}
\begin{tabular}{|c|c|c|c|c|}
\hline
 & CZ & Haar & fSim & fSim$^*$\\
\hline
diffusion rate $D$ & 2/3 & 4/5 & 1 & 1 \\
\hline
reaction rate $R$ & 2/3 & 3/5 & $1/3+\sqrt3/6$ & 2/3\\
\hline
\end{tabular}
\end{center}
\caption{Values of the diffusion rate $D$ and the reaction rate $R$ for a few different entangling gates.}
\label{tab:DR}
\end{table}
 
We note that each column of $T$ is normalized to unity, implying that the matrix indeed describes a transfer matrix for a stochastic process. For example, the entry in the 2nd column and the 4th row specifies the probability of the two sites going from $I\Omega$ to $\Omega\Omega$---this is an example of the ``reaction'' process. Other transitions are given in the following, with probabilities written on top of the arrows,
\begin{eqnarray}
\nonumber \text{vacuum:} && \quad II\xrightarrow{1} II\\
\nonumber \text{stay:} &&  \quad I\Omega\xrightarrow{1-D} I\Omega , \Omega I\xrightarrow{1-D} \Omega I \\
\nonumber \text{move:} && \quad  I\Omega\xrightarrow{D-R} \Omega I ,\\
 \Omega I\xrightarrow{D-R} I\Omega  ,\nonumber
&&\quad\Omega\Omega\xrightarrow{1-2R/\eta} \Omega\Omega\\
\nonumber \text{duplication:} && \quad I\Omega,\Omega I\xrightarrow{R} \Omega\Omega \\
\nonumber \text{recombination:} && \quad  \Omega\Omega\xrightarrow{R/\eta} I\Omega,\Omega I .
\end{eqnarray}
The third process (move) is the ``diffusion'' (i.e., random walk), while the last two (duplication and recombination) are reaction processes, i.e., particle creation and annihilation. Notice that a particle cannot be created from the vacuum or annihilated into the vacuum without interacting with another particle.

\emph{Boundary conditions at the initial state and at the final time.}---
Next, we turn to the boundaries of our tensor network diagram.
First, we contract the input state $\rho_0 \otimes \rho_0$, denoted as $\kket{0^{\otimes 4}}^{\otimes N}$, with tensors associated with all $2^N$ possible particle configurations.
This leads to the vector ${\mathbf u}^{\otimes N}$, where
\begin{equation}\label{eq:ini_vec}
{\mathbf u}=
\begin{pmatrix}
\langle\innerp{I}{0^{\otimes4}}\rangle \\ \langle\innerp{\Omega}{0^{\otimes4}}\rangle
\end{pmatrix}
=
\begin{pmatrix}
1/2 \\ 1/2
\end{pmatrix},
\end{equation}
which follows directly from~\autoref{eq:I_Omega}.
This vector describes the initial distribution of particles: every site is occupied by a particle or remains empty with probabilities 1/2.

Similarly, at the final layer, we contract the $\mathcal{B}_b$ observables with tensors associated with all $2^N$ possible particle configurations, leading to dual vectors
$\mathbf v_\text{\tiny XEB}^{\top\otimes N}$ and $\mathbf v_F^{\top\otimes N}$ for the XEB and the fidelity, respectively, where
\begin{eqnarray}\label{eq:ini_vec 2}
\nonumber \mathbf{v}_\text{\tiny XEB}^\top&=&\begin{pmatrix}
\langle\innerp{I}{\beta_\textrm{XEB}}\rangle & \frac{ \langle\innerp{\Omega}{\beta_\textrm{XEB}}\rangle}{3}
\end{pmatrix} =  
\begin{pmatrix}
2 & 2/3
\end{pmatrix},
\\
\mathbf{v}_F^\top&=&\begin{pmatrix}
\langle\innerp{I}{\beta_F}\rangle & \frac{ \langle\innerp{\Omega}{\beta_F}\rangle}{3}
\end{pmatrix} =  
\begin{pmatrix}
1 & 1
\end{pmatrix}.
\end{eqnarray}
and
\begin{align}
\kket{\beta_\textrm{XEB}}&=2 \sum_{i\in\{0,1\}}\ket i\ket{i}\ket{i}\ket i,\label{eq:xeb_bc}\\
\kket{\beta_F}&= \sum_{i,i^\prime\in\{0,1\}}\ket i\ket{i^\prime}\ket{i^\prime}\ket i,\label{eq:f_bc}
\end{align}
are the single-site versions of $\mathcal{B}_b$, i.e., $\mathcal{B}_b = \beta_b^{\otimes N}$.
We find that $\mathbf{v}_\text{\tiny XEB}$ is distinguished from $\mathbf{v}_F$ by unequal weights between $I$ and $\Omega$ (by a factor of $1/3$) aside from the global normalization factor $2$.
This allows an intuitive explanation: as previously mentioned,  $\bbra{\Omega}$ represents total polarization correlation between two copies of quantum states, but XEB  depends only on correlations measured in the computational basis constituting $1/3$ of the total on average.

Combining the results from bulk transfer matrices, and initial and final boundary conditions, we obtain the expression for the ensemble-averaged XEB and fidelity:
\begin{align}
\label{eq:chi_U_avg}
\chi_\text{av} + 1&\equiv   \mathbb{E}_u [\chi_U]+1 =
    \mathbf v_\text{\tiny XEB}^{\top\otimes N}
    \left(\prod_{j=1}^d \mathcal{T}_j \right)
    \mathbf{u}^{\otimes N}\\
\label{eq:F_U_avg}
F_\text{av} &\equiv
    \mathbb{E}_u[F_U] =
    \mathbf v_F^{\top\otimes N}
    \left(\prod_{j=1}^d \mathcal{T}_j\right)
    \mathbf{u}^{\otimes N},
\end{align}
where $\mathcal{T}_j$ is the transfer matrix for $N$ particles at time-step $j$.

\emph{XEB and fidelity as statistics of a particle distribution.}---
Our results in Eqs.~\eqref{eq:chi_U_avg} and~\eqref{eq:F_U_avg} allow for an intuitive understanding of the XEB and the fidelity in terms of particle distributions in the diffusion-reaction model.
We note that these two quantities differ only by the boundary condition at the final time $t=d$, as defined in Eqs.~\eqref{eq:xeb_bc}-\eqref{eq:f_bc}.
Hence, both the XEB and the fidelity are fully determined by the probability distribution of particle configurations, $\mathbf p$, obtained by evolving the initial uniform distribution $\mathbf{u}^{\otimes N}$ for $d$ time steps:
\begin{equation}\label{eq:Cpath}
{\mathbf p}\equiv\mathcal T_d\cdots\mathcal T_2\mathcal T_1{\mathbf u}^{\otimes N}.
\end{equation}
From this distribution, the XEB and the fidelity can be evaluated by simply contracting either $\mathbf v^{\top\otimes N}_{\rm XEB}$ or $\mathbf v^{\top\otimes N}_{F}$, which corresponds to computing certain statistics of the particle distribution.
For instance, all entries in $\mathbf v_F^{\top\otimes N}$ are unities, implying that $\mathbf v^{\top\otimes N}_F \mathbf{p}$ is equal to the summation over all probabilities:
\begin{align}
    F_\text{av}=\mathbf v^{\top\otimes N}_F \mathbf{p} = \mathbb E_\mathbf{p} [1],\label{eq:XEBFideal01}
\end{align}
where $\mathbb E_\mathbf{p}[\cdot]$ denotes the averaging over the distribution $\mathbf{p}$.
In the absence of any noise or omitted gates, the transfer matrix in~\autoref{eq:transfer_general} preserves the total probability, leading to $F_\text{av} = \mathbb E_\mathbf{p} [1] = 1$.
This result is trivially expected in the quantum circuit picture --- in the absence of any noise or omitted gates, the fidelity must always be unity.
We will soon see how this picture is modified when we introduce noise or omit  gates.

Similarly, the average XEB is
\begin{eqnarray}
\nonumber \chi_\text{av}+1&=&{\mathbf v^{\top\otimes N}_\text{\tiny XEB}}{\mathbf p}=2^N\mathbb E_\mathbf{p}\left[\frac{1}{3^{\#\Omega\text{ in the last layer}}}\right],\label{eq:XEBFideal02}
\end{eqnarray}
where $\#\Omega$ denotes the total number of particles.

\emph{Effects of noise or omitted gates.}---
When unitary dynamics is interspersed by noise channels ($\mathcal{M}_U^{({\rm noisy})}$) or when some of the gates are omitted in our classical algorithms ($\mathcal{M}_U^{({\rm algo})}$), only the bulk part of the tensor network changes, leading to a modified transfer matrix.
For a noisy circuit, the new transfer matrix is
\begin{equation}
\label{eq:transfer_matrix_noisy}
    T^{(G)}_\epsilon= (I_\epsilon\otimes I_\epsilon)\, T^{(G)} \quad\text{with }I_{\epsilon}=
    \begin{pmatrix}
    1 & 0\\
    0 & 1-c\epsilon
    \end{pmatrix},
\end{equation}
where $c$ is a constant depending on the type of noise. For example, $c=4/3$ for the depolarizing noise $\mathcal N_\epsilon(\rho)=(1-\epsilon)\rho+\epsilon/3\sum_{\mu}\sigma^{\mu}\rho\sigma^\mu$, and $c=2/3$ for the amplitude damping noise. 

Unlike the transfer matrix in the ideal case, the noisy-circuit transfer matrix in~\autoref{eq:transfer_matrix_noisy} no longer describes a stochastic process. 
That is, the sum of each column in $T^{(G)}_\epsilon$ is less than unity, implying that the probability is not conserved.
Thus, the effect of noise gives rise to the ``loss of probability'' in our diffusion-reaction model. In general, this leads to an unnormalized final distribution $\mathbf{p}$ and reduced average fidelity $F_\text{av} < 1$.
Crucially, the loss of probability occurs only when a particle ($\Omega$) is present at a given space-time point. The diagonal entries in $I_\epsilon$ imply that the probability associated with a given particle configuration will be damped by a factor $(1-c\epsilon)^{\#\Omega}$ at every time step. 
Therefore, we expect an interesting interplay between the diffusion-reaction dynamics of particles and the probability loss.

For our classical algorithm, it is the omission of gates that modifies the transfer matrix.
In this case, only local transfer matrices associated with an omitted gate are affected
\begin{equation}\label{eq:TP_projector}
    T^{(G)}\rightarrow  (P_I\otimes P_I)\cdot T^{(G)}=P_I\otimes P_I
    \quad\text{with }P_I=
    \begin{pmatrix}
    1 & 0\\
    0 & 0
    \end{pmatrix}.
\end{equation}
Similarly to the noisy circuit case, the omission of gates also causes the loss of probabilities; thus, the fidelity becomes smaller than 1.
More specifically,~\autoref{eq:TP_projector} implies that, at any given time, the probability weights associated with particle configurations containing at least one particle at the site of omitted gates must vanish; such configurations do not contribute to the average XEB or fidelity.
Thus, the only non-vanishing contributions arise from diffusion-reaction processes in which not a single particle ever appears at the sites of omitted gates throughout the entire dynamics.
The average fidelity will then be the total probability of such diffusion-reaction processes, and the average XEB is determined by the resultant unnormalized distribution $\mathbf{p}$.

We remark that the deterministic loss of probability at the positions of omitted gates leads to the factorization of the transfer matrix in~\autoref{eq:TP_projector} (as a product of two projectors).
Due to this factorization, $\mathbf{p}$ for the whole system also factorizes into independent probability vectors for two isolated subsystems.
This feature allows the numerical calculation of the average XEB for system sizes up to the quantum advantage regime (60 qubits, depth 24).

\vspace{2mm}
\noindent\textbf{Dynamics of the XEB and fidelity.} 
Having introduced the mapping of random unitary circuits to the diffusion-reaction model in the previous section, we now leverage this formalism to understand the quantitative behavior of the XEB and the fidelity under various conditions. In particular, we explain the key concepts used to obtain results presented in~\autoref{sec:quantum experimental regime} and~\autoref{sec:quantum complexity results}.

\emph{Ideal circuit. ---}
In the absence of noise and omitted gates, the fidelity remains equal to unity trivially, due to the conservation of the total probability. 
It is non-trivial, however, to see how the average XEB approaches unity in the limit of deep quantum circuits~\cite{arute2019quantum}, which we now explain in terms of diffusion-reaction dynamics.
Both the XEB and the fidelity, at late times (large depths), are determined by the output vector $\bf p$. 
For the transfer matrix in~\autoref{eq:transfer_general}, this distribution converges to a fixed point in the large-depth limit.
In the current case, there are two fixed points for local transfer matrices, ${\bf u}_1=(1/4,3/4)$ and ${\bf u}_2=(1,0)$.
The former represents a nontrivial steady-state solution in which the total normalization, and three different types of correlations (along $x$, $y$, and $z$ directions) are equally distributed, while the latter represents a trivial solution where two copies are both in completely mixed states; hence, no correlation is generated during dynamics.
It can be shown that the global stationary distribution is given as a mixture of ${\bf u}_1^{\otimes N}$ and ${\bf u}_2^{\otimes N}$, whose ratio is determined by the initial condition ${\bf u}^{\otimes N}$:
\begin{align}\label{eq:limiting_distribution}
    \lim_{d\rightarrow \infty} {\bf p} =
    (1-2^{-N}){\bf u}_1^{\otimes N} + 2^{-N} {\bf u}_2^{\otimes N}+O(4^{-N})
\end{align}
The dominant contribution originates from the non-trivial equilibrium configuration $\mathbf u_1$,
whereas the ${\bf u}_2$ term constitutes a small correction.

The nontrivial term describes the homogeneous distribution of  particles with the density 3/4, as shown in~\autoref{fig:stat_population_last}(a), contributing to the XEB
$$
\mathbf{v}_\text{\tiny XEB}^\top{\bf u}_1=2\left(\frac{1}{4}\cdot1+\frac{3}{4}\cdot\frac{1}{3}\right)=1.
$$
The trivial term gives $\mathbf{v}_\text{\tiny XEB}^\top{\bf u}_2=2$ per site.
Combined together with appropriate coefficients, we obtain the average XEB $\chi_\text{av}=(1-2^{-N})\approx 1$ as expected.
We note that the net contribution from the trivial solution (${\bf u}_2$ term) is always $+1$, which exactly cancels the constant term $-1$ in the definition of the XEB.

\emph{Noisy circuit.}---
If noise is introduced to the system, the total probability is no longer conserved, and ${\bf u}_1^{\otimes N}$ does not form a stationary solution. However, we can still predict the behavior of the average XEB and fidelity using our model.
We distinguish two regimes: (a) the weak noise limit where the total probability loss rate $N\epsilon$ is much smaller than the inverse equilibration time $\tau_{\rm eq}^{-1}$ of the particle distribution, $N\epsilon \ll \tau_{\rm eq}^{-1}$, and (b) strong noise limit $N\epsilon \gg \tau_{\rm eq}^{-1}$.
In terms of quantum circuit dynamics, these conditions correspond to the comparison of the total error rate to the scrambling time.

\begin{figure}[h]
\includegraphics[width=7cm]{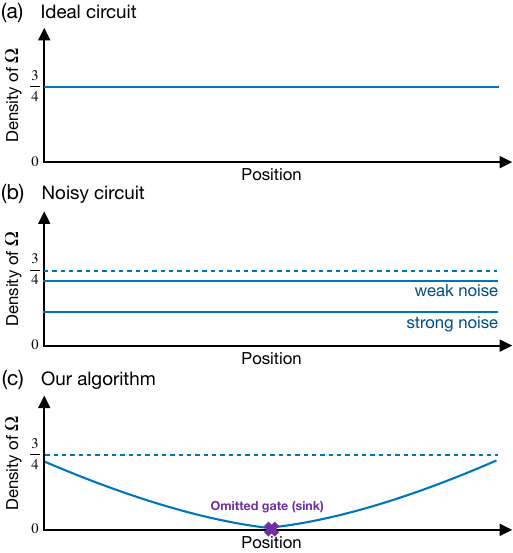}
\centering
\caption{Sketch of the particle population distribution at the last layer.
The vertical axis is the density of particles ($\Omega$) at the final layer normalized by the total probability, and the horizontal axis is the position of sites.
(a) Ideal circuits. 
(b) Noisy circuits.
The density is decreased relatively to the ideal case.
The discrepancy becomes larger for larger noise rates.
(c) Our spoofing algorithm. Close to the position of an omitted gate, or a ``sink'' (purple cross),  the density of particles is suppressed.
}
\label{fig:stat_population_last}
\end{figure}

In the limit of weak noise, the steady state configuration must stay close to that of the equilibrium solution, because the system relaxes quickly before any substantial probability loss occurs.
Thus, the output probability vector at the final time is not severely affected by the probability loss during preceding times, other than a global re-scaling factor.
This leads to the (un-normalized) equilibrium state $\widetilde{\mathbf p}=\widetilde{\mathbf u}_1^{\otimes N}$, where 
\begin{equation}
\label{eq:weak_noise_solution_u}
\widetilde{\mathbf u}_1\approx \alpha\begin{pmatrix}1/4\\3/4 \beta \end{pmatrix}.
\end{equation}
Here $\alpha$ is the re-scaling factor that accounts for the probability loss (per site) during the diffusion-reaction dynamics, and it generally decreases exponentially with depth.
The parameter $\beta $ quantifies the deviation of $\widetilde{\mathbf{u}}_1$ from its equilibrium shape, and generally $\beta \approx 1$ in the weak-noise limit.
The precise value of $\beta$ depends on the strength of noise and the equilibration time.
As long as $\beta \approx 1$, $\widetilde{\mathbf p}$ is a simple re-scaling of the ideal-circuit distribution, and XEB approximates the fidelity well; both quantities are suppressed by the factor of $\alpha^N$.

In the limit of relatively strong noise (slow equilibration), the particle configuration cannot relax to its equilibrium before it is significantly affected by the probability loss.
In this limit, the deviation of $\widetilde{\mathbf{u}}_1$ from the equilibrium becomes significant, and  $\beta < 1$ decreases with the increasing strength of noise.
This is because, generically, the probability loss associated with $\Omega$ particles during dynamics results in a reduced density of particles at the last layer~[see~\autoref{fig:stat_population_last}(b)].
The reduced density of particles implies that the XEB is larger than the fidelity because the boundary vector $\mathbf{v}_{\rm XEB}$ has a higher weight for the vacuum than for the particle state, whereas $\mathbf{v}_{F}$ has the same weight for both states.
Hence, the larger the noise rate, the greater the deviation of the XEB from the fidelity.
Eq.~\eqref{eq:weak_noise_solution_u} no longer holds for greater noise strengths\footnote{
In this work, we focus on the experimentally relevant regime, where the strength of noise and the depth of the circuit are not too large, such that the fidelity remains sufficiently greater than $2^{-N}$. When the fidelity is close to $2^{-N}$, the discussion in this paragraph no longer holds, as the contribution from subdominant terms in Eq.(\ref{eq:limiting_distribution}) becomes significant.}.

\begin{figure}[h]
\includegraphics[width=7cm]{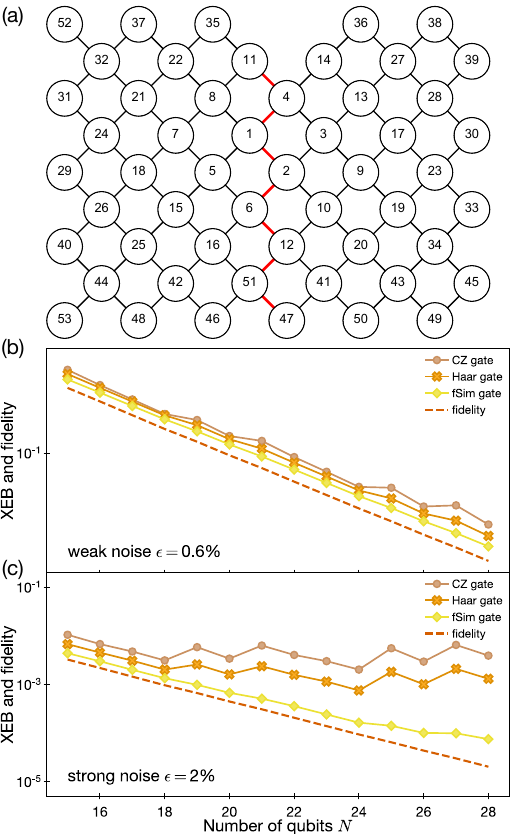}
\centering
\caption{XEB vs. fidelity in noisy circuits. The XEB always overestimates the fidelity, but the deviation depends on the gate ensemble and the strength of noise. (a) For this calculation, we use the original qubit ordering from Fig.~S25 in Ref.~\cite{arute2019quantum} (see also Fig.\ref{fig:intro_partition}). 
(b) Weak noise regime ($\epsilon=0.6$\%). The XEB approximates the fidelity well, and the fidelity values for all gate ensembles are almost the same. (c) Strong noise regime ($\epsilon=2$\%). The quality of the XEB-to-fidelity approximation strongly depends on the choice of the gate ensemble. Among the three ensembles considered here, the fSim ensemble gives the best result.}
\label{fig:stat_noise_fid_vs_xeb}
\end{figure}

\emph{Spoofing algorithm.}--- Our algorithm is designed to leverage the discrepancy between the XEB and the fidelity. 
In contrast to homogeneous errors spread over the bulk of the circuit, the errors in our algorithm are highly inhomogeneous and localized --- they appear only at specific positions where we omit gates. 
This inhomogeneity leads to a particle distribution that is far from its equilibrium counterpart.
More specifically, the position of an omitted gate behaves like a ``sink'' of probabilities --- any configurations containing particles at sink sites, at any time, will acquire vanishing contribution to $\widetilde{\mathbf p}$.
Therefore, in any non-vanishing contribution to $\widetilde{\mathbf{p}}$, the relative density of particles with respect to the density of vacuum states  is substantially lowered near the sink [see~\autoref{fig:stat_population_last}(c)].
This large imbalance (relative to the equilibrium) leads to the large XEB-to-fidelity ratio.
Thus, given the same value of fidelity, which is controlled by the total number of omitted gates, one can achieve high XEB values because vacuum state $I$ has a larger weight in the XEB than in the fidelity.

The non-equilibrium, spatially inhomogeneous dynamics of particles also leads to a distinct scaling behavior. In our algorithm, the average XEB value increases with the system size $N$, when the number of omitted gates is fixed.
This can be intuitively explained: the more space for particles to diffuse to, the less likely it is for them to hit sink sites, leading to an effectively smaller particle loss rate and reduced imbalance in the particle density, relative to the equilibrium. 

Here, we make two remarks.
First, while our analysis remained qualitative and focused on two extreme cases of error models, i.e., one with completely homogeneous noise and another fully localized errors, we emphasize that our intuitive understanding can be straightforwardly generalized to arbitrary circuit geometry with arbitrary inhomogeneous error models in both space and time.
In such cases, one can directly estimate the distribution $\widetilde{\mathbf p}$ by using conventional approaches such as Monte Carlo methods.
Second, we comment that, intuitively, larger diffusion and reaction rates imply shorter time required to reach the equilibrium distribution. In other words, given a circuit architecture, the XEB will be on average a better proxy for the fidelity in circuits consisting of faster scrambling (entangling) gates, with larger $R$ and $D$.

\begin{savequote}[75mm]
There is no scientific study more vital to man than the study of his own brain. Our entire view of the universe depends on it.
\qauthor{Francis H.C.~Crick}
\end{savequote}

\chapter[Preliminary Introduction to Neuroscience]{Preliminary Introduction to\\Neuroscience}\label{ch:prelim neuro}

Neuroscience is the study of the brain and is intrinsically interdisciplinary. While any introductory text on neuroscience will certainly miss some perspectives, the learning journey resembles playing a jigsaw puzzle, where the more one studies, the bigger and more comprehensive pictures will emerge. The goal of this chapter is to provide readers without a neuroscience background a concise crash course on some basics, in order to build up relevant knowledge for appreciating the big picture this thesis aims to address. Furthermore, we hope to stimulate the reader's interest in continuing to learn more about neuroscience.

\section{Different levels and scales in the brain}
The brain is a multi-level information processing system: from nerve cells (a.k.a., neurons) using chemical ions to transmit electrical signals, to a network of neurons, and up to the cognitive realm. In different levels and scales, there are different players as well as different phenomenology. In this section, we will quickly go over some basic facts in neuroscience from the bottom to the top for the reader to prepare some mental picture on the landscape of neuroscience.

\subsection{Cellular and molecular level}
An adult human brain has about 100 billion neurons (also known as nerve cells, or simply neurons), along with 2 to 10 times more glial cells that play a supporting role. However, since neurons are the primary carriers of neural signals, most research focuses on them. Here, we will briefly introduce the basic operation of neurons at the cellular and molecular level from various perspectives.

\medskip \noindent \textbf{Basic structure of a neuron.}
There are thousands of types of neurons, which may vary slightly among different levels of animals. From the most abstract perspective, a neuron consists of four main regions: (i) soma (i.e., cell body), (ii) dendrites, (iii) axon, and (iv) presynaptic terminals.

As the name suggests, the cell body, soma, is responsible for maintaining cellular metabolism and function, and stores genetic information. Dendrites and axons extend from the cell body, playing the roles of receiving and transmitting signals, respectively. Finally, the end of the axon branches into several presynaptic terminals to connect with various receiving neurons.

\begin{figure}[h]
    \centering
    \includegraphics[width=8cm]{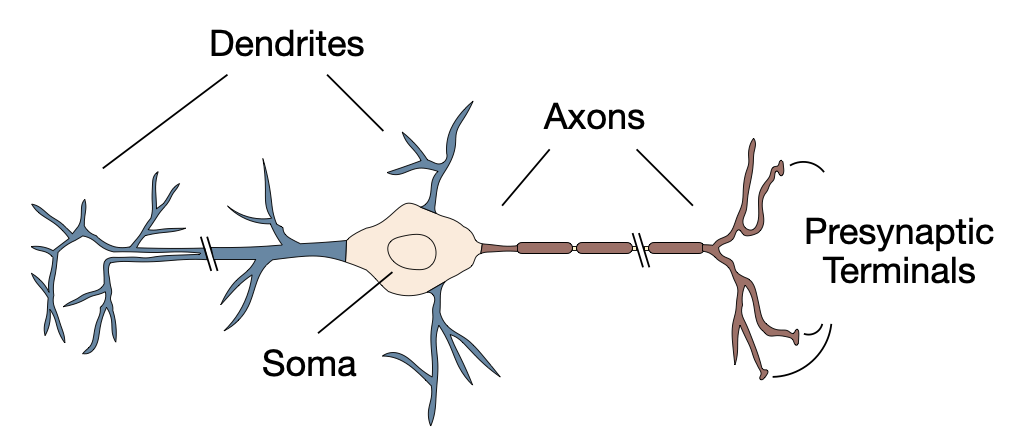}
    \caption{The four basic structures in a neuron.}
    \label{fig:neuron}
\end{figure}

Real neurons have an incredibly diverse range of types and unexpected appearances and behaviors. To this day, countless cellular and molecular neuroscientists are working hard to create a dictionary of different neuronal cell types.

\medskip \noindent \textbf{How a single neuron works.}
The neuronal cell membrane divides the internal and external worlds of a cell, preventing ions from passing through arbitrarily. In particular, four common ions: sodium ($\text{Na}^+$), potassium ($\text{K}^+$), chloride ($\text{Cl}^-$), and calcium ($\text{Ca}^{2+}$), create a difference in electrical potential between the inside and outside of a neuron's cell body based on their relative concentrations. This is referred to as the \textit{membrane potential}.

The cell membrane doesn't completely block ions from passing through it. Many ion channels composed of complex proteins determine whether to allow specific ions to move in or out based on their concentration differences or specific signals. These ion channels play a crucial regulatory role, enabling neurons to sensitively control very subtle changes in electrical potential and intricately alter their membrane potentials.

Under the regulation of ion channels, a neuron's membrane potential typically rests at a \textit{resting potential} of around -70mV. This is maintained by specialized channels, such as the sodium-potassium pump ($\text{Na}^+$ $\text{K}^+$ pump), which use energy (ATP) to maintain the ion concentration difference. If other neurons transmit many signals in a short period, raising the membrane potential to -55mV, an \textit{action potential} is triggered. In a brief moment, the neuron's membrane potential rises (depolarization) and then falls (polarization), entering a temporary refractory period. The dramatic events of an action potential also cause voltage fluctuations in the axon, rapidly transmitting the message (the occurrence of an action potential) from the cell body to the axonal terminals, preparing to send it to other neurons. Based on its appearance, an action potential is often referred to as an \textit{electrical spike}.

\begin{figure}[h]
    \centering
    \includegraphics[width=13cm]{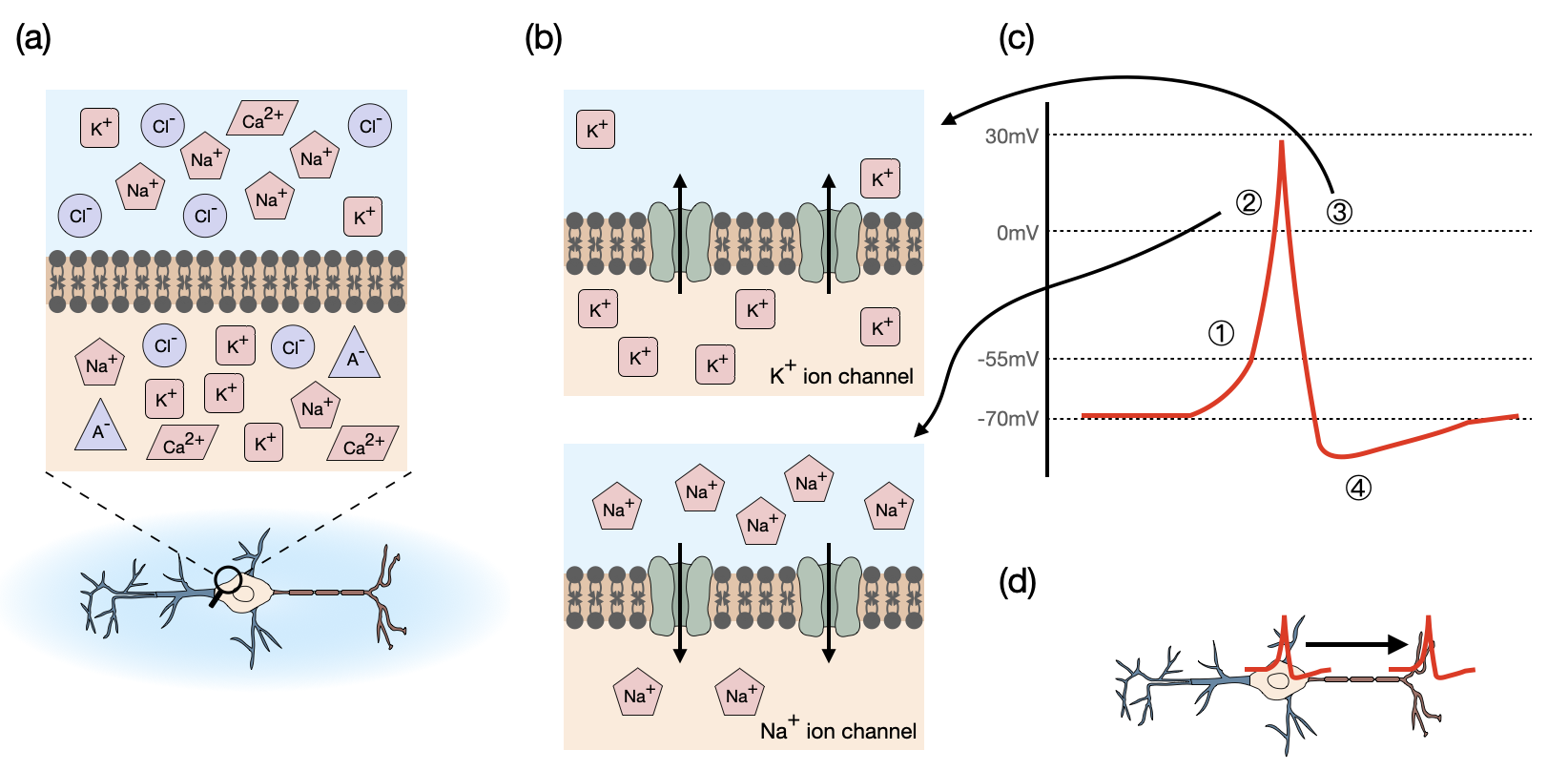}
    \caption{How do neurons work. (a) The cell membrane separates various ions inside and outside the cell, resulting in a difference in concentration. Potassium ions have a higher concentration inside a neuron, while sodium, chloride, and calcium ions have a higher concentration outside a neuron. Additionally, there are some organic anions within the cell. (b) There are ion channels on the cell membrane, and each ion has its own specific channel which allows the ions to diffuse in and out of the neurons based on the concentration differences. The figure above shows the potassium ion channel, and the one below shows the sodium ion channel. Moreover, there are special pumps (not shown in the figure) that use energy to move certain ions against the direction of their concentration gradient. (c) When the membrane potential reaches the threshold value (approximately -55mV), the sodium ion channel will open, allowing a large amount of sodium ions to rush into the cell body, causing a rapid rise in membrane potential. When the membrane potential reaches around 30mV, the sodium ion channels close and the potassium ion channels open. Consequently, the massive influx of potassium ions causes the membrane potential to decrease, sometimes even slightly lower than the resting potential, and enters a refractory period, temporarily making it more difficult to generate an action potential. (d) When an action potential occurs in the cell body, it is transmitted to the end through the axon, preparing to relay the signal to other neurons.}
    \label{fig:action potential}
\end{figure}

\medskip \noindent \textbf{How neurons communicate with each other.}
Dendrites and axons of neurons are akin to their hands. However, when two neurons ``hold hands'', one must use its dendrite while the other uses its axon. The place where they hold their hands is called a \textit{synapse}. When the presynaptic neuron (the one using the axon) generates an action potential and carries the signal to the presynaptic terminal as previously described, the postsynaptic neuron's potential will change in different ways depending on the type of synapse (chemical or electrical).

\begin{figure}[h]
    \centering
    \includegraphics[width=13cm]{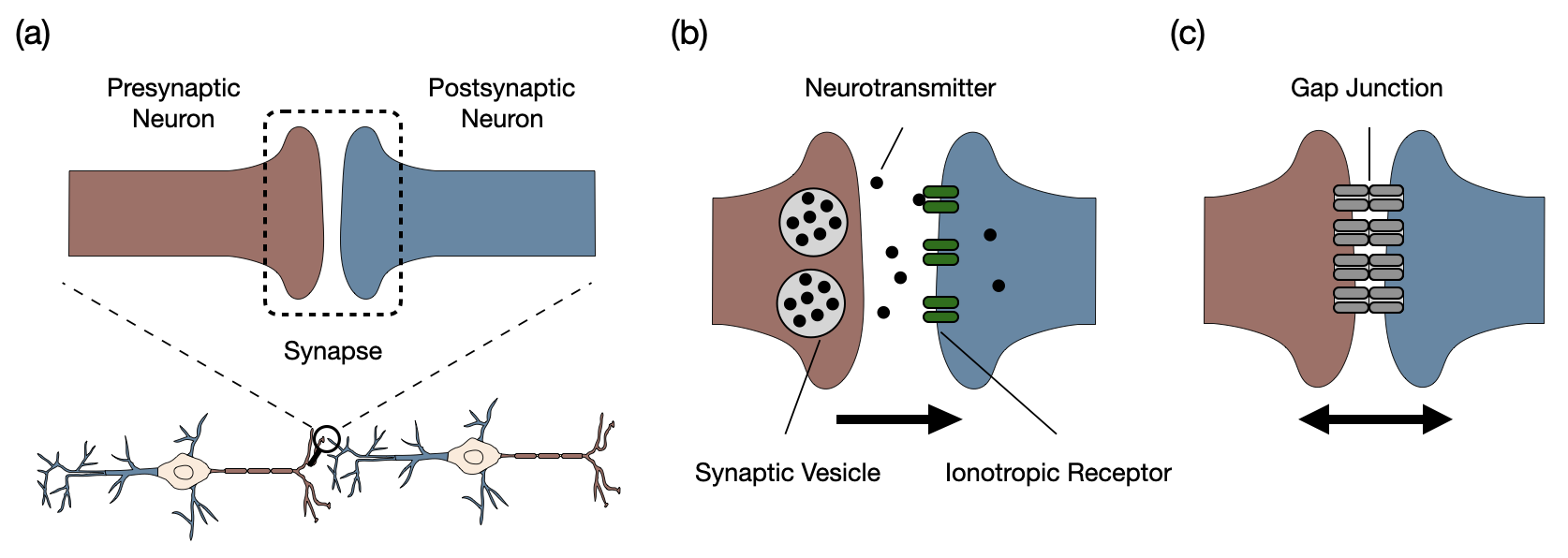}
    \caption{Synapses. (a) A synapse refers to the junction between a presynaptic neuron and a postsynaptic neuron. (b) Chemical synapses transmit signals through neurotransmitters. After the presynaptic neuron generates an action potential, the synaptic vesicles in the presynaptic terminals release their stored neurotransmitters. These neurotransmitters then diffuse to the (ionotropic) receptors on the dendrites of the postsynaptic neuron. (c) Electrical synapses, on the other hand, directly connect two neurons through gap junctions. Generally, electrical synapses transmit signals faster, more delicately, and can be bidirectional. Chemical synapses are slightly slower, but their signal effects can be more extensive, such as enhancing the signal.}
    \label{fig:synapse}
\end{figure}

Similar to how a handshake may become firmer as people become more familiar with each other, the magnitude of the potential change caused by synapses can also vary depending on factors such as the frequency of interaction between the two neurons. This property is referred to as synaptic plasticity. Many neuroscientists believe that this characteristic is an essential mechanism used by the brain during learning. See~\autoref{sec:synaptic plasticity} for more on synaptic plasticity and learning.

\subsection{Circuit level}
When many neurons connect to one another, they jointly form a circuit. Such a network is not arbitrary, but often has a very structured organization in the brain. By examining experimental data from different organisms, neuroscientists abstract out \textit{circuit motifs} for higher-level reasoning.

Four common circuit motifs are: (i) feedforward network, which contains layers of neurons connected in a feedforward manner; (ii) divergent network, which sends information (i.e., neuronal activities) from a few neurons to a large number of neurons that are potentially far apart; (iii) convergent network, which aggregates information from multiple neurons to a few neurons; (iv) recurrent network, which contains neurons that interconnect with each other (not necessarily all-to-all). See~\autoref{fig:circuit motifs} for a pictorial illustration.

\begin{figure}[h]
    \centering
    \includegraphics[width=12cm]{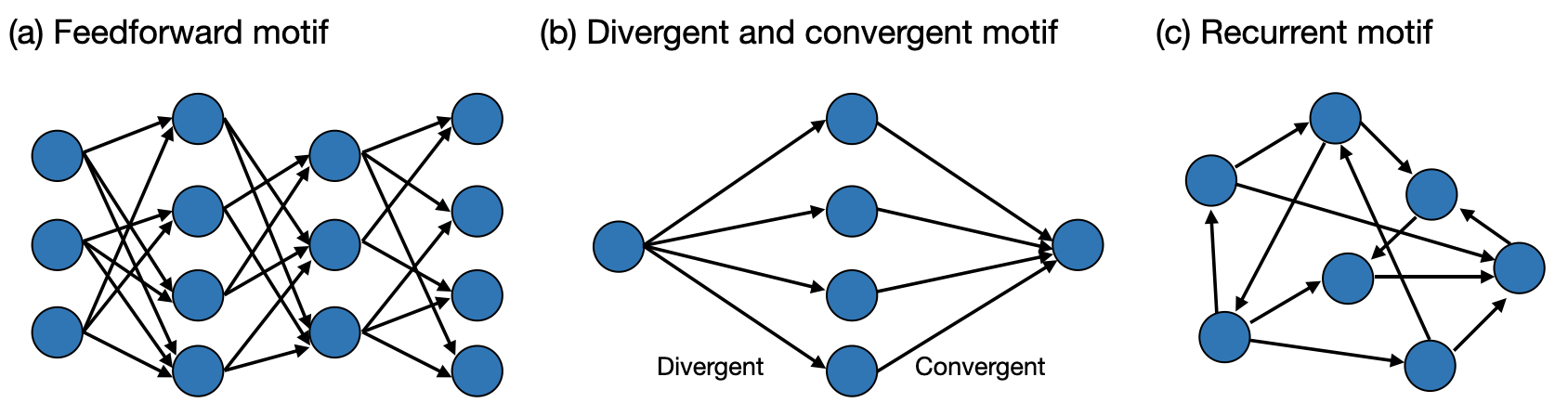}
    \caption{Common circuit motifs.}
    \label{fig:circuit motifs}
\end{figure}

In biology, there are always exceptions, and indeed, any real biological neural network may not precisely fit into the above-mentioned categories. In the following, let us take a look at the \textit{cortical microcircuit} as an example to get a taste of how complex the picture can be.

\medskip \noindent \textbf{An example: cortical microcircuit.}
The \textit{cerebrum} is undoubtedly one of the most important brain areas. It is responsible for various sensory functions, motor sequencing, language, decision-making, emotion, learning, and more. The cerebral cortex, or sometimes abbreviated as \textit{cortex}, consists of the gray matter\footnote{Gray matter contains the soma of a large number of neurons, as opposed to white matter, which contains the axons of neurons.} that covers the cerebrum. It appears folded, resulting in a vast surface area and accounting for approximately half of the total human brain weight. The cortex typically exhibits a distinct six-layer structure and contains hundreds of millions of cortical columns. Each cortical column consists of around 10,000 neurons spanning the six layers of the cerebral cortex and has a large number of internal connections.

Now that we know we can roughly think of the cortex as a massive and important brain area organized in an extremely parallel (i.e., cortical columns) fashion. However, the neurons in the cortex come from hundreds to thousands of different cell types and have recurrent connectivity within their cortical column (both inter- and intra-layer), with other columns, as well as with other brain regions (e.g., thalamus). Fortunately, a neuron in the brain is usually either excitatory or inhibitory\footnote{An excitatory (resp. inhibitory) neuron either only increases (resp. decreases) the activity of the other neurons it connects to.}, allowing us to imagine the cortex as containing a group of excitatory and a group of inhibitory neurons in each layer of each cortical column. Finally, based on anatomical data, we can further refine our understanding of the \textit{cortical microcircuit} as described in~\autoref{fig:cortical microcircuit}.

\begin{figure}
    \centering
    \includegraphics[width=4.5cm]{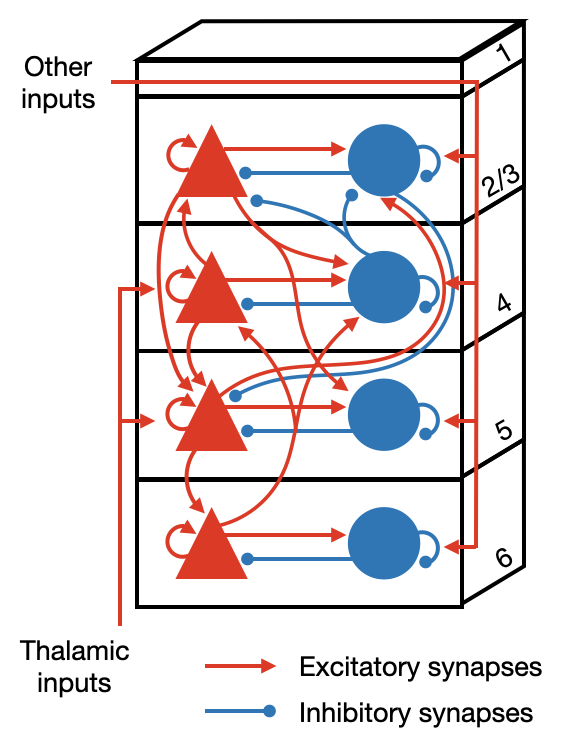}
    \caption{Cortical microcircuit. Each layer of each cortical column can be abstracted as having a group of excitatory neurons and a group of inhibitory neurons, the connection relationship of which is described as a cortical microcircuit.}
    \label{fig:cortical microcircuit}
\end{figure}

\subsection{System level}
When a single network or multiple networks of neurons together correspond to a specific functionality or exhibit a clear anatomical organization, they are often referred to as a system. Systems neuroscience aims to study the functions associated with various structures and systems within the brain, as well as their relationships with perception and behavioral performance. Of course, the biological world is not as clear-cut in its division of labor as the artificial world. Thus, individual subsystems in the brain may not be solely responsible for a single task, and even the boundaries between systems may not be clearly defined.

While different neuroscientists may have their preferred ways to introduce systems neuroscience, we will present an oversimplified classification of brain systems into three functional types: (i) sensory systems, (ii) motor systems, and (iii) integration, learning, and memory  systems, to provide a glimpse into systems neuroscience.

\begin{figure}[h]
    \centering
    \includegraphics[width=12cm]{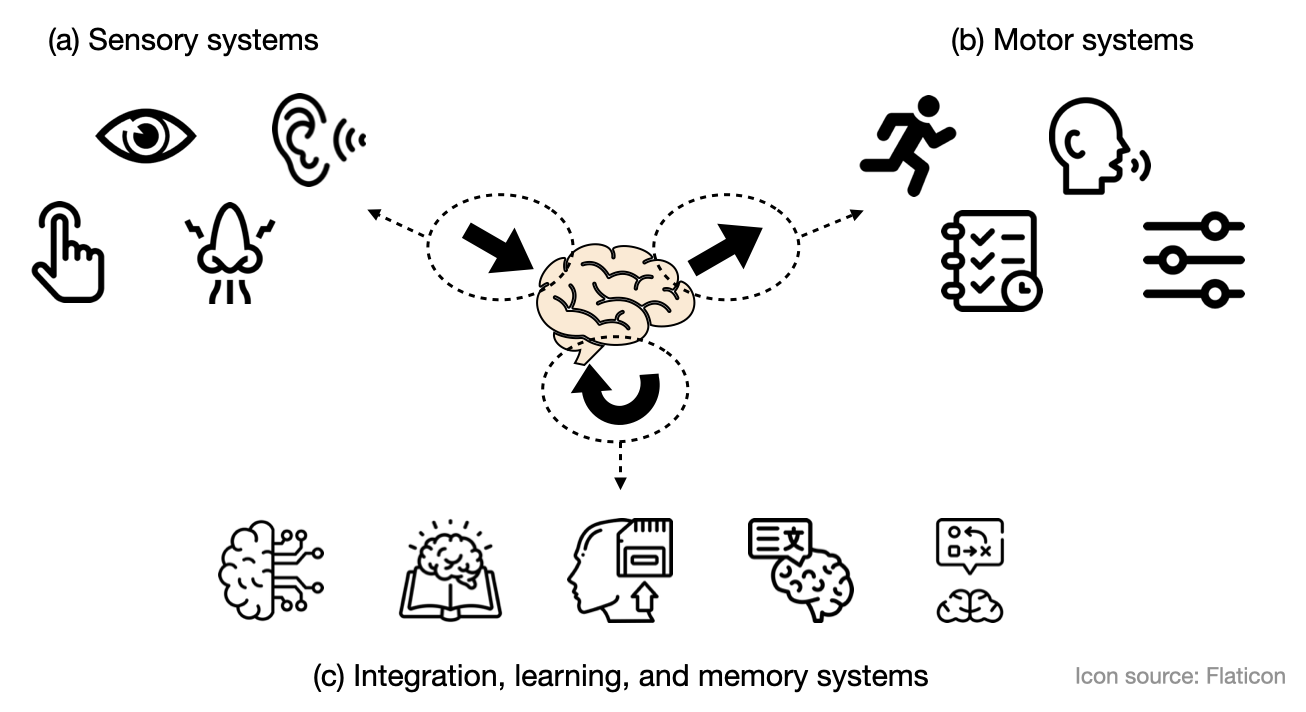}
    \caption{Common types of system in neuroscience.}
    \label{fig:neural systems}
\end{figure}

\medskip \noindent \textbf{Sensory systems.}
The brain acquires information from the outside world through sensory systems. Common sensory systems include the visual system, auditory system, tactile system, and olfactory system, among others.

\medskip \noindent \textbf{Motor systems.}
The brain's signals generate responses to the external world through the motor system. Common motor systems and functions include muscle control, motor planning, and so on.

\medskip \noindent \textbf{Integration, learning, and memory systems.}
Between the sensory and motor systems, the brain undergoes many signal transformations and processing. How to classify these is still an ongoing task and debate. Since many functions cannot be as clearly distinguished as perception and action, here we collectively refer to them as integration, learning, and memory systems.

\subsection{Cognitive level}
Cognition is a loaded term. Broadly speaking, it refers to the mental processes of sensing, knowing, understanding, reasoning, and much more. Traditionally, cognitive science, i.e., the science of cognition, employs a top-down approach, focusing heavily on behaviors and seeking high-level cognitive functions.

Cognitive neuroscience is the study of the biological underpinnings of cognition. By utilizing experimental recordings and computational modeling, researchers aim to provide explanations of cognitive functions at the implementational level.

\section{Toward a holistic understanding: different methodologies}

In the current landscape of neuroscience advancements, numerous studies derive from medical pursuits related to neurological diseases, and theoretical physics research formulating mathematical models to elucidate the brain's complex equations. Only a few years ago, scientists succeeded in constructing the complete brain connectome of a fruit fly, encompassing over three thousand neurons and more than half a million cellular connections. Simultaneously, the latest generation of electrophysiological probes can now record the activity of hundreds of neurons. Within this exceptionally interdisciplinary and diverse field, an inherent downside is the inevitable presence of vast unknown territories for every researcher. However, viewing from a different angle, this very aspect implies an ample space for the application and flourishing of diverse methodologies.

\subsection{The mechanistic approach}
It is so tempting to build up understanding from the ground up, as we have done in physics. From Newton's law to Schr\"{o}dinger's equation, the mechanistic view of the physical world gives us a comforting sense of comprehending the machinery of the universe.

The mechanistic approach has played a central role in neuroscience since its early days. From single-neuron modeling to synaptic plasticity, neuroscientists have constructed numerous mechanistic models (mostly in the form of dynamical systems) that capture empirical phenomena with high precision. One seminal example is the \textit{Hodgkin-Huxley model}~\cite{hodgkin1952quantitative}, which is a mathematical model for the electrophysiology of neurons. Inspired by electrical circuits, the Hodgkin-Huxley model can reproduce the single-neuron activity of dozens of neuronal cell types with a dynamical system using only a few tunable parameters. 

In addition to modeling neuron dynamics, the mechanistic approach is extensively used in representing networks of neurons, cortical circuits, thalamocortical circuits, hippocampal neurons, and so forth. These mechanistic models frequently yield detailed mathematical descriptions of neural substrates, drawing upon information from the chemical level. Consequently, the explanations and predictions produced by the mechanistic approach often provide a relatively high degree of precision. For further resources on the mechanistic approach, I recommend the textbook by Izhikevich~\cite{izhikevich2007dynamical}.

\subsection{The normative approach}
Unlike the mechanistic approach, which begins with a detailed account of the system of interest, the normative approach starts by identifying high-level principles. By postulating normative objectives, neuroscientists can then narrow down the search space of potential theories and establish guidance for discovering new phenomena.

One classic example is the \textit{efficient coding principle}, which is rooted from Barlow's redundancy-reduction hypothesis~\cite{barlow1961possible}. Conceptually, the sensory system can be simplified as a communication channel for transmitting useful information from sensory areas to other down-stream systems for use. How should the sensory system \textit{encode} various input stimuli? It can be imagined that if this were considered an engineering problem, there would be many possible implementation methods. However, using the visual system as an example, neuroscientists have discovered similar \textit{receptive field}~\footnote{The receptive field of a neuron refers to the \textit{field} (in the physical sense, or one can think of it simply as a collection of patterns) of input stimuli that is highly correlated with the activation of the neuron.} in model organisms on different evolutionary paths. Indeed, evolution seems to have consistently chosen specific ways to implement the brain's sensory systems. The efficient coding principle offers a concise answer to this big question: the sensory system will use the ``most efficient way'' to encode information. See~\autoref{sec:efficient coding} for more on efficient coding principle.

Contrary to the mechanistic approach, the normative approach utilizes a top-down perspective, and hence might tie less with strong biophysiological evidences. However, the normative approach offers a significant advantage in its capacity to provide a more succinct and unifying understanding of a neural substrate's functional role. This allows researchers to identify feasible mechanisms and anticipate the higher-level computations that may emerge from the integration of these lower-level mechanisms. Subsequently, the insights derived from this normative understanding can be channeled into other research methodologies to devise experiments for testing the proposed high-level computational principle. They can also serve as foundational stepping stones for a modular understanding of the brain on a grander scale. One example is the beautiful connection between dopamine and reward prediction error (from reinforcement learning), see~\autoref{sec:dopamine} for more details.

\subsection{The computational modeling approach}
As the development of computing power and resources surges, the research rationale of using computational models to build up understanding and theories has influenced many scientific fields, and neuroscience is no exception. Here, computational models broadly refer to mathematical models that can be simulated on a computer~\footnote{The computational model here refers to something totally different from the computational model in the context of theoretical computer science. There a computational model refers to a mathematical model that can perform computation using certain resources and basic rules. And the goal is to capture the computability and computational efficiency etc.}. These models could either be hand-crafted or coming from data analysis. Namely, most mechanistic models are also considered computational models. Similarly, scientists hope that these computational models can accurately reproduce the essential phenomena observed in the subjects of interest.

On the other hand, a computational model may contain many more parameters and often requires fine-tuning. In particular, with the widespread adoption of deep learning in recent years, many scientists now use trained artificial neural networks for modeling purposes. As the search space for model structures and parameters grows boundlessly, researchers turn to the aforementioned normative method to narrow down their choices.

It seems that the computational modeling approach is simply an intersection of the mechanistic approach and the normative approach. What really makes it stand out from merely an intersection of these two previous approaches?

One crucial advantage of computational modeling is its ability to produce useful inferences grounded in data and biology. As seen in many-body physics, traditional mechanistic models can quickly become computationally intractable. Meanwhile, although being able to provide high-level reasoning, normative models often lack justification for implementation details. Computational modeling overcomes these challenges by building up efficiently simulatable models that can serve as \textit{artificial organisms} for thorough examination and analysis.

\begin{examplebox}{Artificial neural networks for neuroscience}
Artificial neural networks broadly refer to computational models that mimic or are inspired by the structure of the biological brain. In recent years, by continually deepening and expanding artificial neural networks, we have achieved previously unimaginable successes in engineering (e.g., Go, protein structure, image recognition, language models, etc.). Can these exciting developments, in turn, advance our understanding of our own brains?

Beyond the direct use of artificial neural networks to assist in data analysis (as shown in~\autoref{fig:ann}(a)), some researchers have started using certain artificial neural networks as models for certain systems for quantitative analysis (as shown in~\autoref{fig:ann}(b)). Another common approach is to use artificial neural networks to learn the input-output patterns of real brain neural data, and then perform analysis on the learned models (as shown in~\autoref{fig:ann}(c)). I recommend the survey papers by Richards et al.~\cite{richards2019deep} an Yang et al.~\cite{yang2020artificial} for further reading.

\vspace{3mm}
\begin{center}
\includegraphics[width=12cm]{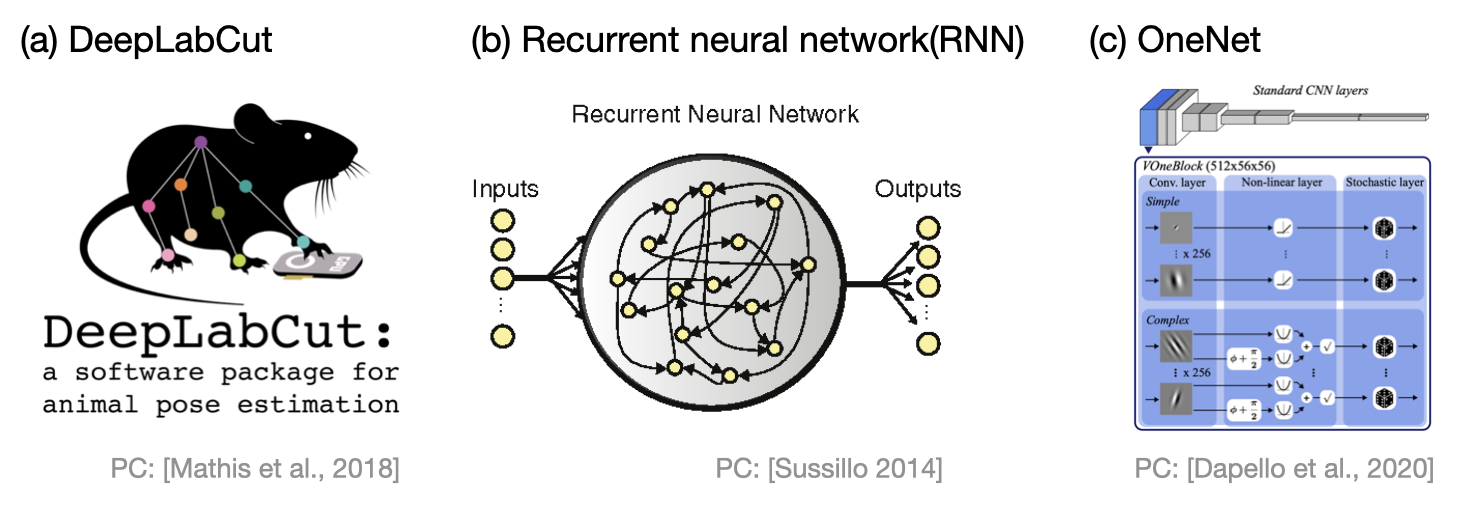}
\captionof{figure}{Artificial neural networks for neuroscience. (a) DeepLabCut~\cite{mathis2018deeplabcut} is a software developed using pre-trained artificial neural networks, which can help experimental scientists annotate features on animal exteriors, such as noses, limbs, etc. (b) In recent years, Recurrent Neural Networks (RNNs) have been commonly used in numerical simulations to depict the neural network within a brain region~\cite{sussillo2014neural}. Researchers would use the RNN trained through input-output to serve as the mathematical model of this brain area for further analysis. (c) The neural networks in perception systems usually have less recurrent connectivity, making it easier to clearly record and analyze neuronal activity in a hierarchical manner. The visual system is an especially popular research direction. In the study shown in the figure by Dapello et al.~\cite{dapello2020simulating}, the input and output neuronal activity of brain area V1 are used as training data for a specially designed feedforward artificial neural network. The researchers claim that this artificial neural network, which mimics the brain's V1 area in certain aspects, performs better than traditional models.}\label{fig:ann}
\end{center}
\end{examplebox}

\section{Examples of computations in neuroscience}\label{sec:neuro computations}
Now that we have had a glimpse into neuroscience, we are ready to discuss the computational lens in neuroscience. The reader might wonder if the computational modeling approach mentioned in the previous section already represents the marriage of computer science and neuroscience. The purpose of this thesis is to convince you that there is much more to be gained from the computational perspective. In this section, we are going to see several examples of computations in neuroscience. Instead of building concrete models or theories, we will focus on extracting out the underlying algorithmic and computational aspects via language from computer science. 

It is important to note that the examples presented here are far from exhaustive. I believe that there is a wealth of additional insights yet to be discovered in neuroscience via the computational lens. The fruition of such discoveries would require effective and respectful communication across the various sub-disciplines in the field.

\subsection{Sensory processing: an example in visual systems}
Retina is the brain of the eyes, where the first layer, the photoreceptor cell, plays a role similar to the photosensitive component in a camera, with each photoreceptor cell dutifully monitoring a small corner of the world. However, when delving deeper into the neural network of the retina, or even entering the vision-related areas of the cerebral cortex, things get even more interesting: neurons no longer just pay attention to the strength of light in a small corner, but start to respond to specific patterns in a certain area. Just like the guards in a prison, each watching over the movements in a certain area, and only reacting to suspicious disturbances. For a neuron, the region and the corresponding pattern that make it particularly responsive are called its \textit{receptive field} (as shown in~\autoref{fig:vision 1}(b)).

\begin{figure}[h]
    \centering
    \includegraphics[width=12cm]{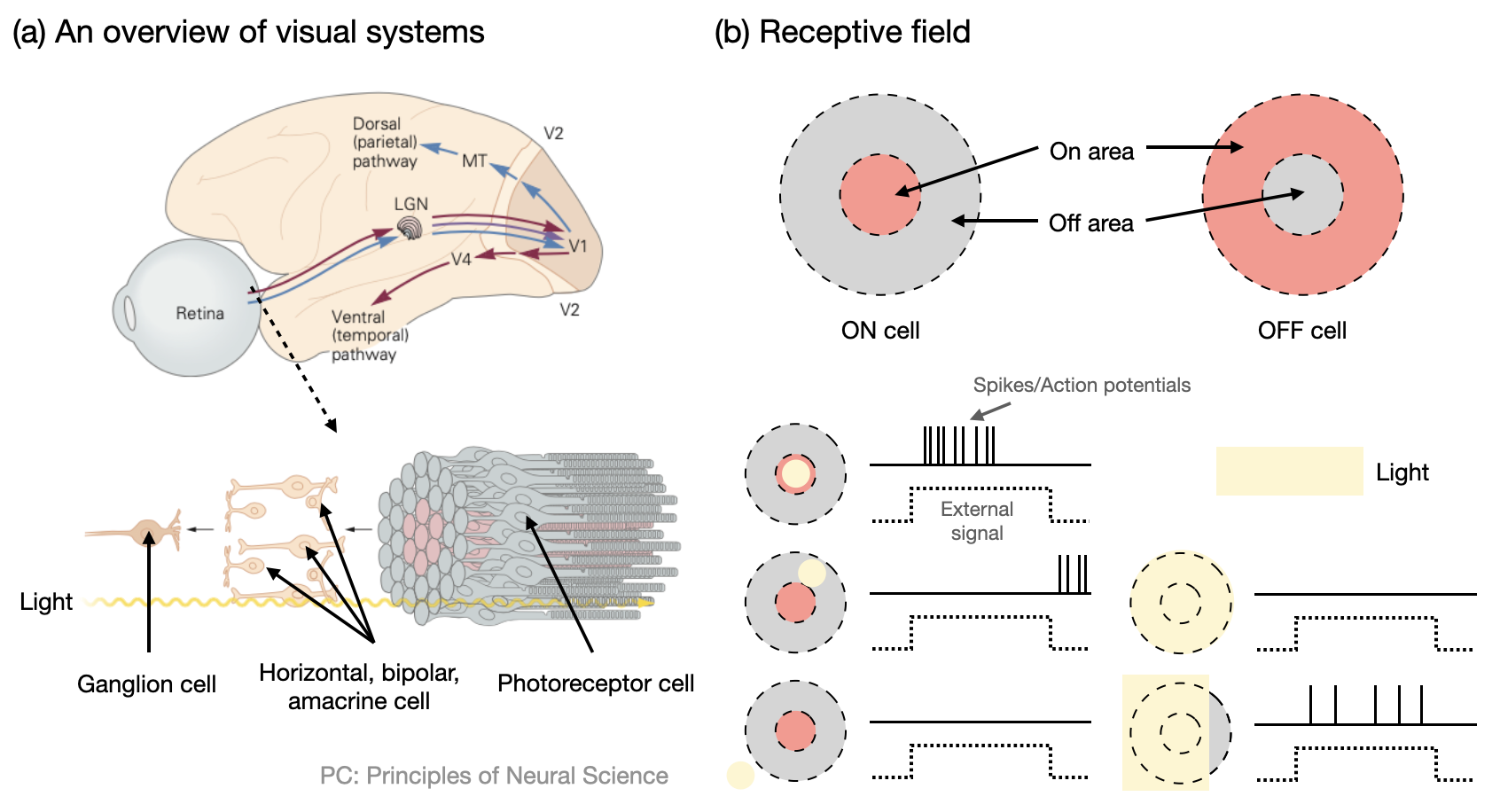}
    \caption{(a) Light triggers signals in the retina, which are then transmitted to the LGN (Lateral Geniculate Nucleus) in the thalamus, and then to the primary visual cortex in the cerebral cortex, and continue to be transmitted along two different paths. In the retina, light first stimulates photoreceptor cells, then proceeds through horizontal, bipolar, and amacrine cells, and finally the ganglion cells converge the signal to the LGN. Note that each type of neuron here has many different subcategories. (b) The receptive field of ganglion cells usually presents two concentric circles, with the central and peripheral areas each having opposite responses. Taking ON cells as an example, when light shines in the middle part, it will stimulate the neuron's potential, but when light shines on the periphery, it will inhibit. The five examples below the figure explain the response of ON cells in different situations.}
    \label{fig:vision 1}
\end{figure}

The first area related to vision in the cerebral cortex is called the \textit{primary visual cortex (V1)}. In 1959, David Hubel and Torsten Wiesel conducted their famous experiment~\cite{hubel1959receptive}, tirelessly showing different visual images to cats in an attempt to clarify the receptive fields of neurons in V1. Initially, they tried to place spots of light in different positions to see if the neurons would respond to specific locations, but things did not go as planned. It was not until later that they discovered many V1 neurons had a special response to a line-shaped light source. Moreover, some neurons reacted violently when a line moved in a certain direction at a certain angle. In other words, the receptive fields of V1 neurons are line-shaped and may even be related to the time axis! This important discovery, along with their subsequent significant contributions to the visual system, unsurprisingly earned them the Nobel Prize in Physiology or Medicine in 1981.

In V1, each neuron may pay attention to different angles, and even the same neuron may emit some pulses for angles close to the receptive field. Therefore, the concept of the \textit{tuning curve} was born: for a certain neuron (for example, a neuron in V1), plot a function of the corresponding pulse frequency based on its parameters for different input stimuli (such as the angle of light stimuli). In this way, the tuning curve can show how a neuron \textit{encodes} a certain angle (as shown in~\autoref{fig:vision 2}(b)).

\begin{figure}[h]
    \centering
    \includegraphics[width=12cm]{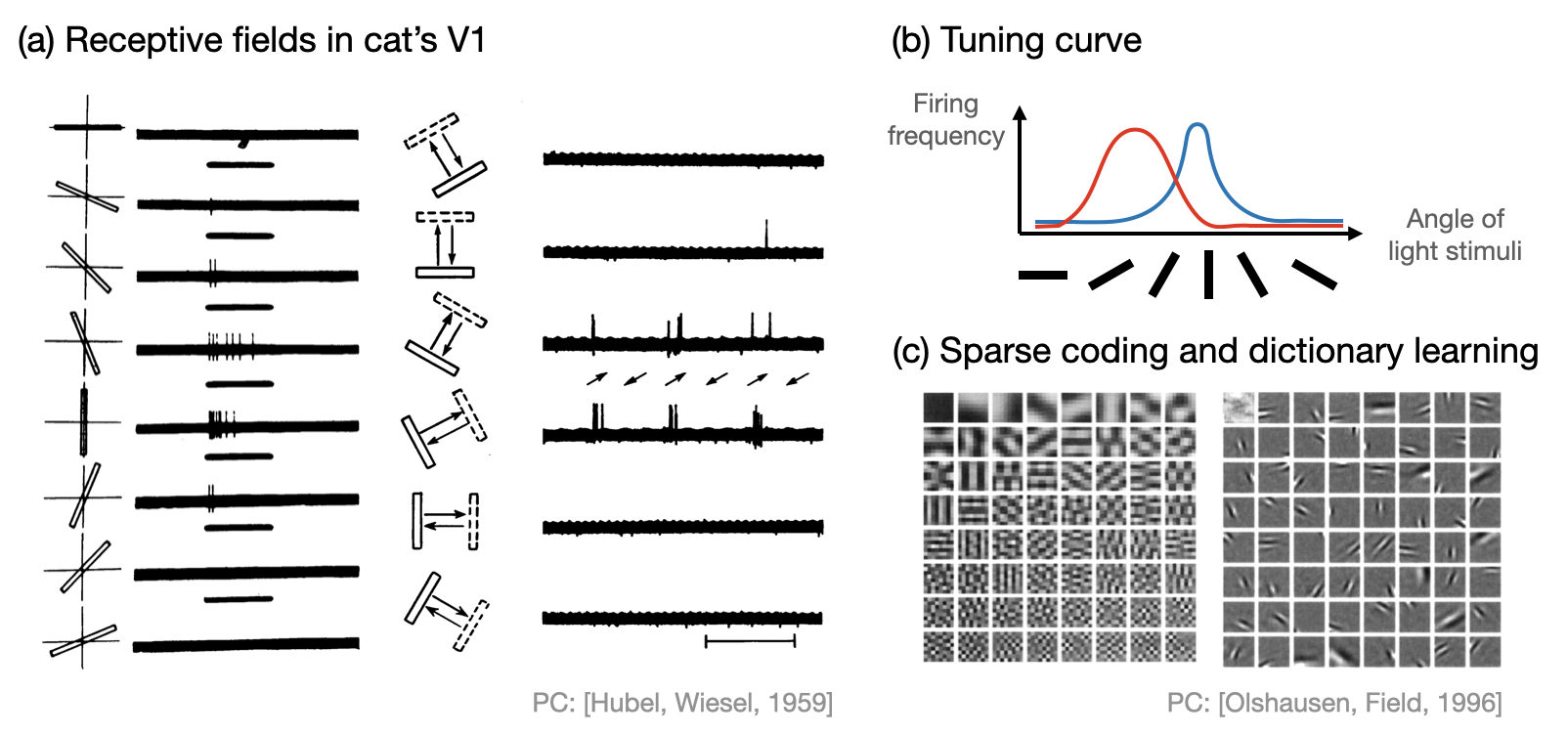}
    \caption{(a) The receptive fields of neurons in V1 are usually linear. In the left figure, for instance, when the light source line is in a vertical direction, the observed neuron has the most frequent pulses. Some V1 neurons' receptive fields are related to the direction of movement. In the right figure, for instance, when the light source line moves to the right at a 30-degree angle, the observed neuron has the most frequent pulses. (b) Tuning curves visualize the degree of a certain neuron's (in this case, there are two in red and blue) response to input stimuli of different directions. This allows for a comparison of whether different neurons cover all possible input stimulus directions. (c) Olshausen and Field used numerical experiments to attempt to explain the shape of receptive fields. They performed a principal component analysis on some image inputs and obtained receptive fields like the one in the left figure. However, once a sparse condition is added, the shape of the receptive field (like the right figure) will resemble the linear shape observed in V1.}
    \label{fig:vision 2}
\end{figure}

In the sensory system, neurons, especially those closer to the input end, can be observed to have a clear receptive field and their tuning curves can be plotted. With these concepts, we can further speculate from a computational perspective how these neurons and related brain areas handle sensory input: the receptive field of each neuron is like a word in a dictionary. When the eyes receive an image, the responsive visual neurons are like telling the next brain area what words appear in this image. So what vocabulary and language does the brain use to communicate?
For example, in 1996, Bruno Olshausen and David Field introduced the concept of \textit{sparse coding}~\cite{olshausen1996emergence}, trying to argue through numerical experiments that the shape of the receptive field is due to the brain's desire for sparser neuron responses (as shown in~\autoref{fig:vision 2}(c)). In terms of vocabulary and language, it is equivalent to saying that they hope to express what they see with just a few simple words. This idea also indirectly opened up the development of \textit{dictionary learning} in computer science.

\subsection{Synaptic plasticity and learning}\label{sec:synaptic plasticity}
Changes in the strength of synapses between neurons are known as synaptic plasticity. In addition to the previously mentioned Long-Term Potentiation (LTP) and Long-Term Depression (LTD), where neurons fire together wire together, neuroscientists have observed many other mechanisms of synaptic strength changes (for example, Spike-Timing-Dependent Plasticity (STDP) shown in~\autoref{fig:synaptic plasticity}(a)). Theoretical neuroscientists have also proposed numerous mathematical models in an attempt to explore the relationship between synaptic plasticity and the computations performed by neural networks (as in~\autoref{fig:synaptic plasticity}(b)).

\begin{figure}[h]
    \centering
    \includegraphics[width=12cm]{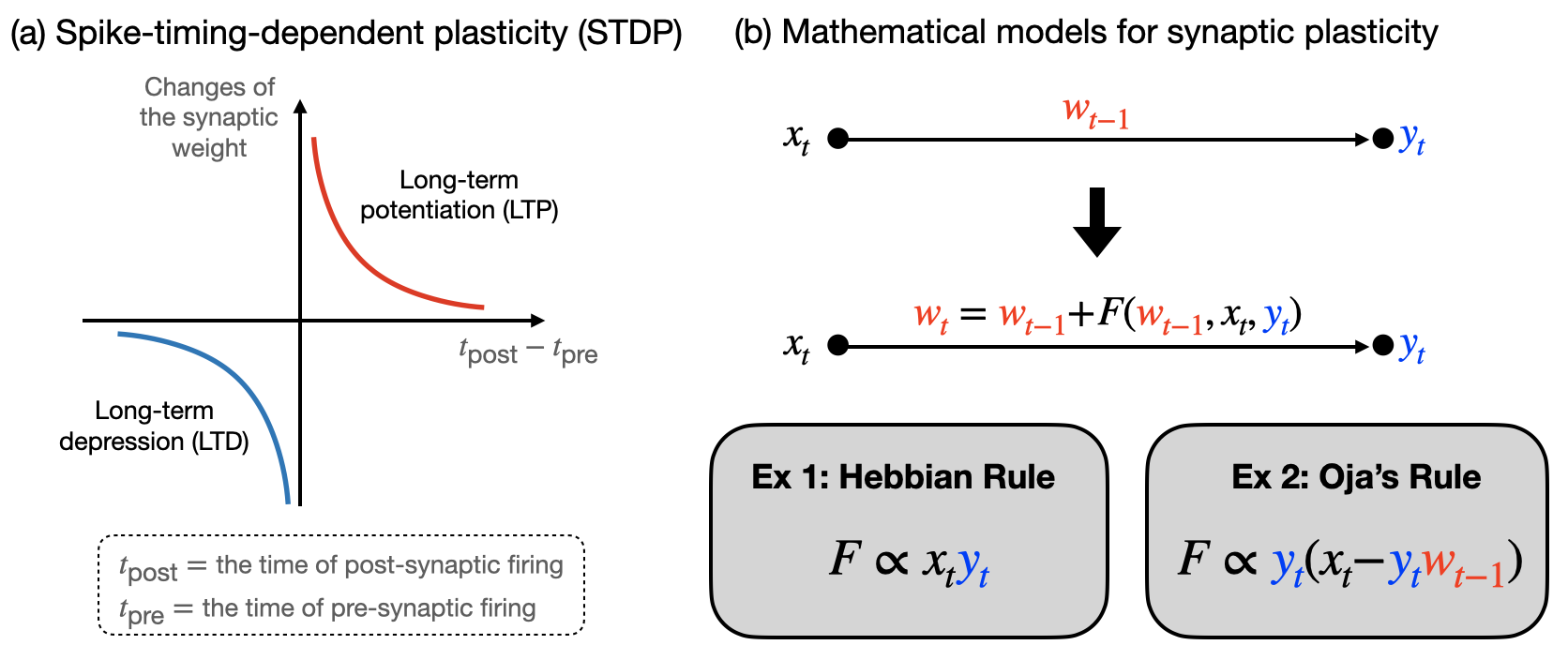}
    \caption{Synaptic plasticity and learning. (a) Spike-Timing-Dependent Plasticity (STDP). When the spike of the postsynaptic neuron closely follows the spike of the presynaptic neuron, the strengthening of the synapse will be greater. Conversely, if the postsynaptic neuron fires before the presynaptic neuron, the strength of the synapse will decrease. (b) Simplified models of synaptic plasticity. In mathematical terms, the activity of the presynaptic neuron (such as the average number of spikes) is marked as $x$, and the activity of the postsynaptic neuron is marked as $y$. Then, the change in synaptic strength w will be related to $x$, $y$, and the previous $w$. Two common synaptic plasticity rules are the Hebbian rule and the Oja's rule. The former corresponds to the idea that if two neurons simultaneously exhibit the same response, the synaptic strength will be enhanced. The latter adds a homeostasis term to control the synaptic strength from becoming too large. Both are linked with certain computational problems, such as Principle Component Analysis (PCA).}
    \label{fig:synaptic plasticity}
\end{figure}

If we consider not just the connection between one neuron and others, but how the synaptic strength between any two neurons in the entire neural network changes, how do we establish and understand the relationship between synaptic plasticity and computation? On an abstract level, since a neural network is like a computing process (for example, mapping sensory input to action output), synaptic plasticity can be seen as a ``computation that changes the computation process''. Therefore, it is naturally linked with learning. Will designing and analyzing different synaptic plasticity rules and their potential/corresponding computational problems/principles allow us to have a more modular understanding of the underlying implementations in the brain?

In the world of artificial neural networks, \textit{backpropagation} is a remarkably successful artificial synaptic plasticity rule. Its main concept is to use the chain rule of calculus to calculate how the strength of each synapse should change to minimize the overall loss with respect to a certain objective function. Although backpropagation has led to unprecedented advancements in artificial neural networks and artificial intelligence, due to biological constraints (for example, specific synaptic connection ways), neuroscientists generally believe that the way the brain learns should be somewhat distanced from backpropagation. How exactly do our brains learn? Could a thorough understanding of synaptic plasticity help us clarify the essence of learning?

\subsection{Navigation}

Imagine you've arrived in a new city. Your phone hasn't connected to the internet, so you can't use online maps. Since you can't speak the language, you can't ask for directions or buy a map. However, staying in the hotel is too boring, so you decide to go out and wander, thinking that with the help of the sun, you should be able to return to the hotel before it gets dark. Walking on unfamiliar streets, what you see are shop signs that are as incomprehensible as hieroglyphs. Although you can't quite figure out what they are selling, you're beginning to recognize some of the recurring chain stores. Maybe the one with the bright red sign is a convenience store, pharmacy, or fast-food restaurant?

Not just for tourists in foreign lands, navigation is also crucial for animals living in the natural world. How does the brain guide ants to find their way home after going out for food? Monarch butterflies migrate from Mexico to North America over several generations, how do they do it? What computational principles are supporting these processes?
Let's peek into the labyrinthine world of the brain through the \textit{Place cells} and \textit{Grid cells} recognized by the 2014 Nobel Prize in Physiology or Medicine.

\begin{figure}[h]
    \centering
    \includegraphics[width=13cm]{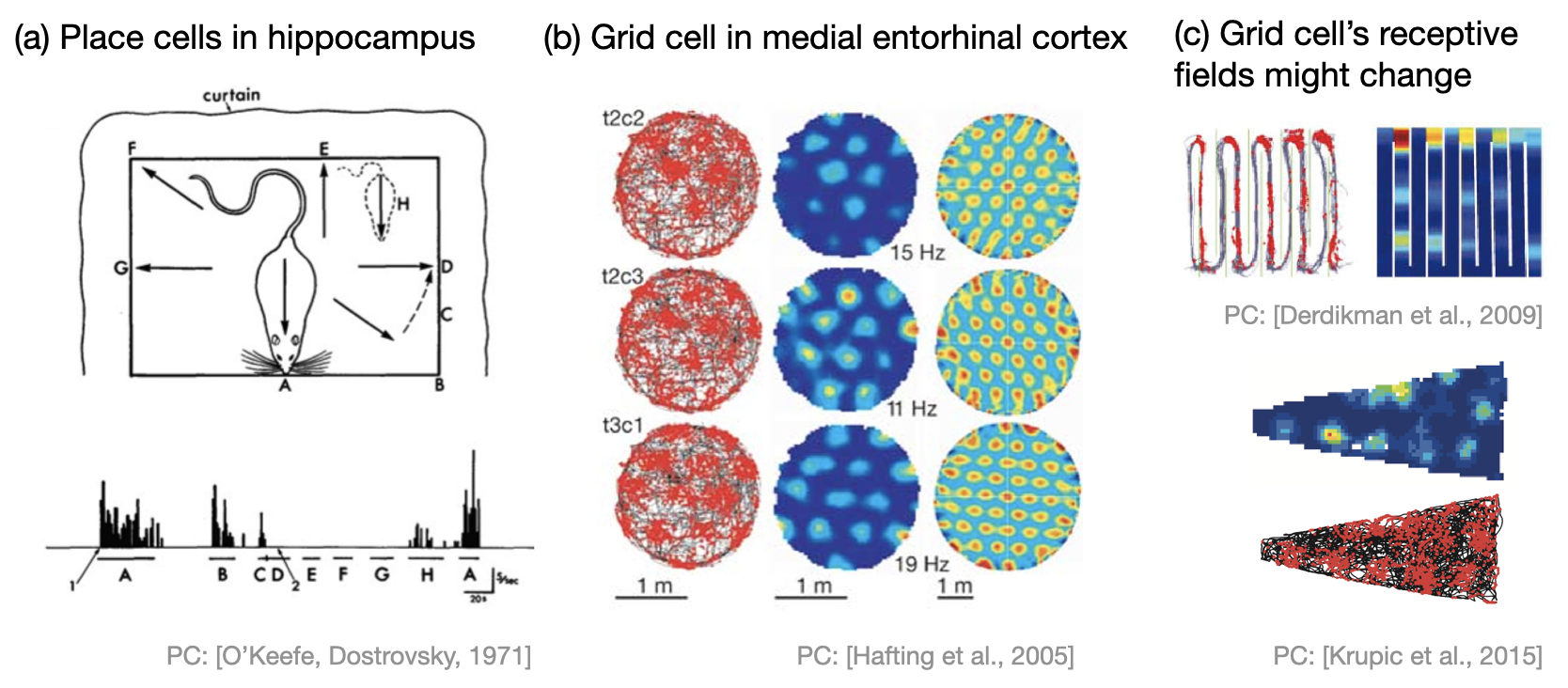}
    \caption{(a) O'Keefe and Dostrovsky discovered in their experiments that certain neurons in the hippocampus of mice only generate impulses when the mouse passes through a specific location. For instance, the number 1 neuron in the picture only has significant activity when the mouse is at location A and location B. (b) Mosers' laboratory found that many neurons in the intermediate dorsal part of the entorhinal cortex have receptive fields at grid points. In the diagram, each horizontal axis corresponds to one neuron. The far-left column in the vertical axis shows the position of the neuron's impulses (red dots) when the mouse moves in a circular space. The middle column is the strength of the receptive field, which is the frequency of the neuron's impulses when the mouse is at different locations. The far-right column is the autocorrelation function of neuron activity in space, which here shows a strong positive correlation in the neuron activity at grid points.}
    \label{fig:navigation}
\end{figure}

John O’Keefe is a neuroscientist trained in psychology. Like most psychologists, O'Keefe places great emphasis on the role of animal behavior. Therefore, when he shifted his research from the amygdala, where he had been working on during his PhD, to the hippocampus, his primary focus was on what abnormal behaviors mice with hippocampal damage would display. Quickly, he and his students found that these mice performed particularly poorly on spatially related tasks, especially when moved to new environments. Through more in-depth experiments, they discovered some special neurons in the CA1 region of the hippocampus that had strong responses when the mouse passed through a particular position (as shown in~\autoref{fig:navigation}(a)). These types of neurons were thus named \textit{place cells}, opening up the exploration of the role of the hippocampus in navigation and memory in neuroscience.

Edvard Moser and May-Britt Moser are the fifth couple to jointly win the Nobel Prize. After completing their doctorates, the couple spent some time learning under O’Keefe in London before returning to their home country of Norway to establish their own laboratory. Inspired by the concept of place cells, the two Mosers decided to investigate how place cells were formed. From anatomical data of neuronal connections, they and their students tried related brain areas, finally discovering in 2005 that many special neurons in the entorhinal cortex responded extremely regularly based on the position of the mouse. If the mouse is in a square space, the receptive field of this kind of neuron relative to the mouse's position would approximate the grid points of a hexagonal grid (as shown in~\autoref{fig:navigation}(b)), hence the name \textit{grid cells}.

So can we imagine navigation in the brain as having grid cells establish a map's longitude and latitude, and place cells record important locations? Unfortunately, this image is a bit oversimplified, as there are always exceptions in the biological world. For instance, with place cells, experiments can observe that when the mouse changes environment, the sensed position also changes, a phenomenon called \textit{remapping}. If the shape of the environment is changed, the receptive field of grid cells will lose their hexagonal grid and become other types of grids (as shown in~\autoref{fig:navigation}(c)).

Place cells and grid cells are among the few types of neurons in the deeper areas of the brain that still have clear receptive fields, hence a lot of research work is focused on understanding their causes and characteristics. However, navigation in the brain is more like a dynamic map and is closely related to other functions (e.g., memory). Unraveling the functional role of place cells and grid cells within the intricate circuitry of the brain remains an ongoing quest for researchers.

\subsection{Dopamine and reinforcement learning}\label{sec:dopamine}
People often hear that the release of dopamine (DA) brings a sense of happiness. In fact, dopamine is a type of neurotransmitter in the brain. Neurons that release dopamine are referred to as dopaminergic neurons and are mainly distributed in the Ventral Tegmental Area (VTA) and Substantia Nigra (SN) in the midbrain (located in the brainstem), as well as other areas like the hypothalamus.

In the famous experiments conducted in the 1950s by James Olds and Peter Milner, they inserted electrodes into certain brain areas of mice. When the mouse pushed a rod, a current was injected to trigger action potentials in the nearby neurons~\cite{olds1954positive}. They found that when electrodes were placed in a particular area, the mouse would continually push the rod, and the attention of the mouse could not be diverted by food, water, or even the attraction of the opposite sex. This obvious behavioral correlation led them to closely study this brain area, and they found that most of the neurons here are connected to dopaminergic neurons!

So, are the sources of happiness in our brains controlled by dopaminergic neurons? It's not that simple! In the late 1990s, Wolfram Schultz and his collaborators discovered that dopaminergic neurons not only respond when receiving rewards (the sources of happiness designed in the experiment), but they also get excited once there are signs that a reward is imminent (for example, smelling the aroma of delicious food)~\cite{schultz1997neural}. Surprisingly, if the mouse doesn't receive a reward in the end, the number of impulses of the dopaminergic neurons actually decreases (as shown in the figure below)!

\begin{figure}[h]
    \centering
    \includegraphics[width=13cm]{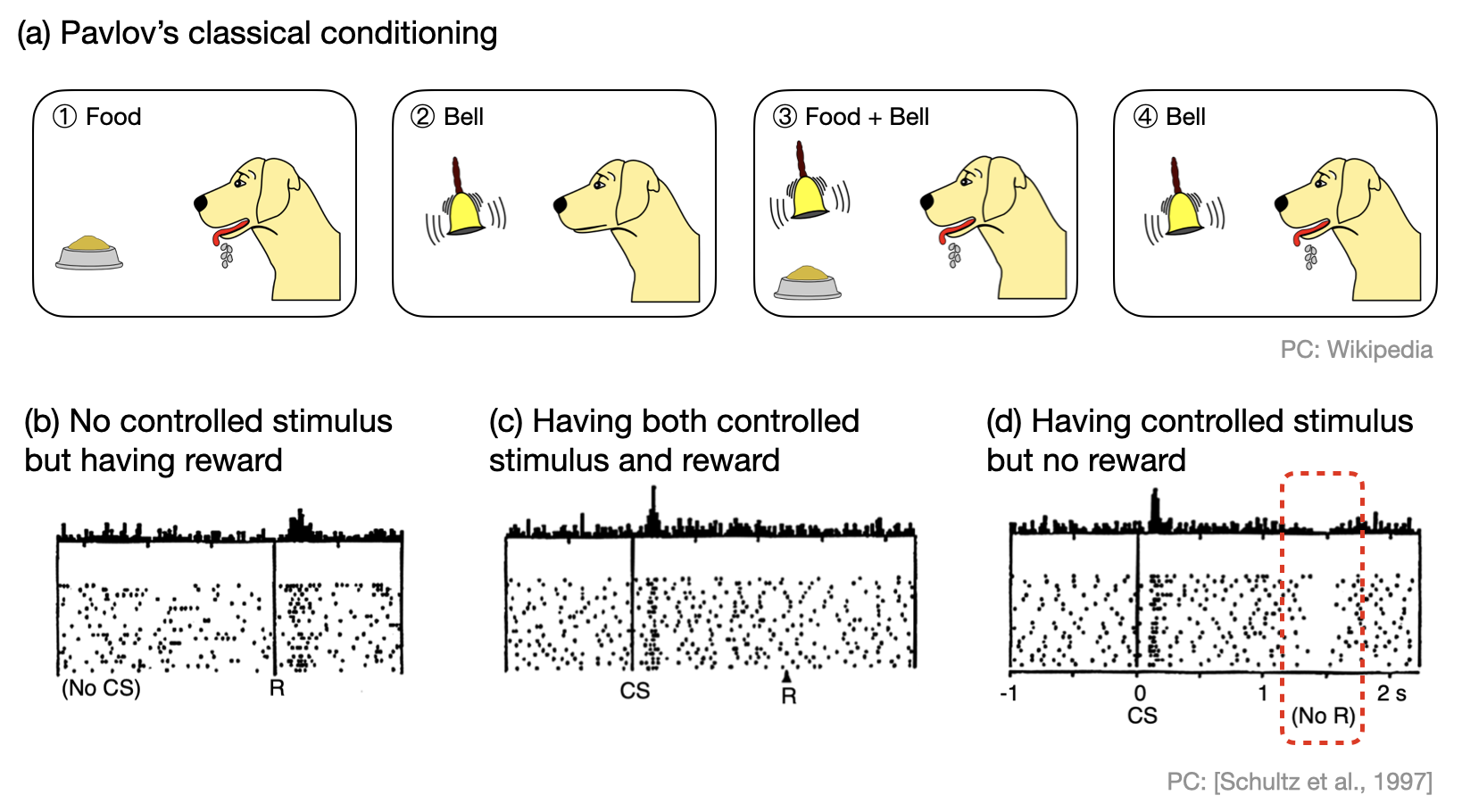}
    \caption{(a) In Pavlov's classical conditioning experiment, the reward (food in the figure) and the stimulus (the bell sound in the figure) are repeatedly presented to the test animal, causing it to develop conditioning. That is, once a stimulus is received (hearing the bell), it will produce the reaction that would occur upon seeing a reward (salivating). (b-d) These are from the monkey experiment mentioned the paper by Schultz et al.~\cite{schultz1997neural}. In the experiment, the monkeys are subjected to visual or auditory stimuli, and juice as a reward. Each row in the figure represents the impulse timing of a dopaminergic neuron over time (leftward), and the histogram at the top represents the total number of impulses. The ``CS'' at the bottom represents the moment the controlled stimulus appears, and ``R'' represents the moment the reward appears. (b) When there is no stimulus but a reward, the response of the dopaminergic neurons quickly increases after the reward appears. (c) But once the monkey has been classically conditioned, as soon as the stimulus appears, the response of the dopaminergic neurons immediately rises. Later, when the reward appears, the activity of the dopaminergic neurons does not particularly fluctuate. (d) If there is only a stimulus and no reward, the dopaminergic neurons will suddenly decrease the impulse frequency after the moment when the expected reward does not appear.}
    \label{fig:dopamine}
\end{figure}

Therefore, the activity of dopaminergic neurons is gradually considered related to \textit{reward prediction}: if expectations are met, then dopaminergic neurons will not be particularly active. Conversely, if expectations are not met, dopaminergic neurons seem to release a \textit{prediction error} signal by reducing activity.

In recent years, neuroscientists have conducted increasingly detailed research on dopaminergic neurons. For example, different signaling patterns or different cellular activities may have different functions. They have found that dopaminergic neurons also seem not to be limited to reward prediction in Pavlovian classical conditioning. Perhaps, because humans are always thinking about pursuing happiness, they simply equate dopamine with happiness. But maybe happiness does not have a single source, and it's not necessarily the end goal. When thinking about the nature of dopaminergic neurons, perhaps we will gain a new understanding of happiness?

\subsection{Efficient coding principle}\label{sec:efficient coding}
Early studies on the visual nervous system were centered around analyzing different receptive fields. As we learn about the receptive fields of more and more neurons, a half-scientific, half-philosophical question naturally arises: What are the functions of these receptive fields? Why do the receptive fields of visual neurons look like this and not something else? As Barlow pointed out in his famous 1961 paper~\cite{barlow1961possible}, if we didn't know that birds could fly, we would still be at a loss no matter how hard we studied the complex structure of their wings. What kind of computational/functional principles govern the sensory system?

In the same paper, Barlow proposed three hypotheses, trying to guide people's research and understanding of the sensory system from the perspective of \textit{computational theory} in Marr's three-level analysis~\cite{marr2010vision}. Among them, the \textit{redundancy-reduction hypothesis} has been passed down to later generations as the more general \textit{efficient coding hypothesis}. As these are top-down hypotheses, here we will follow Barlow's article and focus on abstract analogies, without discussing specific experimental examples.

Let's simplify the sensory system into a channel for transmitting information. It is constrained by some biological limitations (such as the way neurons connect, physiological conditions, noise caused by uncertain factors, etc.), and the goal is to transmit useful information to other systems (such as motion, integration, learning, memory, etc.) for use. How should the sensory system \textit{encode} various input information (such as the scenery seen, the sound heard, the touch felt)? One can imagine that if this is seen as an engineering problem, there will be many possible implementations. However, using the visual system as an example, neuroscientists have found similar receptive field shapes in model organisms on different evolutionary paths. At first glance, evolution seems to have regularly chosen some specific ways to implement the sensory system in the brain.

The answer to this big question given by the efficient coding hypothesis is simple: the sensory system will use the ``most efficient way'' to encode information.

Usually, when such an abstract grand principle is proposed, it will be seriously questioned and examined for its specific correspondences and implementations in the biological world. What does ``most efficient'' mean? What unit is used to calculate it? What is the biological significance? As the lengthy questions and debates continue, it would not be excessive to write another book. The \textit{efficient coding hypothesis} is still a hot concept in the (theoretical) neuroscience world today, but for each researcher, their understanding of the definition may be slightly different.

\section{From neuroscience to computation}
In~\autoref{sec:neuro computations}, we've seen several examples of computations drawn from neuroscience. Conversely, neuroscience has also had a profound impact on research in computer science. In this brief section, we will conclude this chapter by highlighting two notable examples of this: neuromorphic computing and artificial intelligence.

\subsection{Neuromorphic Computing}
While a single supercomputer cluster today can consume more electrical power than a small city, the human brain operates incredibly efficiently, requiring only about 20 watts of power. This stark contrast has inspired scientists and engineers to explore new computing frameworks beyond the traditional von Neumann and transistor-based paradigm. This research direction, known as \textit{neuromorphic computing}, aims to achieve advancements at three different levels of abstraction: materials, circuits, and algorithms.

At the material level, the focus is on developing electronic neurons or other physical neuron analogs that can emulate the computational power and energy efficiency of biological neurons. For instance, demonstrating the ability to use spiking and event-driven computing.
At the circuit level, research is dedicated to understanding how to integrate these so-called \textit{cold neurons} into circuits or chips to perform standard computations. Specifically, one critical challenge here is the issue of implementing distributed computing and memory.
At the algorithmic level, the goal is to design algorithms from a neuron's perspective. That is, to create algorithms that are specially tailored for achieving computational efficiency and advantages in neural networks.
For a more comprehensive introduction to neuromorphic computing, I would recommend referring to a survey paper by Schuman et al.~\cite{schuman2017survey} or a recent roadmap paper by Christensen et al.~\cite{christensen20222022}.

\subsection{Artificial intelligence}
The intertwined history of artificial intelligence (AI) and neuroscience is filled with constant inspiration and cross-pollination. Today, we marvel at the remarkable successes of deep neural networks in AI, with many of their groundbreaking and historically significant architectures being heavily inspired by their biological counterparts. For instance, the perceptron~\cite{rosenblatt1958perceptron} is modeled after neurons, while the key idea in convolutional networks~\cite{lecun1995convolutional} draws from the primate visual system~\cite{hubel1959receptive}.

But the knowledge that computer scientists have extracted from biology is not limited to the replication of brain-based structures. On a normative level, animal learning has been a wellspring of inspiration for reinforcement learning~\cite{sutton2018reinforcement}. On a system level, cognitive processes such as attention, episodic memory, working memory, continual learning, and more, have informed the development of novel architectures and algorithms.

The reciprocal relationship between AI and neuroscience continues, with each field making considerable contributions to the other. As we uncover more about the intricate workings of the brain, we can anticipate further advancements and innovations in the realm of AI. For a deeper exploration of this topic, the review article by Hassabis et al.~\cite{hassabis2017neuroscience} provides an excellent overview and serves as a valuable reference.

\section{Concluding remarks}
Neuroscience is a topic that is both related to ourselves, intriguing, and full of possibilities. For readers who haven't had much exposure to the topic before, I hope this chapter can ignite your foundation and passion for further learning. For those already in the field, I hope the arrangement and the narrative of this chapter can bring you new insights.

Although computational tools and concepts have started to be extensively used in various research in neuroscience in recent years, the perspective of ``understanding the brain through computational thinking'' is just beginning to emerge. I personally believe that ``computation'' is the most comfortable and consensual formal language for humans to engage in ``mechanical and logical thinking''. However, how to establish a \textit{computational language} and corresponding quantitative analysis techniques, experimental designs, as well as theoretical frameworks and interpretations, still have a long way to go.

\begin{savequote}[75mm]
A theory is valuable only insofar as it proposes detailed and particular mechanisms to explain a wide variety of phenomena in its domain and insofar as it stimulates new experiments.
\qauthor{Gerald M.~Edelman}
\end{savequote}

\chapter[Algorithmic Neuroscience and Emergent Computations]{Algorithmic Neuroscience and\\Emergent Computations}\label{ch:example EI}

From the previous chapter, we have seen many examples of the theoretical and computational studies in neuroscience. In this chapter, we aim to propose two different angles/methodologies through the computational lens: the algorithmic neuroscience and emergent computations.

\medskip \noindent \textbf{Algorithmic neuroscience} refers to studying and modeling systems in neuroscience as performing certain algorithms. By focusing on the algorithmic ingredients, we care less about the underlying detailed implementations (as opposed to the modelings via dynamical systems) and emphasize the computational roles and composability with other parts of the broader brain regime. Through the algorithmic angle, it might provide a more flexible analytical framework to study neural circuits as well as their compositions and interactions.

\medskip \noindent \textbf{Emergent computations} refer to the computational aspects that arise from the inductive biases of biological constraints, underlying tasks, evolutionary relics, and input structures. By building up a dictionary between these inductive biases and computations, it might bring up clearer high-level recipe for the intuitions to connect the cognitive side of neuroscience to the bottom level.

In the rest of this chapter, we will first see a concrete example of how the algorithmic perspective can shed light on a simple and classic model from theoretical neuroscience, and lead to the discovery of a emerging computation. In~\autoref{ch:conclusion}, we will also see a proposal for future research agenda along this line of thoughts.

\section{E/I balanced neural networks}\label{sec:E/I balanced}
Neurons in many brain areas exhibit irregular neural activities. For example, the firing patterns in the central nervous systems share similar statistics (e.g., interspike interval distribution) to that of a Poisson process~\cite{abeles1991corticonics,bair1994power,softky1993highly}. Meanwhile, some controlled experiments found that the firing pattern of a cortical neuron becomes regular when it is stimulated by a constant current~\cite{mainen1995reliability,holt1996comparison}, suggesting that the irregularity might carry information instead of purely being noisy. These experimental observations raise several questions regarding the irregularity of neural coding at the population level:
\begin{itemize}
\item \textbf{(The implementation problem)} What are the underlying mechanisms and implementations in the brain that lead to the generation of irregular spike trains?
\item \textbf{(The coding problem)} How do neurons encode and decode information via such a seemingly unreliable manner? What is the coding efficiency?
\item \textbf{(The computation problem)} What computation can be accommodated by these irregular firing activities? What would be the computational advantage over regular neural code?
\end{itemize}

The balance of excitatory and inhibitory neural signals, referred to as E/I balance, serves as a conceptual framework for investigating the aforementioned research questions. This concept originated from two commonly observed anatomical implementations in the brain, particularly within cortical areas. Firstly, Dale's law~\cite{kandel1968dale,strata1999dale} postulates that a neuron releases only one type of chemical neurotransmitter at all its synapses, either exciting or inhibiting the connected neurons. Secondly, the cortex predominantly consists of pyramidal cells as the main excitatory neurons, while multiple types of inhibitory interneurons exist~\cite{markram2004interneurons}, contributing to the modulation of neural activity. Supported by numerous experimental findings~\cite{haider2006neocortical,xue2014equalizing,tan2009balanced,wilent2005dynamics}, it is widely believed that the neural signals from excitatory neurons are delicately balanced by those from inhibitory neurons.

From a theoretical standpoint, E/I balance has served as both a biological constraint and a guiding principle for researchers investigating the enigma of irregular spiking activities in neural networks. Numerous efforts have been devoted to developing mathematical models and dynamics that phenomenologically resemble experimental findings. Specifically, there are two common approaches: (i) the normative approach, which employs computational objectives to derive biologically plausible E/I balanced networks~\cite{boerlin2013predictive,bourdoukan2012learning,barrett2013firing}; and (ii) the dynamical approach, which utilizes mathematical models such as randomly connected recurrent neural networks to explore the parameter regimes that yield E/I balance and Poisson-like output statistics~\cite{renart2010asynchronous}. These endeavors have shed light on the potential roles of E/I balance in the brain, including predictive coding~\cite{boerlin2013predictive}, local supervised learning~\cite{bourdoukan2015enforcing}, working memory~\cite{boerlin2011spike}, neuron loss~\cite{barrett2016optimal}, and more. However, both approaches have certain limitations. In the normative approach, the derived models often require fine-tuning and do not anatomically align with experimental observations. In the dynamical approach, while fine-tuning issues are less of a concern, the computational perspective remains relatively understudied. For readers interested in delving deeper into the theoretical aspects of E/I balance, we recommend the survey by Den\`{e}ve and Machens~\cite{deneve2016efficient}.

\section{An algorithmic investigation on E/I balanced SNNs}\label{sec:SNNs}
To bridge the normative and dynamical approaches in studying E/I balance, we advocate the usage of an algorithmic lens to discover underlying emergent computations. In this section, we will explore an example from a joint work with Chung and Lu~\cite{CCL19}. This example showcases how sparse computation emerges as a result of the biological constraint of E/I balance and a normative principle of energy conservation. Through this example, we aim to provide the reader with a glimpse of how the computational lens can potentially offer new insights into neuroscience, enticing further exploration and understanding.

We start with introducing the integrate-and-fire model~\cite{lapicque1907recherches,burkitt2006review1,burkitt2006review2} for biological neural networks in~\autoref{sec:IAF}.
Next, in~\autoref{sec:optimal SNN} we narrow down the focus to the \textit{optimal E/I balanced SNNs}, which is normatively derived from a computational objective by Barrett et al.~\cite{barrett2013firing}. Finally, we present our results on the emergent computation in optimal E/I SNNs in~\autoref{sec:SNNs results} and provide an overview on the algorithmic and computation aspects in~\autoref{sec:SNNs algorithm} and~\autoref{sec:SNNs emergence} respectively.

\subsection{Integrate-and-fire neural networks}\label{sec:IAF}
Let $m$ be the dimension of the input stimuli and $n$ be the number of integrate-and-fire neurons in the E/I balanced networks. Each neuron is associated with a membrane potential that vary over time and together this forms a potential vector $\bv(t)\in\Real^n$ for every time $t\geq0$. The dynamic of the potential vector is governed by the following differential equation
\begin{equation}\label{eq:SNN continuous}
\frac{d \bv(t)}{dt} = -\tau \bv(t) -\Omega \bs(t) + I
\end{equation}
where $\tau$ is a leaky parameter, $\Omega\in\Real^{n\times n}$ is the recurrent connectivity matrix, $\bs(t)$ is the spike train at time $t$, and $I\in\Real^m$ is the external input current. Concretely, the spike train of neuron $i$ is defined as $s_i(t)=\sum_{t_i^{(j)}<t}\delta(t-t_i^{(j)})$ where $\{t_i^{(j)}\}_{j}$ are the times when the potential of neuron $i$ exceeds the firing threshold. Note that in the setting above, most parameters (e.g., $\Omega,I$) are time-independent, i.e., the network is \textit{static}. It is of great interest to study the \textit{dynamic} setting where the network parameters might change over time due to learning. Nonetheless, here we focus on the static scenario and leave the question of dynamic setting for future exploration.

For convenient (or complication, depending on your background), one can also consider the following discrete dynamic where here the time $t$ is indexed by non-negative integers.
\begin{equation}\label{eq:SNN discrete}
\bv(t+1) = -\tau \bv(t) - \Omega \bs(t) + I\cdot\Delta t
\end{equation}
where $\bs(t)$ is the indicator vector of whether a neuron fire a spike at time $t$ and $\Delta t>0$ is the discrete time-step. Concretely, $s_i(t)=1$ when $v_i(t)>1$; otherwise, $s_i(t)=0$. The following is an example of a network with $m=2$ and $n=5$.

Finally, the firing rate, i.e., the average number of spikes, is defined as $\br(t) = \int_0^t\bs(t')dt'/t$ in the continuous case and $\br(t) = \sum_{t'=0}^{t-1}\bs(t')/t$ in the discrete case. In the following, we provide a simple example for the reader to get familiar with the setup.

\begin{examplebox}{Example (An integrate-and-fire spiking neural network).}
Let us consider a network of three neurons with the following parameter settings. First, all of them have $\tau=0$, i.e., non-leaky, and the firing threshold to be $1$. Next, the connectivity matrix $\Omega$ and the external input current vector $I$ are
\[
\Omega = \begin{pmatrix}
    1 & 0 & \sqrt{2}/2 \\
    0 & 1 & \sqrt{2}/2 \\
    \sqrt{2}/2 & \sqrt{2}/2 & 1
\end{pmatrix} \ ,\ I = \begin{pmatrix}
    1\\ 2\\ 3\sqrt{2}/2
\end{pmatrix} \, .
\]
Finally, we numerically simulate this integrate-and-fire SNN and plot the membrane potential and firing rate of each neuron in~\autoref{fig:SNNs example potential}.
\end{examplebox}

\begin{figure}[h]
    \centering
    \includegraphics[width=13cm]{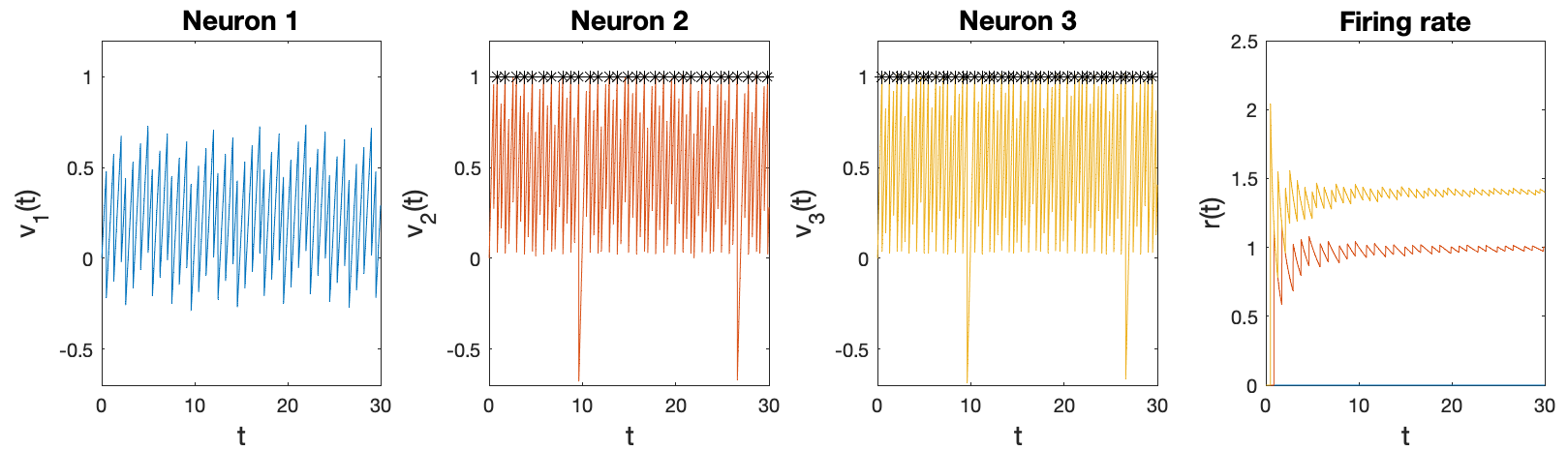}
    \caption{A three-neuron integrate-and-fire spiking neural network. The parameter setting is described in the above example box. The figure plots the membrane potential and firing rate of each neuron and the black asterisk indicates the timing when the neuron fires a spike.}
    \label{fig:SNNs example potential}
\end{figure}

As a natural summary statistics of a SNN, it is of great interest for neuroscientists to analyze the firing rate of a given network. Of course, one can always simply simulate the SNN and empirically calculate the firing rate. But how fast would such an empirical estimation of firing rate converge? Moreover, is there any analytical connection between the firing rate and the network parameters?

\subsection{Optimal E/I balanced networks}\label{sec:optimal SNN}
From~\autoref{eq:SNN continuous} and~\autoref{eq:SNN discrete}, one can sense that a spiking neural network is in general hard to analyze due to the non-linearity from the activation rule. Specifically, as there are so mcuh possibilities of network parameters (i.e., $D,F,\Omega$), the dynamic of a spiking neural network is in principle analytically intractable. Nevertheless, in a work of Boerlin, Den\`{e}ve, and Machens~\cite{boerlin2011spike}, they derived, under the normative assumption of predictive coding, the network connectivities should follow certain algebraic relation, which they termed the optimal E/I balanced condition.

\begin{definition}[Optimal E/I balanced~\cite{boerlin2011spike}]
Let $F\in\Real^{n\times m}$ be the feedforward synaptic weight, let $C\in\Real^{n\times n}$ be the recurrent connectivity, and let $D\in\Real^{n\times m}$ be the decoding synaptic weight. We say this spiking neural network is optimal E/I balanced if $\Omega=FF^\top$ and $D=F^\top$.
\end{definition}

So now the dynamics of the SNNs, i.e., ~\autoref{eq:SNN continuous} and~\autoref{eq:SNN discrete}, become
\begin{equation}\label{eq:optimal SNN continuous}
\frac{d \bv(t)}{dt} = -\tau \bv(t) -FF^\top \bs(t) + F\bx
\end{equation}
for the continuous case and
\begin{equation}\label{eq:optimal SNN discrete}
\bv(t+1) = -\tau \bv(t) - FF^\top \bs(t) + F\bx\cdot\Dt
\end{equation}
for the discrete case where the external input current is modeled as $I=F\bx$.

Furthermore, in a subsequent work, Barrett, Den\`{e}ve, and Machens~\cite{barrett2013firing} found that when the SNN is optimally E/I balanced, then there is a analytical expression for the firing rate of the network. Concretely, they used a normative theory to derive the following optimal E/I balanced condition.

\begin{proposition}[Firing rate prediction in optimal E/I balanced SNNs\cite{barrett2013firing}]\label{prop:BDM}
In an optimal balanced SNN, the firing rate will converge to minimize $E(\br)=-\br^\top FF^\top \br - 2\br^\top F\bx$ under the constraint of each coordinate of $\br$ being non-negative, i.e.,
\[
\lim_{t\to\infty} E(\br(t)) = \min_{\br\geq0}E(\br) \, .
\]
\end{proposition}

Finally, the following is an example of optimal E/I balanced SNN.

\begin{examplebox}{Example (An optimal E/I balanced SNN).}
Let us consider the same SNN described in the previous example. It is not difficult to check that it is actually an optimal E/I balanced SNN with
\[
F    = \begin{pmatrix}
    1 & 0  \\
    0 & 1  \\
    \sqrt{2}/2 & \sqrt{2}/2
\end{pmatrix} \ ,\ \bx = \begin{pmatrix}
    1\\ 2
\end{pmatrix} \, .
\]
\end{examplebox}

\subsection{Our results}\label{sec:SNNs results}
Barrett, Den\`{e}ve, and Machens~\cite{barrett2013firing} made a good first step in analytically characterizing the firing rate of optimal E/I balanced networks. They numerically supported the above proposition in several toy networks and sketched the underlying physical intuition via the concept of tight E/I balanced~\footnote{In the theory of E/I balance, there are two types of balancedness that have been studied: (i) the loose balance where the inhibitory signal only cancel out the slow time scale part of the excitatory signal and (ii) the tight balance where the inhibitory signal tightly track the excitatory signal with only a tiny time shift. The phenomenology, functionality, and dynamic of these two types of E/I balance are quite different and it is of both theoretical and experimental interest to explore deeper into their computational aspects.}. Meanwhile, from the computer science perspectives, their observation left out several interesting research questions:

\begin{enumerate}
    \item Would the proposition hold for all the optimal E/I balanced networks?
    \item How fast is the convergence of the firing rate?
    \item If there are more than one solution to the error function $E(\br)$, which one would the firing rate converge to?
\end{enumerate}

In a collaboration with Chung and Lu~\cite{CCL19}, we answered all the above questions using rigorous mathematics.
First, the non-negative least squares problem defined in the following exactly captures the observation of Barrett, Den\`{e}ve, and Machens~\cite{barrett2013firing} as described in~\autoref{prop:BDM}.
\begin{definition}[Non-negative least squares problem]\label{def:non-negative least squares}
Let $F\in\Real^{m\times n}$ and $\bx\in\Real^m$, define
\begin{equation}
	\begin{aligned}
	& \underset{\br\in\Real^n}{\text{minimize}}
	& & \frac{1}{2}\|\bx-F^\top\br\|_2 \\
	& \text{subject to}
	& & \br\geq\mathbf{0}.
	\end{aligned}
    \tag{Non-negative least squares}
\end{equation}
\end{definition}

The following is an informal version of the theorem we prove regarding the convergence to the (non-negative) least squares problem. See~\autoref{app:SNNs least squares} for the formal theorem statement as well as a complete proof.
\begin{theorem}[Informal]\label{thm:SNNs main least squares}
The firing rate of an optimal balanced SNN will efficiently converge to a solution to the non-negative least squares problem. 
\end{theorem}

Note that when $m<n$, the non-negative least squares problem might have more than one optimal solution. Thus, a natural question would be: which solution would the firing rate of SNNs converge to? In~\autoref{thm:SNNs main least squares} we show that when the SNNs are non-leaky, the firing rate will converge to the solution with the least $\ell_1$ norm as captured by the following optimization problem.
\begin{definition}[Non-negative $\ell_1$ minimization problem]\label{def:l1 min}
Let $F\in\Real^{n\times m}$ and $\bx\in\Real^m$, define
\begin{equation}
	\begin{aligned}
	& \underset{\br\in\Real^n}{\text{minimize}}
	& & \|\br\|_1 \\
	& \text{subject to}
	& & F^\top\br=\bx,\ \br\geq\mathbf{0}
	\end{aligned}
    \tag{$\ell_1$ minimization}
\end{equation}
where $\|\br\|_1$ is defined as $\sum_i|r_i|$.
\end{definition}
As $\br$ correspond to the firing rate vector of SNNs, its $\ell_1$ norm corresponds to the firing rate of the whole SNN. That is, our result indicates that the SNN will converge to the most ``spike-efficient'' solution.

When the SNNs are leaky, the firing rate will instead converge to the optimal solution of a (non-negative) Lasso (Least absolute shrinkage and selection operator) problem (or equivalently, the basis pursuit denoising (BPDN) problem), which is essentially the non-negative least squares problem with an extra $\ell_1$ regularized term.
\begin{definition}[BPDN/Lasso]\label{def:lasso}
Let $F\in\Real^{n\times m}$, $\bx\in\Real^m$, and $\Lambda>0$, define
\begin{equation}
	\begin{aligned}
	& \underset{\br\in\Real^n}{\text{minimize}}
	& & \frac{1}{2}\|\bx-F^\top\br\|_2^2\\
	& \text{subject to}
	& & \|\br\|_1\leq \Lambda,\ \br\geq\mathbf{0}.
	\end{aligned}
    \tag{BPDN}
\end{equation}
Or, equivalently for some $\beta>0$,
\begin{equation}
	\begin{aligned}
	& \underset{\br\in\Real^n}{\text{minimize}}
	& & \frac{1}{2}\|\bx-F^\top\br\|_2^2 + \beta\|\br\|_1 \\
	& \text{subject to}
	& & \br\geq\mathbf{0}.
	\end{aligned}
    \tag{Lasso}
\end{equation}
\end{definition}
The constant $\beta$ in~\autoref{def:lasso} determines the tradeoffs between the sparsity of the solution (i.e., the $\ell_1$ norm of the solution) and the error (i.e., $\|\bx-F^\top\br\|_2^2/2$). Intuitively, the optimal solution of Lasso corresponds to the solution that minimize the error under certain firing budget constraint.

Now, let us summarize our results in the following theorem (an informal version).

\begin{theorem}[Informal]\label{thm:SNNs main sparse}
The firing rate of an optimal balanced SNN will efficiently converge to a sparse solution to the non-negative least squares problem. In particular, if the SNN is non-leaky, the firing rate will converge to the least $\ell_1$ norm solution. If the SNN is leaky, the firing rate will converge to the (non-negative) Lasso (with a proper choice of the regularizing constant) solution.
\end{theorem}

The mathematically rigorous version of the above theorem will be stated later in~\autoref{thm:l1}, where several technical conditions would arise. Nevertheless, the main contribution here is the introduction of the first connection between optimal E/I balanced SNNs and sparse recovery through the lens of duality and geometry in convex optimization.

\subsection{An algorithmic perspective of optimal E/I balanced SNNs}\label{sec:SNNs algorithm}

\medskip \noindent \textbf{The duality of optimization problems.}
In (convex) optimization theory, \textit{duality} is an important tool to investigate the optimality, geometry, and stability of optimization problems. In the work with Chung and Lu~\cite{CCL19}, we discover a beautiful interplay between optimal balanced SNNs and sparse recovery through the lens of duality. For readers who are new to optimization theory, we recommend a standard textbook by Boyd and Vandenberghe~\cite{boyd2004convex}. Now, let us start with the dual problems of the optimization problems of our interest.

The (non-negative) $\ell_1$ minimization (\autoref{def:l1 min}) has the following dual problem (see~\autoref{app:l1 min derivation} for a derivation).

\vspace{\minipagebeforesep}
\begin{minipage}{\linewidth}
	\begin{minipage}{0.45\linewidth}
		\begin{equation}\label{op:l1 min}
		\begin{aligned}
		& \underset{\br\in\Real^n}{\text{minimize}}
		& & \|\br\|_1 \\
		& \text{subject to}
		& & F^\top\br=\bx,\ \br\geq0
		\end{aligned}
		\end{equation}
	\end{minipage}
	\begin{minipage}{0.45\linewidth}
		\begin{equation}\label{op:l1 min dual}
		\begin{aligned}
		& \underset{\bu\in\Real^m}{\text{maximize}}
		& & \bx^\top\bu \\
		& \text{subject to}
		& & F\bu\leq1.
		\end{aligned}
		\end{equation}
	\end{minipage}
\end{minipage}
\vspace{\minipageaftersep}

The (non-negative) Lasso (\autoref{def:lasso}) has the following dual problem (see~\autoref{app:lasso derivation} for a derivation).

\vspace{\minipagebeforesep}
\begin{minipage}{\linewidth}
	\begin{minipage}{0.45\linewidth}
		\begin{equation}\label{op:lasso}
		\begin{aligned}
		& \underset{\br\in\Real^n}{\text{minimize}}
		& & \frac{1}{2}\|\bx-F^\top\br\|_2^2 + \beta\|\br\|_1\\
	    & \text{subject to}
	    & & \br\geq\mathbf{0}.
		\end{aligned}
		\end{equation}
	\end{minipage}
	\begin{minipage}{0.45\linewidth}
		\begin{equation}\label{op:lasso dual}
		\begin{aligned}
		& \underset{\bu\in\Real^m}{\text{maximize}}
		& & \frac{1}{2}\|\bx\|_2^2 - \frac{1}{2}\|\bx-\beta\bu\|_2^2 \\
		& \text{subject to}
		& & F\bu\leq1.
		\end{aligned}
		\end{equation}
	\end{minipage}
\end{minipage}
\vspace{\minipageaftersep}

\medskip \noindent \textbf{The dual space and the dual process of SNNs.}
The dual problems (\autoref{op:l1 min dual} and~\autoref{op:lasso dual}) naturally inspire the definition of the dual space of SNNs. Here we motivate the definition via defining the \textit{dual dynamics} of optimal balanced SNNs as follows. Let $\tau,F$ be the parameters as one would use in~\autoref{eq:optimal SNN continuous} (and in~\autoref{eq:optimal SNN discrete} for the discrete case), define processes $\bu(t)$ by the following differential equation or update rule:
\begin{equation}\label{eq:optimal SNN continuous dual}
\frac{d \bu(t)}{dt} = -\tau \bu(t) -F^\top \bs(t) + \bx
\end{equation}
for the continuous case and
\begin{equation}\label{eq:optimal SNN discrete dual}
\bu(t+1) = -\tau \bu(t) - F^\top \bs(t) + \bx\cdot\Delta t
\end{equation}
for the discrete case. Here $\bs(t)$ follows the same definition as discussed in~\autoref{eq:SNN continuous} and~\autoref{eq:SNN discrete}. Once we set $\bu(0)$ properly so that $\bv(0)=F\bu(0)$, then by construction, we have $\bv(t)=F\bu(t)$ for every $t$.

\medskip \noindent \textbf{A geometric view of the dual dynamics.}
We can now try to interpret the dual dynamics by examining the spiking activities. First, the spiking condition of neuron $i$ at time $t$ becomes the following.
\begin{equation}
\bv_i(t) > \eta \Leftrightarrow F_i^\top \bu(t) > \eta
\end{equation}
where $F_i$ is (the transpose of) the $i$-th row vector. Moreover, the effect of the $i$-th neuron's spiking on the dual dynamic is $-F_i$. Namely, this gives us a geometric picture to interpret the dual dynamic: whenever the dual process $\bu(t)$ surpasses the hyperplane defined as $W_i:=\{\bu:\ F_i^\top\bu=\eta\}$, then it will bounce back in the direction of the normal vector of this hyperplane. See~\autoref{fig:SNNs dual geometry} for an illustration.

\begin{figure}[h]
    \centering
    \includegraphics[width=13.5cm]{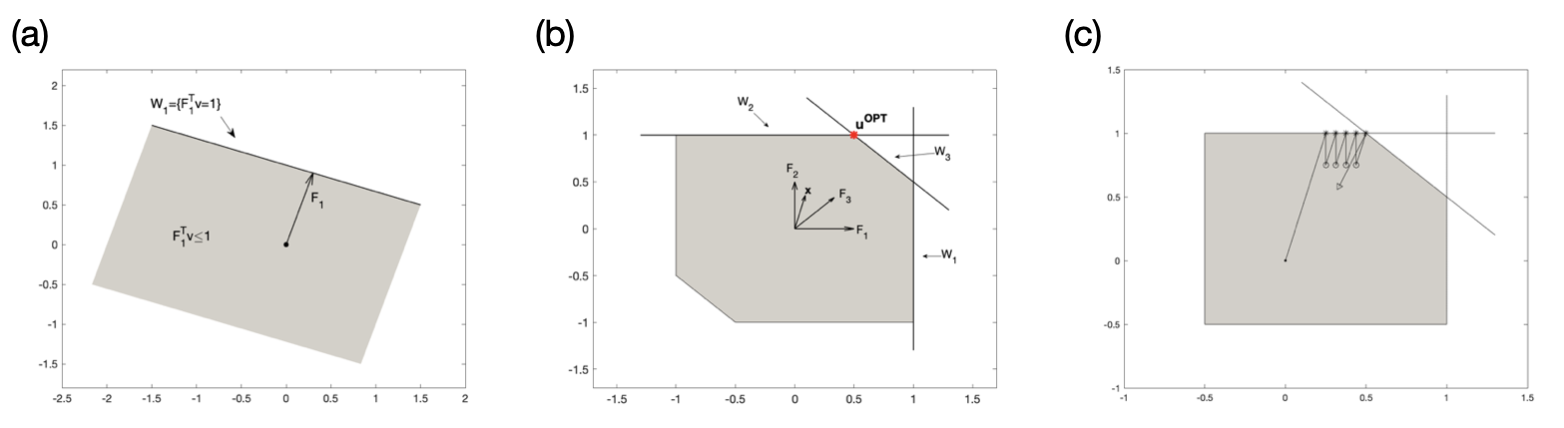}
    \caption{A geometric view of the dual dynamics. (a) An example of the dual space geometry of a single neuron. The ``dual wall'' $W_1$ and the feasible region (plotted gray). (b) An example of the dual space geometry of a three-neuron SNN (the same example as described earlier). (c) The dual space dynamic of the same three-neuron SNN with external current $\bx=[1,2]^\top$.}
    \label{fig:SNNs dual geometry}
\end{figure}

\medskip \noindent \textbf{The dual dynamic is a greedy projected gradient descent algorithm for the dual problem.}
Finally, we are ready to put every pieces together: the dual dynamic (\autoref{eq:optimal SNN continuous dual} and~\autoref{eq:optimal SNN discrete dual}) is a greedy projected gradient descent algorithm for the dual problems (\autoref{op:l1 min dual} and~\autoref{op:lasso dual}). To make this claim mathematically rigorous, it requires extra mathematical assumptions as well as another 20 pages of tedious derivations. Hence, we defer that part to~\autoref{app:SNNs l1 min details} and focus on the high-level intuition here.

Let us use the non-leaky and continuous optimal balanced SNNs as a starting example. Recall that the primal (i.e., the original membrane potential vector) and the dual process are

\vspace{\minipagebeforesep}
\begin{minipage}{\linewidth}
	\begin{minipage}{0.5\linewidth}
		\begin{equation*}
		\frac{d\bv(t)}{dt} = -FF^\top\bs(t) + F\bx
        \tag{Primal}
		\end{equation*}
	\end{minipage}
	\begin{minipage}{0.45\linewidth}
		\begin{equation*}
		\frac{d\bu(t)}{dt} = -F^\top\bs(t) + \bx \, .
        \tag{Dual}
		\end{equation*}
	\end{minipage}
\end{minipage}
\vspace{\minipageaftersep}

The primal and dual optimization problem of (non-negative) $\ell_1$ minimization are the following.

\vspace{\minipagebeforesep}
\begin{minipage}{\linewidth}
	\begin{minipage}{0.5\linewidth}
		\begin{equation*}
		\begin{aligned}
		& \underset{\br\in\Real^n}{\text{minimize}}
		& & \|\br\|_1 \\
		& \text{subject to}
		& & F^\top\br=\bx,\ \br\geq0
		\end{aligned}
        \tag{Primal}
		\end{equation*}
	\end{minipage}
	\begin{minipage}{0.45\linewidth}
		\begin{equation*}
		\begin{aligned}
		& \underset{\bu\in\Real^m}{\text{maximize}}
		& & \bx^\top\bu \\
		& \text{subject to}
		& & F\bu\leq1.
		\end{aligned}
        \tag{Dual}
		\end{equation*}
	\end{minipage}
\end{minipage}
\vspace{\minipageaftersep}

While at first glance, it's unclear the connection between the primal dynamic and the primal problem, the dual dynamic is naturally solving the dual problem: the external current term $\bx$ in the dual dynamic corresponds to increasing the objective value (i.e., $\bx^\top\bu$) and the spiking effect in the dual dynamic corresponds to making sure $\bu(t)$ remain feasible (i.e., $F\bu\leq1$).

Similarly, for the leaky case the primal and dual process are

\vspace{\minipagebeforesep}
\begin{minipage}{\linewidth}
	\begin{minipage}{0.5\linewidth}
		\begin{equation*}
		\frac{d\bv(t)}{dt} = -\tau\bv(t) -FF^\top\bs(t) + F\bx
        \tag{Primal}
		\end{equation*}
	\end{minipage}
	\begin{minipage}{0.45\linewidth}
		\begin{equation*}
		\frac{d\bu(t)}{dt} = -\tau\bu(t) -F^\top\bs(t) + \bx \, .
        \tag{Dual}
		\end{equation*}
	\end{minipage}
\end{minipage}
\vspace{\minipageaftersep}

The primal and dual optimization problem of (non-negative) Lasso are the following.

\vspace{\minipagebeforesep}
\begin{minipage}{\linewidth}
	\begin{minipage}{0.5\linewidth}
		\begin{equation*}
		\begin{aligned}
		& \underset{\br\in\Real^n}{\text{minimize}}
		& & \frac{1}{2}\|\bx-F^\top\br\|_2^2 + \beta\|\br\|_1\\
	    & \text{subject to}
	    & & \br\geq\mathbf{0}.
		\end{aligned}
        \tag{Primal}
		\end{equation*}
	\end{minipage}
	\begin{minipage}{0.45\linewidth}
		\begin{equation*}
		\begin{aligned}
		& \underset{\bu\in\Real^m}{\text{maximize}}
		& & \frac{1}{2}\|\bx\|_2^2 - \frac{1}{2}\|\bx-\beta\bu\|_2^2 \\
		& \text{subject to}
		& & F\bu\leq1.
		\end{aligned}
        \tag{Dual}
		\end{equation*}
	\end{minipage}
\end{minipage}
\vspace{\minipageaftersep}

Now the dual objective is to minimize the distance between $\bx$ and $\beta\bu$. And indeed, the gradient of $\|\bx-\beta\bu\|_2^2/2$ is $\beta^2\bu-\beta\bx$, which is exactly right hand side of the dual dynamic when there's no spike. From this view, one can also see that when $\beta$ goes to $0$ corresponds to the amount the leakage goes to $0$.

\subsection{An emergent computation of optimal E/I balanced SNNs}\label{sec:SNNs emergence}
Through the dual dynamic and the lens of algorithms, we further establish the connection between the firing rate of an optimal balanced SNN and sparse recovery. We now state the formal version of the two main theorems.

\begin{restatable}[Lease squares]{theorem}{snnleastsquares}\label{thm:linearsystem} 
For every $F\in\Real^{n\times m}$, $\bx\in\Real^m$, and $\epsilon>0$, let $\br(t)$ be the firing rate of an SNN with $\Omega=FF^\top$, $I=F\bx$, $\alpha=1$, $\eta\geq\lambda_{\max}$, and $\Dt<\frac{\sqrt{\lambda_{\min}}}{24\sqrt{n}\cdot\|\bx_F\|_2}$ where $\bx_F$ is the orthogonal projection of $\bx$ to the orthogonal complement of the null space of $F$. We have that $\|\bx_F-F^\top\br(t)\|_2\leq2\sqrt{\kappa\eta n}\cdot\|\bx_F\|_2/(t\cdot\Dt)$ for every $t\in\N$.

In particular, for every $\epsilon>0$, if we set $\eta = \lambda_{\max}$ and $\Dt=\frac{\sqrt{\lambda_{\min}}}{24\sqrt{n}\cdot\|\bx_F\|_2}$, then $\|\bx_F-F^\top\br(t)\|_2\leq\epsilon$ for every $t\geq\Omega(\kappa n/\epsilon)$.
\end{restatable}

Next, to state the theorem about sparse recovery, we need to introduce an extra parameter to the SNN dynamic: the spike strength $\alpha$, which is a non-zero real number capturing the strength of the spiking effect. Namely,~\autoref{eq:optimal SNN continuous} and~\autoref{eq:optimal SNN discrete} become
\begin{equation}\label{eq:optimal SNN continuous alpha}
\frac{d \bv(t)}{dt} = -\tau \bv(t) - \alpha\cdot FF^\top \bs(t) + F\bx
\end{equation}
and
\begin{equation}\label{eq:optimal SNN discrete alpha}
\bv(t+1) = -\tau \bv(t) - \alpha\cdot FF^\top \bs(t) + F\bx\cdot\Dt
\end{equation}
respectively.

\begin{restatable}[Sparse recovery]{theorem}{snnlonemin}\label{thm:l1}
For every $F\in\Real^{n\times m}$ and $\bx\in\Real^m$ where all the row of $F$ has unit norm, let $\gamma(F)$ be the niceness parameter of $F$ defined later in~\autoref{def:nice}. Suppose $\gamma(F)>0$ and there exists a solution for $F^\top\br=\bx$. There exists a polynomial $\alpha(\cdot)$ such that for any $t\geq0$, let $\br(t)$ be the firing rate of the SNN with $\Omega = FF^\top$, $I = F\bx$, $\eta=1$, $0<\alpha\leq \alpha(\frac{\gamma(F)}{n\cdot\lambda_{\max}})$. Let $\OPT^{\ell_1}$ be the optimal value of the $\ell_1$ minimization problem (\autoref{op:l1 min}). For any $\epsilon>0$, when $t\geq\Omega(\frac{m^2\cdot n\cdot\|\bx\|_2^2}{\epsilon^2\cdot\lambda_{\min}\cdot\OPT^{\ell_1}})$, then $\|\bx_F- F^\top\br(t)\|_2\leq\epsilon\cdot\|\bx_F\|_2$ and $\|\br(t)\|_1\leq(1+\epsilon)\cdot\OPT^{\ell_1}$.
\end{restatable}

For those who are not familiar with theoretical computer science, the constants (e.g., 24, 48 in~\autoref{thm:linearsystem}) and the big O notations in the theorem statements are not optimized. Namely, the right way to interpret the results is to focus on the asymptotic scaling. For example,~\autoref{thm:linearsystem} says that as long as the parameters of an optimal SNNs satisfy the required condition, then its firing rate will converge to an optimal solution of the least squares problem with rate $O(1/t)$. Generally speaking, the firing rate could converge much faster than the upper bound in the theorem statement, however, the theorems assert that the guarantees would still hold in the worst-case scenario.

Next, a few words on the introduction of spike strength $\alpha$ and the technical condition $\gamma(F)$ in~\autoref{thm:l1}. As the first attempt in establishing the connection between optimal E/I balanced SNNs and sparse recovery, we focus on the regime where an SNN would solve the corresponding optimization problem when $t$ goes to infinity. Moreover, as we adopt the methodology of provable analysis and hence we sacrifice a little bit bio-plausibility (i.e., in biological SNNs, $\alpha$ is close to $1$). We pose it as an interesting future work to systematically examine the connection between optimal E/I balanced SNNs and sparse recovery in a wider parameter regime.

The complete proof of the above two theorems are provided in~\autoref{app:SNNs least squares} and~\autoref{app:SNNs l1 min details} respectively. We encourage the reader with mathematical background and interest to glimpse through. Meanwhile, for the broader audience, we discuss the high-level ideas and takeaways as follows.

\medskip \noindent \textbf{Optimal balance leads to energy conservation in a matrix norm.}
As pointed out by Barrett et al.~\cite{barrett2013firing}, the optimal E/I balanced SNNs track the input signal by having each neuron's membrane potential correspond to the residual error $\bx-F^\top\br$. Hence, the integrate-and fire rule naturally ensures that the $\ell_\infty$ norm (i.e., the largest coordinate value in a vector) of the potential vector is upper bounded. However, this does not provide a tight characterization of the signal recovery performance as the residual error is typically measured in terms of the $\ell_2$ norm, i.e., $\|\bx-F^\top\br\|_2$.

The key idea in the proof of~\autoref{thm:linearsystem} is to analyze the energy conservation of the potential vector $\bv(t)$ in the matrix norm $\|\bv(t)\|_{(FF^\top)^\dagger}:=\sqrt{\bv(t)^\top(FF^\top)^\dagger\bv(t)}$, which has a tight connection to the residual error as revealed by simple algebraic manipulations (see \autoref{lemma:DSNN-potentialbounded} for more details). In addition to the mathematical elegance, this insight also hints at potential biological relevance, since the feedforward matrix $F$ defines the geometry of the residual error space. Specifically, by applying the matrix norm $\|\cdot\|_{(FF^\top)^\dagger}$, one can gain a clearer understanding of which signal direction $\bx$ the SNN can efficiently represent. This could pave the way for future investigations into synaptic learning rules for $F$.

\medskip \noindent \textbf{Optimal balance plus integrate-and-spike induce a sparse computation.}
Let us revisit the overarching research questions on E/I balanced networks as discussed in~\autoref{sec:E/I balanced}. The primary motivation stems from the enigma of irregular activity patterns observed in the brain, and E/I balance suggested as a conceptual framework to address the three cornerstone questions: the implementation problem, the coding problem, and the computation problem. In particular, the theory of optimal balanced SNNs demonstrated how a coding principle leads to a specific network implementation.

In the joint work with Chung and Lu~\cite{CCL19} as presented in~\autoref{sec:SNNs}, we further establish a link between optimal balanced SNNs and an emergent computation: sparse recovery. Our results offer an illustration of how a coding principle (i.e., optimal balance) combined with a biological constraint (i.e., integrate-and-fire neurons) can lead to a non-trivial computation (i.e., sparse recovery). This might bring up two future research directions. First, within the context of optimal balanced SNNs, what are the further computational or coding implications of sparse recovery? In particular, as this naturally connects to dictionary learning when it comes to the learning of feedforward connectivity matrix $F$, would this emergent sparse computation lead to further non-trivial downstream computations? Second, regarding the grand question of irregular spiking, could sparse recovery serve as a new computational principle, stimulating the identification of a next coding principle and networks organization for future studies?

\begin{savequote}[75mm]
The best theory is inspired by practice. The best practice is inspired by theory.
\qauthor{Donald E.~Knuth}
\end{savequote}

\chapter{Conclusion}\label{ch:conclusion}

Thus far in this dissertation, we have examined numerous instances of prior work employing the computational lens across various scientific fields in~\autoref{ch:introduction}. We have also scrutinized the author's previous efforts within quantum physics and neuroscience in~\autoref{ch:example RCS} and~\autoref{ch:example EI} respectively.
Indeed, these instances underscore the pervasiveness of the computational perspective. However, unlike the two transformative waves of computing~–~the popularization of digital computers and the rise of artificial intelligence~–~the approach of thinking computationally has not yet achieved widespread adoption across different fields at a conceptual level. Consequently, in this final chapter, I will endeavor to answer the question, ``How might we employ the computational lens more effectively?'' by reflecting on what may be currently lacking or absent in the prevailing methodology. Ultimately, I will delineate a series of thoughts on potential future research directions in both quantum physics and neuroscience.

\section{How to employ the computational lens?}
Rather than attempting to provide a recipe for others to follow, the remaining chapters of this dissertation will present my view on the challenges of studying interdisciplinary problems through the computational lens. 

% I hope that these examples will inspire future endeavors in applying the computational lens across various research domains. In the interim, I would like to share a few key insights gained from these adventurous journeys.

\medskip \noindent \textbf{On the difference of terminologies, methodologies, and appreciations.}
When it comes to interdisciplinary collaboration, the first reaction of many people is probably about the challenge of learning a new language in the other field. Indeed, the language of a field often contains three difficult ingredients to grasp (listed by increasing order): (i) terminologies, i.e., the jargon and technical concepts; (ii) methodologies, i.e., methods of reasoning, making statements, and constructing results; (iii) appreciations, i.e., the types of questions to ask, which results are interesting, and the kind of understanding to pursue.

When I venture into a new field, I typically like to tackle these challenges in reverse order: first, getting to know the key questions people are asking and the kind of understanding they are pursuing. Then, try to resonate myself with their way of thinking without forgoing my own previous mindset. Next, I delve into the methodology in use, comparing it with methods I know from other fields to see if there's a chance for integration. Through this process of understanding the appreciations and methodologies, I often naturally begin to learn the field's terminologies. When I decide to delve deeper into a field, I study their terminologies systematically, not just memorizing terms, but contemplating why particular terms were created for specific concepts. Using the correct terminology can indeed sound professional, but understanding why those terms are used can, in my opinion, lead to a deeper understanding of the field.

\medskip \noindent \textbf{On the tango between questions and approaches.}
While this is a dissertation about the proposal of a new scientific methodology, I want to emphasize the importance of the way people ask questions. In the end, our goal is to elevate the understanding of the subject of interest. Though people with different training will approach the same problem in unique ways, I believe the research question should always be \emph{relevant}~\footnote{Here, I intentionally avoid using the word ``application'' as it might be interpreted differently by different people.}.
Indeed, the word relevant can have a broad interpretation. For instance, cooking could be seen as relevant to baseball because athletes require well-designed diets to maintain their optimal condition. Yet, this doesn't imply that every research in cooking is relevant to baseball. In the broadest sense I can conceive, I would say a research question is relevant as long as the authors genuinely care about the subject of interest and reflect in the question they study.

At the same time, with my theoretical background, I acknowledge that some major scientific breakthroughs have emerged from a focus on research methodology, rather than a concrete question of relevance. It is always fascinating to see how seemingly unrelated logical or mathematical derivations can yield insightful revelations about real-world problems. For instance, Schwarzschild's solution to Einstein's field equations led to the discovery of blackholes, and Parisi's glorious replica calculations opened a new chapter in spin glass theory. What makes these theoretical works stand out is that they not only demonstrate profound abstract reasoning, but they also circle back to asking and answering questions that are relevant to the subject of interest.
I believe that this tango between the question and approach is particularly important in the interdisciplinary endeavor of applying a computational lens to various scientific fields.

\medskip \noindent \textbf{On a web of scientific understanding from various perspectives.}
As the complexity of scientific inquiries continues to grow, traditional disciplinary boundaries begin to blur, propelled by two principal forces: an intrinsic pressure, driven by the necessity to seek novel paradigms and methodologies to address intricate problems; and an extrinsic pressure, originating from the fact that the subjects of interest frequently reside at the intersection of multiple disciplines. 

This implies that we are inexorably moving towards constructing a ``web of scientific understanding from various perspectives''. By this, I mean a specific scientific question could be approached through methodologies X, Y, and Z. None of these methodologies may stand out as the singular, definitive approach. Instead, it is the holistic picture, created by the convergence of these diverse methodologies, that truly illuminates and provides deep insights into the original question.

In my view, the computational lens can function not only as an independent perspective but also as a component embedded within various methodologies.
However, it's important to clarify that I am proposing a \textit{web} of understanding rather than a single \textit{unifying} perspective. Consequently, a crucial aspect I wish to emphasize is that when selecting a methodology to tackle a problem, the underlying assumptions and guiding principles must be thoroughly discussed upfront. Additionally, a web of understanding only comes after examining the same problem with multiple angles \textit{separately}.

Throughout my doctoral journey, I encountered numerous instances where I conflated rigorous mathematical approaches with empirical phenomenological methodologies within a single project, leading me to nowhere. These experiences taught me that while the computational lens could operate independently (for example, when analyzing a problem algorithmically), when integrating it with other scientific approaches, it is crucial to maintain clarity about the dominant methodology.

\section{Future research directions}
In this section, I will present some thoughts and suggestions regarding potential future research questions in the domains of quantum computation and neuroscience, as examined through the computational lens. The discussions are meant to be high-level and primarily conceptual, aiming to identify promising and important avenues for further exploration.

\subsection{Quantum physics}

\medskip \noindent \textbf{Benchmarking methods for the near-term quantum computational advantage.}
Quantum computing has attracted considerable attention, funding, and intellectual resources over the past decade. The seminal paper by Preskill~\cite{preskill2012quantum}, which introduced the concept of quantum supremacy (also known as quantum computational advantage), has been used as a significant intermediate milestone guiding both theoretical and experimental development within the field. Although Google claimed to have achieved this goal, subsequent research~\cite{pan,GKCLBC21} has cast serious doubts on the benchmarks they employed. A similar cat-and-mouse dynamic occurred recently when, just two weeks after IBM published a demonstration of quantum advantage on a quantum simulation of a 127-qubit kicked Ising quantum system~\cite{kim2023evidence}, a group of researchers promptly announced an efficient classical simulation on arXiv~\cite{tindall2023,begusic2023}.

In my view, the primary issue at hand revolves around the tension between asymptotic scalings and finite-size considerations. Theoretical researchers often focus on discussing computational speed-up in terms of scaling input sizes towards infinity and comparing the order of leading terms. In contrast, experimentalists usually prioritize concrete numbers for benchmarking purposes. It is unquestionable that the momentum of quantum computing was initiated by the exponential speed-up suggested by Shor's algorithm. However, at this juncture in the evolution of the field, we need quantitative measures that can accurately indicate our current position and guide future advancements.

Reflecting on the recent advances in machine learning, it is not an exaggeration to state that many breakthroughs have been propelled by well-structured task-specific datasets. A notable example is ImageNet~\cite{deng2009imagenet}, which soon paved the way for the breakthrough of AlexNet~\cite{krizhevsky2017imagenet}. This suggests that a similar approach could be instrumental within the realm of quantum computing. I propose the development of public benchmarking framework that encapsulate both computational and engineering aspects of quantum computing, ideally also incorporating factors such as error-correction. The tasks within these benchmarks should be designed in alignment with established quantum computing criteria, such as the seminal Bell test~\cite{bell1964einstein} or those outlined by DiVincenzo~\cite{divincenzo2000physical}. A well-structured and comprehensive benchmarking suite of this nature could significantly advance the field by offering clear, standardized assessment metrics, driving competitive progress, and highlighting areas requiring further research and development.

\medskip \noindent \textbf{Quantum computations for science.}
It is a widely accepted view in the field that quantum computing is not intended to accelerate all types of computations. While Shor's algorithm has given hope for exponential speedup on problems originating from computer science, it is important to remember Feynman's visionary anticipation of using quantum computers to more accurately simulate quantum systems. This reminds us that the potential advantage of quantum computing other than speedup; it could also entails the ability to accommodate complex quantum phenomena that are inherently intractable for classical computers.

Hence, I extend an invitation to my fellow computer scientists and students to invest time in studying the beautiful world of quantum physics and other relevant scientific fields, such as chemistry. The insights you can learn from these fields will not be a waste of time; rather, they have the potential to seed your mind with fresh inspirations and perspectives. In particular, through the application of the computational lens, we may strive to identify phenomena or pose questions from these fields using the language of computation, thereby unveiling potentially novel applications of quantum computing. This exploration has the potential to significantly enrich our understanding and further expand the horizon of possibilities in the realm of quantum computation.

\subsection{Neuroscience}

\medskip \noindent \textbf{Population geometry as an intermediate quantitative language.}
Neurons in the brain generate patterns of activity in response to external stimuli or while an animal is performing a specific task. These patterns, known as neural representations, are pivotal for understanding how the brain encodes information and performs computations. Traditional methods, which primarily focus on the neural representations of single neurons or small groups of neurons, have had considerable success in detailing the intricate tuning properties of various cell types. However, as we transition into an era of large-scale neural recordings, the need to decipher complex, high-dimensional, and noisy neural representations demands the development of new frameworks to facilitate quantitative investigations.

Population geometry approaches neural representations as high-dimensional geometric objects, providing a language that enables both quantitative analysis and intuitive understanding. 
Meanwhile, the powerful modeling method and analytical techniques in statistical physics allow us to systematically derive order parameters~—~macroscopic observables characterizing critical phenomena in a system~—~that link population geometry to computational principles at the level of neural representations. For instance, a recent work by Chung, Lee, and Sompolinsky~\cite{chung2018classification} utilized spin glass theory and geometry, has exemplified this by deriving the notion of manifold capacity in connection to the neural manifold's geometrical properties under the context of invariant object recognition.
As neural representations are emerging language used by the brain, it is an exciting research direction to build up a quantitative framework to systematically bridge lower-level implementations with higher-level functionalities. Specifically, the computational lens could potentially help identifying useful computational objectives to investigate and modularizing our understanding.

\medskip \noindent \textbf{Algorithmic neuroscience as a modeling tool.}
A recurrent theme in this thesis is the emphasis of algorithmic thinking providing a modular and mechanical understanding of an information processing system. It can complement the traditional mathematical and physical approach to studying the brain. In the research direction of algorithmic neuroscience, I propose a systematic approach to model various neural systems algorithmically. This involves characterizing their input-output behaviors, abstracting their core algorithmic concepts (possibly represented in the form of pseudocodes), extracting out computational principles, analyzing their computational complexity, and so on.  While these steps may appear to be standard procedures within computer science, I want to highlight that the mindset should differ significantly in this context.
Notably, our aim here is not to \textit{design} new algorithms with a focus on speed or optimization, as one would typically aim for in computer science. Instead, the objective here is to dissect the complex systems of interest into a modular, clean, and comprehensible form. 

\medskip \noindent \textbf{Emergent computations as a probe for high-level reasoning.}
Once an algorithmic description for a neural system is established, it presents us with a unique opportunity to gain further insights through the lens of emergent computations. Leveraging mathematical theories (for instance, geometry, optimization, algebra), statistical physics tools (such as approximations, heuristics, and physical picture), along with the computational lens, the overarching goal is to consolidate computational principles and/or normative objectives for high-level reasoning. One concrete direction that we could explore is developing composition rules for algorithmic components. For example, in~\autoref{sec:SNNs emergence}, we derived an example of ``spiking neurons + an energy constraint $\rightarrow$ a sparse computation''. Therefore, emergent computations can serve as an effective tool to establish a comprehensive dictionary of neural mechanisms at various scales.

\section{Concluding remarks}
Donald Knuth, a Turing award laureate often hailed as ``father of the analysis of algorithms'', once said, ``The best theory is inspired by practice. The best practice is inspired by theory''. His wisdom, in my opinion, truly captures the essence of scientific progression - a realm that is continuously evolving, consistently posing challenges, yet invariably anchored by the harmonious blend of theory and practice. 
Trained as a theoretical computer scientists, I find myself riding the transformative waves of computing on the vast ocean of knowledge. Much like the way rational thinking ushered in a revolution in the 17th century and logical thinking reshaped the 20th, perhaps the deployment of the computational lens, with its inherent advantages of composability and modularity, could help us better decipher, and ultimately comprehend the ever-increasing complexities of the world we live in.

\begin{appendices}
    \chapter{Details for the Quantum Part}\label{app:quantum}

\section{Circuit architectures}\label{app:circuit architecture}
For practical relevance, we focus on 1D and 2D circuit architectures.
For 1D circuits, we theoretically and numerically show that our basic algorithm can achieve, in linear time, a higher average XEB value than noisy quantum systems. More specifically, we show that setting subsystem size to be constant ($ l=O(1)$) is sufficient for our algorithm to obtain a higher XEB value than that of $\epsilon$-noisy quantum simulations, for every constant $\epsilon>0$, for sufficiently large $N$.
This is due to the distinct scaling behavior of the XEB value for noisy circuits and our algorithm; we discuss in detail the origin of this difference in the scaling behavior in~\autoref{sec:Ising}.

\begin{figure}[h]
\includegraphics[width=8cm]{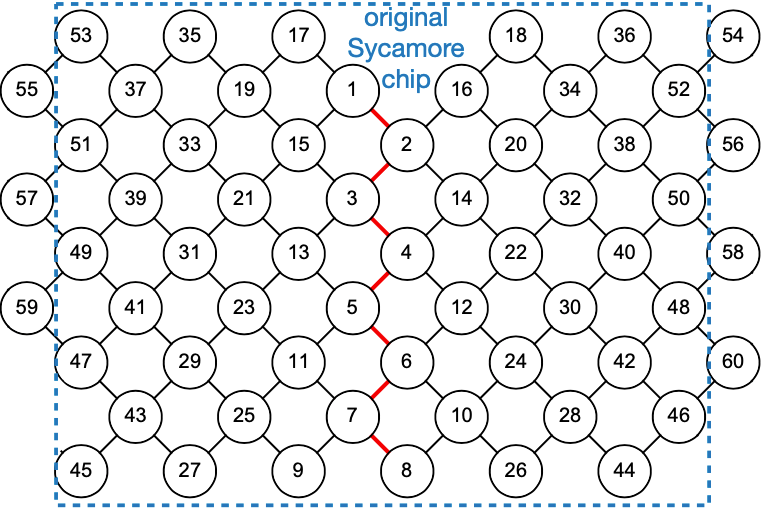}
\centering
\caption{
Sycamore circuit architecture from Ref.~\cite{arute2019quantum} and its horizontal extension. The gates marked with red lines are omitted in our algorithm. The Zuchongzhi architecture is very similar; see Ref.~\cite{USTC,zhu2021quantum} for more detail.
}
\label{fig:intro_partition}
\end{figure}

For 2D circuits, we consider Google's Sycamore architecture, which has $N=53$ qubits~\cite{arute2019quantum}, and we choose $ l\approx\left\lceil{N/2}\right\rceil = 27$ (Fig.~\ref{fig:intro_partition}). We also consider USTC's Zuchongzhi architectures which have 56 qubits and 60 qubits respectively, and we choose $l\approx28$ for both cases (with some qubits being omitted).
A subsystem of this size can be simulated by one NVIDIA Tesla V100 GPU with 32GB memory in about 1 second~\cite{julia1,julia2}. 
We analyze the performance of our algorithms on circuits constructed from  the following different quantum-gate ensembles:
\begin{description}
\item[CZ ensemble] Each random two-qubit gate is composed of the control-Z gate surrounded by four independent single-qubit Haar random gates [see~\autoref{fig:intro_alg}(b)].
\item [Haar ensemble] Each random two-qubit gate is a two-qubit Haar random gate.
\item[fSim ensemble] Similar to CZ ensemble, but replacing the control-Z gate by the fSim gate, which is defined as
\begin{equation}
\label{eq:fSim_formula}
\text{fSim}_{\theta,\phi}=
\begin{pmatrix}
1 & 0 & 0 & 0 \\
0 & \cos(\theta) & -i\sin(\theta) & 0 \\
0  & -i\sin(\theta)  & \cos(\theta) & 0  \\
0 & 0 & 0 & e^{-i\phi} 
\end{pmatrix},
\end{equation}
with parameters $\theta=90^{\circ},\phi=60^{\circ}$~\cite{arute2019quantum} (denoted as fSim); we also define a new gate fSim$^*$ which has $\theta=90^{\circ},\phi=0^{\circ}$.

\item[fSim with discrete 1-qubit ensemble]
Similar to fSim ensemble, but replacing the 1-qubit Haar random gate by $Z(\theta_1)VZ(\theta_2)$
where $V$ is chosen randomly from $\{\sqrt X,\sqrt Y, \sqrt W\}$ ($W=(X+Y)/\sqrt2$) but the two $V$s between two successive layers on the same qubit should be different; and $Z(\theta_i)$ is chosen randomly from $[0,2\pi)$. 
\end{description}
The last ensemble is closely modelled after quantum circuits used in recent experiments~\cite{arute2019quantum,USTC,zhu2021quantum}.
The only modification is that, in experiments, the single qubit rotation angles $\theta_i$' are not actively controlled, but rather determined by the specific ordering of quantum gates and the qubit specification at hardware level.
We expect that this difference does not influence the performance of our algorithm significantly, because we also consider the case where $\theta_i$ is chosen randomly from either 0 or $\pi$ (which corresponds to $I$ or $Z$ operator, respectively). The numerical result shows that the average XEB values for the top-1 method in the two cases are similar: $0.00018$ ($\theta_i \in [0, 2\pi)$) and $0.0004$ ($\theta_i \in \{0, \pi\}$), respectively, for the Sycamore architecture (53 qubits, 20 depth). Therefore, we argue that the $z$-rotation part does not influence the XEB value too much.

\subsection{Improving the algorithm}\label{app:algorithm improve}

While our basic algorithm is simple and relatively straightforward to implement, it already
has significant consequences for the computational hardness of obtaining high XEB values.
Moreover, its practical performance can be 
further improved via the following modifications. 

{\bf Top-$k$ post-processing method.} Given the output distribution $q_C(U)$ produced by our algorithm $C$,  which is correlated with the ideal distribution $p_U(x)$, it is possible to amplify such correlations by using the so-called \emph{top-$k$ post-processing heuristic}.
In this method, one modifies the bitstring distribution $q_C(x)$ by ordering the bitstrings $x_i \in \{x\}$ from largest $q_C(x_i)$ to the smallest, selecting first $k$ of them (or equivalently setting the probability of the others to 0),
\begin{equation}
    q_C(x_i)\to\tilde{q}_C(x_i)=\begin{cases} 0 &{\rm if }~ i \leq k\\
    1/k & {\rm if }~ i>k
    \end{cases}.
\end{equation}
Since we can efficiently compute the probability distribution $q_C(x)$ produced by the original algorithm, we can also efficiently compute the amplified probability distribution. As an example, we illustrate this algorithm (with slight modification for simplicity) in the case of $l=2$ and assume it is efficient to get the entire distributions $q_1$ and $q_2$ of the two subsystems respectively. Thus $q_C=q_1q_2$. Then we sort $q_1$ and $q_2$ in a decreasing order and enumerate the bitstrings corresponding to $k$-largest probability value $p_1$ and $p_2$ respectively. Finally, we get $k^2$ bitstrings from our classical algorithm.

The intuition behind this heuristic can be understood as follows.
The XEB is equivalent to evaluating the average of $p_U(x)$ weighted by $q(x)$ up to an unimportant scaling factor $2^N$, and a constant $-1$. 
If $q(x)$ is modified such that $q(x)$ is increased (decreased) for bitstrings $x$ with relatively large (small) values of $p_U(x)$, then the weighted average will increase.
Given that $q(x)$ and $p_U(x)$ are already positively correlated, such behavior is naturally expected for our top-$k$ post-processing heuristic, at least on average.

In fact, we can prove that the top-$k$ method increases the XEB if its value is positive and the STD over circuit realizations is not too large. The second requirement is necessary to avoid the situation where some occasional $x$ with small $p_x$ but large $q_x$ will be amplified (in another words, ``over-fitting").
Unfortunately, this second criterion is not satisfied by our basic algorithm where we simply omit gates.
This issue, however, can be straightforwardly addressed using the following method.

{\bf Self-averaging algorithm.} 
In order to decrease the STD, we make a small modification to our basic algorithm: instead of omitting gates, we insert maximal depolarizing noise or equivalently take average over different realizations of our basic algorithm with random single qubit unitary at the position of omission. This \emph{self-averaging algorithm} guarantees the positivity and small  STD conditions. However, the computational resources required are larger since we need to simulate mixed state evolution.
Interestingly, for a certain class of entangling gates (including the one used in recent experiments~\cite{arute2019quantum,USTC,zhu2021quantum}) that exhibit the ``maximal scrambling speed'' and that hinders the application of our basic algorithm, one can substantially reduce the computational resources needed for such mixed-state simulation.
This is possible because for that class of entangling gates the effect of depolarizing noise can be propagated efficiently.

{\bf Combining algorithmic  improvements.} In Figure~\ref{fig:intro_advantage_regime}, we present the increase of the XEB for the modified version of Google's gate set ensemble by several orders of magnitude after the application of the top-$k$ method on the self-averaging algorithm.
While the discussion above is mostly focused on the mean value of the XEB, it is important to show that our result also holds for typical, individual instances of quantum circuits with a high probability. In~\autoref{sapp:improved}, we show that the self-averaging algorithm offers a much better control over the STD, and guarantees the benefit of using the top-$k$ method. Additionally, we show evidence that the STD of the top-$k$ method decreases as $1/\sqrt{k}$.

\section{A diffusion-reaction model for XEB and fidelity}

\subsection{Dynamics of the XEB and fidelity}\label{app:stat_DR_dynamics}

\emph{Numerical demonstration.}---
To corroborate our predictions based on the diffusion-reaction model, we present the results of our numerical simulations.
First, we confirm that the XEB overestimates the fidelity, and that the discrepancy is larger for higher noise rates, as shown in~\autoref{fig:stat_noise_fid_vs_xeb}.
We find that the fSim ensemble has the smallest XEB-to-fidelity ratio. The reason for this is clear from the diffusion-reaction model: among the three gates we considered, their reaction rates $R$ are similar (between 0.6 and 0.67), but the fSim gate has the largest possible diffusion rate $D=1$, as shown in Table.~\ref{tab:DR}. 

We use this intuition to devise an even better gate, which we call the fSim$^*$. By fixing $D=1$, we find that the fSim$^*$  gate has a larger $R=2/3$. Moreover, these values of $R$ and $D$ are now optimal, which we prove in~\autoref{ssec:fSim}. Thus, fSim$^*$ has the smallest possible discrepancy between the  XEB and the fidelity.

Next, we verify that the average XEB value of our algorithm for a specific circuit architecture (the Sycamore chip) can be very accurately predicted by our diffusion-reaction model. These results are shown in~\autoref{fig:stat_DR_mean}. 
We find that our diffusion-reaction model can predict  even the fine details of the scaling with the system size $N$. For example, in~\autoref{fig:stat_DR_mean}, the rise and fall in the value of XEB is caused by the lattice structure [see~\autoref{fig:stat_noise_fid_vs_xeb}(a)] and its effect on the diffusion process.

\begin{figure}[h]
\includegraphics[width=7cm]{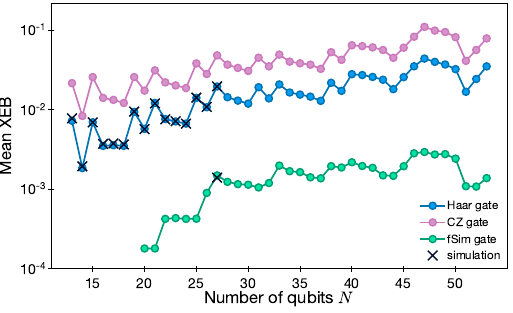}
\centering
\caption{Mean XEB value obtained by our algorithm as a function of the system size $N$, using the  ordering in~\autoref{fig:stat_noise_fid_vs_xeb}(a) and the $d=14$ circuit architecture from Ref.~\cite{arute2019quantum}. The average XEB values are calculated using the diffusion-reaction model for three different gate ensembles: Haar (blue), CZ (purple), and fSim (green).
When compared with the results of the direct simulation of quantum circuits (crosses), both methods agree very well.}
\label{fig:stat_DR_mean}
\end{figure}

\subsection{Detailed derivation of the diffusion-reaction model}
\label{app:DR}
In this section, we assume the gate ensemble consists of single-qubit Haar random gates (or more generally, single-qubit unitaries with the 2-design property) and any 2-qubit entangling gates,
present a detailed derivation of the diffusion-reaction model and discuss some properties relevant to the results in the main text. In~\autoref{sapp:1qubit}, we briefly review the properties of the unitary 2-design ensemble consisting of single-qubit unitaries, and use them in subsequent subsections to derive the diffusion-reaction model. In~\autoref{sapp:DR_transfer}, we show that the transfer matrix $T^{(G)}$ in the resulting diffusion-reaction model has the form of~\autoref{eq:transfer_general}; we also discuss the physical meaning of its parameters. Then, in~\autoref{sapp:DR_stationary}, we compute the stationary distribution using $T^{(G)}$  and present numerical evidence which shows that depth 20 in the Sycamore architecture is sufficiently deep to reach the equilibrium distribution.  In~\autoref{sapp:DR_non_ideal}, we study the effects of introducing defective gates, including noisy gates and omitted gates. We also discuss how to detect the type of noise present in the system by an algorithm similar to the one used for spoofing the XEB. Finally, in~\autoref{sapp:DR_alg}, we analyze the scaling behavior of our classical algorithm and its dependence on  the properties of omitted gates and the transfer matrix $T^{(G)}$.
We show that, using the language of the diffusion-reaction model, even the fine details of the $N$-dependence (e.g., ups and downs in~\autoref{fig:stat_DR_mean}) can be explained in an intuitive way.

\subsection{Brief review of 2-design properties of single qubit unitaries}\label{sapp:1qubit}
We explain how the behavior of quantum circuits, averaged over an ensemble of unitary gates, can be expressed in a simple form.
In particular, we will consider averaging a single-qubit unitary gate over the Haar ensemble, which is a uniform distribution over all unitaries in $\textsf{SU}(2)$.

Consider quantum states $\otimes_{i=1}^t \rho_i$, in the $t$ copies of a Hilbert space, undergoing the same unitary evolution $u$.
We are interested in the resultant quantum state averaged over the Haar random unitary $\mathbb E_{u\in\text{Haar}}\left[\bigotimes_{i=1,\cdots, t}u\rho_i u^\dag\right]$. 
The ensemble of Haar-random unitaries is defined by the invariance of the averaged quantity by both the left and right multiplications of any unitary $v\in \textsf{SU}(d)$; i.e., for any $t\in \mathbb{Z}_{+}$,
\begin{equation}\label{Seq:HaarId}
 \mathbb E_{u\in \text{Haar}}\left[ \left(u\otimes u^*\right)^{\otimes t}\right]=\mathbb E_{u\in \text{Haar}}\left[ \left(vu\otimes v^*u^*\right)^{\otimes t}\right] = \mathbb E_{u\in \text{Haar}}\left[ \left(uv\otimes u^*v^*\right)^{\otimes t}\right].
\end{equation}
This is a natural definition for a uniform distribution: if the ensemble is uniformly distributed, any application of extra rotation by $v$ should only ``permute'' the elements of $\textsf{SU}(d)$ from $u\mapsto vu$ or $u\mapsto uv$, and the average should not be affected.
We focus solely on $d=2$ (i.e., qubits) in this work.
By considering the average behavior of the $t$-copy wavefunction under the same random unitary $u$, we can study the behavior of observables in the extended ($t$-copy) Hilbert space, which contains  observables that are nonlinear (up to power $t$) in a single-copy density matrix. 
In this work, since we focus on the expectation values of the XEB and the fidelity, it suffices to study the case where $t=1$ and $t=2$. In other words, the Haar ensemble can be replaced by any other ensemble of unitaries that behaves identically to the Haar ensemble in $t=1$- and $t=2$-copy Hilbert spaces on average, which is the defining property of the so-called \emph{unitary 2-design}.

First, we consider the case of $t=1$ to introduce some useful notations and identities, which will be helpful for the $t=2$ case. Let $u$ be sampled uniformly from $\textsf{SU}(2)$ and $\rho_1$ be any 2-qubit input state. Observe that the quantity $\mathbb E_{u\in\text{Haar}}[u\rho_1 u^\dag]$ is invariant under the action of an arbitrary $v\in\textsf{SU}(2)$, i.e., $\mathbb E_{u\in\text{Haar}}[u\rho_1 u^\dag]=\mathbb E_{u\in\text{Haar}}[uv\rho_1v^\dag u^\dag]=\mathbb E_{u\in\text{Haar}}[vu\rho_1u^\dag v^\dag]$. Thus, we have
\begin{equation}\label{eq:1-design_og}
\mathbb E_{u\in\text{Haar}}[u\rho_1 u^\dag]=(\tr\rho_1)\frac{I}{2}
\end{equation}
because only $\tr\rho_1$ and $I$ are the invariant quantities with respect to $\textsf{SU}(2)$. 
We adopt two standard representations of quantum many-body states that will be mathematically convenient for later calculation: diagram (also known as tensor network representation~\cite{bridgeman2017hand}) and Choi representation \cite{choi1975completely} (which is already used in the main text and also known as Choi-Jamio\l kowski isomorphism). The former provides intuitive graphics for gates and states while the later associates a density matrix to an ``entangled state" by the map~\cite{watrous2018theory} 
$$\text{vec:}\sum_{ij}\rho_{ij}\ket i\bra j\mapsto\sum_{ij}\rho_{ij}\kket{ij}.$$ We use $\kket\rho=\text{vec}(\rho)$ to represent this ``entangled state". Then 
\begin{equation}
\text{vec}(u\rho u^\dag)=u\otimes u^*\kket{\rho} \text{ and } \tr\rho=\sqrt2\langle\innerp{\text{Bell}}{\rho}\rangle
\end{equation}
where $\kket{\text{Bell}}=\text{vec}(I)/\sqrt{2}=(\kket{00}+\kket{11})/\sqrt2$ is the ``Bell state".
Thus~\autoref{eq:1-design_og} can be rewritten as
\begin{equation}\label{eq:1-design}
\mathbb E_{u\in\text{Haar}}[u\otimes u^*\kket{\rho_1}]=\kket{\text{Bell}}\langle\langle\text{Bell}\kket{\rho_1},
\end{equation}
the diagram of which is shown in~\autoref{fig:12design}(a), where a line denotes $\text{vec}(I)$.
Focusing on the effect of unitaries averged over an ensemble, we can identify 
\begin{equation}\label{eq:1_design}
\mathbb E_{u\in\text{Haar}}[u\otimes u^*]=\kket{\text{Bell}}\bbra{\text{Bell}},
\end{equation}
which is a projector to the ``Bell state". Here we use the variation of the Dirac notation $\kket\cdot$ ($\bbra\cdot$) to represent the vector (dual vector) in the Choi representation, instead of an ordinary quantum mechanical state.
The conclusion is that, if we only have a single copy of a quantum state, all directional information on the Bloch sphere is erased after averaging over the Haar ensemble, except the normalization condition (i.e., $\tr\rho_1$ is preserved). The output state is always the maximally mixed state (or equivalently $\kket{\text{Bell}}$, in terms of the Choi representation), no matter what the initial state is.

\begin{figure*}
\includegraphics[width=0.8\textwidth]{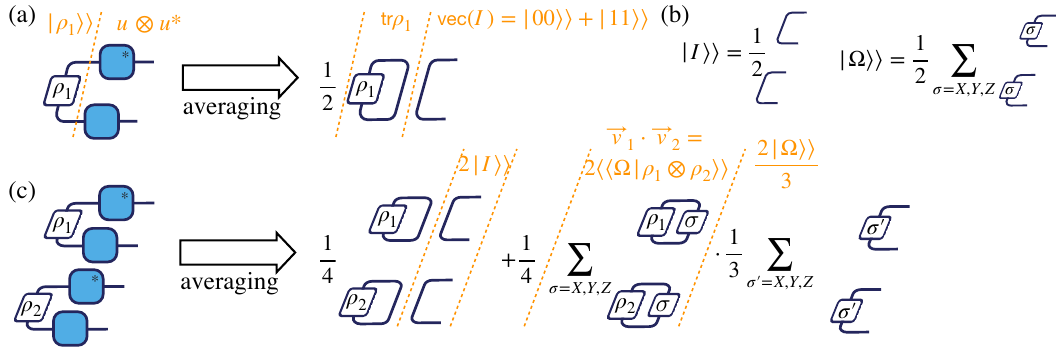}
\caption{
Diagram of averaging over Haar ensemble. The blue box represents a single qubit Haar random unitary $u$ and the blue box with a $^*$ represents $u^*$.
(a) Averaging for $t=1$ copy (i.e., unitary 1-design property). The diagram can be understood as either a tensor network representation or a circuit (non-unitary after averaging) of a state in Choi representation in which $\tr\rho_1=\langle\innerp{00}{\rho_1}\rangle+\langle\innerp{11}{\rho_1}\rangle$ . (b) Diagram of $\kket{I}$ and $\kket{\Omega}$ which are defined in~\autoref{Seq:I_Omega}. (c) Averaging for $t=2$ copies (i.e., unitary 2-design property) which is a diagram of~\autoref{eq:app_2design} by using (b).
}
\label{fig:12design}
\end{figure*}

In the case of $t=2$, as mentioned in the main text, we always have two output states after averaging, denoted as $\kket I$ and $\kket\Omega$ in the Choi representation. In particular, we use these two states as the degree of freedom in the diffusion-reaction model. The central result for a single-qubit gate is
\begin{equation}\label{eq:1qubit}
\mathbb E_{u\in\text{Haar}}[u\rho_1 u^\dag\otimes u\rho_2 u^\dag] = \frac{1}{4}\left(I\otimes I+  \frac{{\vec{\mathrm v}}_1\cdot {\vec{\mathrm v}}_2}{3}{\vec\sigma}\cdot {\vec\sigma} \right) \xrightarrow[]{\text{vec}}\frac{1}{2}\left(\kket{I}+  \frac{{\vec{\mathrm v}}_1\cdot {\vec{\mathrm v}}_2}{3}\kket{\Omega} \right)
\end{equation}
where $\vec{\mathrm v}_i = (\tr{\rho_i X},\tr{\rho_i Y},\tr{\rho_i Z})$ is the vector representation of $\rho_i$ in the Pauli-matrix basis (such that $\rho_i=(I+\vec{\mathrm v}_i\cdot \vec\sigma)/2$), ${\vec\sigma}\cdot {\vec\sigma}=\sum_{\sigma=X,Y,Z}\sigma\otimes \sigma$, $X,Y,Z$ are the 3 Pauli matrices, and 
\begin{eqnarray}\label{Seq:I_Omega}
\nonumber\kket{I}&=& \text{vec}\left(\frac{I\otimes I}{2}\right)= \kket{\text{Bell}}^{\otimes2},     \\
\kket{\Omega}&=&\text{vec}\left(\frac{{\vec\sigma}\cdot {\vec\sigma}}{2}\right)=\sum_{\sigma=X,Y,Z}[(\sigma\otimes I)\kket{\text{Bell}}]^{\otimes2} \, .
\end{eqnarray}
This can be explained according to the invariance of the expectation value under the application of an arbitrary unitary $v\otimes v$ to the density matrix  (that is two $v$s and two $v^\dag$s): (1) the output should be a linear combination of $I\otimes I$ and $\vec\sigma\cdot\vec\sigma$ because they are the only invariant 2-qubit operators (up to linear combination
); (2) by computing the expectation value of the trace with $I\otimes I$ or ${\vec\sigma}\cdot {\vec\sigma}$, the coefficients $1/4$ and $\vec{\mathrm v}_1\cdot \vec{\mathrm v}_2/(3\cdot 4)$ can be determined. The appearance of $\vec{\mathrm v}_1\cdot \vec{\mathrm v}_2$ is a consequence of this invariance, since $v$ is mapped to a rotation on the Bloch sphere while this inner product is invariant under $\text{SO}(3)$. See~\autoref{ssapp:Haar} for a more rigorous proof. The conclusion is that all the directional information on the Bloch sphere is deleted after the averaging process, except the normalization condition and the total polarization correlation between the two states, i.e., $\vec{\mathrm v}_1\cdot \vec{\mathrm v}_2$.

The above result can be formulated in the Choi representation (which is the same as~\autoref{eq:2design}):
\begin{equation}\label{eq:app_2design}
 \mathbb E_{{u\in\text{Haar}}} [u\otimes u^*\otimes u\otimes u^*]
=\kket{I}\bbra{I}+\frac{1}{3}\kket\Omega\bbra\Omega,
\end{equation}
where
$$
\langle\innerp{\Omega}{\rho_1\otimes\rho_2}\rangle=\sum_{\sigma=X,Y,Z} \bbra{\text{Bell}}\sigma \otimes I \kket{\rho_1}  \bbra{\text{Bell}}\sigma \otimes I \kket{\rho_2}=\sum_{\sigma=X,Y,Z}\frac{1}{2}\tr(\sigma\rho_1)\tr(\sigma\rho_2)=\frac{\vec{\mathrm v}_1\cdot\vec{\mathrm v}_2}{2}.
$$
The second equality is due to $\tr O=\sqrt2\langle\innerp{\text{Bell}}{O}\rangle$ for an arbitrary operator $O$ (here $O=\sigma\rho_i$ and $\kket O=\text{vec}$($O$)).
 The corresponding diagram is shown in~\autoref{fig:12design}(b,c). Intuitively, $\bbra{I}$ and $\bbra{\Omega}$ encode the normalization information and the total polarization correlation information $\vec{\mathrm v}_1\cdot \vec{\mathrm v}_2$, respectively. $\kket I$ and $\kket{\Omega}/3$ represent the propagation of the corresponding information to the next time step. The $1/3$ factor in $\kket\Omega$ could be understood as the 3 polarization correlations (represented by $\text{vec}(\sigma\otimes \sigma/2)$ with $\sigma=X,Y,Z$) with equal probability being propagated. In the next subsection (subsection \ref{sapp:DR}), we will elaborate on this interpretation in terms of the diffusion-reaction model, where $I$ and $\Omega$ represent vacuum and particle states, respectively.

\subsubsection{Proof of 2-design properties}\label{ssapp:Haar}

In the following, we will prove~\autoref{eq:1-design_og} and~\autoref{eq:1qubit}. They are special cases of the Weingarten formula~\cite{weingarten1978asymptotic} for $d=2$ and $t=1,2$, respectively. In this special situation, we present simple proofs for completeness.

First, we prove~\autoref{eq:1-design_og}. The density matrix $\rho$ of a single qubit can be written in the Pauli basis as follows.
\begin{equation}\label{Seq:densitym}
    \rho=\frac{\tr(\rho) I+\vec{\mathrm v}\cdot \vec\sigma}{2},
\end{equation}
where $\vec{\mathrm v}$ is a 3-dimensional vector and the Pauli matrices $\vec\sigma=(X,Y,Z)$.
By expanding the expression of $\rho$ in this way, it suffices to understand $\mathbb{E}_{u\in\text{Haar}}[uIu^\dagger]$ and $\mathbb{E}_{u\in\text{Haar}}[u\sigma u^\dagger]$ for all $\sigma\in\{X,Y,Z\}$. First, $\mathbb E_{u\in\text{Haar}}[u\rho u^\dag]$ is straightforwardly
\begin{equation*}
   \mathbb E_{u\in\text{Haar}}[uI u^\dag]= I. 
\end{equation*}
For each $\sigma\in\{X,Y,Z\}$, we use~\autoref{Seq:HaarId} with $v=\sigma^{\prime}\in\{X,Y,Z\}\backslash\{\sigma\}$,
\begin{align}
   \mathbb E_{u\in\text{Haar}}[u \sigma u^\dag]&=\mathbb E_{u\in\text{Haar}}[u \sigma^\prime\sigma\sigma^\prime u^\dag]\nonumber \\
   &= -\mathbb E_{u\in\text{Haar}}[u \sigma u^\dag] \nonumber\\ 
   &= 0,\label{Seq:1des2} 
\end{align}
where we used the identity $\sigma^\prime\sigma\sigma^\prime= -\sigma$ for $\sigma\neq \sigma^\prime$. Putting the above results together gives
\begin{equation*}
    \mathbb E_{u\in\text{Haar}}[u \rho u^\dag] = \tr(\rho)\frac{I}{2},
\end{equation*}
which proves~\autoref{eq:1-design_og}. .

Next, we prove the 2-design property from~\autoref{eq:1qubit}. Using the parametrization from~\autoref{Seq:densitym}, the tensor product of two density matrices is
\begin{eqnarray*}
\rho_1\otimes \rho_2 &=& \frac{I\otimes I}{4}\\
&+& \frac{I\otimes \vec{\mathrm v}_2\cdot \vec\sigma}{4}+\frac{ \vec{\mathrm v}_1\cdot \vec\sigma\otimes I}{4}\\
&+& \sum_{\sigma_1\ne \sigma_2\in\{X,Y,Z\}}\frac{\mathrm v_1^{(\sigma_1)} \mathrm v_2^{(\sigma_2)}}{4}\sigma_1\otimes\sigma_2\\
&+& \sum_{\sigma\in\{X,Y,Z\}}\frac{\mathrm v_1^{(\sigma)}\mathrm v_2^{(\sigma)}}{4}\sigma\otimes \sigma.
\end{eqnarray*}
where we used the notation, where $\mathrm v_i^{(X)}$ is the $x$-th component of $\vec{\mathrm v}_i$.
The Haar-average of the first line is simply
$$
\mathbb E_{u\in\text{Haar}}[uI u^\dag \otimes uI u^\dag]=I\otimes I.
$$
The terms in the second line become zero after averaging due to the same reasoning as in~\autoref{Seq:1des2}; as an example, the first one is
\begin{equation}\label{Seq:2-design_og}
\mathbb E_{u\in\text{Haar}}[(uI u^\dag) \otimes (u \vec{\mathrm v}_2\cdot \vec\sigma u^\dag)]=I \otimes \mathbb E_{u\in\text{Haar}}[ u \vec{\mathrm v}_2\cdot \vec\sigma u^\dag]
=0.
\end{equation}
The average of the third-line term ($\sigma_1\neq \sigma_2$) also vanishes because
\begin{equation*}
\mathbb E_{u\in\text{Haar}}[u\sigma_1 u^\dag \otimes u \sigma_2 u^\dag]=E_{u\in\text{Haar}}[u(\sigma_1)^3 u^\dag \otimes u \sigma_1\sigma_2\sigma_1 u^\dag]
=-E_{u\in\text{Haar}}[u\sigma_1 u^\dag \otimes u \sigma_2 u^\dag]
= 0,
\end{equation*}
where we again used~\autoref{Seq:HaarId} with $v=\sigma_1$ and $\sigma_1\sigma_2\sigma_1= - \sigma_2$ for $\sigma_1\neq \sigma_2$.
Finally, for the fourth term, we have
\begin{equation}\label{Seq:2 design symmetry of pauli}
\mathbb E_{u\in\text{Haar}}[uX u^\dag \otimes u X u^\dag]=\mathbb E_{u\in\text{Haar}}[uY u^\dag \otimes u Y u^\dag]
=\mathbb E_{u\in\text{Haar}}[uZ u^\dag \otimes u Z u^\dag],
\end{equation}
which can be seen from the invariance under unitary rotations, namely $Y$ and $Z$ operators are related to $X$ operator by unitary transformations (e.g. Hadamard and $\pi/2$-phase gate).
Then we use the identity
\begin{equation}\label{Seq:SWAP}
    2S=I\otimes I + X\otimes X+Y\otimes Y+ Z\otimes Z,
\end{equation}
where $S$ is a SWAP operator. Crucially, $S$ commutes with any tensor product of two identical operators
\begin{equation*}
    (u\otimes u)S = S(u\otimes u),
\end{equation*}
by definition. %which can be verified by direct matrix multiplication. 
This means that $\mathbb E_{u\in\text{Haar}}[u\otimes u  S\, u^\dag \otimes u^\dag] = S$ and by~\autoref{Seq:2 design symmetry of pauli} and~\autoref{Seq:SWAP} we have
\begin{align}
    &\mathbb E_{u\in\text{Haar}}[(u\otimes u)  (X \otimes X) \, (u^\dag \otimes u^\dag)]\nonumber\\
    &=\mathbb E_{u\in\text{Haar}}[(u\otimes u)  (Y \otimes Y) \, (u^\dag \otimes u^\dag)]\nonumber\\
    &=\mathbb E_{u\in\text{Haar}}[(u\otimes u)  (Z \otimes Z) \, (u^\dag \otimes u^\dag)]\nonumber\\
    &= \mathbb E_{u\in\text{Haar}}\left[(u\otimes u)  \frac{X\otimes X + Y \otimes Y + Z\otimes Z}{3} \, (u^\dag \otimes u^\dag)\right]\nonumber\\
    &= \mathbb E_{u\in\text{Haar}}\left[(u\otimes u)  \frac{2S-I\otimes I}{3}\, (u^\dag \otimes u^\dag)\right]\nonumber\\
    &= \frac{2S-I\otimes I}{3}= \frac{X\otimes X+Y\otimes Y+ Z\otimes Z}{3}.\label{Seq:heiseninv}
\end{align}%\\&=&
We use this formula to obtain the expression for the fourth term
\begin{equation*}
    \sum_{\sigma\in\{X,Y,Z\}}  \mathrm v_1^{(\sigma)}\mathrm v_2^{(\sigma)} \mathbb E _{u\in\text{Haar}}[(u\otimes u) \sigma\otimes \sigma (u^\dag \otimes u^\dag)] 
    = \vec{\mathrm v}_1\cdot \vec{\mathrm v}_2\frac{\vec\sigma\cdot\vec\sigma }{3}.
\end{equation*}
Putting all these results together, we
proved~\autoref{eq:1qubit}.

\begin{figure*}[tbp]
\includegraphics[width=1\textwidth]{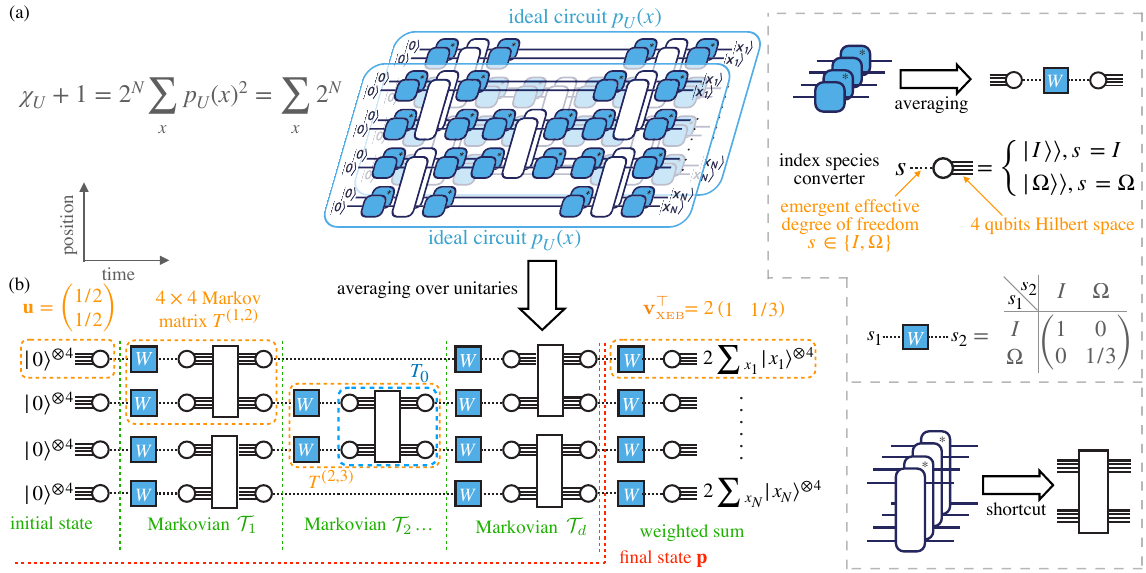}
\caption{
Illustration of the mapping from the average XEB of the ideal circuit to the diffusion-reaction model. (a) The XEB of ideal circuit can be computed by considering two copies of a state evolving under the same random quantum circuit.
Gray dashed boxes defines simple representations of tensors that appear in our tensor network diagrams. The first box gives a shortcut of the diagram for the average behavior of each single qubit gate as discussed in~\autoref{fig:12design}(c) and~\autoref{eq:app_2design}.
A circle is labelled by a classical variable $s\in \{I, \Omega\}$ and represents the corresponding ``4-qubits states" (or vectorized density matrix in duplicated Hilbert space) $\kket I$ and $\kket\Omega$ defined in~\autoref{fig:12design}(b) and~\autoref{eq:app_2design} (denoted as 4 lines). The $W$ is a diagonal matrix defined for the classical degree of freedom gives.
The second box defines a simplified diagram for the four copies of a 2-qubit entangling gate.
(b) The tensor network of the diffusion-reaction model where the horizontal direction is viewed as time evolution of a Markovian process, described by $\mathcal T_1,\cdots,\mathcal T_d$, on the classical degree of freedom.
In each $\mathcal T_i$, a matrix $T_0$ (where we omitted the gate dependence $^{(G)}$ here displayed in main text for $T_0^{(G)}$) on the classical degree of freedom is defined as the two copies of 2-qubit entangling gate combined with the white circles.
Then combining $T_0$ with $W$, we get the transfer matrix $T$ which is a Markovian (will be proved in subsection \ref{sapp:DR_transfer}).
In (a), there are two sets of independent single qubit Haar random gates on a wire between two successive entangling gates. They can be merged into one because the product of two independent Haar random unitary gates equals a single Haar random unitary gate. This is why we only have one layer of  $W$'s between two successive layers of entangling unitary gates in (b). Here we only present the tensor network diagram for XEB, fidelity only differs at the right boundary condition as discussed in~\autoref{fig:stat_mapping_outline}(a) and (b).
}
\label{fig:DR}
\end{figure*}

\subsection{Deriving the diffusion-reaction model}\label{sapp:DR}
In this subsection, we present a detailed derivation of the diffusion-reaction model. We consider the average XEB of an ideal circuit:
\begin{equation}
\chi_{\text{av}}=\mathbb E_{U\in\text{Haar}^{\otimes N_\text{\tiny single}}}\left[2^N\sum_xp_U(x)^2-1\right],
\end{equation}
where we use $\chi_{\text{av}}$ to denote $\mathbb E_{U}[\chi_U]$, and $N_\text{\tiny single}$ is the number of single-qubit Haar gates $u$.
The tensor network diagram representing this quantity is shown in~\autoref{fig:DR}(a). By applying the 2-design properties (inserting~\autoref{eq:app_2design}) in the middle of two successive layers of entangling gates, i.e., applying the upper-right gray box in~\autoref{fig:DR} to each single-qubit gate, we get a path integral of the diffusion-reaction model in terms of only $\{I,\Omega\}$ variables shown in~\autoref{fig:DR}(b). The path integral turns out to be a Markovian evolution, in the sense that each 2-qubit gate is mapped to a transition matrix $T_0$ over the state space $\{I,\Omega\}^2$, and each single-qubit gate is mapped to a weighted diagonal matrix $W$ (see the gray boxes in~\autoref{fig:DR}); then, we combine two $W$s and $T_0$s together and define $T$ to be the transition matrix over the state space $\{I,\Omega\}^2$ as follows:
\begin{equation}
T=T_0 (W\otimes W).
\end{equation}
We can show that this $T$ is indeed a stochastic matrix (\autoref{sapp:DR_transfer}).
When these gates are applied to the $(i,j)$ qubit pair, we denote $T^{(i,j)}$ to be the corresponding transition matrix over the state space $\{I,\Omega\}^2$. Also, we let $\mathcal{T}_t=\otimes_{(i,j)\in t\text{-th layer}}T^{(i,j)}$ be the transition matrix of the $t$-th layer of the circuit.
To sum up, $\chi_{\text{av}}+1$ can be written as a Markovian evolution in terms of the transition matrices $\mathcal{T}_1,\dots,\mathcal{T}_d$ over the state space $\{I,\Omega\}^N$, with appropriate  boundary conditions as described in~\autoref{eq:ini_vec}. See also~\autoref{eq:chi_U_avg} for a summary of the whole diffusion-reaction process. 

As discussed in~\autoref{fig:stat_mapping_outline}(a) and (b), the corresponding diffusion-reaction model is the same for the fidelity, except the right boundary condition which is given in~\autoref{eq:ini_vec 2}; see also~\autoref{eq:F_U_avg}.

\subsection{Properties of the transfer matrix $T^{(G)}$}\label{sapp:DR_transfer}
In this subsection, we discuss in more detail the properties of the transfer matrix $T$.
Here and below, we omit the superscript in $T^{(G)}$ in order to simplify our notations whenever doing so does not lead to ambiguity.
We will focus on the connection between $T$ and the amount of entanglement generated by an entangling gate.
\begin{figure}[tbp]
\centering
\includegraphics[width=\textwidth]{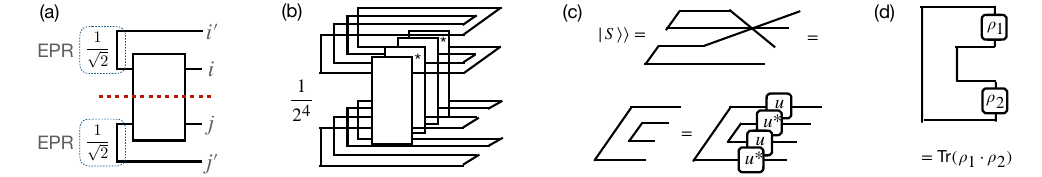}
\caption{``Apparent entanglement" and $\kket S$. (a) We consider two qubits $i$ and $j$ that are initially maximally entangled with their respective partners $i'$ and $j'$. The ``apparent entanglement'' is defined by the amount of the entanglement  (quantified by the smallness of the purity of reduced density matrices) between $ii^\prime$ and $jj^\prime$ that is generated by a unitary gate acting on $i$ and $j$.
(b) Tensor network diagram representation of  $\tr(\rho_{i^\prime i}^2)$, where $\rho_{i^\prime i}$ is the reduced density matrix of the subsystem labeled by $i$ and $i'$.
(c) Diagramatic representations of the state $\kket S$ and its invariance under the action of $u\otimes u^*\otimes u\otimes u^*$. The SWAP operator is applied to the first and third lines (could be also between the second and fourth). Here we only display the invariance under a single-qubit unitary, it is straightforward to see the invariance under 2-qubit unitaries for $\kket S^{\otimes2}$. (d) $\langle\innerp{\rho_1\otimes\rho_2}{S}\rangle=\tr{\rho_1\rho_2}$. }
\label{fig:app_D}
\end{figure} 

Recall the expression for $T$, presented in~\autoref{eq:transfer_general}, which we reproduce here for convenience
\begin{equation*}
T=\begin{pmatrix}
1 & 0 & 0 & 0\\
0 & 1-D & D-R & R/\eta\\
0 & D-R & 1-D & R/\eta\\
0 & R & R & 1-2R/\eta
\end{pmatrix}.
\end{equation*}
This matrix has the following properties, which we prove later in this section:
\begin{itemize}
\item The first column and the first row are all zero except the first entry. This reflects the fact that the polarization correlation can only be produced by the propagation from the polarization correlation with other qubits. Similarly, this holds for the reverse process. In terms of the diffusion-reaction model, a particle cannot be created or annihilated from vacuum. The interaction with other particles is necessary in order to change the particle number.
\item $T$ is symmetric with respect to switching the two sites (exchanging the second and third columns and rows). This reflects the fact that $T$ describes the process of entanglement changes, and that entanglement is a  concept symmetrical between the two qubits.
\item We define the quantities $R=T_{I\Omega\rightarrow \Omega\Omega}$ and $R/\eta=T_{\Omega\Omega\rightarrow I\Omega}$ to study the reaction process. We prove that the reaction ratio $\eta=3$ for every 2-qubit gate is set. This reflects the fact that there are 3 species corresponding to $\Omega$ (3 polarization directions) while there is only 1 species for $I$. In the next subsection, we will show that $\eta$ almost fully determines the stationary distribution $\mathbf p_\infty=\lim_{d\rightarrow\infty}\mathbf p$ (except some retrograde cases like $D=0$ or $R=0$) under the evolution of $\mathcal{T}_1,\cdots,\mathcal{T}_d$, when $d$ is large enough, for arbitrary circuit architectures in subsection.
\item Since $\eta=3$ is independent of the choice of the gate set, the reaction rate $R$ fully determines the reaction process. We prove that $R$ is quantitatively related to the ``2-body entanglement productivity" (also denoted as ``entanglement power"~\cite{zanardi2000entangling}) for $G^{(i,j)}\ket{\psi_i}\ket{\psi_j}$, which is defined as
\begin{equation}\label{eq:2_body_ent}
 S_2=1-\mathbb E_{|\psi_i\rangle,|\psi_j\rangle\in\text{Haar}}[\tr\rho_i^2],
\end{equation}
where $\rho_i$ is the reduced density matrix of the subsystem labeled by $i$. We have
\begin{equation}\label{eq:2_body_R}
R=3S_2.
\end{equation}
It can be shown that $0\le R\le 2/3$ (according to the result in Ref.~\cite{vatan2004optimal} that ``entanglement power" $S_2$ is at most 2/9). 
In terms of the diffusion-reaction model, $R$  characterizes the ability to change the particle number.
Generally, larger $R$ implies faster equilibration to the stationary distribution $\mathbf p_\infty$. 
\item We define the diffusion rate $D=1-T_{I\Omega\rightarrow I \Omega }$ to describe the process of particle diffusion or, equivalently, the random walk speed (by noting that the duplication process also includes a movement: e.g., $I\Omega \rightarrow \Omega\Omega$ should be viewed as the second particle is moved over 1 site and then duplicated). Intuitively, the more entanglement the 2-qubit gate can produce, the easier the polarization correlation propagates (the faster the particles move). In fact, we prove that
\begin{equation}\label{eq:apparent_D}
D=\frac{4}{3}S_\text{a} \text{ with }
S_\text{a}=(1-\tr\rho_{i^\prime i}^2),
\end{equation}
where $\rho_{i^\prime i}$ is the reduced density matrix of the subsystem labeled by $i$ and $i'$, as explained in~\autoref{fig:app_D}.
Note that $0\le D\le 1$ because $0\le S_\text{a}\le 3/4$ (3/4 can be fulfilled when all the eigenvalues of the reduced density matrix are $1/4$). The quantity $S_\text{a}$ represents the ``apparent entanglement productivity'' of the state $G^{(i,j)}\ket{\text{Bell}}_{i^\prime i}\ket{\text{Bell}}_{jj^\prime}$,
as shown in~\autoref{fig:app_D}(a). 
We say ``apparent'' because this quantity measures the effective bond dimension that increases when the gates is applied in the tensor network state before optimization (e.g., before truncating the bond dimension by the SVD decomposition~\cite{white1992density,vidal2003efficient,verstraete2008matrix,paeckel2019time,zhou2020limits}). However, it does not always characterize the true entanglement productivity.

We note that $D$ and $R$ characterize different aspects of an entangling gate.
For example, SWAP has $D=1$ but $R=0$, i.e., $T_{I\Omega\rightarrow \Omega I}=1$ and $T_{I\Omega\rightarrow I \Omega }=T_{I\Omega\rightarrow \Omega \Omega }=0$. 
When only SWAP gates are applied to the initial state $\ket0^{\otimes N}$, there is no entanglement produced--- no matter how many gates are applied. In terms of the  diffusion-reaction model, the particle distribution will never approach the equilibrium $\mathbf p_{\infty}$ in this case.
When $R>0$, larger $D$ implies faster equilibration time because, for example, when there are no particles around a given site, particles from other sites need to come and interact in order for the particle number to increase.
\item Finally, the quantity $T_{I\Omega\rightarrow\Omega I}=D-R$ must be non-negative, which roughly reflects the intuition that the ``apparent entanglement productivity" is larger than the ``2-body entanglement productivity", up to a rescaling factor, because the former is only ``apparent".
\end{itemize}
In summary, the reaction ratio $\eta$ is independent of the choice of a gate set. ``Reaction'', governed by the reaction rate $R>0$, is the only mechanism that changes the particle number and leads to the equilibration to $\mathbf p_\infty$. Therefore, it is necessary to produce the scrambling state. The diffusion rate $D>0$ can accelerate the equilibration if $R>0$. These quantities are essential to understand the realtion between the XEB and the fidelity for non-ideal random quantum circuits, which will be discussed in~\autoref{sapp:DR_non_ideal}.
In the rest of this subsection,
we prove the above properties. In the next subsection, we solve for $\mathbf p_\infty$ and explain why the average XEB $\chi_\text{av}\approx1$ for deep ideal circuits.

\subsubsection{Proofs of these properties}
In the rest of this subsection, to make equations shorter, we use $T_{s_1s_2,s_3s_4}$ to denote $T_{s_3s_4\rightarrow s_1s_2}$. This notation is consistent with the convention of using column vectors to represent a distribution (such that a Markov matrix is applied from the left).
We use the same convention for $T_0$.
First, we prove a few properties of $T_0$, defined in~\autoref{eq:transfer G} or~\autoref{fig:DR}(b): (1) each entry is non-negative; (2) the entry in the first row and the first column is 1; (3) all the other entries in the first row and the first column are 0; (4) this matrix is symmetric; (5) this matrix is invariant under switching of the first site and the second site, i.e., switching the second and third rows and columns.
To prove (1), we denote $\kket{\bar\sigma}^{\otimes2}=[(\bar\sigma\otimes I)\kket{\text{Bell}}]^{\otimes2}$, where $\bar\sigma$ extends $\sigma$ by including the identity $I$,
such that 
$$\kket I=\sum_{\bar\sigma\in\{I\}}\kket{\bar\sigma}^{\otimes2}
\text{ and }\kket\Omega=\sum_{\bar\sigma\in\{X,Y,Z\}}\kket{\bar\sigma}^{\otimes2}
,$$
which could be summarized as
$$
\kket s=\sum_{\bar\sigma(s)}\kket{\bar\sigma(s)}^{\otimes2}
$$
where $s\in\{I,\Omega\}$ and the summation is over $\bar\sigma(s)\in\{I\}$ if $s=I$ and $\bar\sigma(s)\in\{X,Y,Z\}$ if $s=\Omega$. Then,
\begin{eqnarray*}
T_{0; s_1s_2,s_3s_4}&=&
\sum_{\bar\sigma(s_i)}
\bbra{\bar\sigma(s_1)}^{\otimes2} \bbra{\bar\sigma(s_2)}^{\otimes2} (G\otimes G^{*})^{\otimes 2}
\kket{\bar\sigma(s_3)}^{\otimes2} \ket{\bar\sigma(s_4)}^{\otimes2}\\
&=&\sum_{\bar\sigma(s_i)}\left(
\bbra{\bar\sigma(s_1)} \bbra{\bar\sigma(s_2)}(G\otimes G^{*})
\kket{\bar\sigma(s_3)} \kket{\bar\sigma(s_4)}\right)^2\ge0,
\end{eqnarray*}
where $G$ is the 2-qubit gate. 
To prove (2), we choose $s=I$ in the last equation, 
$$T_{0;II,II}=\left(\bbra{\text{Bell}} \bbra{\text{Bell}}(G\otimes G^{*})
\kket{\text{Bell}} \kket{\text{Bell}}\right)^2=(\tr( GG^{\dag}))^2/2^4=1.$$
To prove (3), similarly, as an example, $$T_{0;II,I\Omega}=\sum_{\sigma=X,Y,Z}(\tr(\sigma\otimes I))^2/2^4=0$$ and similarly for other entries with $II$ on the left or right. To prove (4) and (5), we need to prove the following equalities
\begin{eqnarray*}
T_{0;I\Omega,\Omega I}&=&T_{0;\Omega I,I\Omega}\\
T_{0;I\Omega, I\Omega }&=&T_{0;\Omega I,\Omega I}\\
T_{0;I\Omega, \Omega\Omega }=T_{0;\Omega I,\Omega \Omega}&=&T_{0;\Omega\Omega, I\Omega }=T_{0;\Omega \Omega,\Omega I}.
\end{eqnarray*}
As an example, we prove the first equality in detail while the others follow from the same idea. Define $\kket S=\text{vec}(S)$ where $S$ is the SWAP operators (see~\autoref{fig:app_D}(c)). According to~\autoref{Seq:SWAP},
$$
\kket S=\kket I+\kket\Omega,$$
which is basically $\kket\downarrow$ as we introduced in~\autoref{sec:Ising} in which we discuss the mapping to the Ising spin model. Recall the symmetry of the exchange $2\kket I \leftrightarrow \kket S$ (which is equivalent to the permutation symmetry between the two $G$s or the two $G^*$ in $G\otimes G^*\otimes G\otimes G^*$):
\begin{eqnarray*}
T_{0;IS,SI}&=& T_{0;II,II}+T_{0;I\Omega,II}+T_{0;II,\Omega I}+T_{0;I\Omega,\Omega I}=T_{0;II,II}+T_{0;I\Omega,\Omega I}  \\
=T_{0;SI,IS}&=& T_{0;II,II}+T_{0;\Omega I,II}+T_{0;II,I\Omega}+T_{0;\Omega I,I\Omega}=T_{0;II,II}+T_{0;\Omega I,I\Omega}, 
\end{eqnarray*}
where the equalities in the last column are due to terms like $T_{0;II,I\Omega}=0$. Thus, we have proved all the general properties (1)-(5) of $T_0$. Then, we use them to prove properties of $T$.

\emph{Gate set independent properties of $T$.}---
Next, we prove that for all choices of gate sets, the corresponding transition matrix $T$ must have the form shown in~\autoref{eq:transfer_general}. Recall that $T=T_0\cdot W^{\otimes 2}$. The matrix $W$ does not change the first row and the first column of $T_0$, so $T$ also has the properties (1)--(3) of $T_0$. Because of the symmetry of exchanging $I\Omega\leftrightarrow \Omega I$ for $T_0$ and $W^{\otimes 2}$, $T$ also has this symmetry; this explains why $T$ is invariant under exchanging the second and third columns and rows. Then, the last row and the last column are almost the same except an extra factor $\eta=3$ due to the transpose symmetry of $T_0$ and $W$ (which caused the $1/3$ factor). Finally, we need to prove that each column is normalized. Recall that $\kket{S}=\kket{I}+\kket{\Omega}$, and in the following we associate $I$ with $0$ and $\Omega$ with $1$ for convenience of writing equations. For each $s_1,s_2\in\{I,\Omega\}$, the sum of the column indexed by $s_1s_2$ is
\begin{eqnarray*}
\sum_{s_a,s_b}T_{s_as_b,s_1s_2}&=&  T_{SS,s_1s_2}\\
&=& \frac{T_{0;SS,s_1s_2}}{3^{s_1+s_2}} \\
&=& \frac{\langle\innerp{S}{s_1}\rangle}{3^{s_1}}\cdot \frac{\langle\innerp{S}{s_2}\rangle}{3^{s_2}}\\
&=& \sum_{\bar\sigma(s_1),\bar\sigma(s_2)} \frac{\tr\left(\bar\sigma(s_1)^2\right)}{2\cdot3^{s_1}}\cdot \frac{\tr\left(\bar\sigma(s_2)^2\right)}{2\cdot3^{s_2}}\\
&=& \frac{3^{s_1}}{3^{s_1}} \cdot \frac{3^{s_2}}{3^{s_2}}=1,
\end{eqnarray*}
where the third and fourth equalities are due to~\autoref{fig:app_D}(c) and (d), respectively (where the later is obtained by considering $\rho_i=\bar\sigma(s_i)$). These  fully determine the form of $T$, as shown in~\autoref{eq:transfer_general}, where $D$ and $R$ are just two parameters depending on the 2-qubit gate.

\emph{Gate set dependent properties of $T$.}---
Next, we consider gate set dependent terms by linking $D$ and $R$ to the entanglement properties of the 2-qubit entangling gate $G^{(i,j)}$. First, we consider the ``2-body entanglement productivity" $S_2$ shown in~\autoref{eq:2_body_ent},
\begin{eqnarray*}
1-S_2=&&\mathbb E_{u_1,u_2\in\text{1-qubit Haar}}\tr(\rho_2^2)\\
=&&\mathbb E_{u_1,u_2\in\text{1-qubit Haar}}2\bbra I\bbra S(G^{(1,2)}u_1\ket0 u_2\ket0)\\
&&\otimes(G^{(1,2)*}u^*_1\ket0 u^*_2\ket0)\otimes(G^{(1,2)}u_1\ket0 u_2\ket0)\otimes(G^{(1,2)*}u^*_1\ket0 u^*_2\ket0)\\
=&&\mathbb E_{u_1,u_2\in\text{1-qubit Haar}}2\bbra I\bbra S
(G^{(1,2)}\otimes G^{(1,2)*}\otimes G^{(1,2)}\otimes G^{(1,2)*})\\
&&\cdot 
\left(
(u_1\otimes u_1^*\otimes u_1\otimes u_1^* \ket0^{\otimes4})(u_2\otimes u_2^*\otimes u_2\otimes u_2^* \ket0^{\otimes4})\right)\\
=&&2\bbra I\bbra S
(G^{(1,2)}\otimes G^{(1,2)*}\otimes G^{(1,2)}\otimes G^{(1,2)*})
\frac{1}{4}
\left(\kket I+\frac{1}{3}\kket\Omega \right) \left(\kket I+\frac{1}{3}\kket\Omega \right)\\
=&& \frac{1}{2}\left(T_{0;II,II}+\frac{1}{3}T_{0;I\Omega,\Omega I}+\frac{1}{3}T_{0;I\Omega,I\Omega}+\frac{1}{9}T_{0;I\Omega,\Omega \Omega}\right)\\
=&& \frac{1}{2}\left(T_{0;II,II}+\frac{1}{3}T_{0;\Omega I,I\Omega }+\frac{1}{3}T_{0;I\Omega,I\Omega}+\frac{1}{9}T_{0;\Omega \Omega,I\Omega}\right)\\
=&& \frac{1}{2}\left(T_{II,II}+T_{\Omega I,I\Omega}+T_{I\Omega,I\Omega}+\frac{1}{3}T_{\Omega\Omega,I\Omega}\right)
=\frac{1}{2}\left(
1+1-R+\frac{1}{3}R
\right)\\
=&&1-\frac{1}{3}R,
\end{eqnarray*}
where in the second line $2\kket I$ plays the role of a partial trace and $\kket S$ plays the role of matrix multiplication and then trace as shown in~\autoref{fig:app_D}(d); thus it represents the purity $\tr(\rho_2^2)$ where $\rho_2$ represents the reduced density matrix of the second qubits for the output state after $G^{(1,2)}$; the fourth equality is because of $W\cdot\mathbf u$ according to~\autoref{fig:DR}(b); we omitted terms which are 0 in the fifth line and notice the change of the second and fourth terms in the sixth line. Note that this proves the relation between the reaction rate $R$ and the 2-body entanglement productivity, as described in~\autoref{eq:2_body_R}.

Second, we consider the ``apparent entanglement'' $S_\text{a}$, which is defined by the state shown in~\autoref{fig:app_D}(a). According to~\autoref{fig:app_D}(a) and (b),
\begin{eqnarray*}
1-S_\text{a}&=&\frac{1}{4}T_{0;SI,SI}=\frac{1}{4} (T_{0;II,II}+T_{0;\Omega I,\Omega I})=\frac{1}{4} (T_{II,II}+3T_{\Omega I,\Omega I})=\frac{1+3-3D}{4}\\
&=&1-\frac{3D}{4}
\end{eqnarray*}
This proves the relation between the diffusion rate $D$ and the apparent entanglement productivity, as described in~\autoref{eq:apparent_D}.

Direct calculations could give $D$ and $R$ for arbitrary 2-qubit entangling gates like CZ, fSim and fSim$^*$. In the following, we present an example for the 2-qubit Haar ensemble.

\subsubsection{Calculation of $T_\text{Haar}$}
$T_\text{Haar}$ can be found in Ref.~\cite{dalzell2020random},  but for completeness we show the derivation here.
%Finally, we calculate $T_\text{Haar}$ by establishing some relations between $D$ and $R$ and solving the system of equations. 
Denote $\kket{\sigma}=\kket{\bar\sigma}^{\otimes2}$, then $\kket\Omega=\sum_{\sigma=X,Y,Z}\kket{\sigma}$, and consider one of the entries of $T_0$ for $G^{(i,j)}$ chosen from the Haar 2-qubit unitary ensemble
$$
T_{0;I\Omega,I\sigma}=\bbra I\bbra\Omega G^{(i,j)}\otimes G^{(i,j)*}\otimes G^{(i,j)}\otimes G^{(i,j)*}\kket I\kket\sigma
$$
. We can prove
$$
T_{0,\text{Haar};I\Omega,IX}=T_{0,\text{Haar};I\Omega,IY}=T_{0,\text{Haar};I\Omega,IZ}
$$
by using the definition of the Haar random ensemble and inserting $V=I\otimes H,I\otimes HS$ into $G^{(i,j)}$. Then,
$$
T_{\text{Haar};I\Omega,I\Omega}=\frac{1}{3}\sum_{\sigma=X,Y,Z}T_{0,\text{Haar};I\Omega,I\sigma}=T_{0,\text{Haar};I\Omega,IX}.
$$
Next, consider
$$
T_{0,\text{Haar};I\Omega,IX}=T_{0,\text{Haar};I\Omega,ZX}=T_{\text{Haar};I\Omega,\Omega\Omega},
$$
where the first equality is obtained by inserting $V=\text{CZ}$, and the last equality follows similarly with $V=H\otimes I,S\otimes I$. By inserting $V=\text{SWAP}$, we can further prove $T_{\text{Haar};I\Omega,I\Omega}=T_{\text{Haar};I\Omega,\Omega I}$. Together, these formulae show that for $T_\text{Haar}$, $1-D=D-R=R/\eta$, which implies $1-D=1/5,D-R=1/5,R=3/5$.

\subsection{Stationary distribution for ideal circuits}\label{sapp:DR_stationary}

One of the important features of the XEB is that, for ideal circuits, its value approaches 1 in the large-depth limit~\cite{arute2019quantum}. Here, we  reproduce this property using our diffusion-reaction model. This familiar example will set the stage for the more complicated cases of noisy circuits and our classical algorithms.

%Since $\chi_\text{av}+1$ is a weighted sum of the probability distribution at the last layer, as defined in~\autoref{eq:Cpath}, 
Our problem is reduced to computing the stationary distribution $ \mathbf p_\infty=\lim_{d\rightarrow\infty}\mathbf p$ under the evolution of $\mathcal{T}_1,\cdots,\mathcal{T}_d$.
We try the factorizable distribution as an ansatz first, and then show that such an ansatz is the only solution:
\begin{equation}
    \mathbf p_\infty=\bigotimes_{i=1}^N\mathbf u^{(i)}_\infty,
\end{equation}
where $\mathbf u^{(i)}_\infty$ is a proposed single-bit probability in the product distribution ansatz on the $i$-th site, and we assume that they are identical for all sites $i$. For deep-enough circuits, we expect the diffusion-reaction process to reach a fixed point, as is the case in usual Markov processes. Hence, we can write a self-consistent equation for the above ansatz
\begin{equation}
T^{(i,j)}\mathbf u^{(i)}_\infty\mathbf u^{(j)}_\infty=\mathbf u^{(i)}_\infty\mathbf u^{(j)}_\infty.
\end{equation}
This equation has two solutions that are, surprisingly, independent of both $D$ and $R$
\begin{equation}\label{eq:fixed point sol}
{\mathbf u_1}=
\begin{pmatrix}
\frac{1}{1+\eta} \\ \frac{\eta}{1+\eta}
\end{pmatrix}
=
\begin{pmatrix}
\frac{1}{4} \\ \frac{3}{4}
\end{pmatrix}
\text{ or }
{\mathbf u_2}=
\begin{pmatrix}
1\\0
\end{pmatrix}.
\end{equation}
These two vectors fully determine all the solutions of $\mathcal T_d\cdots \mathcal T_2 \mathcal T_1\mathbf p_\infty=\mathbf p_\infty$.

We note that these two vectors are the only solutions to the Markovian dynamics of the system.
This can be seen by using the Perron–Frobenius theorem, which implies that the steady state solution is unique for any Markovian process, as long as the process is ergodic, i.e., all configurations have non-zero transition probabilities to one another.
In our diffusion-reaction model, it can be easily checked that any pair of particle configurations with at least one particle in the system have non-zero transition probabilities upon the multiplication of the transfer matrix for a finite time. 
Therefore, all these configurations form an ergodic sector.
The trivial configuration with no particles in the entire system forms its own ergodic sector.  
Therefore, our Markovian process has at most two stationary solutions, corresponding to ${\mathbf u_1}^{\otimes N}-1/4^N{\mathbf u_2}^{\otimes N}$ and $\mathbf{u}_2^{\otimes N}$. Furthermore, it is not difficult to check that our Markovian dynamics is also aperiodic, i.e., for time steps larger than 2, there is always a non-zero transition probability from one configuration to itself, which avoids the problem of periodic solutions.

The probability of the all-vacuum state in the initial distribution is $1/2^N$, which results in the final probability
\begin{align}
\mathbf p_\infty &= \frac{1-1/2^N}{1-1/4^N}({\mathbf u_1}^{\otimes N}-\frac{1}{4^N}\mathbf {u_2}^{\otimes N})+\frac{1}{2^N}{\mathbf u_2}^{\otimes N}\nonumber\\
&\approx
\left(1-\frac{1}{2^N}\right){\mathbf u_1}^{\otimes N}+\frac{1}{2^N}{\mathbf u_2}^{\otimes N}+O(1/4^N).
\end{align}
Finally, we apply the appropriate boundary conditions to $\mathbf p_\infty$ at the final time (${\mathbf v^\top}_\text{\tiny XEB}^{\otimes N}$) to get the XEB for deep circuits
\begin{eqnarray}\label{eq:ideal_XEB}
\nonumber\chi_{\infty;\text{av}}&=&{\mathbf v^\top}_\text{\tiny XEB}^{\otimes N}{\mathbf p_\infty}-1\\
\nonumber&\approx& \left(1-\frac{1}{2^N}\right) \left(\mathbf v_\text{\tiny XEB}^\top\mathbf u_1 \right)^N+\frac{1}{2^N}  \left(\mathbf v_\text{\tiny XEB}^\top\mathbf u_2 \right)^N-1\\
\nonumber&\approx&  \left(\mathbf v_\text{\tiny XEB}^\top\mathbf u_1 \right)^N+\frac{1}{2^N}2^N-1\\
\nonumber&=&   \left(\mathbf v_\text{\tiny XEB}^\top\mathbf u_1 \right)^N\\
&=&\left(\frac{1}{2}+\frac{3}{2}\cdot\frac{1}{3}\right)^N= 1.
\end{eqnarray}
The small contribution $\mathbf u_2$, which represents the all-vacuum state, cancels out with the $-1$ term in the definition of XEB.

\begin{figure}[tbp]
\centering
\includegraphics[width=0.7\textwidth]{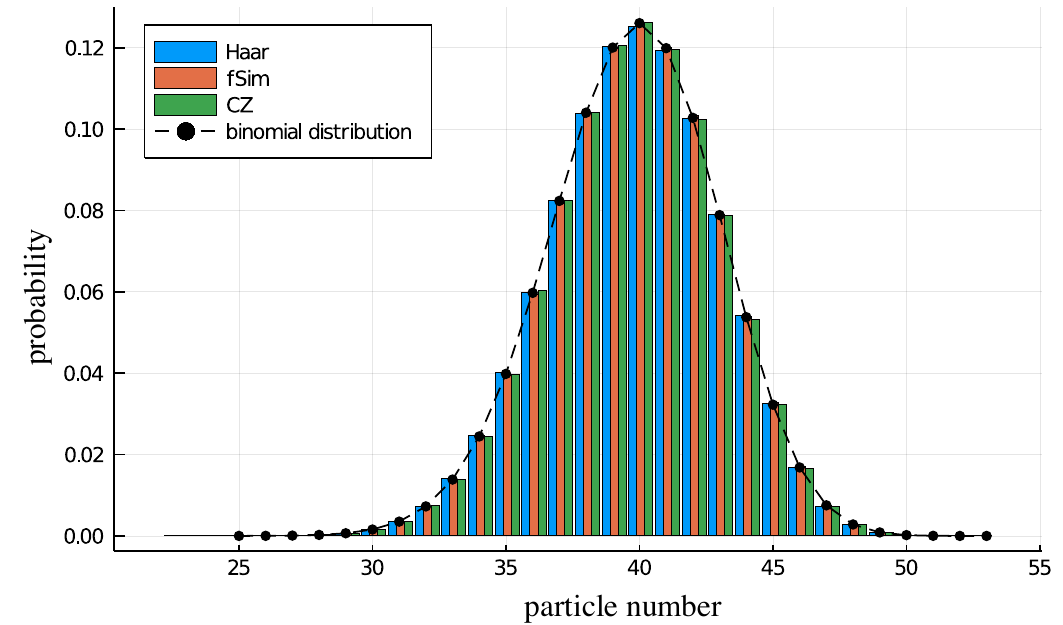}
\caption{The stationary distribution histogram fitted with a binomial distribution $B(53,3/4)$.
We use $10^7$ samples (instances of the diffusion-reaction process) to draw the normalized histogram. The circuits correspond to the Sycamore architecture, with $N=53$ and $d=20$.
}
\label{fig:stationary}
\end{figure} 

Finally, we study the distribution induced by the finite-depth random circuits. We have shown that the stationary distribution is $\mathbf p=(1/4,3/4)^{\otimes N}$ in~\autoref{eq:fixed point sol} up to an exponentially small correction, thus the particle number distribution obeys the binomial distribution $B(N,3/4)$ and this can be used to diagnose whether the depth is large enough for the equilibration to occur (strictly speaking, this is necessary but not sufficient). This can be achieved by simulating the diffusion-reaction model, which is a classical stochastic process and thus much easier to simulate (e.g., using Monte Carlo sampling). For the Sycamore architecture with 53 qubits and depth 20, we numerically test corresponding diffusion-reaction models and find that the particle number distribution can not be distinguished from $B(N,3/4)$, as shown in~\autoref{fig:stationary}, which gives strong evidence that this class of circuits with $d=20$ is deep enough.

\subsection{The effects of defective gates}\label{sapp:DR_non_ideal}
For ideal circuits, the corresponding probability distribution $\mathbf p$ of a pure-state ensemble is normalized to 1 since its average fidelity is 1. However, for mixed states or in the presence of correlations between two non-identical pure states, this is not always the case. Concretely, let $\rho_1$ and $\rho_2$ be two distinct density matrices. We consider the corresponding distribution $\mathbf p$ defined as
\begin{equation*}
p(s_1\cdots s_N)=\langle\innerp{s_1\cdots s_N}{\rho_1\otimes \rho_2}\rangle \text{ such that } \sum_{s_1\cdots s_N}\langle\innerp{s_1\cdots s_N}{\rho_1\otimes \rho_2}\rangle=\langle\innerp{S^{\otimes N}}{\rho_1\otimes \rho_2}\rangle=\tr\rho_1\rho_2.
\end{equation*}
Note that $\mathbf p$ is not always normalized to 1 because the fidelity of $\tr\rho_1\rho_2$ is less than 1 in general. In this subsection, we consider two situations: (1) noisy circuits, i.e., $\rho_1=\ket{\psi}\bra{\psi}$ and $\rho_2$ is its noisy version; (2) our algorithm: $\rho_1=\ket{\psi_1}\bra{\psi_1}$ and $\rho_2=\ket{\psi_2}\bra{\psi_2}$, where $\psi_2$  is a related, but not identical, pure state. In the following, before presenting the results of noisy circuits and our algorithm respectively, we discuss single-qubit examples first, in order to build intuition.

\subsubsection{Noisy gates}

\begin{figure}[tbp]
\centering
\includegraphics[width=1\textwidth]{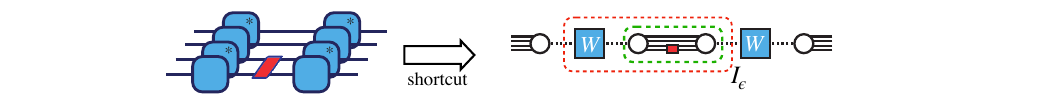}
\caption{The effect of noise. The diagram on the right represents a simplified picture of the left after averaging over two single-qubit Haar random unitaries. The matrix $I_\epsilon$ (appearing in~\autoref{eq:transfer_matrix_noisy}) is defined as the matrix in the space of classical degrees of freedom bounded by the red, dashed box. The part inside the green box is computed as $\bbra{s_1}\hat I\otimes\hat\Phi_\epsilon\kket{s_2}$.}
\label{fig:app_noise}
\end{figure} 

Recall that in~\autoref{eq:1qubit}, the only non-trivial part is $\vec{\mathrm v}_1\cdot\vec{\mathrm v}_2$, which is the coefficient of  $\kket{\Omega}$. For two identical pure states, this inner product is equal to $|\vec{\mathrm v}_1|^2=1$. However, for noisy circuits, e.g. with depolarizing noise, $\vec{\mathrm v}_2=(1-\epsilon)\vec{\mathrm v}_1$, which follows from the density matrix represention in~\autoref{Seq:densitym}; thus, the inner product is equal to $1-\epsilon$. Intuitively, a small amount of polarization correlation is lost.
In terms of the diffusion-reaction model, the noise introduces the reduction of the probability by a factor of $1-\epsilon$ if there is a particle at the given site, which is a probability loss process. Furthermore, the picture of probability loss also works for any other type of noise. For example, for a coherent noise, we have $\vec{\mathrm v}_2\sim (1-\epsilon)\vec{\mathrm v}_1+\sqrt{2\epsilon} \vec{\mathrm v}_1^\perp$ and so the inner product is also $1-\epsilon$. The amplitude damping noise is similar: $\vec{\mathrm v}_2$ is a combination of the displacement, rotation and possibly shrinkage of $\vec{\mathrm v}_1$ by a total amount $1-\epsilon$. In summary, any type of uncorrelated noise appears in the same way in  the diffusion-reaction model.

Below, we provide a more quantitative analysis of the effect of noise.
We denote $\Phi_\epsilon$ as the quantum channel of the noise, and denote its Choi representation as $\hat\Phi_\epsilon$. Suppose a quantum channel in Pauli basis is given by
$$
\Phi_\epsilon(\rho)=\frac{1}{2}\sum_{\sigma_1,\sigma_2\in\{I,X,Y,Z\}}c_{\sigma_1,\sigma_2}\tr(\rho \sigma_1)\sigma_2,
$$
which is basically the Pauli-Liouville representation for quantum channel (see e.g., Ref.~\cite{greenbaum2015introduction}) such that $c_{I,I}=1$ and $c_{\sigma,\sigma}=1-O(\epsilon)$. Similar to $T$ for entangling gate, we could also compute the corresponding matrix element ($I_\epsilon$ in~\autoref{eq:transfer_matrix_noisy}) for $\hat I\otimes \hat \Phi_\epsilon$ (we use $\hat I$ to denote the Choi representation of the identity operation for ideal circuits) in the $I,\Omega$ basis, explicitly:
\begin{eqnarray}\label{Seq:n}
\nonumber\bbra I  \hat I\otimes \hat\Phi_\epsilon \kket I &=& \frac{\tr(I\Phi_\epsilon(I)) \tr(I\Phi_\epsilon(I))}{4}=1\\
\frac{\bbra I  \hat I\otimes \hat\Phi_\epsilon \kket\Omega}{3}&=&\frac{\sum_{\sigma=X,Y,Z}\tr(I\sigma) \tr(\Phi_\epsilon(I)\sigma)}{3\cdot4}=0\\\nonumber
\frac{\bbra \Omega  \hat I\otimes \hat\Phi_\epsilon \kket I}{3}&=&\frac{\sum_{\sigma=X,Y,Z}\tr(\sigma I) \tr(\Phi_\epsilon(\sigma)I)}{3\cdot4}=0\\\nonumber
\frac{\bbra \Omega  \hat I\otimes \hat\Phi_\epsilon \kket\Omega}{3}&=&\frac{\sum_{\sigma,\sigma^\prime=X,Y,Z}\tr(\sigma^\prime\sigma) \tr(\sigma^\prime\Phi_\epsilon(\sigma))}{3\cdot4}=\frac{\sum_{\sigma=X,Y,Z} \tr(\sigma\Phi_\epsilon(\sigma))}{3\cdot2}\\\nonumber
&=&\frac{\sum_{\sigma\in\{X,Y,Z\}}c_{\sigma,\sigma}}{3}=1-O(\epsilon),
\end{eqnarray}
where the $1/3$ factor comes from $W$ [see~\autoref{fig:app_noise}]. 

As examples, we consider depolarizing noise (parameterized by $\mathcal N_\epsilon(\rho)=(1-\epsilon)\rho+\epsilon/3\sum_{\sigma=X,Y,Z}\sigma\rho\sigma$) and amplitude damping noise. For the depolarizing noise, $c_{X,X}=c_{Y,Y}=c_{Z,Z}=1-4\epsilon/3$. For the amplitude damping noise, $c_{X,X}=c_{Y,Y}=\sqrt{1-\epsilon}$ and $c_{Z,Z}=1-\epsilon$, so the second diagonal element is roughly $1-2\epsilon/3$.

\subsubsection{Omitting gates}
Now, we consider the effect of omitting gates. Recall that there are two distinct density matrices $\rho_1$ and $\rho_2$  involved in the definition of $\mathbf{p}$, where the former corresponds to the ideal circuit and the latter would correspond to our algorithm. In $\rho_2$, the single-qubit Haar gates applied to $\rho_1$ are omitted for $\rho_2$ if the corresponding 2-qubit gates in the circuit are removed as part of the algorithm. We recall that each entangling gate is accompanied by 4 single-qubit Haar unitaries, as shown in~\autoref{fig:intro_alg}. 
Therefore, omitting a 2-qubit gate implies that we omit not only the entangling gate but also four single-qubit Haar random gates associated with it. Upon averaging, this belongs to the $t=1$ case in~\autoref{sapp:DR} and so only normalization survives the process. This effectively corresponds to a maximally depolarizing noise: any directional information is deleted. Another way to view this is to think of the remaining $u$ on $\rho_1$ as an extra unitary which is effectively a rotation (denoted as $\hat R$) on $\vec{\mathrm v}_2$. Then, the inner product between this vector, which is denoted as $\vec{\mathrm v}_1=\hat R\vec{\mathrm v}_2$, and $\vec{\mathrm v}_2$ is $\langle\hat R\vec{\mathrm v}_1,\vec{\mathrm v}_1\rangle$. Its average value is clearly 0. In terms of the diffusion-reaction model, this corresponds to a strong probability loss at the position of omitted gates: once a particle hits this region, the probability density of this particle configuration of $I$ and $\Omega$ over the entire space-time is set to 0. Namely, in the diffusion-reaction model, only configurations  without any probability loss would contribute to $\mathbf p$ at the last layer.

Formalizing the above discussion, $I_\epsilon$ appearing in~\autoref{eq:transfer_matrix_noisy}, will be replaced by the projector $P_I$, since this corresponds to the situation of a noisy circuit with maximal depolarizing noise (such that $c\epsilon=1$), according to the 1-design property.

\subsubsection{Detecting noise type by generalizing XEB}
 As a remark, we note that  if one replaces the ideal circuit with a non-trivial quantum channel (as a reference state), our previous results can be used to extract information about the noise type. Mathematically, this changes $\vec{\mathrm v}_1$ (ideal circuit) to $\vec{\mathrm v}_1'$ (non-trivial quantum channel). Then, the inner product between $\vec{\mathrm v}_1'$ and $\vec{\mathrm v}_2$ (noisy circuit whose properties we want to detect) can reveal the information about the noise type. For example, if we introduce a noisy circuit with amplitude damping noise $\Phi_{\epsilon_0}$ as the reference circuit (instead of the ideal circuit), one of the entries in $I_\epsilon$ becomes $\bbra{I}\hat{\Phi}_{\epsilon_0}\otimes\hat{\Phi}_{\epsilon}\kket{\Omega}/3=\epsilon_0\epsilon/12\ne0$, which is different from~\autoref{Seq:n}.

\subsection{Analysis of the scaling behavior of our algorithm through the diffusion-reaction model}\label{app:our_algorithm}\label{sapp:DR_alg}

Applying the diffusion-reaction model we developed, it is intuitive to understand the scaling behavior of our algorithm, as discussed in the main text. Increasing $N$ while keeping the number of omitted gates fixed usually increases XEB for our circuit but the opposite is expected from noisy circuits. This can be explained through the diffusion-reaction model. For our algorithm, larger space for particles undergoing random walk will decrease the probability loss rate, because the particles that are far away from the boundary are less likely to hit the loss region (positions of omitted gates).

Another way to describe the same phenomena is based on the following observation: XEB behaves more like an additive (rather than multiplicative) quantity; thus, we should consider the average loss over particles.
In contrast, fidelity behaves multiplicatively, such that we should consider the total loss over particles, thus increasing the system size will decrease fidelity. In~\autoref{fig:app_scaling}, we present a step-by-step argument regarding the scaling behavior of the XEB for our algorithm.

\begin{figure}[tbp]
\centering
\includegraphics[width=0.7\textwidth]{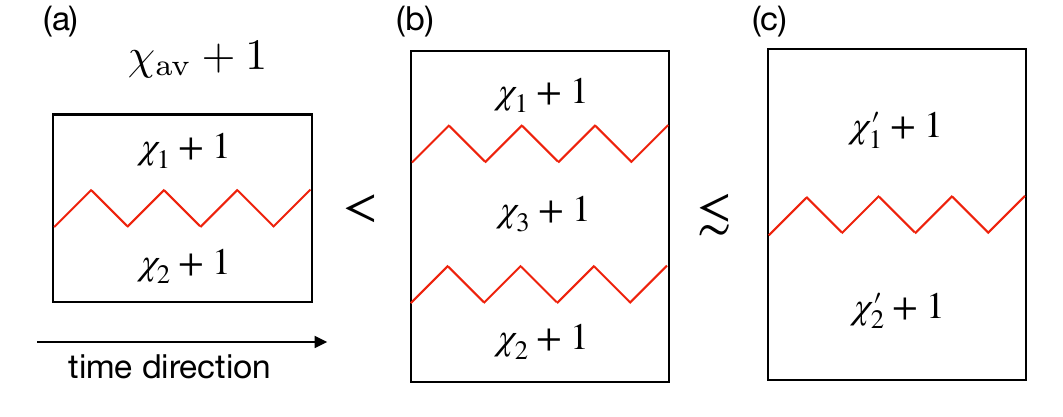}
\caption{Scaling of the XEB when the number of omitted gates is fixed. (a) The behavior of the XEB in our algorithm with two disconnected subsystems. The total $\chi_\text{av}+1$ can be written as $(\chi_1+1)(\chi_2+1) = 1+\chi_1 + \chi_2 + \chi_1 \chi_2$ owing to the decoupling of the diffusion-reaction model by omitting gates.
This quantity is larger than $1+\chi_1$ or $1+\chi_2$.
In fact, the total XEB is approximately additive, $\chi_\text{av}\approx \chi_1 + \chi_2$ when $\chi_1\chi_2$ is small.
(b) Introducing more subsystems only increases XEB, as long as $\chi_3>0$.
(c) It is highly likely that by reducing the number of omitted gates, the XEB can be further increased.}
\label{fig:app_scaling}
\end{figure}

\autoref{fig:stat_DR_mean} shows that the scaling behavior with the system size is quantitatively similar for different gate sets. We observe that the CZ ensemble has much larger XEB than the fSim ensemble. This is because the diffusion rate of the fSim ensemble is the largest (see~\autoref{tab:DR}), which means that particles require the least amount of time to hit the loss region. At the first sight, it might seem this result suggests that the smaller XEB for the fSim ensemble is due to the fact that each fSim gate produces larger (at least ``apparent'') entanglement, or larger bond dimension in the tensor network representation--- thus, it is more dangerous to omit fSim gates. However, this is not correct: if all the omitted gates were fSim gates, but the remaining gates would belong to the CZ ensemble, the mean value of the XEB would be the same as in the case where all the omitted gates were CZ gates. In short, the XEB value does not depend on the properties of the omitted gates, but rather on those in the rest of the circuit.

\subsubsection{Fine structure of the scaling in the Sycamore architecture}

The intuitive picture based on the random walk (diffusion) can explain even finer details of the scaling with the system size $N$. For example, in~\autoref{fig:stat_DR_mean}, the rise and fall in the value of the XEB is caused by the lattice structure [see~\autoref{fig:stat_noise_fid_vs_xeb}(a)] and its effect on the diffusion process.
For examples, the large fall occurring at $N=51$ is caused by adding 2 omitted gates that enlarge the loss region and connect qubit 16 and 46 closer to the loss region.
Other examples include adding qubits 14, 16, 20: they shorten the distance for particles at position 13, 15, 19 from the loss region. The qubits that have only a single connection to the rest of the system (before adding subsequent qubits), such as qubits 15, 27, 47, 52, 53, for example, contribute a lot to increasing the XEB value since particles at these positions have only one way out and are kept away from the loss region. This is reflected in the sudden rises in the XEB value when those qubits are included.

\subsection{An Ising model for 1D Haar ensembles}\label{sec:Ising}

The diffusion-reaction model is useful for analyzing our system qualitatively and numerically, for general circuit architectures and two qubit gate sets. However, for a certain class of systems, such as 1D circuits with Haar two-qubit gates, one can further simplify the classical statistical physics model to the 2D Ising model.
This can be understood as a  special case of the diffusion-reaction model, related to it mathematically through a basis transformation. This mapping has been studied previously in Refs.~\cite{hayden2016holographic, you2018machine, hunter2019unitary,zhou2019emergent, jian2020measurement, bao2020theory, napp2019efficient,dalzell2020random}.
The Ising model allows us to obtain more quantitative results. We find that the behavior of the XEB is related to symmetry, symmetry breaking, and magnetization.

The basis change from the diffusion-reaction model to the Ising model is 
\begin{eqnarray*}
\kket{\uparrow}&=&2\kket I,\\
\kket{\downarrow}&=&\kket I+\kket\Omega
\end{eqnarray*}
such that 
\begin{eqnarray*}
\nonumber  \bbra{a,b,c,d} \uparrow\rangle \rangle&=& \delta_{ab} \delta_{cd},\\
\bbra{a,b,c,d} \downarrow\rangle \rangle
&=&\delta_{ad} \delta_{bc},
\end{eqnarray*}
where the second equation indicates that $\kket{\downarrow}$ corresponds to a swap between indices $a$ and $c$ (or $b$ and $d$). This new basis reflects the symmetry in $u\otimes u^{*}\otimes u\otimes u^{*}$ between the two copies: the state is invariant if we exchange the positions of the two $u$s or $u^{*}$s (labeled by $a,c$ and $b,d$, respectively). 

\begin{figure}[h]
\includegraphics[width=7cm]{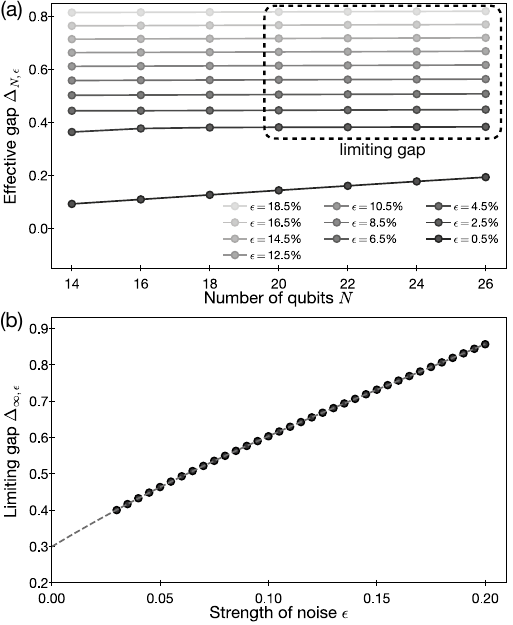}
\centering
\caption{Effective gaps of 1D noisy circuits. (a) For any noise strength $\epsilon$, the gap $\Delta_{N,\epsilon}$ saturates, for sufficiently large $N$, at the  \emph{limiting gap} value $\Delta_{\infty,\epsilon}$. (b) The limiting gap as a function of the noise strength.  Polynomial extrapolation indicates the $\epsilon\to0$ limit of the gap to be $\approx 0.03$. We define this limiting value as $\Delta_3:=\lim_{\epsilon\rightarrow0}\Delta_{\infty,\epsilon}$; in ~\autoref{fig:intro_1D_gaps}, it is represented by the orange, dotted horizontal line. The subsystem considered here has only one boundary with omitted gates as the total system has open boundary condition.
}
\label{fig:1D_mean_noisy}
\end{figure} 

We regard $\uparrow$ and $\downarrow$ as the up and down spins, and the path integral of the diffusion-reaction dynamics is mapped to the partition function of the spin model [see~\autoref{fig:stat_mapping_outline}(d)]. 
In the absence of noise or omitted gates, the partition function has a global $\mathbb Z_2$ Ising symmetry, such that $\kket{\uparrow}\leftrightarrow\kket{\downarrow}$ applied to all spins does not change the partition function.

After the basis change, XEB$+1$ corresponds to the partition function of the $\mathbb Z_2$-symmetric Ising spin model with identical boundary conditions at both the initial and final times. In the special case of Haar entangling gates, this model is the ordinary Ising model with 2-body interactions, which are detailed in the~\autoref{sapp:Ising_model}  and Refs.~\cite{hayden2016holographic, you2018machine, hunter2019unitary,zhou2019emergent, jian2020measurement, bao2020theory, napp2019efficient,dalzell2020random}. 

This mapping allows us to write the XEB for the ideal circuit in the following form:
\begin{equation}\label{eq:1D Ising partition}
\chi_\text{ideal}+1=Z=\bbra\psi \mathcal T_\text{Ising}^{(d-1)/2}\kket\psi,
\end{equation}
where $\kket\psi$ and $\bbra\psi$ are the boundary conditions, and $\mathcal T_\text{Ising}$ is the transfer matrix of the Ising model along the horizontal direction in~\autoref{fig:stat_mapping_outline}(d); it is semi-definite positive and can be computed from $T_0^{(\text{Haar})}$ and~\autoref{eq:2design}. We defer the details of this calculation to the~\autoref{sec:stat_DR_mapping}.  Here, we only need to know that this Ising model is in the ferromagnetic phase. Thus, the largest eigenvalue of $\mathcal T_\text{Ising}$ is doubly degenerate, which gives $Z=2$ and so  XEB$=1$ in the large-$d$ limit.

Once noise or gate defects are introduced, the Ising symmetry is violated. In the case of noisy circuits, the symmetry is violated everywhere, with each local interaction modified by the presence of effective magnetic fields with strength $\epsilon$. Then, there will be a spectral gap $\Delta_{N,\epsilon}=\lambda_1-\lambda_2$ in the modified $\mathcal T_\text{Ising}$, which we evaluate exactly.
Figure~\ref{fig:1D_mean_noisy}(a) shows the gap as a function of the system size for various error rates.
We show that in this case
\begin{equation}
   \chi_\text{noisy}=O\left(e^{-\Delta_{N,\epsilon}d}\right).
\end{equation}
If the violation is small enough ($N\epsilon\ll1$), the spectral gap is $\Delta_{N,\epsilon}\propto N\epsilon$ because the total magnetic field is only a small perturbation from the ideal (symmetric) case. However, if we consider the asymptotic behavior of noisy circuits, $\epsilon$ is assumed constant, but $N$ could be very large. In this limit, the gap will saturate to a fixed value $\Delta_{\infty,\epsilon}$, as shown in~\autoref{fig:1D_mean_noisy}(a). This corresponds to the thermodynamic limit in terms of statistical physics (taking $N\rightarrow\infty$ first then $\epsilon\rightarrow0$). In this case, even if $\epsilon$ tends to 0, as long as $N\epsilon$ is still large, there is a finite gap in $\mathcal T_\text{Ising}$. This corresponds to the phenomena of spontaneous magnetization: even if the magnetic field fades away, most of the spins still point in the same direction leading to a non-vanishing decay rate which is the indicator of the symmetry breaking. We numerically extrapolate the limiting gap to the vanishing noise rate $\epsilon$ and get
\begin{equation}
    \Delta_3=\lim_{\epsilon\rightarrow0}\Delta_{\infty,\epsilon}=\lim_{\epsilon\rightarrow0}\lim_{N\rightarrow+\infty}\Delta_{N,\epsilon}\approx0.3,
\end{equation}
as shown in~\autoref{fig:1D_mean_noisy}(b). This corresponds to the orange dashed line in~\autoref{fig:intro_1D_gaps}.

\begin{figure}[h]
\includegraphics[width=8cm]{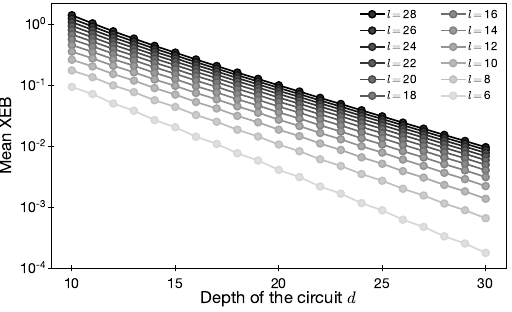}
\centering
\caption{Exponential decay of the average XEB value with increasing circuit depth $d$ for our algorithm. Through linear interpolation (on the semi-log plot), we compute the slope of the lines at each subsystem size $l$ and extract the spectral gap $\Delta_1$. The dependence of $\Delta_1$ on $l$ is shown in~\autoref{fig:intro_1D_gaps} (blue solid curve).
}
\label{fig:1D_mean_our}
\end{figure}

For our algorithm, the omitted gates are mapped to a tensor product of projectors, as shown in~\autoref{eq:TP_projector}, so the partition function will also be separated into the product of partition functions of isolated subsystems
\begin{align}
\chi_\text{algo}=Z-1
&=\prod_{i=1}^{\lceil N/l\rceil} Z_l^{(i)}-1\\
&\approx \prod_{i=1}^{\lceil N/l\rceil} (e^{-\Delta_l^{(i)}d}+1) - 1\\
&\approx \sum_{i=1}^{\lceil N/l\rceil} e^{-\Delta^{(i)}_l d}\sim \frac{N}{l}e^{-\Delta_l d},\label{eq:Isingpart}
\end{align}
where $\Delta_l^{(i)}$ is the gap of the $i$-th subsystem, and $\Delta_l$ is the typical gap among these subsystems, assuming they have similar sizes.~\autoref{eq:Isingpart} shows that the XEB increases with the system size if the subsystem size $l$ is fixed. The decay rate is mainly determined by the subsystem with $\Delta_1=\min_i \Delta^{(i)}_l$. For each subsystem, the omitted gates correspond to strong magnetic fields at the bottom (or top) boundary, which have been previously identified as ``sinks'' in our diffusion-reaction model. These fields violate the $\mathbb Z_2$ symmetry, which causes the gap to open. The gap decreases if the subsystem size $l$ increases; see the discussion in the previous subsection. We numerically compute the gap for different circuit parameters and present the results in~\autoref{fig:1D_mean_our}. We find that when $l\ge15$, $\Delta_1$ approaches to a constant $\Delta_1\approx0.25$. Crucially, we see that $\Delta_1 < \Delta_3$; this means that our algorithm generates a higher XEB value, in the large-depth limit, than noisy 1D circuits --- even with arbitrarily weak noise. 

\emph{A remark.-- In the above discussion, we ignored the factor in front of the exponential decay with depth. In the case of our algorithm, it is a constant (which could depend on the subsystem size $l$) for each subsystem because the subsystem can not distinguish how large of the total system it belongs to. Thus the factor in the total XEB is proportional to system size $N$. In the case of noisy circuit, the factor is possible to grow at most $\poly(N)$ because $d=O(\log N)$ is enough to guarantee the XEB of noisy circuit is less than 1. This is due to that the XEB of noisy circuit should be smaller than ideal circuit, XEB of ideal circuit is exactly the anti-concentration constant and anti-concentration depth is the order of $\log N$~\cite{lightcone,dalzell2020random}.
If the factor grows faster than any polynomial, $d=O(\log N)$ would not be enough to converge to 1. Since the decay rate of XEB in our algorithm is smaller than that of noisy circuit, $d=\Omega(\log N)$ makes former XEB larger.}

We note that many of the qualitative behaviors discussed in this section  also hold in general architectures and two qubit gate set. For example,~\autoref{eq:Isingpart} shows that the XEB obtained by our algorithm behaves more like an additive quantity, i.e., the total XEB approximately equals the sum of XEB values for each subsystems, if they are decoupled (in our algorithm) or only weakly coupled (in noisy circuits).
In contrast, fidelity exhibits multiplicative behavior, i.e., every error contributes to reducing the fidelity of the total system exponentially.

\subsection{The numerical results for STD}\label{sapp:1D_fSim}

In the main text, we mainly focus on the mean value of the XEB. To complete our understanding, we also need to address the fluctuations of the XEB value caused by the random unitaries. If the fluctuations turned out to be much larger than the mean value, then it suggests that our result might not hold for individual instances of random circuits with a large probability. 
If the fluctuation of XEB over different instances of random circuits is sufficiently small, it implies that our algorithm can spoof the XEB for any given randomly choosen instance with a large probability. 
In this section, we numerically estimate the STD of our algorithm in several settings. Additionally, we propose a variant of our algorithm to significantly decrease the STD, however, with the cost of a higher running time.

\begin{figure*}[tbp]
\includegraphics[width=1\textwidth]{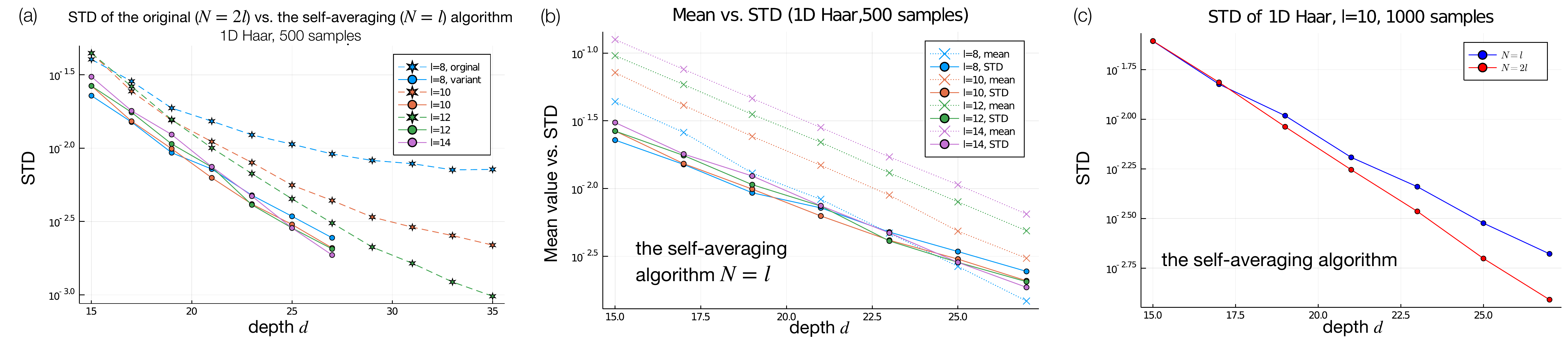}
\caption{(a) The STD of the original algorithm vs. a self-averaging version of our original algorithm (see (b) for more detail) for 1D Haar ensembles with open boundary condition. The former saturates  for a sufficiently large   depth $d$, while the latter does not and is smaller than the former even when the depth is small.
(b) Mean value vs. STD of the self-averaging algorithm (by inserting maximal depolarizing noise instead of omitting gates, see subsection \ref{app:new_algorithm}). Here, the STD is estimated for a subsystem. 
By computing the slope of the solid lines, we extract the decay rate $\Delta_3$, shown in~\autoref{fig:intro_1D_gaps} (green curve). 
(c) Comparison between the STD in (b), i.e., the $N=l$ case and the $N=2l$ case. This indicates that (b) actually overestimates the actual value of the STD.
}
\label{fig:1D_STD}
\end{figure*} 

It is likely that the STD of the original algorithm saturates to a depth-independent value $2^{-O(l)}$, as suggested by~\autoref{fig:1D_STD}(a).
This is expected because there is no mechanism that would decrease fluctuations further below $2^{-O(l)}$ for disconnected evolution in increasing depth (limited by the Hilbert space dimension of the subsystem).
More specifically,  we note that the (sub)system size is the only characteristic length in random circuits~\cite{brandao2016local,harrow2018approximate} (because the subsystems decouple with each other thus one of them are not influenced by other subsystems). For a sufficiently deep circuit, this is analogous to considering a fully scrambling system. There, the variance of observables is indeed depth-independent, since the system wavefunction approaches Haar-random (or more precisely, 4-design) states within each subsystem of size $l$.

For complexity-theoretic purposes, we must consider the limit of deep circuits. Thus, in 1D systems, the original algorithm does not provide a good asymptotic scaling with $d$, because the mean XEB value will eventually drop below the STD value. However, it is still practical for finite-depth systems.

Therefore, in 1D systems, we focus only on the STD of a self-averaging version of our algorithm (see subsection \ref{app:1D_fSim_STD}) and estimate its value numerically. In this case, we expect the fluctuations to decrease with the depth of the circuit since the maximal depolarizing noise on the boundary adds entropy to the system and causes it to decay to the maximally-mixed state. Compared to the original algorithm, the numerical analysis of such mixed state evolution requires further computational resources. To reduce the amount of required computational resources, we focus on the analysis of only one subsystem.
We argue that it overestimates the STD, which means that the true magnitude of fluctuations is even smaller.
This is because the STD of a single subsystem turns out to be smaller than the STD of the joint distribution of $N/l$ identical subsystems that comprise the whole circuit. We demonstrate this numerically in~\autoref{fig:1D_STD}(c) on the example of $N=2l$. Intuitively, this is because the joint system makes the scrambling more complete (for both ideal circuit and the self-averaging algorithm; the later is regarded as a fully connected system but with very strong depolarizing noise at the positions of omitted gates). Thus, the joint system experiences smaller  fluctuations around the typical cases. 

For more discussion of the self-averaging algorithm and a more efficient implementation for fSim gate, see~\autoref{sapp:fSim_trick}.

\section{Properties of circuits with fSim entangling gates}\label{ssec:fSim}

In this section, we discuss several special properties of quantum circuits consisting of single qubit rotations and the fSim gate.
We will see that these properties both improve and obstruct the performance of our spoofing algorithm: on the one hand, fSim gives rise to the optimal ``scrambling speed" such that our original algorithm becomes relatively less efficient;
on the other, we can take advantage of this ``optimal scrambling'' property to design an improved, more efficient algorithm for spoofing the XEB.
To illustrate the role of this ``optimal scrambling'' property, we first study the effect of maximally depolarizing noise in fSim circuits in~\autoref{ssec:fSim identity}, and then analyze the effect of omitted gates on the XEB in~\autoref{ssec:fSim omitting gate effect}.

\subsection{Maximally depolarizing noise
in fSim circuits}\label{ssec:fSim identity}

Here we present a useful property of maximally depolarizing noise (MDN), when applied to fSim circuits.
Formally, MDN is defined as
\begin{align}
    \rho \mapsto \mathcal{D}[\rho] \equiv  \tr{[\rho]} \;  \mathbf{1}/2,
\end{align}
where $\rho$ is the density matrix of a single qubit.
When the MDN is applied twice to a qubit: before and after an fSim gate, its effect is equivalent to removing the fSim gate and applying MDN to both qubits [see Fig.\ref{fig:app_fSim_identity}(a)].
Concretely,
\begin{align}
   \mathcal{D}_1\left[ U_\textrm{fSim}\mathcal{D}_1[\rho_{12}] U_\textrm{fSim}^\dagger\right] = \mathcal{D}_2[\mathcal{D}_1[\rho_{12}]],\label{seq:fsimmdn}
\end{align}
where $\rho_{12}$ is a two-qubit density operator, $\mathcal{D}_i$ is the MDN applied to qubit $i$, and $U_\textrm{fSim}$ is the unitary representing the fSim gate.
In fact, this relation holds much more generally for any unitary gate that is equivalent (up to single-qubit rotations) to $U =U_1 U_2$, where $U_1$ and $U_2$ represent the SWAP and the controlled-phase gate, respectively.

This identity can be understood in the following way.
The SWAP operation moves the MDN from the first qubit to the second qubit, while the controlled phase operation preserves the MDN.
More formally, we can write the action of the unitary $U$ on a density matrix in the tensor notation $U_{b_1c_1,b_2c_2}U^*_{b^\prime_1c^\prime_1,b^\prime_2c^\prime_2}$, and similarly the action of the MDN channel $\delta_{aa^\prime}\delta_{bb^\prime}/2$, where the index with $^\prime$ labels the complex conjugate part. When $U$ is of the above-mentioned ``SWAP$+$control-phase gate'' form, we have  $U_{b_1c_1,b_2c_2}=\delta_{b_1c_2}\delta_{b_2c_1}e^{i\phi_{b_1b_2}}$. Then, the left-hand side of~\autoref{fig:app_fSim_identity}(a) is 
\begin{eqnarray*}
  &&\sum_{b_1,b_1^\prime,c_1,c_1^\prime} \frac{1}{2}\delta_{a_1a_1^\prime}
   \delta_{b_1b_1^\prime}\cdot
   \delta_{b_1c_2}\delta_{b_2c_1}e^{i\phi_{b_1c_2}}
   \delta_{b^\prime_1c^\prime_2}\delta_{b^\prime_2c^\prime_1}e^{-i\phi_{b^\prime_1c^\prime_2}}
   \cdot \frac{1}{2}\delta_{c_1c_1^\prime}
   \delta_{d_1d_1^\prime}\\
   =&&\sum_{c_1,c_1^\prime}\frac{1}{2}\delta_{a_1a_1^\prime}\delta_{c_2c_2^\prime}\delta_{b_2c_1}\delta_{b^\prime_2c^\prime_1}e^{i\phi_{b_1c_2}}e^{-i\phi_{b_1c_2}}\cdot \frac{1}{2}\delta_{c_1c_1^\prime}
   \delta_{d_1d_1^\prime}\\
   =&& 
   \frac{1}{2}\delta_{a_1a_1^\prime}\delta_{d_1d_1^\prime}
    \frac{1}{2}\delta_{b_2b_2^\prime}\delta_{c_2c_2^\prime},
\end{eqnarray*}
where the final result corresponds exactly to the right-hand side of~\autoref{fig:app_fSim_identity}(a) and the statement in~\autoref{seq:fsimmdn}

\begin{figure}[tbp]
    \centering
    \includegraphics[width=\linewidth]{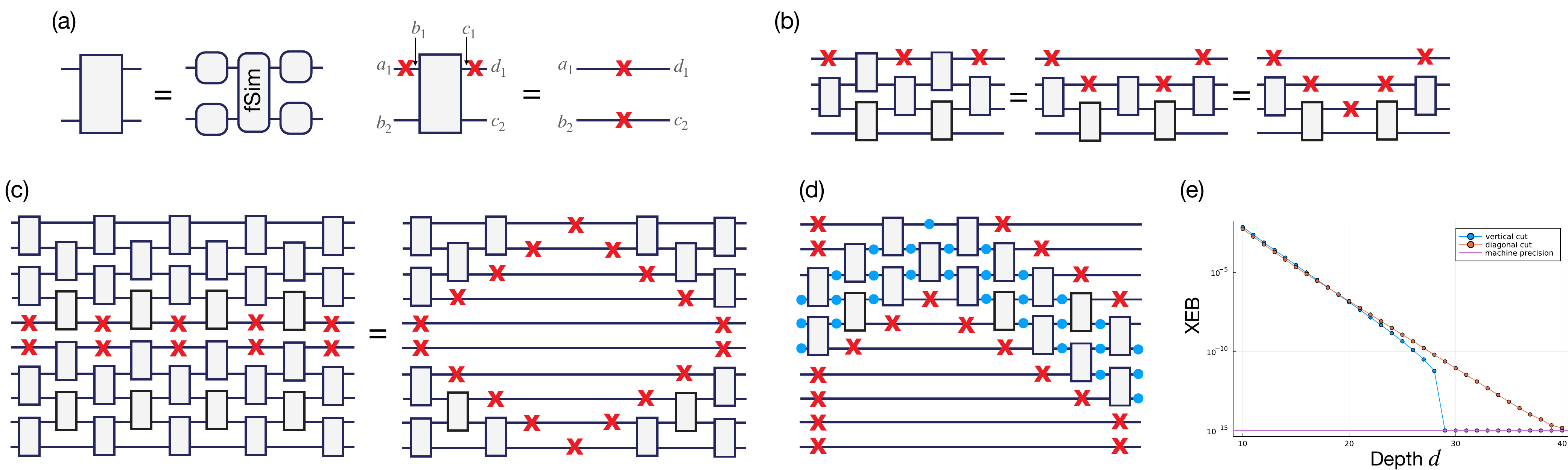}
    \caption{Properties of MDN applied to fSim circuits. The red crosses represent instances of inserting MDN. (a) When MDN is applied to the same qubit before and after the action of an fSim gate (or any gate of the form $U_1 U_2$, where $U_1$ and $U_2$ are the SWAP and controlled phase gate, respectively) while surrounded by single-qubit unitaries, the fSim gate can be effectively replaced by the MDN acting on the other qubit. (b) By repeatedly applying the identity, one can propagate the MDNs in an fSim circuits. (c) For a sufficiently deep fSim circuit, the XEB equals 0, since all the gates in the middle are removed an replaced by MDN; this effectively stop all information flow from the input to the output of the circuit. (d) Different pattern of MDN-insertions, such that the XEB is non-zero even for deep circuits. The blue dots represent particles from the diffusion-reaction model, which in the language of Sec.~\ref{app:DR} corresponds to a
    path with non-zero contribution to the XEB. (e) Numerical simulations verifying (c) and (d), where the subsystem size is 29, and we used periodic boundary conditions and  Haar-random single-qubit gates.}
    \label{fig:app_fSim_identity}
\end{figure}

Note that this MDN property can be used to explain the propagation of noise in this system. 
For example, as shown in~\autoref{fig:app_fSim_identity}(b), the MDNs applied to the first qubit (represented by the top line) can propagate to the second, and the third qubit by repeatedly applying the identity depicted pictorially in~\autoref{fig:app_fSim_identity}(a).

\subsection{Limitations of our original algorithm applied to fSim circuits}\label{ssec:fSim omitting gate effect}

The MDN identity described in~\autoref{ssec:fSim identity} is helpful in studying the effect of omitting gates on the XEB value, in fSim circuits. This is because, once averaged over random unitary gates, omitting gates and applying MDN leads to the same fidelity and XEB in both cases.
Therefore, we consider our algorithm, in which, instead of omitting gates, we apply MDNs for every gates in the middle of the circuit, as depicted in~\autoref{fig:app_fSim_identity}(c).
The MDN property tells us that the MDNs would propagate all the way to the (top and bottom) boundaries of the circuit as long as the circuit is sufficiently deep.
In this case, we find that inputs and outputs for all qubits are completely disconnected by MDNs, implying that the output of the quantum circuit is exactly the maximally mixed state [see~\autoref{fig:app_fSim_identity}(c)].
In other words, when the fSim circuit is sufficiently deep, our algorithm with omitted gates in the middle of the circuit cannot produce any meaningful output bitstring distribution; i.e., the XEB value of our algorithm using this particular positioning of omitted gates will be zero.

There is a simple way to bypass this catastrophic situation by judiciously choosing the position of omitted gates.
As an example, see the ``zig-zag'' pattern in~\autoref{fig:app_fSim_identity}(d).
In this case, the input and output of the circuit remains connected. 
Consequently, the XEB value of our algorithm using this particular choice of omitted gates will be positive even when the circuit is deep. For example, if the subsystem size is 4, as shown in~\autoref{fig:app_fSim_identity}(d), the expected XEB value will scale at least (e.g., there exist other non-zero transition paths) as 
$$
\left(T^{(\text{fSim})}_{I\Omega\rightarrow \Omega\Omega}T^{(\text{fSim})}_{\Omega\Omega\rightarrow \Omega I}\right)^d=
[(1/3+\sqrt{3}/6)^2/3]^d\approx0.13^d,
$$
where $d$ is the depth of the circuit, and we obtained this result using a direct calculation within the diffusion-reaction model.

\subsection{Deriving the Ising model from the diffusion-reaction model}\label{sapp:Ising_model}

The mapping from the XEB to the partition function of an Ising model is diagramatically illustrated in Fig.~\ref{fig:ising}.
This mapping is obtained simply by a basis change from the diffusion-reaction model.
The new basis is motivated by the symmetry in the Choi representation of unitary operation $u\otimes u^{*}\otimes u\otimes u^{*}$ on two copies of the states when we exchange the position of the two $u$s or $u^{*}$s (labeled by $1,2$ or $\bar 1,\bar 2$ respectively). This symmetry is hidden in the basis $\kket I$ and $\kket\Omega$, but explicitly exhibited by $\kket I$ and $\kket S$ (which is a SWAP operator acting on the first and the second line of $2\kket I$, as shown in Fig.~\ref{fig:app_D}(c)). This inspires us to try the following basis transformation 
\begin{eqnarray*}
\kket{\uparrow}&=&2\kket I,\\
\kket{\downarrow}&=& \kket S=\kket I+\kket\Omega.
\end{eqnarray*}
Then, the corresponding transfer matrices for ideal circuits, as shown in Fig.~\ref{fig:ising}(d) and (e), are unchanged by exchanging $\kket{\uparrow}\leftrightarrow\kket{\downarrow}$. We regard these two variables as the spins in an Ising model. Concretely, the mapping works in the following way:
(i) turn each local $2$-qubit gate into two spins and a blue box $A$, (ii) turn each wire connecting two 2-qubit gates into a red box $B_\epsilon$ ($\epsilon=0$ and $\epsilon=1$ correspond to ideal gates and omitted gates, respectively; other non-trivial $\epsilon$ values correspond to the noise strength in noisy circuits), (iii) the boundary condition for input and output states are simply equal-weight summations over all possible spin configurations, (iv) the presence of noise or gate defects corresponds to magnetic field towards $\uparrow$ directions with non-zero $\epsilon$. See Fig.~\ref{fig:ising} and Fig.~\ref{fig:partition_function} for a pictorial explanation.

Next, after the basis change, $\chi_\text{av}+1$ from Fig.~\ref{fig:DR}, as well as its versions for noisy circuits (by inserting $I_\epsilon$) and our algorithm (by replacing the omitted gates by projectors $P\otimes P$), correspond to the partition functions of the respective Ising spin models shown in Fig.~\ref{fig:partition_function}. 

We note that there are negative Boltzmann weights in the matrix $A$, which is represented by blue boxes in Fig.~\ref{fig:ising}. Na\"ively, it seems to indicate that we have a non-classical Ising model. However, this model can be transformed into a spin problem with non-negative Boltzmann weights by either integrating out some of the spins~\cite{hunter2019unitary} (also known as the ``star-triangle transformation''), or via the Kramers-Wannier duality~\cite{baxter2016exactly}. Ref.~\cite{hunter2019unitary} makes use of the first approach and shows that this model is in the ferromagnetic phase by counting the domain walls in the model after the transformation.

\begin{figure}
\centering
\includegraphics[width=0.45\textwidth]{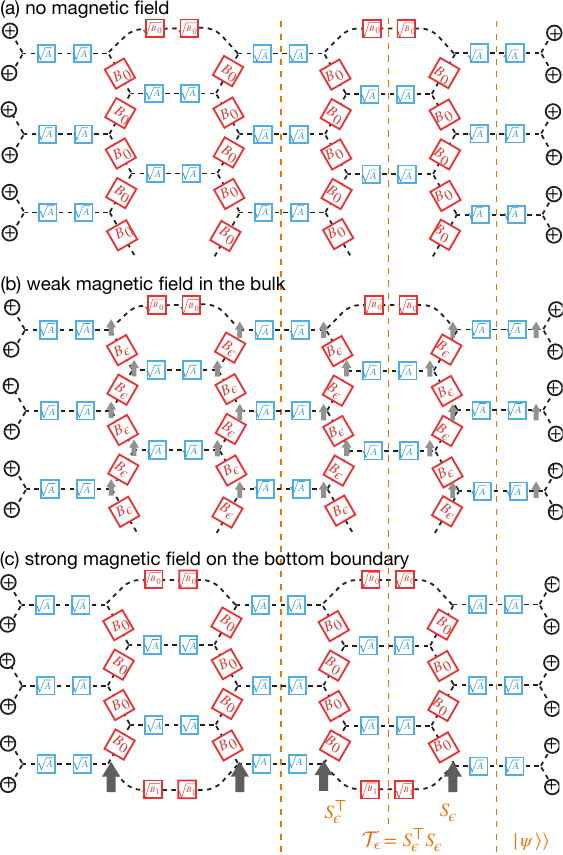}
\caption{$\chi_\text{av}+1$ as a partition function of the Ising spin model. Here we only draw the top part instead of the whole circuit. (a) Ideal circuits. There is a global $\mathbb Z_2$ Ising symmetry.
(b) Noisy circuits. There are weak magnetic fields over the entire bulk with strength $\sim \epsilon$, which breaks the Ising symmetry.
(c) Our algorithm. There are strong magnetic fields on the boundary (at the positions of omitted gates).  The bottom is the transfer matrix view of the partition function. Because $\mathcal T_\epsilon=S^\top_\epsilon S_\epsilon$, the transfer matrix is positive semi-definite.
}
\label{fig:partition_function}
\end{figure} 

We start with the XEB of ideal circuits. Let $Z$ be the partition function of an ideal circuit. The key observation is that $Z$ (i.e., $\chi_\text{av}+1$) can be written in the following form ,when the depth $d$ is odd,
\begin{equation}\label{eq:1D Ising partition 2}
\chi_\text{av}+1=Z=\bbra\psi \mathcal T_0^{(d-1)/2}\kket\psi
\end{equation}
where $\kket\psi$ and $\mathcal T_0$ are shown in the bottom of Fig.~\ref{fig:partition_function} for the case of $\epsilon=0$. The transfer matrix $\mathcal T_0$ is  positive semi-definite and, hence, it has an eigen-decomposition $\mathcal T_0=\sum_i \lambda_i \kket i\bbra i$ with $\lambda_0\geq\lambda_1\geq\cdots\geq0$ and $\langle\innerp{i}{j}\rangle=\delta_{ij}$. Thus, the partition function can be further written as follows,
\begin{eqnarray}\label{eq:Z}
Z&=&\lambda_0^{(d-1)/2}|\langle\innerp{0}{\psi}\rangle|^2
\left(1+\sum_{i>0}\left(\frac{\lambda_i}{\lambda_0}\right)^{(d-1)/2}\frac{|\langle\innerp{i}{\psi}\rangle|^2}{|\langle\innerp{0}{\psi}\rangle|^2}
\right)\nonumber\\
&\xrightarrow{\text{large }d}&\lambda_0^{(d-1)/2}|\langle\innerp{0}{\psi}\rangle|^2\left(
1+\left(\frac{\lambda_1}{\lambda_0}\right)^{(d-1)/2}\frac{|\langle\innerp{1}{\psi}\rangle|^2}{|\langle\innerp{0}{\psi}\rangle|^2}\right)\nonumber\\
&=&
1+\lambda_1^{(d-1)/2}|\langle\innerp{1}{\psi}\rangle|^2.\nonumber
\end{eqnarray}
As discussed in the section about the diffusion-reaction model, we know that when $d\rightarrow +\infty$, $Z\rightarrow 2$. Thus, $\lambda_0=\lambda_1=1$ and $|\langle\innerp{0}{\psi}\rangle |=|\langle\innerp{1}{\psi}\rangle |=1$.  This coincides with the $\mathbb Z_2$ symmetry in the ferromagnetic phase.
For noisy circuits or our algorithm, $\lambda_0$ and $\lambda_1$ are no longer equal because of the violation of the Ising symmetry caused by the presence of magnetic fields, as shown in Fig.~\ref{fig:partition_function}(b,c).
We denote $\Delta=(\lambda_0-\lambda_1)/2=(1-\lambda_1)/2$ as the gap and obtain
\begin{equation}
\chi_\text{av}=O\left(e^{-\Delta d}\right),
\end{equation}
where $\lambda_0=1$ and $|\innerp{0}{\psi}|=1$ because $Z\rightarrow 1$ when $d\rightarrow +\infty$. Note that for ideal circuits, we have $\Delta=0$.

\section{Improved algorithm: Mixed state simulation and top-$k$ heuristics for fSim circuits}\label{sapp:improved}

In this section, we provide more details for our improved algorithm. This algorithm has two steps: (1) replace the omission of gates by inserting maximal depolarizing noise (MDN) and get a probability distribution $\{\widetilde q_x\}$; (2) sort $\{\widetilde q_x\}$ in the decreasing order and choose the first $k$ bistrings with the largest $\widetilde q_x$ as our samples.
In many cases, the distribution from step (1) has exactly the same expected XEB value as the original algorithm due to the 2-design property of the gate set, e.g., for all the ensembles that contain single-qubit Haar ensemble. However, for Google's gate set or its modification (fSim+discrete single-qubit gate), this is not guaranteed, and one has to rely on numerical results.

The motivation of step (1) is two-fold: (i) to get a provable positive XEB; (ii) to have a reduced STD over random circuits. Both (i) and (ii) are important for the step (2) of our improved algorithm. 
This is because as the step (2) amplifies the weight of $x$ with large $\widetilde q_x$, it may also amplify the fluctuation, i.e., the difference between $\widetilde q_x$ and true probabilities.
The two properties (i) and (ii) of the step (1) of our improved algorithm can aid successful amplification in the step (2).
To be more precise, we restate the two properties as follows. (i) $\widetilde q_x$ should have positive correlation with $p_x$; (ii) The STD should be small in order to avoid the case that some occasional $x$ with small $p_x$ but large $\widetilde q_x$ will be amplified (in another words, ``over-fitting'').

In the rest of this section, we first prove the 1-design property of Google's gate set which is the key to prove (i). Second, we prove properties (i) and (ii) of step (1). Third, we discuss in detail how to get $\widetilde q_x$ by simulating the mixed state evolution for the gate set with fSim as the entangling gate. Finally, we discuss the top-$k$ amplification method.

\subsection{1-design property of the  modified Google's single-qubit gate set}

Google's gate set has two ingredients: the single qubit random gate of the form $Z(\theta_1)VZ(\theta_2)$, and the two-qubit fSim entangling gate. For the single-qubit gate, $V$ is chosen randomly from $\{\sqrt X,\sqrt Y, \sqrt W\}$ except with a constraint that two $V$s in successive layers on the same qubit must be different; $Z(\theta_i)$ is a $z$-axis rotation on the Bloch sphere with a site-dependent angle $\theta_i$ which is not actively chosen but is known to be constant and can be potentially calibrated. For simplicity, we introduce two analogous ensembles, with small modifications to the behavior of the $Z$ gate.
To the best of our knowledge, these modifications do not lead to any significant changes in the behavior of quantum circuits.

Ensemble 1: $\theta_i$ is chosen randomly from $[0,2\pi)$. Ensemble 2: $\theta_i$ is chosen randomly from either 0 or $\pi$ (which is $I$ or $Z$ operator respectively). For these ensembles, numerical simulations show that the average XEB values using the top-1 method are $0.00018$ and $0.0004$, respectively, for the Sycamore architecture (53 qubits, 20 depth). Since these values are similar, we argue that the details of the $z$-rotation do not influence the XEB value too much at least not in orders of magnitudes. 

Next, we prove that the single qubit random gates (after the slight modification) form a 1-design ensemble, even with the constraint that two successive $V$'s on the same qubit must be different. 
To see this, we first observe that a single qubit rotation always maps computational states $\ket{0}$ and $\ket{1}$ into their equal superpositions with different relative phases:
\begin{equation*}
    V\ket0=\frac{\ket0+e^{i\phi}\ket1}{\sqrt2}\text{ and } V\ket1=\frac{\ket0-e^{-i\phi}\ket1}{\sqrt2},
\end{equation*}
where $\phi$ depends on the specific gate $V$ we are considering.
We note that this property makes the circuit scramble faster for a fixed entangling gate, compared to the Haar single-qubit gate ensemble. Finally, consider a matrix $M$ under the action of $Z(\theta_1)VZ(\theta_2)$
\begin{eqnarray*}
&&M=\begin{pmatrix}
a & b\\
c & d
\end{pmatrix}\\
\Longrightarrow &&
\mathbb E_{\theta_2}[Z(\theta_2)MZ^\dag(\theta_2)]=
\begin{pmatrix}
a & 0\\
0 & d
\end{pmatrix}\\
\Longrightarrow &&
V\mathbb E_{\theta_2}[Z(\theta_2)MZ^\dag(\theta_2)]V^\dag
=\frac{a}{2}
\begin{pmatrix}
1 & e^{-i\phi}\\
e^{i\phi} & 1
\end{pmatrix}
+
\frac{d}{2}
\begin{pmatrix}
1 & -e^{i\phi}\\
-e^{-i\phi} & 1
\end{pmatrix}\\
\Longrightarrow&&
E_{\theta_1,\theta_2}[Z(\theta_1)V\mathbb Z(\theta_2)MZ^\dag(\theta_2)V^\dag Z^\dag(\theta_1)]=
\frac{a+d}{2}
\begin{pmatrix}
1 & 0\\
0 & 1
\end{pmatrix}
=\tr{M}\cdot\frac{I}{2},
\end{eqnarray*}
where the expectation is over either the ensemble 1 or 2. This proves that these ensembles form a 1-design, no matter which $V$ is chosen.

\subsection{Self-averaging algorithm with maximally depolarizing noise}\label{app:new_algorithm}\label{app:1D_fSim_STD}

Recall that we denote the bitstring distribution from an ideal quantum circuit as $p_x$ and denote the probability distribution after inserting MDNs as $\widetilde q_x$. Since the two subsystems decouple after the insertion of MDNs, we have $\widetilde q_x=\widetilde q_{x_1}\widetilde q_{x_2}$. In the rest of this Supplementary Material, we use $\langle\cdot\rangle$ to denote expectation value instead of $\mathbb E[\cdot]$, since the Dirac notation is no longer used.

We first show property (i) for step (1) of our improved algorithm. The XEB between $p_x$ and $\widetilde q_x$ is
\begin{equation}\label{seq:pq_qq}
    2^{N}\sum_{x}\langle p_x\widetilde q_x \rangle-1
    =2^{N}\sum_{x_1,x_2}\langle p_{x_1x_2}\widetilde q_{x_1}\widetilde q_{x_2} \rangle-1
    =2^{N}\sum_{x_1,x_2}\langle \widetilde q_{x_1}\widetilde q_{x_2}\widetilde q_{x_1}\widetilde q_{x_2} \rangle-1
    =2^{N}\sum_{x}\langle \widetilde q_x^2 \rangle-1
    >
    0,
\end{equation}
where the average is taken over the ensemble of omitted gates. The second equality is due to the 1-design property of the gate set, and the last equality is due to the fact that the XEB is 0 if and only if the distribution $\widetilde q_x$ is uniform (which can be shown by applying the Cauchy-Schwartz inequality). Since $\widetilde q_x$ is non-uniform, as long as the subsystem density matrix is not maximally mixed, the last inequality is strict, leading to property (i).

Next, we move on to show property (ii) for step (1) of our improved algorithm.
Compared to our original algorithm, which directly removes gates along a cut dividing the circuit into two, it is reasonable to expect that the STD of this new algorithm is smaller since the introduction of MDNs intuitively produces averaging effects. The mean value, however,  remains the same at least in the case that the single-qubit gate set is Haar random (or has the 2-design property). Concretely, the improved algorithm can also be realized by averaging over many realizations of the original algorithm by inserting different Paulis (plus $I$) or other arbitrary $t$-design single-qubit gates with $t\ge1$ at the positions of omitted gates. 
Each realization is equivalent to the original algorithm, because the omitted gates are random as well, which would effectively introduce these extra Paulis (plus $I$) (by choosing $v$ as the corresponding Paulis in Eq.~\eqref{Seq:HaarId}). The mean value of the new algorithm is exactly the same as that of the original one because of the 1-design property of random Paulis. Meanwhile, this improved algorithm effectively averages over  different instances of the original algorithm,  hence reducing the STD.
For this reason,  we call this step ``self-averaging''.
Because the output bitstrings could be influenced after the propagation of these single qubit gates in the middle of the circuit, the STD over different circuits is also associated with the STD in the property (ii).

\begin{figure}\includegraphics[width=0.8\textwidth]{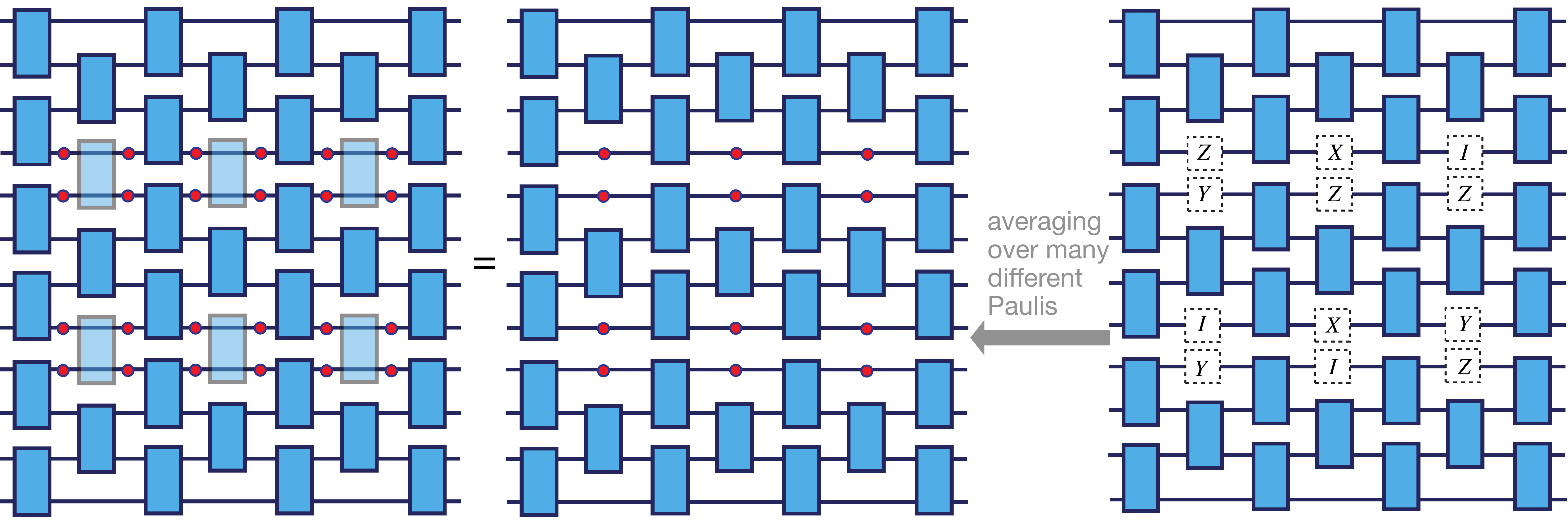}\caption{Illustration of the self-averaging algorithm: simulating the target circuit by inserting MDNs (red circles) or, equivalently, taking partial trace and preparing maximally-mixed states at the positions of light blue gates in the target circuit. In this case, the light blue gates can be omitted since the state is still maximally-mixed even after applying arbitrary unitary gates. This algorithm can also be realized by taking the average over many realizations of the original algorithm while inserting different Paulis (plus $I$) at the positions of omitted gates. }
\label{fig:algorithm2}
\end{figure}

\subsection{Numerical techniques for simulating fSim circuits with MDNs}\label{sapp:fSim_trick}
In order to exactly simulate MDN, the memory resources necessary for simulating the dynamics increase substantially as one needs to use density matrices to represent mixed states.
Naively, this is equivalent to doubling the system size.
Therefore, a direct and exact simulation of 53 qubits is no longer numerically viable with our resources.
Instead, one can sample many different realizations of the original algorithm, by applying single-qubit gates, randomly chosen from a 1-design ensemble, at the positions of omitted gates.

In the special case where the entangling gate is fSim, however, we can take advantage of the property discussed in the previous section, and presented in Fig.~\ref{fig:app_fSim_identity}(a), to efficiently simulate the dynamics of mixed states. For completeness, Fig.~\ref{fig:app_identities} shows all the identities used for simplifying the tensor network that represents the mixed state evolution of the quantum circuit.
After simplifying a quantum circuit using these identities, we use a tensor network contraction algorithm based on a Julia package OMEinsum\footnote{``https://under-peter.github.io/OMEinsum.jl/dev/"} in which the contraction order is found using the algorithm in Ref.~\cite{kalachev2021recursive}. The subsystems we consider are given in Table~\ref{tab:app_subsystems}.

In this work, we only consider the most direct way to insert MDNs which is time/depth independent.
The choice of the subsystem is not optimized either.
By generalizing the way of inserting MDNs, e.g., making their locations to be time/depth dependent, the tensor network contraction algorithm can still be used straightforwardly. We expect that this type of optimization can produce higher XEB without increasing the necessary computational resources too much.

\begin{figure}
\centering
\includegraphics[width=1\textwidth]{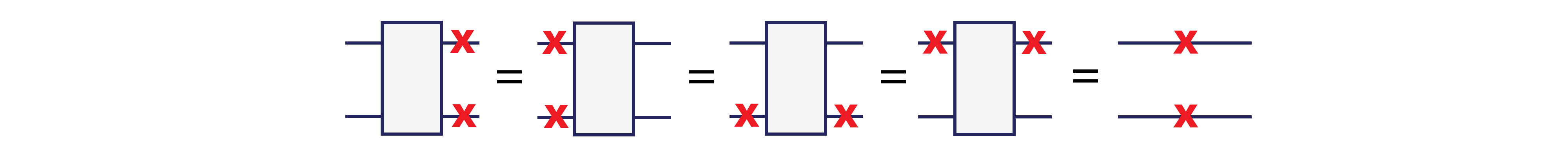}
\caption{Identities used for simplifying the tensor network representing the mixed state evolution of the fSim quantum circuits (or any other SWAP+control-phase gate).}
\label{fig:app_identities}
\end{figure}

\begin{table*}
\begin{center}
\begin{tabular}{|c|c|c|}
\hline
 & subsystem 1 & subsystem 2\\
\hline
\hline
   Google: Fig.~S27 in Ref.~\cite{arute2019quantum} &
   the left part of red lines &
   the right part of red lines\\
\hline
USTC-1: Fig.~S11 in Ref.~\cite{USTC} &
   the blue part &
   exclude black part with 1\\
\hline
USTC-2: Fig.~3(a) in Ref.~\cite{zhu2021quantum} & \makecell{
   1,2,3,7,8,9,10,12,\\13,14,15,18,19,20,\\21,24,25,26,30,31,\\32,36,37,42,43,48 }&
   \makecell{
   23,27,28,29,33,34,35,\\38,39,40,41,44,45,\\46,47,49,50,51,52,53,\\55,56,57,58,59}\\
\hline
\hline
   Google (another partition) &
   \makecell{
   40,53,26,44,25,48,15,42,\\16,46,6,51,2,12,47,\\10,41,9,20,50,19,\\43,23,34,49,33,45 } &
   the rest qubits\\
\hline

USTC-2 (another partition) & 0$\sim$34 exclude 4,5,17,11,22  & 40$\sim$60
   \\
\hline
\end{tabular}
\end{center}
\caption{
The subsystems in our simulation. They are not optimized and chosen for simulating the mixed state evolution using 1 GPU (NVIDIA Tesla V100). The last two rows are different partitions with better XEB but longer running time.}
\label{tab:app_subsystems}
\end{table*}

\subsection{Top-$k$ method}
Suppose we replace $\widetilde q_x$ by another distribution $r_x=r_{x_1}r_{x_2}$, then
\begin{equation}
   \chi_\text{av}= 2^{N}\sum_{x}\langle p_x r_x \rangle-1
    =2^{N}\sum_{x_1,x_2}\langle p_{x_1x_2} r_{x_1} r_{x_2} \rangle-1
    =2^{N}\sum_{x_1,x_2}\langle \widetilde q_{x_1}\widetilde q_{x_2} r_{x_1}r_{x_2} \rangle-1
    =2^{N}\sum_{x}\langle \widetilde q_xr_x \rangle-1,
\end{equation}
where the expectation value is taken over the omitted gates, and the second equality follows from the 1-design property. Here, we choose $r_x$ in the following way:
\begin{equation}
    r_x=\begin{cases}
       \frac{1}{k}, \text{ if }\widetilde q_x \text{ is in the first }k\text{ largest probabilities;}\\
       0, \text{ otherwise.}
    \end{cases}
\end{equation}
This is called a ``top-$k$ method'', and it can substantially amplify the XEB. The resulting XEB is
\begin{equation}
    2^{N}\sum_{x^*}\frac{\widetilde q_{x^*}}{k}-1 \text{ where }x^*\in\{\widetilde q_{x^*}\text{ is in the first }k\text{ largest probabilities}\}.
\end{equation}

In the actual simulation, we choose top-$k_1$ and top-$k_2$ bitstrings from $\widetilde q_{x_1}$ and $\widetilde q_{x_2}$ respectively. More accurately, we only choose top-$k_i$ bitstrings from the non-trivial part of $\widetilde q_{x_i}$.
Here we mention ``non-trivial" because many bitstrings share exactly the same value for $\widetilde q_{x_i}$.
This is because usually some output qubits experience MDN right before measurements, and thus the corresponding distribution is perfectly uniform, leading to degeneracies in $\widetilde q_{x_i}$.
In this case, we say the output qubits are ``trivial''.
Denoting the total number of trivial qubits as $m$, we can get $k_1k_22^m$ distinct bitstrings with $k_1 k_2$ distinct values of $\widetilde q_{x_i}$ (up to potential accidental degeneracies).
In Fig.~\ref{fig:app_top_k}, we present the performance of the top-$k$ method for the non-trivial part of subsystem 1 of Google's Sycamore architecture using the modification of their gate set. We can see that the mean value hardly decreases when increasing $k$ until $k$ becomes very large $\sim 10^5$.
However, the STD decreases $\propto 1/\sqrt k$. Intuitively, this suggests that the top-$k$ bitstrings are roughly independent due to strong scrambling.

\begin{figure}\includegraphics[width=1.0\textwidth]{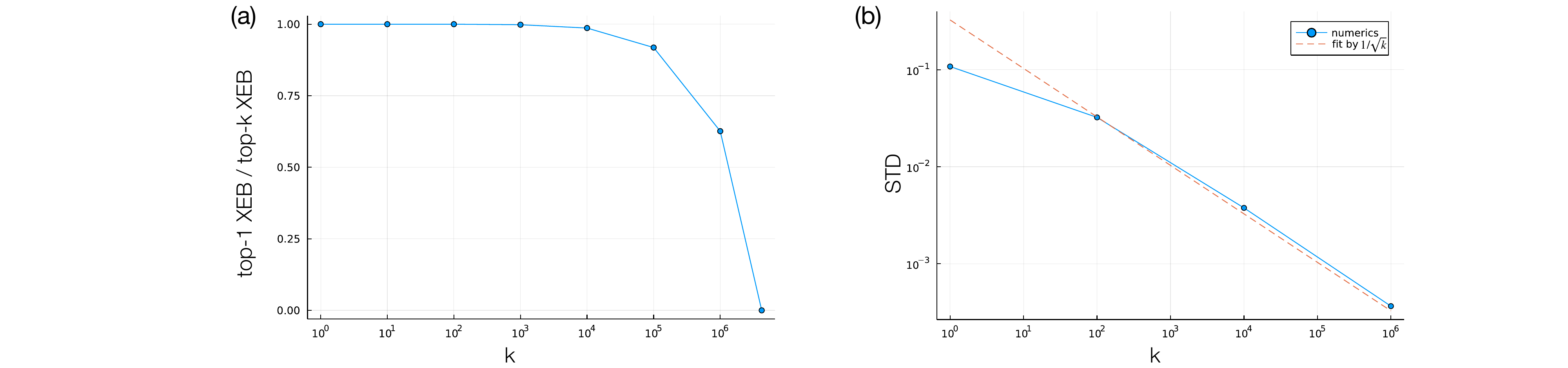}
\caption{Performance of the top-$k$ method. (a) The ratio between XEB values of top-$1$ and top-$k$ methods. (b) The STD of top-$k$ is approximately $\propto 1/\sqrt{k}$. Direct verification of the STD is difficult, since it requires simulating the ideal circuit. Here, we replace the omitted gates with the tensor product of two random, single-qubit gates in the ideal circuit as a reasonable approximation.}
\label{fig:app_top_k}
\end{figure}

\section{Refuting XQUATH}
The Linear Cross-Entropy Quantum Threshold Assumption (XQUATH) is proposed by Aaronson and Gunn~\cite{aaronson2019classical} and serves as the complexity-theoretic foundation of the XEB-based quantum computational advantage. The assumption is that there is no efficient classical algorithm to estimate the output probability of a string from a randomly sampled quantum circuit. Formally, we restate XQUATH from Ref.~\cite{aaronson2019classical} as follows.

\begin{conjecture}[Linear Cross-Entropy Quantum Threshold Assumption (XQUATH)~\cite{aaronson2019classical}]
Given a random circuit description $U$, there is no polynomial time classical algorithm to compute an estimation $q_U(0^N)$ of $p_U(0^N)$ (the probability of getting $0^N$ for the ideal quantum circuit $U$) such that
\begin{equation}\label{Seq:XQUATH}
   2^{2N}\ex{\left(p_U(0^N)-\frac{1}{2^N}\right)^2-\left(p_U(0^N)-q_U(0^N)\right)^2}{U}=\delta
\end{equation}
where $\delta=\Omega(2^{-N})$ and the expectation is over a random circuit ensemble.
\end{conjecture}

In this section, we show that our techniques can refute XQUATH for every random circuits using single qubit 2-design gates. Note that the refutation applies to any circuit architecture.

\subsection{Reduction from XQUATH to the hardness of average XEB}\label{ssec:XQUATHreduction}
Here we adopt our self-averaging algorithm shown in Fig.~\ref{fig:algorithm2} to refute XQUATH, since this algorithm has smaller STD as discussed in Sec.~\ref{app:1D_fSim_STD}. The key idea is to reduce Eq.~\eqref{Seq:XQUATH} in XQUATH to the average XEB value of our algorithm.

First, we show that the quantity on the left-hand side of Eq.~\eqref{Seq:XQUATH} is exactly the same as the average XEB if $\widetilde q_U$ is computed from the self-averaging algorithm introduced in Fig.~\ref{fig:algorithm2} and if the circuits are random enough:
\begin{eqnarray*}
&&2^{2N}\ex{\left(p_U(0^N)-\frac{1}{2^N}\right)^2-\left(p_U(0^N)-\widetilde q_U(0^N)\right)^2}{U}\\
&=&2^{2N}\left(
2\ex{p_U(0^N)\widetilde q_U(0^N)}{U}-\ex{\widetilde q^2_U(0^N)}{U}
-\frac{2\ex{p_U(0^N)}{U}}{2^N}+\frac{1}{2^{2N}}\right)\\
&=&2^{2N}\left(
2\ex{p_U(0^N)\widetilde q_U(0^N)}{U}-\ex{\widetilde q^2_U(0^N)}{U}
\right)-1\\
&=&
2^{2N}\ex{p_U(0^N)\widetilde q_U(0^N)}{U}-1\\
&=&
2^{N}\ex{\sum_xp_U(x)\widetilde q_U(x)}{U}-1\\
&=&\ex{\chi_U(C)}{U}
\end{eqnarray*}
where the second line is due to $\ex{p_U(0^N)}{U}=1/2^N$; the third line is due to
$\ex{p_U(0^N)\widetilde q_U(0^N)}{U}=\ex{\widetilde q^2_U(0^N)}{U}$ (see Eq.~(\ref{seq:pq_qq})); the fourth line is due to that the circuit behaves identically for all bitstrinigs, i.e., there is nothing special about the particular choice of $0^N$, i.e., $\ex{p_U(0^N)\widetilde q_U(0^N)}{U}=\ex{p_U(x)\widetilde q_U(x)}{U}$ for every $x\in\{0,1\}^N$. 

In summary, we proved that, for our self-averaging algorithm, the $\delta$ defined in Eq.~\eqref{Seq:XQUATH} of XQUATH is exactly the same as the average XEB. In the next subsection, we discuss the value of $\delta$ obtained from our algorithm.

\subsection{Refuting XQUATH}
In~\autoref{app:1D} in this supplementary material, we have shown that, for 1D circuits, $\chi_\text{av}=\Omega(e^{-\Delta_1 d})$ where $\Delta_1\approx 0.25$. Here, we show that for a broad family of circuit architectures, our algorithm achieves $\chi_\text{av}=\Omega(e^{-\Delta d})$ with some constant $\Delta>0$. Moreover, for every circuit architecture with every two-qubit gate being surrounded by single-qubit Haar random gates, we can design a polynomial time algorithm that achieves $\chi_\text{av}=\Omega(e^{-\Delta d})$ with some constant $\Delta>0$.
As a consequence, we refute XQUATH for these random circuits.

\subsubsection{Our algorithm refutes XQUATH for a broad family of circuit architectures}

The key observation is that after our algorithm breaks a circuit into several subsystems, the corresponding diffusion-reaction model (or Ising spin model for 1D) is also decoupled into $\lceil N/l\rceil$ isolated subsystems, where $l$ is the subsystem size (measured in the number of qubits). Then, similar to~\autoref{eq:Isingpart}, we have
$$
\chi_\text{av}=\prod_i^{\lceil N/l\rceil}(\chi^{(i)}_\text{av}+1)-1,
$$
where $\chi_\text{av}^{(i)}+1$ is equal to the partition function in the $i$-th subsystem. Then, we have
$$
\chi^{(i)}_\text{av}=c_l^{(i)}e^{-\Delta_l^{(i)}d},
$$
where $c_l^{(i)}$ and $\Delta_l^{(i)}$ are some constants that depend only on the subsystem size $l$ (up to the detailed arrangement of these $l$ qubits). Specifically, $c_l^{(i)}$ and $\Delta_l^{(i)}$ would not have any dependency on $N$ because each subsystem can only see the maximal depolarizing noise at the boundaries (positions of omitted gates) and is disconnected from any information about other subsystems. Next, we choose $l$ as a constant so that each $\chi^{(i)}_\text{av}$ has the form $\Omega(e^{-\Delta d})$. Thus,
$$
\chi_\text{av}=\prod_i^{\lceil N/l\rceil}(\chi^{(i)}_\text{av}+1)-1\approx\sum_i^{\lceil N/l\rceil}\chi^{(i)}_\text{av}
=\sum_i^{\lceil N/l\rceil}c_l^{(i)}e^{-\Delta_l^{(i)} d} = \Omega(e^{-\Delta d})
$$
for some constant $\Delta>0$, as desired. Together with the reduction in subsection~\ref{ssec:XQUATHreduction}, this shows that our algorithm refutes XQUATH for circuit architectures with non-maximal-entangling 2-qubit gates because $d\sim N^{1/D}$ is a reasonable requirement for a sufficiently scrambling circuit dynamics to demonstrate quantum computational advantage. 
We start with a generic analysis for non-maximal-entangling 2-qubit gates like Haar-random and CZ gates. For special gate sets, such as fSim, one needs to carefully cut the subsystem, as described in the next subsubsection.

\subsubsection{Efficient algorithms that refute XQUATH for every circuit architectures}

We now describe how to modify our algorithms to refute XQUATH for every circuit architectures with every two-qubit gate being surrounded by single-qubit Haar random gates. The key idea is using path integral in Pauli basis so we start with explaining the underlying intuition.

\paragraph{Path integral in Pauli basis.}
The proposal of XQUATH was based on the observation that sub-sampling the path integral for XEB over the computational basis only achieves $2^{-\Omega(Nd)}$ XEB value with high probability. Concretely, if we consider the path integral for XEB using the computational basis, there are roughly $2^{Nd}$ different paths and their weights (i.e., their contribution to the XEB) are roughly equal. The belief underlying XQUATH is that a polynomial time classical algorithm can only compute the contribution from polynomial number of paths and thus the attainable XEB is only $\text{poly}(Nd)2^{-Nd}$.

However, the weight of each path is far from uniform when we consider performing a path integral with respect to the Pauli basis ($I$ and $X,Y,Z$ as we discussed in section~\ref{app:DR}). Here, we denote 
\begin{equation}
\widetilde\sigma_0=\frac{1}{\sqrt2}
\begin{pmatrix}
1&0\\
0&1
\end{pmatrix}, \quad
\widetilde\sigma_1=\frac{1}{\sqrt2}
\begin{pmatrix}
0&1\\
1&0
\end{pmatrix}, \quad
\widetilde\sigma_2=\frac{1}{\sqrt2}
\begin{pmatrix}
0&-i\\
i&0
\end{pmatrix},\quad
\widetilde\sigma_3=\frac{1}{\sqrt2}
\begin{pmatrix}
1&0\\
0&-1
\end{pmatrix},
\end{equation}
where the factor $1/\sqrt2$ makes them different from conventional Pauli matrices.
For any $2\times 2$ density matrix $\rho$, we have the following decomposition:
\begin{equation}\label{eq:dec_Her}
\rho=\sum_{i=0,1,2,3}(\tr{\widetilde\sigma_i\rho})\widetilde\sigma_i.
\end{equation}
This is the Pauli-Liouville representation~\cite{greenbaum2015introduction}, and since $\rho$ is Hermitian, $\tr{\widetilde\sigma_i\rho}$ is real. For a $4\times4$ matrix $\rho$
\begin{equation}
\rho=\sum_{i,j=0,1,2,3}(\tr{\widetilde\sigma_i\otimes\widetilde\sigma_j \rho})\widetilde\sigma_i\otimes\widetilde\sigma_j.
\end{equation}
Note that the $\{\widetilde\sigma_i\}^{\otimes n}$ basis is complete for all Hermitian matrices acting on $n$-qubits.

Similarly to inserting $I=\ket0\bra0+\ket1\bra1$ at every space-time position in the quantum circuit to get the path integral in the computational basis, we can formulate a path integral in the Pauli basis by inserting a different and appropriately chosen resolution of identity. To illustrate the idea of this path integral, we use a very simple example. First, we consider a matrix element of a single-qubit density matrix $\rho$ after the application of two single-qubit gates $U$ and $V$ as follows.
\begin{eqnarray}\label{eq:path}
\notag &&\bra x V^\dag U^\dag \rho U V \ket x\\
\notag &=&\sum_{i=0,1,2,3}(\tr{\widetilde\sigma_i\rho})\bra xV^\dag U^\dag \widetilde\sigma_iUV \ket x\\
\notag&=&\sum_{i,j=0,1,2,3}(\tr{\widetilde\sigma_i\rho}) ( \tr{U^\dag \widetilde\sigma_iU\widetilde\sigma_j })
 \bra xV^\dag \widetilde\sigma_j  V \ket x\\
 &=&\sum_{i,j,k=0,1,2,3}(\tr{\widetilde\sigma_i\rho}) ( \tr{U^\dag \widetilde\sigma_iU\widetilde\sigma_j} )
(\tr{V^\dag \widetilde\sigma_j  V \widetilde\sigma_k}) (\tr{\widetilde\sigma_k\ketbra{x}{x}}),
\end{eqnarray}
where we used Eq.~(\ref{eq:dec_Her}) iteratively.

While now there are roughly $4^{Nd}$ different paths (because we have $\propto Nd$ space-time positions and 4 basis matrices), the weights of different paths can be vastly different, as contrasted with the nearly-uniform distribution of weights in the computational-basis path integral. In particular,  
a path in the Pauli basis has a weight exponentially decaying with the number of non-trivial Paulis (i.e., other than the identity). This can be understood by the fact that the transition $I\to I$ satisfies  $|\tr{U^\dag I U I}|>|\tr{U^\dag \sigma U \sigma^\prime}|$, if any of $\sigma$,$\sigma^\prime$ is not $I$ (Clifford gates are a special case where  should use $\ge$). Thus, a transition involving a non-trivial Pauli ($X,Y,Z$) will cause a decay relative to the transition involving only $I$. The path involving only $I$s corresponds to the contribution which is cancelled by $-1$ in the definition of the XEB.

\paragraph{Modified algorithms to refute XQUATH.}
For a circuit architecture with every two-qubit gates being surrounded by single-qubit Haar random gates, we first apply our analytical tools to get its corresponding diffusion-reaction model. The modified algorithm for this circuit architecture will first cut out a subsystem (which might not be the straightforward way we did in the original algorithm) of constant width in space. Then the algorithm fully simulates this subsystem in polynomial time (since it has constant width in space) and samples output string according to the marginal distribution of the output layer with the output qubits outside being sampled from the uniform distribution. In the following, we describe how to select a subsystem for a circuit architecture and then argue that the modified algorithm indeed achieves average XEB at least $e^{-O(d)}$.

First, recall from~\autoref{eq:ini_vec} and~\autoref{eq:chi_U_avg} that the partition function of the diffusion-reaction model $\chi_\text{av}+1$ has input boundary ${\mathbf u}=(1/2\ 1/2)$ and output boundary $\mathbf{v}_\text{\tiny XEB}=(2\ 2/3)$ in the Pauli basis $\{I,\Omega\}^l$ where $\Omega$ denotes the non-trivial Pauli. 
As the contribution of both the input and output boundary being all $I$s is $1$, we have that $\chi_\text{av}$ corresponds to the contribution from boundaries containing non-Pauli's. 

Next, we focus on a boundary condition with both input and output boundary having exactly one non-trivial Pauli and argue that its contribution to the partition function is at least $e^{-O(d)}$. Before continuing, we reorganize ${\mathbf u}=(1 1)$ and output boundary $\mathbf{v}_\text{\tiny XEB}=(1 1/3)$ for convenience.
Notice that the transfer matrix of our diffusion reaction model provides a way to calculate the weight of each path in the Pauli basis (see subsubsection~\ref{sapp:DR_transfer}). In particular, as the reaction rate $R$ is at most $2/3$, there is always a non-zero probability to move from a configuration of exactly one non-trivial Pauli to another configuration of exactly one non-trivial Pauli. Thus, we can start with an input boundary having exactly one non-trivial Pauli and repeat the above argument to move it to an output boundary with exactly one non-trivial Pauli. Specifically, this path has weight at least $\max(1-D,D-R)^d/3\ge6^{-d}/3$.

Finally, we simply select a path of exactly one non-trivial Pauli and select the subsystem containing the gates that correspond to the locations of these non-trivial Pauli's. Then the algorithm simply fully simulates this subsystem and samples output string from the marginal distribution with the other output qubits being sampled from the uniform distribution.
Note that the partition function of the corresponding diffusion-reaction model is at least $1$ plus the weight of this subsystem, which is at least $e^{-O(d)}$ as discussed in the previous paragraph. Thus, we have an efficient algorithm that achieves average XEB at least $e^{-O(d)}$.

To conclude, the supporting argument for XQUATH, based on the conventional path integral formulation, does not hold if we consider a Pauli-basis path integral. This is because the contribution of each path becomes polarized (not equally weighted) and our subsystem algorithm is equivalent to outputting paths with weight $e^{-O(d)}$ in the Pauli basis.

    \chapter{Details for the Neuroscience Part}\label{app:neuroscience}

\section{Mathematical preliminaries}

\begin{definition}[Matrix norm]
For every positive semi-definite matrix $A$, the $A$-norm of a vector $\bx$ is defined as $\|\bx\|_A=\sqrt{\bx^\top A\bx}$.
\end{definition}
Note that when $A$ is the identity matrix, then $A$-norm is simply the Euclidean norm.

\begin{lemma}[Cauchy-Schwarz inequality]
Let $\bx,\by\in\Real^n$ be two vectors, we have $|\bx^\top\by|\leq\sqrt{\|\bx\|_2\cdot\|\by\|_2}$.
\end{lemma}

\section{Optimal Balanced SNNs Solves the Least Squares Problem}\label{app:SNNs least squares}

In this section, we give a complete proof for the convergence of optimal balanced SNN to the (non-negative) least squares problem as discussed in~\autoref{sec:SNNs results},~\autoref{thm:linearsystem} (with an informal version,~\autoref{thm:SNNs main least squares}), and~\autoref{thm:l1} (with an informal version,~\autoref{thm:SNNs main sparse}).

Recall that given a matrix $F \in \Real^{n\times m}$ and a vector $\bx \in \Real^m$, the goal of the least squares problem is to find $\br \in \Real^n$ that minimizes $\|\bx-F^\top\br\|_2^2/2$. Here we consider the discrete and non-leaky optimal balanced SNNs (\autoref{eq:optimal SNN discrete}):
\begin{equation}\label{eq:app discrete SNN}
\bv(t+1) = -FF^\top\bs(t) + F\bx\cdot\Dt
\end{equation}
where $\bv(0)$ is instantiated to the all zero vector. For the other parameters, we require the threshold $\eta \geq \lambda_{\max}$ and the timestep size $\Dt \leq \frac{ \sqrt{\lambda_{\min}}}{24\sqrt{n} \cdot \|\bx\|_2}$, where $\lambda_{\min}$ and $\lambda_{\max}$ are the minimal and maximal (non-zero) eigenvalues of $FF^\top$. Here we restate~\autoref{thm:linearsystem} which shows that the residual error $\|\bx-F^\top\bv(t)\|_2^2/2$ converges to its minimum with rate $O(1/t)$.

\snnleastsquares*

Before proving the theorem, we first give some intuition about why the firing rate would converge to an optimal least squares solution. An important fact here is that the neurons' potential $\bv(t)$ will remain bounded (specifically, we will show that the $(FF^\top)^{\dagger}$-norm of $\bv(t)$ is bounded). Now, recall that the potential is affected by the constant external charging current $I\cdot\Dt = F\bx\cdot\Dt$ and the spike $-FF^\top\bs$. Intuitively, to ensure the potential being bounded, the constant external charging current needs to be balanced out by the spiking effect over time. Thus, in the limit, we should expect 
\[
FF^\top \br \approx F\bx \, .
\]
Namely, the average effect of the spike should roughly be the external charging for one unit of time. Let $\bx_F$ be the orthogonal projection of $\bx$ onto the orthogonal complement of the null space of $F$. Note that as $F\bx=F\bx_F$ and $\bx_F$ lies in the range space of $F^\top$, this implies that $F^\top\br \approx \bx_F$. Moreover, the minimal residual error $\|\bx-F^\top\br\|_2/2$ would be $\|\bx-\bx_F\|_2/2$ since $F^\top\br$ always lies in the range space of $F^\top$. Thus, in the following, we focus on upper bounding $\|\bx_F-F^\top\br(t)\|_2$ instead.

In the following, we start with a lemma about the properties of matrix norm and then prove the boundedness of $\bv(t)$. The proofs of these two lemmas are provided in the subsequent subsection. Finally, we use the ``potential conservation argument'' to prove~\autoref{thm:linearsystem}.

\begin{lemma}\label{lemma:PSDmatrix}
For every $F\in\Real^{n\times m}$, we have
\begin{enumerate}
	\item For every $\br\in\Real^n$, $\|\br\|_{FF^\top}=\|F^\top\br\|_2=\|FF^\top\br\|_{(FF^\top)^{\dagger}}$.
	\item For every $\bx\in\Real^m$, $\|F\bx\|_{(FF^\top)^{\dagger}}=\|\bx_F\|_2$.
\end{enumerate}
\end{lemma}

\begin{lemma}\label{lemma:DSNN-potentialbounded}
	For every time $t\in\N$, when the threshold $\eta\geq\lambda_{\max}$ and the discretized step size $\Dt<\frac{\sqrt{\lambda_{\min}}}{24\sqrt{n}\cdot\|\bx_F\|_2}$, we have $\|\bv(t)\|_{(FF^\top)^{\dagger}}\leq 2\sqrt{\kappa\eta n}$.
\end{lemma}

We now prove~\autoref{thm:linearsystem} using ~\autoref{lemma:DSNN-potentialbounded}.
\begin{proof}[Proof of~\autoref{thm:linearsystem}]
The proof consists of two steps. First, using the dynamics of $\bv(t)$ to derive a connection between $\bv(t)$ and the residual error of the least squares problem. Second, use ~\autoref{lemma:DSNN-potentialbounded} to upper bound the residual error.

Let us start with telescoping~\autoref{eq:app discrete SNN} from time 1 to time $t$, 
\begin{align*}
\sum_{s=1}^{t}\bv(s+1) &= \sum_{s=1}^{t}\left(\bv(s) - FF^\top\bs(s) + F\bx\cdot\Dt\right)\\
&= \left(\sum_{s=1}^{t}\bv(s)\right) - (t\cdot\Dt)\cdot\left(FF^\top\br(t) -F\bx\right) \, .
\end{align*}
Note that the second equality uses the fact that $\br(t)=\sum_{s=1}^t\bs(s)$ for non-leaky SNNs. Next, subtract $\sum_{s=1}^{t}\bv(s)$ on both sides and divide with $t\cdot\Dt$. Now, we have established an elegant connection between the firing rate and the potential vector.
\begin{equation}\label{eq:conservation}
\frac{\bv(t+1)-\bv(1)}{t\cdot\Dt} = F\bx -FF^\top\br(t) \, .
\end{equation}

Note that the right-hand side of~\autoref{eq:conservation} is $F$ multiplying the residual error $\bx-F^\top\br(t)$. To get rid of the $F$ term, we apply the $(FF^\top)^{\dagger}$-norm on the both sides of \autoref{eq:conservation}. By ~\autoref{lemma:PSDmatrix} and the fact that $\bx_F$ is the projection of $\bx$ on the range space of $F^\top$, we have

\begin{equation}\label{eq:deterministicSNN-conservation}
\frac{\|\bv(t+1)-\bv(1)\|_{(FF^\top)^{\dagger}}}{t\cdot\Dt}=\|F\bx-FF^\top\br(t)\|_{(FF^\top)^{\dagger}}=\|\bx_F - F^\top\br(t)\|_2 \, .
\end{equation}

As a result, to upper bound the error $\|\bx_F-F^\top\br(t)\|_2$, it suffices to upper bound $\|\bv(t+1)-\bv(1)\|_{(FF^\top)^{\dagger}}$. Moreover, as we take $\bv(1)=\mathbf{0}$, now we only need to show $\|\bv(t+1)\|_{(FF^\top)^{\dagger}}$ is bounded by some $\epsilon$ multiplicative factor of $\|\bx_F\|_2$. As we take $\Dt\leq\frac{\sqrt{\lambda_{\min}}}{24\sqrt{n}\cdot\|\bx_F\|_2}$ and $\eta\geq\lambda_{\max}$, ~\autoref{lemma:DSNN-potentialbounded} gives that $\|\bv(t+1)\|_{(FF^\top)^{\dagger}}\leq 2\sqrt{\kappa\eta n}$ for every $t\in\N$. Combine this upper bound for $\|\bv(t+1)\|_{(FF^\top)^{\dagger}}$ with~\autoref{eq:deterministicSNN-conservation} and the choice of $t\geq2\sqrt{\kappa\eta n}\cdot\|\bx_F\|_2/(\epsilon\cdot\Dt)$, we have
\begin{equation}
	\|\bx_F-F^\top\br(t)\|_2\leq\frac{\|\bv(t+1)\|_{(FF^\top)^{\dagger}}}{t\cdot\Dt}\leq\frac{2\epsilon\sqrt{\kappa\eta n}\cdot\|\bx_F\|_2}{2\sqrt{\kappa\eta n}}= \epsilon\cdot\|\bx_F\|_2 \, .
\end{equation}
\end{proof}

\subsection{Proof of~\autoref{lemma:PSDmatrix}}
In the following, we prove two properties of the matrix norm. First, for every $\br\in\Real^n$, $\|\br\|_{FF^\top}=\|F^\top\br\|_2=\|FF^\top\br\|_{(FF^\top)^{\dagger}}$. Second, $\|F\bx\|_{(FF^\top)^{\dagger}}=\|\bx_F\|_2$.
\begin{enumerate}
	\item By definition, we have $\|\br\|_{FF^\top}=\sqrt{\br^\top FF^\top\br}=\|F^\top\br\|_2$. As for the second part of the equality, recall that if a vector $\bv\in\Real^n$ lies in the range space of $FF^\top$, then
    \begin{equation}\label{eq:pinv}
    (FF^\top)^{\dagger}(FF^\top)\bv=\bv.
    \end{equation}
    Let $\br_F$ be the orthogonal projection of $\br$ such that $F^\top\br=F^\top\br_F$ and $\br_F$ lies in the range space of $FF^\top$. By~\autoref{eq:pinv}, we have
    \begin{align*}
    	\|FF^\top\br\|_{(FF^\top)^{\dagger}}^2 &= (FF^\top\br)^\top(FF^\top)^{\dagger}(FF^\top\br)\\
        &=(FF^\top\br_F)^\top(FF^\top)^{\dagger}(FF^\top\br_F)\\
        &=(FF^\top\br_F)^\top\br_F\\
        &=\|F^\top\br_F\|_2^2=\|F^\top\br\|_2^2 \, .
    \end{align*}
    
	\item Let $\br\in\Real^n$ be a solution of $F^\top\br=\bx_F$ where $\bx_F$ is the orthogonal projection of $\bx$ to the orthogonal complement of the null space of $F$. By~\autoref{eq:pinv}, we have
    \begin{align*}
    \|F\bx\|_{(FF^\top)^{\dagger}}^2 &= (F\bx)^\top(FF^\top)^{\dagger}(F\bx)\\
    &= (FF^\top\br)^\top(FF^\top)^{\dagger}(FF^\top\br)\\
    &= \|F^\top\br\|_2^2=\|\bx_F\|_2^2.
    \end{align*}
\end{enumerate}

\subsection{Proof of~\autoref{lemma:DSNN-potentialbounded}}
At timestep t, denote the set of neurons that exceeds the threshold as $\Gamma(t) := \{i \in [n] \, |\,  |\bv_i(t)| > \eta\}$. We have the following observation.

\begin{lemma}\label{lemma:DSNN-keylemma}
	For every $t\in\N$, $\bv(t)^\top \bs(t)\geq \eta\cdot|\Gamma(t)|$.
\end{lemma}
\begin{proof}
	Observe that
	\begin{align}
	\bv(t)^\top \bs(t) &= \sum_{i\in[n]}\bv_i(t)\cdot\bs_i(t)\\
	&= \sum_{i\in\Gamma(t)}\bv_i(t)\cdot\mbox{sign}(\bv_i(t))= \sum_{i\in\Gamma(t)}|\bv_i(t)|\\
	&\geq\sum_{i\in\Gamma(t)}\eta=\eta\cdot|\Gamma(t)|.
	\end{align}
\end{proof}

\begin{lemma}\label{lemma:potentialliesinrangespaceofA}
	For every $t\in\N$, $\bv(t)^\top(FF^\top)^{\dagger}FF^\top\bs(t) = \bv(t)^\top \bs(t)$.
\end{lemma}
\begin{proof}
	First, note that $(FF^\top)^{\dagger}FF^\top$ is the orthogonal projection matrix for the range space of $FF^\top$. Next, by the dynamics of SNN and $\bv(1)=\mathbf{0}$, we know that $\bv(t)$ is in the range space of $FF^\top$. Let $\bz\in\Real^n$ be the unique vector in the range space of $FF^\top$ such that $\bv(t)=FF^\top\bz(t)$. Let $\bs_F(t)$ be the orthogonal projection of $\bs(t)$ onto the range space of $FF^\top$. Note that $FF^\top\bs(t)=FF^\top\bs_F(t)$. By~\autoref{lemma:PSDmatrix}, we have
    \begin{align*}
    \bv(t)^\top (FF^\top)^{\dagger}FF^\top\bs(t) &= \bz^\top FF^\top(FF^\top)^{\dagger}FF^\top\bs_F(t)\\
    &= \bz^\top FF^\top\bs_F(t)\\
    &= \bz^\top FF^\top\bs(t) = \bv(t)^\top\bs(t).
    \end{align*}
\end{proof}

Finally, we are ready to prove~\autoref{lemma:DSNN-potentialbounded}.
\begin{proof}[Proof of~\autoref{lemma:DSNN-potentialbounded}]
The proof is based on an induction over the number of iterations $t\in\N$. Let $B=2\sqrt{\kappa\eta n}$, our goal is to show that $\|\bv(t)\|_{(FF^\top)^{\dagger}}\leq B$ for every $t$ when $\Dt\leq\frac{\sqrt{\lambda_{\min}}}{24\cdot\sqrt{n}\cdot\|\bx_F\|_2}$. Clearly that $\|\bv(1)\|_{(FF^\top)^{\dagger}}=0\leq B$. Now, assume that $\|\bv(t)\|_{(FF^\top)^{\dagger}}\leq B$, let us consider
\begin{align}
\|\bv(t+1)\|_{(FF^\top)^{\dagger}}^2 &= \|\bv(t) - FF^\top\bs(t) + F\bx\cdot\Dt\|_{(FF^\top)^{\dagger}}^2 \nonumber\\
&= \|\bv(t)\|_{(FF^\top)^{\dagger}}^2 + \|FF^\top\bs(t) - F\bx\cdot\Dt\|_{(FF^\top)^{\dagger}}^2\nonumber\\
&- 2\bv(t)^\top (FF^\top)^{\dagger}\Big(FF^\top\bs(t)-F\bx\cdot\Dt\Big).\label{eq:proofof-lemma-DSNN-potentialbounded-1}
\end{align}
Expand the $\|FF^\top\bs(t) - F\bx\cdot\Dt\|_{(FF^\top)^{\dagger}}^2$ term, by ~\autoref{lemma:PSDmatrix}, we have
\begin{align}
\|FF^\top\bs(t) - F\bx\cdot\Dt\|_{(FF^\top)^{\dagger}}^2 &= \|F^\top\bs(t)-\bx_F\cdot\Dt\|_2^2 \nonumber\\
&=\|F^\top\bs(t)\|_2^2+\|\bx_F\cdot\Dt\|_2^2-2\bs(t)^\top F\bx_F\cdot\Dt \nonumber\\
&\leq\lambda_{\max}\cdot|\Gamma(t)| + \|\bx_F\|_2^2\cdot\Dt^2+2\sqrt{\lambda_{\max}\cdot |\Gamma(t)|}\|\bx_F\|_2\cdot\Dt.\label{eq:proofof-lemma-DSNN-potentialbounded-2}
\end{align}
The last inequality is from Cauchy-Schwartz inequality and the fact that $\|\bs(t)\|^2_2=|\Gamma(t)|$. Next, by~\autoref{lemma:PSDmatrix},~\autoref{lemma:DSNN-keylemma},~\autoref{lemma:potentialliesinrangespaceofA}, Cauchy-Schwartz inequality, and the induction hypothesis, we have
\begin{align}
- 2\bv(t)^\top (FF^\top)^{\dagger}\Big(FF^\top\bs(t)-F\bx\cdot\Dt\Big)&=- 2\bv(t)^\top \bs(t)+ 2\bv(t)^\top (FF^\top)^{\dagger}\Big(F\bx\cdot\Dt\Big) \\
&\leq- 2\bv(t)^\top \bs(t)+ 2\|\bv(t)\|_{(FF^\top)^{\dagger}}\cdot\|\bx_F\|_2\cdot\Dt\\
&\leq -2\eta|\Gamma(t)| + 2B\|\bx_F\|_2\cdot\Dt.\label{eq:proofof-lemma-DSNN-potentialbounded-3}
\end{align}
Combine~\autoref{eq:proofof-lemma-DSNN-potentialbounded-1},~\autoref{eq:proofof-lemma-DSNN-potentialbounded-2},~\autoref{eq:proofof-lemma-DSNN-potentialbounded-3}, and induction hypothesis, we have
\begin{align*}
\|\bv(t+1)\|_{(FF^\top)^{\dagger}}^2 &\leq \|\bv(t)\|_{(FF^\top)^{\dagger}}^2 + (\lambda_{\max}-2\eta)\cdot|\Gamma(t)|\nonumber\\
&+ (2\sqrt{\lambda_{\max}|\Gamma(t)|}\cdot\|\bx_F\|_2+2B\cdot\|\bx_F\|_2)\cdot\Dt+\|\bx_F\|_2^2\cdot\Dt^2.
\end{align*}
By the choice of $\eta, \Dt$, and $B$, we have
\begin{align}
\|\bv(t+1)\|_{(FF^\top)^{\dagger}}^2&\leq \|\bv(t)\|_{(FF^\top)^{\dagger}}^2 - \lambda_{\max}\cdot|\Gamma(t)|+ \frac{\sqrt{\lambda_{\max}\lambda_{\min}|\Gamma(t)|}}{6\sqrt{n}} + \frac{\lambda_{\max}}{6} + \frac{\lambda_{\min}}{144n}\\
&\leq \|\bv(t)\|_{(FF^\top)^{\dagger}}^2 - \lambda_{\max}|\Gamma(t)|+ \frac{\lambda_{\max}|\Gamma(t)|}{2}+\frac{\lambda_{\max}}{2}.\label{eq:deterministicSNN-keyinequality}
\end{align}
Now, there are two cases to consider: $\Gamma(t)$ is empty or not. From (\autoref{eq:deterministicSNN-keyinequality}) we can see that if $\Gamma(t)$ is nonempty, then by the induction hypothesis, $\|\bv(t+1)\|_{(FF^\top)^{\dagger}}\leq \|\bv(t)\|_{(FF^\top)^{\dagger}}\leq B$. If $\Gamma(t)$ is empty, then for any $i\in[n]$, $|\bu_i(t)|\leq\eta$. Thus, we have
\begin{align}\label{eq:mainlemma-1}
\|\bv(t+1)\|_{(FF^\top)^{\dagger}}&\leq\|\bv(t)\|_{(FF^\top)^{\dagger}}+\frac{\lambda_{\max}}{2}.
\end{align}
Since $|\Gamma(t)|=0$, we know that $\|\bv(t)\|_{\infty}\leq\eta$ and thus	
\begin{equation*}
\|\bv(t)\|_{(FF^\top)^{\dagger}}^2\leq\frac{\|\bv(t)\|_2^2\leq n\cdot\max_{i\in[n]}|\bv_i(t)|^2}{\lambda_{\min}}\leq\frac{\eta^2n}{\lambda_{\min}}.
\end{equation*}
As a result, for the $|\Gamma(t)|=0$ case, we have
\begin{align*}
\|\bv(t+1)\|_{(FF^\top)^{\dagger}}&\leq\|\bv(t)\|_{(FF^\top)^{\dagger}}+\frac{\lambda_{\max}}{2}\\
&\leq\frac{\eta\sqrt{n}}{\sqrt{\lambda_{\min}}}+\frac{B}{2}\\
&\leq\sqrt{\kappa\eta n}+\frac{B}{2}\leq B.
\end{align*}
We conclude that the induction holds and hence finish the proof of~\autoref{lemma:DSNN-potentialbounded}.
\end{proof}

\section{Derivations of the dual programs}

\subsection{$\ell_1$ minimization}\label{app:l1 min derivation}
Recall that the (non-negative) $\ell_1$ minimization problem is defined as follows.

\begin{equation}
	\begin{aligned}
	& \underset{\br\in\Real^n}{\text{minimize}}
	& & \|\br\|_1 \\
	& \text{subject to}
	& & F^\top\br=\bx,\ \br\geq0.
	\end{aligned}
    \tag{$\ell_1$ minimization}
\end{equation}

To derive the dual problem, we first write down the Lagrangian:
\begin{equation}
L(\br,\blambda,\nu) = \mathbf{1}^\top\br - \blambda^\top\br + \bnu^\top(F^\top\br-\bx)
\end{equation}
where we implicitly replace $\|\br\|_1$ with $\mathbf{1}^\top\br$ as they are the same when $\br\geq0$. Next, we derive the Lagrangian dual function:
\begin{align}
g(\blambda,\bnu) &= \inf_{\br} L(\br,\blambda,\bnu)\\
&= \inf_{\br} \left\{[(\mathbf{1}-\blambda)^\top+\bnu^\top F^\top]\br - \bnu^\top\bx\right\}\\
&= \left\{\begin{array}{ll}
-\infty     &  \text{, if }(\mathbf{1}-\blambda)^\top+\bnu^\top F^\top\neq0 \\
-\bnu^\top\bx     & \text{, if }(\mathbf{1}-\blambda)^\top+\bnu^\top F^\top=0.
\end{array}
\right. \label{eq: l1 min dual 1}
\end{align}

Finally, the dual problem is defined as the maximum of $g(\blambda,\bnu)$ subject to $\blambda\geq0$. As the first condition in~\autoref{eq: l1 min dual 1} won't lead to the minimizer, the dual problem is the following.

\begin{equation*}
	\begin{aligned}
	& \underset{\blambda\in\Real^n,\bnu\in\Real^m}{\text{maximize}}
	& & -\bnu^\top\bx \\
	& \text{subject to}
	& & (\mathbf{1}-\blambda)^\top+\bnu^\top F^\top=0,\ \blambda\geq0.
	\end{aligned}
\end{equation*}

Note that the condition in the above program is equivalent to $-\bnu^\top F^\top\leq1$. By setting $\bu=-\bnu$ and slightly rearranging the equations, we can rewrite the dual program in the form as we saw in~\autoref{op:l1 min dual}.

\begin{equation*}
	\begin{aligned}
	& \underset{\bu\in\Real^m}{\text{maximize}}
	& & \bx^\top\bu \\
	& \text{subject to}
	& & F\bu\leq1.
	\end{aligned}
\end{equation*}

\subsection{Lasso}\label{app:lasso derivation}
Recall that the (non-negative) Lasso problem is defined as follows.

\begin{equation}
	\begin{aligned}
	& \underset{\br\in\Real^n}{\text{minimize}}
	& & \frac{1}{2}\|F^\top\br-\bx\|_2^2 + \beta\|\br\|_1 \\
	& \text{subject to}
	& & \br\geq\mathbf{0}.
	\end{aligned}
    \tag{Lasso}
\end{equation}

To derive the dual problem, we first write down the Lagrangian:
\begin{equation}
L(\br,\blambda) = \frac{1}{2}\|F^\top\br-\bx\|_2^2 + \beta\mathbf{1}^\top\br - \blambda^\top\br
\end{equation}
where we implicitly replace $\|\br\|_1$ with $\mathbf{1}^\top\br$ as they are the same when $\br\geq0$. Next, we derive the Lagrangian dual function:
\begin{align}
g(\blambda) &= \inf_{\br} L(\br,\blambda)\\
&= \inf_{\br} \left\{(\beta\mathbf{1}-\blambda)^\top\br + \frac{1}{2}\|F^\top\br-\bx\|_2^2 \right\} \, .
\end{align}
Note that the minimizer should be of the form $(F^\dagger)^\top\by$ where $F^\dagger$ is the (right) pseudo-inverse of $F$. Thus we have
\begin{align}
g(\blambda) =&\ \inf_{\by} \left\{ (\beta\mathbf{1}-\blambda)^\top (F^\dagger)^\top\by + \frac{1}{2}\|\by-\bx\|_2^2\right\} \\
=&\ \frac{1}{2}\|\by-(\bx-F^\dagger(\beta\mathbf{1}-\blambda))\|_2^2 \\
+&\ \frac{1}{2}\|\bx\|_2^2 - \frac{1}{2}\|\bx-F^\dagger(\beta\mathbf{1}-\blambda)\|_2^2 \\
=&\ \frac{1}{2}\|\bx\|_2^2 - \frac{1}{2}\|\bx-F^\dagger(\beta\mathbf{1}-\blambda)\|_2^2 \, .
\end{align}

Finally, the dual problem is defined as the maximum of $g(\blambda)$ subject to $\blambda\geq0$:

\begin{equation*}
	\begin{aligned}
	& \underset{\blambda\in\Real^n}{\text{maximize}}
	& & \frac{1}{2}\|\bx\|_2^2 - \frac{1}{2}\|\bx-F^\dagger(\beta\mathbf{1}-\blambda)\|_2^2 \\
	& \text{subject to}
	& & \blambda\geq0 \, .
	\end{aligned}
\end{equation*}

Note that by setting $\bu=F^\dagger(\beta\mathbf{1}-\blambda)/\beta$, the constraints become $F\bu\leq\mathbf{1}$. Hence, we can rewrite the dual program in the form as we saw in~\autoref{op:lasso dual}.

\begin{equation*}
\begin{aligned}
& \underset{\bu\in\Real^m}{\text{maximize}}
& & \frac{1}{2}\|\bx\|_2^2 - \frac{1}{2}\|\bx-\beta\bu\|_2^2 \\
& \text{subject to}
& & F\bu\leq1 \, .
\end{aligned}
\end{equation*}

\section{Optimal Balanced SNNs Solves the $\ell_1$ Minimization Problem}\label{app:SNNs l1 min details}

\snnlonemin*

Two remarks on the statement of~\autoref{thm:l1}. First, we consider the \textit{continuous SNN} instead of the discrete SNN, which is of interest for simulation on classical computer. In discrete SNN, the \textit{step size} is some non-negligible $\Dt>0$ instead of $dt$. The main reason for considering continuous SNN is that this significantly simplify the proof by avoiding a huge amount of nasty calculations. We suspect that the proof idea would hold for discrete SNN with discretization parameter $\Dt\leq\Dt(\frac{\gamma(F)}{n\cdot\lambda_{\max}})$ for some polynomial $\Dt(\cdot)$.
Second, the parameters in~\autoref{thm:l1} have not been optimized and we believe all the dependencies can be improved. Since the parameters highly affect the efficiency of SNN as an algorithm for $\ell_1$ minimization problem, we pose it as an interesting open problem to study what are the best dependencies one can get.

\subsection{Overview of the proof for~\autoref{thm:l1}}
% Before we move on, we make two simplifications on the simple SNN to make the analysis easier. First, we consider the continuous setting where $\Dt\rightarrow0$. Second, if there are some neurons fire at time $t$ and they trigger consecutive firing, we add the external charging \textit{after} the consecutive firing. Note that the proof still holds for the simple SNN for some $\Dt$ small enough.

The proof of~\autoref{thm:l1} consists of two main steps as mentioned in the previous subsection. The first step argues that the dual SNN $\bu(t)$ would converge to the neighborhood of the optimal dual solution $\bu^\OPT$. The second step is connecting the dual solution (\textit{i.e.,} the dual SNN) to the primal solution (\textit{i.e.,} the firing rate).

In the first step, we try to identify a \textit{potential function}\footnote{Potential function is widely used in the analysis of many gradient-descent based algorithm. In the theory of dynamical systems, sometimes people use the term ``Lyapunov function''. The difficulty lies in the search of a good potential function for the algorithm.} that captures how close is $\bu(t)$ to the optimal dual solution $\bu^\OPT$.
It turns out that this is not an easy task since the effect of spikes makes the behavior of dual SNN very non-monotone. We conquer the difficulty via a technique that we call \textit{ideal coupling} (see~\autoref{def:ideal coupling} and Figure~\ref{fig:ideal coupling}). The idea is to associate the dual SNN $\bu(t)$ with an \textit{ideal SNN} $\bu^\text{ideal}(t)$ for every $t\geq0$ such that the ideal SNN would have \textit{smoother} behavior comparing to the spiking phenomenon in the dual SNN. We will formally define the ideal SNN in Section~\ref{sec:ideal coupling}. There are two advantages of using ideal SNN instead of handling dual SNN directly: 
(i) Ideal SNN is smoother than dual SNN in the sense that it would not change after spikes (see~\autoref{lem:ideal SNN unchaged}). Further, by introducing some auxiliary processes (\textit{i.e.,} the auxiliary SNNs defined in~\autoref{def:auxiliary}), we are able to identify a potential function that is strictly improving at any moment and measures how well the dual SNN has been solving the $\ell_1$ minimization problem (see~\autoref{lem:strict improvement}). 
(ii) ideal SNN is naturally associated with an \textit{ideal solution} (defined in~\autoref{def:ideal solution}) which is easier to analyze than the firing rate. Using these good properties of ideal SNN, we can prove in~\autoref{lemma:idealalgorithm-l2bound} that the $\ell_2$ residual error of the ideal solution will converge to $0$.

After we are able to show the convergence of the $\ell_2$ residual error in~\autoref{lemma:idealalgorithm-l2bound}, we move to the second step where the goal is showing that the $\ell_1$ norm of the solution is also small. We look at the KKT conditions of the $\ell_1$ minimization problem and observe that the primal and dual solutions of SNN satisfy the KKT conditions of a \textit{perturbed} program of the $\ell_1$ minimization problem. Finally, combine tools from perturbation theory, we can upper bound the $\ell_1$ error of the ideal solution by its $\ell_2$ residual error in~\autoref{lemma:idealalgorithm-OPTbounds}.

Theorem~\ref{thm:l1} then follows from~\autoref{lemma:idealalgorithm-l2bound} and~\autoref{lemma:idealalgorithm-OPTbounds} with some special cares on how to transform everything for ideal solution to the firing rate. See Figure~\ref{fig:overview proof for thm l1} for an overall structure of the proof for~\autoref{thm:l1}.

In the rest of this section, we are going to start from some definitions on the \textit{nice conditions} we need for the input matrix in Section~\ref{sec:nice}. Next, we define the ideal coupling in Section~\ref{sec:ideal coupling} and prove~\autoref{lem:ideal SNN unchaged} and~\autoref{lem:strict improvement} in Section~\ref{sec:unchange} and Section~\ref{sec:strict} respectively. Finally, we wrap up the proof for~\autoref{thm:l1} in Section~\ref{sec:l1 convergence}.

\subsection{Some nice conditions on the input matrix}\label{sec:nice}
We need some \textit{nice conditions} for the input matrix as follows.
\begin{definition}[non-degeneracy]
	Let $F\in\Real^{n\times m}$ where $m\leq n$. We say $F$ is non-degenerate if for any size $m\times m$ submatrix of $F$ has full rank. For any $\gamma>0$, we say $F$ is $\gamma$-non-degenerate if for any $\Gamma\subseteq[n]$, $|\Gamma|=m$, and $i\in\Gamma$, $\|F_i-\Pi_{F_{\Gamma\backslash\{i\}}}F_i\|_2\geq\gamma$ where $\Pi_{F_{\Gamma\backslash\{i\}}}\bv$ is the projection of $\bv$ onto subspace spanned by $\{F_j:\ j\in\Gamma\backslash\{i\}\|\}$ for any $\bv\in\Real^m$.
\end{definition}

Note that if $F$ is non-degenerate, then for any $S\subseteq[n]$ and $|S|=m$ and $\bx\in\{-1,1\}^m$, there exists an unique solution $\bv\in\Real^m$ to $F_S^\top\bv=\bx$ where $F_S$ is the submatrix of $F$ restricted to columns in $S$. We call such $\bv$ a \textit{vertex} of the polytope $\mathcal{P}_{F,1}$. Note that in this definition, a vertex might not lie in $\mathcal{P}_{F,1}$. An important parameter for future analysis is the minimum distance between two distinct vertices of $\mathcal{P}_{F,1}$.

\begin{definition}[nice input matrix]\label{def:nice}
	Let $F\in\Real^{n\times m}$ and $\gamma\geq0$. We say $F$ is $\gamma$-nice if all of the following conditions hold.
	\begin{enumerate}[label=(\arabic*)]
		\item $F$ is $\gamma$-non-degenerate.
		\item The distance between any two distinct vertices of $\mathcal{P}_{F,1}$ is at least $\gamma$.
		\item For any $\bx\in\{-1,1\}^m$, $\Gamma\subseteq[n]$, and $|\Gamma|=m$, let $\bx=(F_\Gamma^\top)^{-1}\bx$. For any $i\in[m]$, $|\bx_i|\geq\gamma$.
	\end{enumerate}
	Define $\gamma(F)$ to be the largest $\gamma$ such that $F$ is $\gamma$-nice. We say $F$ is \textit{nice} if $\gamma(F)>0$.
\end{definition}

To motivate the definition of niceness, the following lemma shows that the $\ell_1$ minimization problem defined by matrix $F$ has unique solution if $\gamma(F)>0$.
\begin{lemma}
	Let $F\in\Real^{n\times m}$. If $\gamma(F)>0$, then for any $\bx\in\Real^m$, the $\ell_1$ minimization problem for $(F,\bx)$ has an unique solution.
\end{lemma}
\begin{proof}
	We prove the lemma by contradiction. Suppose there exists $\bx\in\Real^m$ such that there are two distinct solutions $\bx_1\neq\bx_2$ to the $\ell_1$ minimization problem for $(F,\bx)$. Let $\bv^*$ be the optimal solution of the dual program as in equation~\autoref{op:basispursuit-dual}. By the complementary slackness in the KKT condition, for any optimal solution $\bx$ to the primal program, $\text{supp}(\bx)\subseteq\{i\in[n]:\ |F_i^\top\bv^*|=1 \}$. Let $S=\{i\in[n]:\ |F_i^\top\bv^*|=1 \}$, then both $\bx_1$ and $\bx_2$ are solution to $F_S\bx=\bx_S$ where $F_S$ and $\bx_S$ are restrictions to index set $S$.  As $\gamma(F)>0$, we have $|S|\leq m$. By the non-degeneracy of $F$, $F_S$ has full rank and thus $F_S\bx=\bx_S$ has unique solution. That is, $\bx_1=\bx_2$, which is a contradiction.
	
	We conclude that if $F$ is non-degenerate and $\gamma(F)>0$, then for any $\bx\in\Real^m$, the $\ell_1$ minimization problem for $(A,\bx)$ has unique solution.
\end{proof}

In general, it is easy to find a matrix $F$ such that $\gamma(F)=0$. However, we would like to argue that most of the matrices are actually nice.
The following lemma shows that random matrix $F$ sampled from the \textit{rotational symmetry model (RSM)} is nice. In RSM, each column of $F$ is an uniform vector on the unit sphere of $\Real^m$. Note that such matrix for $\ell_1$ minimization problem is commonly used in practice such as compressed sensing.
\begin{lemma}\label{lem:gamma lb of RSM}
Let $F\in\Real^{n\times m}$ be a random matrix samples from RSM, then $\gamma(F)>0$ with high probability.
\end{lemma}
\begin{proof}
First, we show that $F$ is non-degenerate with high probability. For any $\Gamma\subseteq[n]$ and $i\in\Gamma$, denote the event where $F_i = \Pi_{F_{\Gamma\backslash\{i\}}}F_i$ as $E_{\Gamma,i}$. Note that this event is measured zero for all choice of $\Gamma$ and $i$ and thus by union bound, we have $F$ being non-degenerate with high probability. For the other two properties, similar arguments hold.
\end{proof}

We remark that giving a lower bound in terms of $m$ and $n$ for $\gamma(F)$ would result in a better asymptotic bound for our main theorem and could have applications in other problems too. Since the goal of this paper is giving a provable analysis, we do not intend to optimize the parameter. Note that for $F$ sampled from RSM, $\gamma(F)$ has an inverse exponential lower bound directly from union bound when $n$ and $m$ are polynomially related. As for upper bound, there are inverse quasi-polynomial upper bound if $n\geq\textsf{polylog}(m)\cdot m$ and inverse exponential upper bound if $n\geq m^{1+\Omega(1)}$ as pointed out by the anonymous reviewer from ITCS 2019. We leave it as an open question to understand the correct asymptotic behavior of $\gamma(F)$ when $F$ is sampled from RSM.

\subsection{Ideal coupling}\label{sec:ideal coupling}
Ideal coupling is a technique to keeping track of the dual SNN $\bu(t)$ by associating any point in the dual polytope to a point in a smaller polytope. Concretely, let $\mathcal{P}_{F,1}=\{\bu\in\Real^m:\ \|F\bu\|_\infty\leq1\}$ be the dual polytope and $\mathcal{P}_{F,1-\tau}$ be the \textit{ideal polytope} where $\tau\in(0,1)$ is  an important parameter that will be properly chosen\footnote{The choice of $\tau$ depends on $F$ and $1$ and will be discussed later.} in the end of the proof. Observe that $\mathcal{P}_{F,1-\tau}\subsetneq\mathcal{P}_{F,1}$. The idea of ideal coupling is associating each $\bu\in\mathcal{P}_{F,1}$ with a point $\bu^{\text{ideal}}$ in $\mathcal{P}_{F,1-\tau}$. In the analysis, we will then focus on the dynamics of $\bu^\text{idael}$ instead of that of $\bu$.

Before we formally define the coupling, we have to define a \textit{partition} of $\mathcal{P}_{F,1}$ with respect to $\mathcal{P}_{F,1-\tau}$ as follows.

\begin{definition}[partition of $\mathcal{P}_{F,1}$]\label{def:ideal partition}
Let $\mathcal{P}_{F,1}$ and $\mathcal{P}_{F,1-\tau}$ be defined as above. For each $\bu^\text{ideal}\in\mathcal{P}_{F,1-\tau}$, define
\begin{equation*}
S_{\bu^\text{ideal}} = \{\bu^\text{ideal}+\mathcal{C}_{F,\Gamma(\bu^{ideal})}\}\cap\mathcal{P}_{F,1}.
\end{equation*}
where $\Gamma(\bu^\text{ideal})=\{i\in[\pm n]:\ \langle F_i,\bu^\text{ideal}\rangle=1-\tau\}$ is the active walls of $\bu^\text{ideal}$ and $\mathcal{C}_{F,\Gamma(\bu^{ideal})}=\{\sum_{i\in\Gamma(\bu^\text{ideal})}a_iF_i,\ \forall a_i\geq0\}$ is the cone spanned by the column of $F$ indexed by $\Gamma(\bu^\text{ideal})$.
\end{definition}

Consider the example where $A=\bigl( \begin{smallmatrix}1&0\\0&1\end{smallmatrix}\bigr)$ and $\tau\in(0,1)$. The dual polytope (resp. ideal polytope) is the square with vertices in the form $(\pm1,\pm1)$ (resp. $(\pm1-\tau,\pm1-\tau)$). For a arbitrary $\bu^\text{ideal}=(x,y)\in\mathcal{P}_{F,1-\tau}$, let us see what $S_{\bu^\text{ideal}}$ is:
\begin{itemize}
	\item When $|x|,|y|<1-\tau$, \textit{i.e.,} $\bu^\text{ideal}$ strictly lies inside $\mathcal{P}_{F,1-\tau}$, $\Gamma(\bu^\text{ideal})=\emptyset$ and thus $C_{F,\Gamma(\bu^\text{ideal})}=\emptyset$. Namely, $S_{\bu^\text{ideal}}=\bu^\text{ideal}$.
	\item When $|x|=1-\tau$ and $|y|<1-\tau$, \textit{i.e.,} $\bu^\text{ideal}$ lies on an edge of the ideal polytope, $\Gamma(\bu^\text{ideal})=\{\text{sgn}(x)\cdot1\}$ and thus $C_{F,\Gamma(\bu^\text{ideal})}=\{(a,0):\ a\geq0\}$. Namely, $S_{\bu^\text{ideal}}=\{(a,y):\ a\in[1-\tau,1] \}$.
	\item When $|x|<1-\tau$ and $|y|=1-\tau$, \textit{i.e.,} $\bu^\text{ideal}$ lies on an edge of the ideal polytope, $\Gamma(\bu^\text{ideal})=\{\text{sgn}(y)\cdot2\}$ and thus $C_{F,\Gamma(\bu^\text{ideal})}=\{(0,b):\ b\geq0\}$. Namely, $S_{\bu^\text{ideal}}=\{(x,b):\ b\in[1-\tau,1] \}$.
	\item When $|x|=|y|=1-\tau$, \textit{i.e.,} $\bu^\text{ideal}$ lies on a vertex of the ideal polytope, $\Gamma(\bu^\text{ideal})=\{\text{sgn}(x)\cdot1,\text{sgn}(y)\cdot2\}$ and thus $C_{F,\Gamma(\bu^\text{ideal})}=\{(a,b):\ a,b\geq0\}$. Namely, $S_{\bu^\text{ideal}}=\{(a,b):\ a,b\in[1-\tau,1] \}$.
\end{itemize}

The following lemma checks that~\autoref{def:ideal partition} does give a partition for $\mathcal{P}_{F,1}$.
\begin{lemma}\label{lem:ideal partition}
	$\{S_{\bu^\text{ideal}}\}_{\bu^\text{ideal}\in\mathcal{P}_{F,1-\tau}}$ is a partition for $\mathcal{P}_{F,1}$.
\end{lemma}
\begin{proof}[Proof of~\autoref{lem:ideal partition}]
	The proof is basically doing case analysis and using some basic properties from linear algebra. See Section~\ref{sec:missing proofs ideal auxiliary SNN} for details.
\end{proof}

\begin{definition}[ideal coupling]\label{def:ideal coupling}
	Let $\mathcal{P}_{F,1}$ and $\mathcal{P}_{F,1-\tau}$ be defined as above. For any $\bv\in\mathcal{P}_{F,1}$, define $\bu^\text{ideal}(\bv)$ be the unique $\bu^\text{ideal}$ such that $\bv\in S_{\bu^\text{ideal}}$. We denote $\bu^\text{ideal}(\bv)$ as $\bu^\text{ideal}$ when the context is clear. Specifically, for any $t\geq0$, we denote $\bu^\text{ideal}(t)=\bu^\text{ideal}(\bv(t))$ as the ideal SNN at time $t$.
\end{definition}

See Figure~\ref{fig:ideal coupling} for an example of the ideal coupling.

\begin{figure}[h]
	\centering
	\includegraphics[width=10cm]{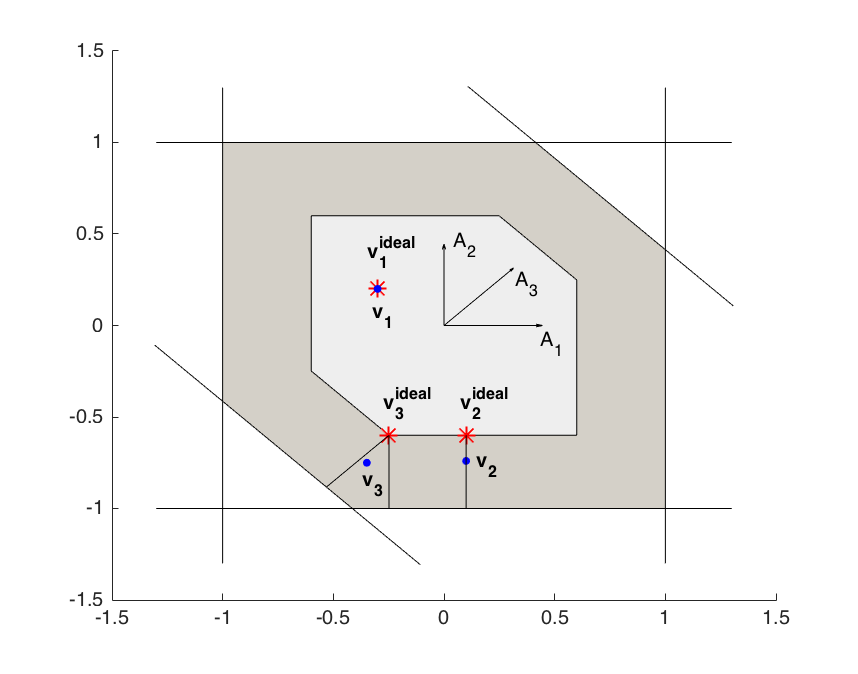}
	\caption{This is an example of ideal coupling in $\Real^2$ where $\tau=0.4$, $F_1=[1\ 0]^\top$, $F_2=[0\ 1]^\top$, and $F_3=[\frac{1}{\sqrt{2}}\ \frac{1}{\sqrt{2}}]^\top$. The dots (\textit{i.e.,} $\bv_1,\bv_2,\bv_3$) are dual SNN and the stars (\textit{i.e.,} $\bv_1^\text{ideal},\bv_2^\text{ideal},\bv_3^\text{ideal}$) are the corresponding ideal SNN. The whole gray area is the dual polytope $\mathcal{P}_{F,1}$ and the gray area in the middle is the ideal polytope $\mathcal{P}_{1-\tau}$.}
	\label{fig:ideal coupling}
\end{figure}
\vspace{3mm}

Note that~\autoref{def:ideal coupling} is well-defined due to~\autoref{lem:ideal partition}. With the ideal coupling, we are then switching to analyze the \textit{ideal SNN} $\bu^\text{ideal}(t)$ instead of the dual SNN $\bu(t)$. In the following, we are going to show that the ideal SNN is indeed tractable for analysis, though it is highly non-trivial and is very sensitive to the choice of parameters.

To show the convergence of ideal SNN, we need a notion to measure how close $\bu^\text{ideal}(t)$ and the optimal point is. To do so, we define the \textit{ideal solution} of ideal SNN at time $t$ as follows.

\begin{definition}[ideal solution]\label{def:ideal solution}
	For any $t\geq0$, define the ideal solution $\br^\text{ideal}(t)$ at time $t$ as
	\begin{equation*}
	\br^\text{ideal}(t) = \argmin_{\substack{\br\geq0,\\\bx_i=0,\ \forall i\in\Gamma(\bu^\text{ideal}(t))}}\|\bx-F^\top\br\|_2.
	\end{equation*}
	Also, let $\Gamma^*(\bu^\text{ideal}(t))=\{i\in\Gamma(\bu^\text{ideal}(t)):\ \br^\text{ideal}(t)\neq0\}$ to be the set of super active neurons.
\end{definition}

In the later proof, we need one more definition on a variant of ideal SNN called the \text{super SNN}. Similar to~\autoref{def:ideal solution}, we define the super ideal SNN $\bu^\text{super}(t)$ as the projection of $\bv(t)$ to the ideal polytope \textit{without} those non-super ideal neurons. Formally, define $\bu^\text{super}(t)$ be the unique solution of the following equations: $\bu=\bu(t)-F_{\Gamma^*(\bu^\text{ideal}(t))}\bz$ and $F_i^\top\bu=1-\tau$ for each $i\in\Gamma^*(\bu^\text{ideal}(t))$. See Figure~\ref{fig:super} for example. Note that the uniqueness of the solution is guaranteed by the non-degeneracy of $F$.

\begin{figure}[h]
	\centering
	\includegraphics[width=10cm]{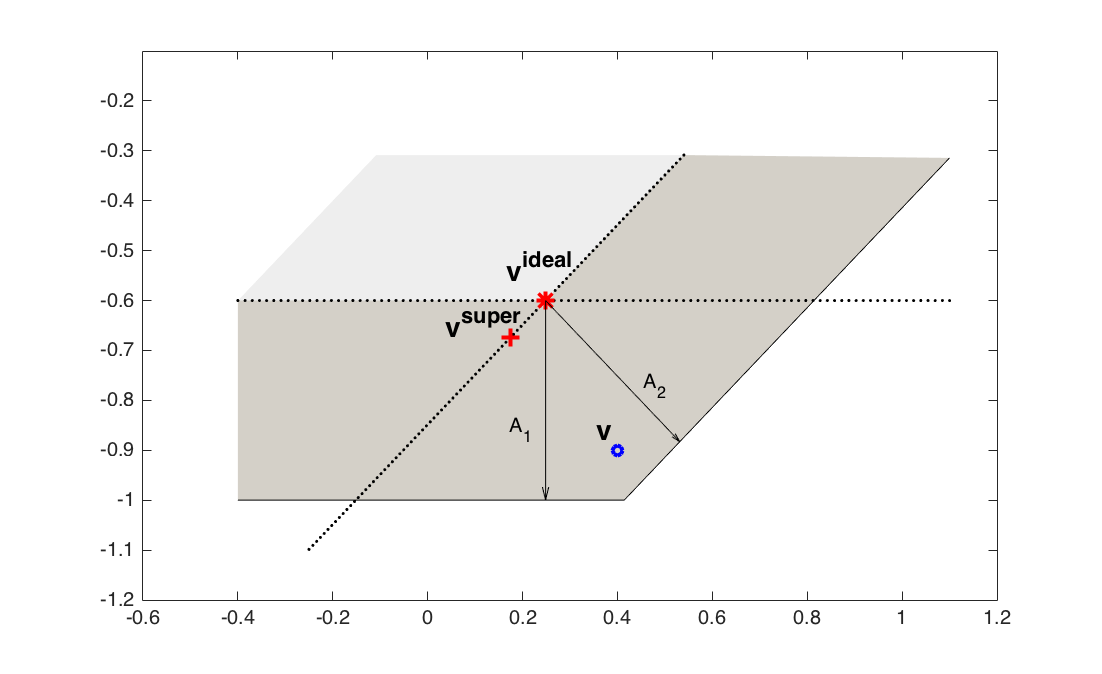}
	\caption{This is an example of $\bu^\text{super}$ in $\Real^2$ where $\tau=0.4$, $F_1=[0\ -1]^\top$, $F_2=[\frac{1}{\sqrt{2}}\ -\frac{1}{\sqrt{2}}]^\top$, $\bx=[1\ 0]^\top$, and $\bu=[0.4\ -0.9]^\top$. The light gray area is the ideal polytope and the dark gray area is the dual polytope. In this example, we have $\Gamma(\bu)=\{1,2\}$ while $\Gamma^*(\bu)=\{2\}$. As a result, $\bu^\text{super}$ is defined as the projection of $\bu$ onto the ideal polytope that only contains neuron $2$.}
	\label{fig:super}
\end{figure}
\vspace{3mm}

It is indeed unclear why we need these definitions at this stage of the proof. It would be clearer why we need them in the next two subsections once we go into the main analysis. Before we move on to more details, see Figure~\ref{fig:ideal coupling} and Figure~\ref{fig:super} again to familiarize with the definitions.

\subsection{Ideal SNN remains unchanged after firing spikes}\label{sec:unchange}
In this subsection, we are going to prove the following important lemma saying that the dual SNN would not change its ideal SNN after firing spikes.

\begin{lemma}[ideal SNN remains unchanged after firing spikes]\label{lem:ideal SNN unchaged}
There exists a polynomial $\alpha(\cdot)$ such that if $F$ is nice and $0<\alpha\leq\alpha(\frac{\tau\cdot\gamma(F)}{n\cdot\lambda_{\max}})$, then $\bu(t)-\alpha F^\top\bs(t)\in S_{\bu^\text{ideal}(t)}$ for each $t\geq0$.
\end{lemma}

\begin{proof}[Proof of~\autoref{lem:ideal SNN unchaged}]
	First, note that for each $\bu\in S_{\bu^\text{ideal}(t)}$, by the property of dual polytope, there exists an unique $\bz\in\Real_{\geq0}^{|\Gamma(\bu^\text{ideal}(t))|}$ such that $\bu=\bu^\text{ideal}(t)+F_{\Gamma(\bu^\text{ideal}(t))}^\top\bz$ where $\bz$ can be thought of as the \textit{coordinates} of $\bu$ in $S_{\bu^\text{ideal}(t)}$. With this concept in mind, it is then sufficient to show that whenever neuron $i$ fires, $\bz_i>\alpha$. The reason is that
	\begin{align}
	\bu(t)-\alpha F^\top\bs(t) &= \bu^\text{ideal}(t) + F_{\Gamma(\bu^\text{ideal}(t))}^\top\bz - \sum_{i\in\Gamma(\bu(t))}\alpha F_i\nonumber\\
	&= \bu^\text{ideal}(t) + \sum_{i\in\Gamma(\bu^\text{ideal}(t))\backslash\Gamma(\bu(t))}\bz_iF_i + \sum_{i\in\Gamma(\bu(t))}(\bz_i-\alpha)F_i.\label{eq:ideal SNN fire}
	\end{align}
	As a result, if $\bz_i-\alpha>0$ for every $i\in\Gamma(\bu(t))$, then we have $\bu(t)-\alpha F^\top\bs(t)\in S_{\bu^\text{ideal}(t)}$ because every new coordinates are still non-negative. See Figure~\ref{fig:coordinate} for an example.
	
	\begin{figure}[h]
		\centering
		\includegraphics[width=10cm]{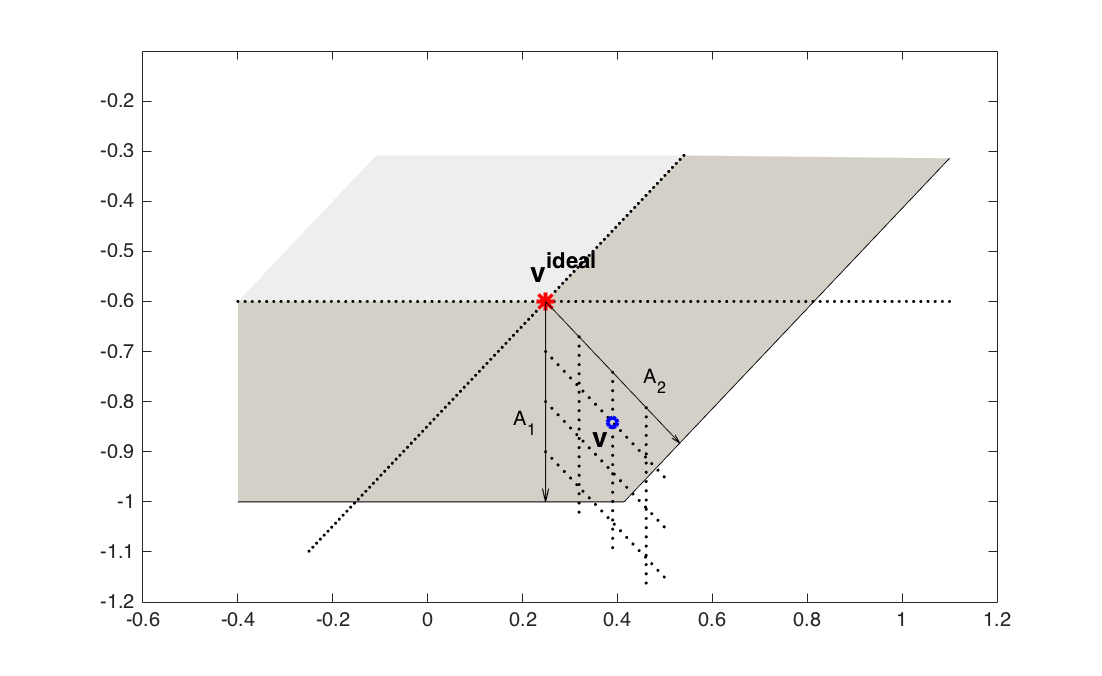}
		\caption{This is an example of \textit{coordinates} of  $S_{\bu^\text{ideal}}$ in $\Real^2$ where $\tau=0.4$, $F_1=[0\ -1]^\top$ and $F_2=[\frac{1}{\sqrt{2}}\ -\frac{1}{\sqrt{2}}]^\top$. The light gray area is the ideal polytope and the dark gray area is the dual polytope. In this example, the dot lines are the \textit{level set} of each coordinates in $S_{\bu^\text{ideal}}$. For instance, the $\bu$ in the figure has coordinate $\bz=[0.1\ 0.2]^\top$ and thus we have $\bu=\bu^\text{ideal}+F^\top\bz$.}
		\label{fig:coordinate}
	\end{figure}
	\vspace{3mm}
	
	\begin{claim}\label{claim:ideal SNN stay after fire coordinate}
		There exists a polynomial $\alpha(\cdot)$ such that when $0<\alpha\leq\poly(\frac{\tau\cdot\gamma(F)}{n\cdot\lambda_{\max}})$ and $\bu(t)=\bu^\text{ideal}(t)+F_{\Gamma(\bu^\text{ideal}(t))}^\top\bz\in S_{\bu^\text{ideal}(t)}$ for some $t\geq0$, if $i\in\Gamma(\bu(t))$, then $\bz_i>\alpha$.
	\end{claim}
	\begin{proof}[Proof of Claim~\ref{claim:ideal SNN stay after fire coordinate}]
		The proof consists of two steps. First, we are going to show that for any $t\geq0$, $\bu(t)$ is close to $\bu^\text{ideal}(t)$. Concretely, if $\alpha\leq\frac{\tau}{m}$, then $\|\bu(t)-\bu^\text{ideal}(t)\|_2\leq\tau\lambda_{\max}$. Second, we are going to show that once we pick $\alpha$ small enough, then for any $i\in\Gamma(\bu^\text{ideal}(t))$, the wall $W_i$ is far away from the $\alpha$-level set in $S_{\bu^\text{ideal}(t)}$. Thus, whenever neuron $i$ fires, $\bz_i>\alpha$.
		
		The first step is a key observation that the distance between $\bu(t)$ and $\bu^\text{ideal}(t)$ would not increase after the neurons fire spikes. The main reason is that neuron $i$ fires at time $t$ if and only if $F_i^\top\bu(t)>1$. As a result,
		\begin{align*}
		&\|\left(\bu(t)-\alpha F^\top\bs(t)\right)-\bu^\text{ideal}(t)\|_2^2\\
        =&\ \|\bu(t)-\bu^\text{ideal}(t)\|_2^2 + \alpha^2\|F^\top\bs(t)\|_2^2 - 2\alpha\left(F^\top\bs(t)\right)^\top\left(\bu(t)-\bu^\text{ideal}(t)\right)\\
		=&\ \|\bu(t)-\bu^\text{ideal}(t)\|_2^2 + \alpha^2\|F^\top\bs(t)\|_2^2 -2\alpha\sum_{i\in\Gamma(\bu(t))}F_i^\top\left(\bu(t)-\bu^\text{ideal}(t)\right)\\
		\leq&\ \|\bu(t)-\bu^\text{ideal}(t)\|_2^2 + \alpha^2|\Gamma(\bu(t))|^2 -2\alpha\tau|\Gamma(\bu(t))| \, .
		\end{align*}
		That is, if $\alpha\leq\frac{\tau}{m}$, then $\|\left(\bu(t)-\alpha F^\top\bs(t)\right)-\bu^\text{ideal}(t)\|_2$ would not increase after some neurons fire. Furthermore, the longest distance between $\bu(t)$ and $\bu^\text{ideal}(t)$ would then be $\tau\lambda_{\max}$.
		
		The second step is rather complicated. Let us start with some definitions. Recall that for any $i\in[\pm n]$, the wall $i$ is defined as $W_i=\{\bu\in\Real^m:\ F_i^\top\bu=1 \}$. Now, define the $\alpha$-level set of $i$ in $S_{\bu^\text{ideal}(t)}$ as
		$$
		L_{\bu^\text{ideal}(t),i,\alpha} = \{\bu\in S_{\bu^\text{ideal}(t)}:\ \bu=\bu^\text{ideal}(t)+F_{\Gamma(\bu^\text{ideal}(t))}^\top\bz,\ \bz_i=\alpha\}. 
		$$
		That is, $L_{\bu^\text{ideal}(t),i,\alpha} $ consists of the set of points in $S_{\bu^\text{ideal}(t)}$ that has the $i$-th coordinate to be $\alpha$.
		
		\begin{claim}[furtherest point in $S_{\bu^\text{ideal}}$]\label{claim:furthest point in safe region}
			For any $t\geq0$, let $\bu_{\Gamma(\bu^\text{ideal}(t))}$ be the unique point $\bu\in S_{\bu^\text{ideal}(t)}$ such that for any $i\in\Gamma(\bu^\text{ideal}(t))$, $F_i^\top\left(\bu-\bu^\text{ideal}(t)\right)=\tau$. Then, we have $\|\bu_{\Gamma(\bu^\text{ideal}(t))}-\bu^\text{ideal}(t)\|_2=\max_{\bu\in S_{\bu^\text{ideal}(t)}}\|\bu-\bu^\text{ideal}(t)\|_2$.
		\end{claim}
		\begin{proof}[Proof of Claim~\ref{claim:furthest point in safe region}]
			Let us prove by contradiction. Suppose $\bu^*\in S_{\bu^\text{ideal}(t)}$ such that $\|\bu_{\Gamma(\bu^\text{ideal}(t))}-\bu^\text{ideal}(t)\|_2<\|\bu^*-\bu^\text{ideal}(t)\|_2$. To simplify the notations, let $\bu_{\Gamma}=\bu_{\Gamma(\bu^\text{ideal}(t))}-\bu^\text{ideal}(t)$ and $\bu=\bu^*-\bu^\text{ideal}(t)$.
			
			By definition, we have $F_i^\top\bu_\Gamma=\tau$ for all $i\in\Gamma(\bu^\text{ideal}(t))$ and $\bu=F_{\Gamma(\bu^\text{ideal}(t))}^\top\bz_\Gamma$ for some $\bz_\Gamma\in\Real_{>0}$. On the other hand, we also have $0\leq F_i^\top\bu\leq\tau$ for all $i\in\Gamma(\bu^\text{ideal}(t)$.
			
			Now, look at the quantity $\bu_{\Gamma}^\top\left(\bu-\bu_\Gamma\right)$. Note that since $\|\bu\|_2>\|\bu_\Gamma\|_2$, we have $\bu_{\Gamma}^\top\left(\bu-\bu_\Gamma\right)>0$. Also, for any $i\in\Gamma(\bu^\text{ideal}(t))$, we have $F_i^\top\left(\bu-\bu_\Gamma\right)\leq0$. Using the fact that $\bu=F_{\Gamma(\bu^\text{ideal}(t))}^\top\bz_\Gamma$ for some $\bz_\Gamma\in\Real_{>0}$, we have
			\begin{align*}
			0 < \bu_{\Gamma}^\top\left(\bu-\bu_\Gamma\right)&=\bz_\Gamma^\top F_{\Gamma(\bu^\text{ideal}(t))}\left(\bu-\bu_\Gamma\right)\\
			&= \bz_\Gamma^\top\bu\leq0,
			\end{align*}
			where $\bu=F_{\Gamma(\bu^\text{ideal}(t))}\left(\bu-\bu_\Gamma\right)\in\Real_{\leq0}^{|\Gamma(\bu^\text{ideal}(t))|}$. That is, we reach a contradiction and thus $\|\bu\|_2\leq\|\bu_\Gamma\|_2$ and we conclude that $\bu_{\Gamma(\bu^\text{ideal}(t))}$ is the furtherest point from $\bu^\text{ideal}(t)$ in $S_{\bu^\text{ideal}(t)}$.
		\end{proof}
		\begin{claim}[intersection of wall and $\alpha$-level set is far]\label{claim:ideal SNN intersection is far}
			When $0<\alpha\leq\tau^2\cdot\gamma(F)^3$, for any $t\geq0$ and $i\in\Gamma(\bu^\text{ideal}(t))$, we have
			$$
			\min_{\bu:\ \bu\in W_i\cap L_{\bu^\text{ideal}(t),i,\alpha}}\|\bu-\bu^\text{ideal}(t)\|_2>\|\bu_{\Gamma(\bu^\text{ideal}(t))}-\bu^\text{ideal}(t)\|_2.
			$$
		\end{claim}
		\begin{proof}[Proof of Claim~\ref{claim:ideal SNN intersection is far}]
			First, let us write $\bu_{\Gamma(\bu^\text{ideal}(t))}=\bu^\text{ideal}(t)+\sum_{i\in\Gamma(\bu^\text{ideal}(t))}\bz_iF_i$ where $\bz_i\geq\tau\cdot\gamma(F)$ by~\autoref{def:nice}. Furthermore, for any $i\in\Gamma(\bu^\text{ideal}(t))$, we have
			\begin{equation*}
			\text{dist}\left(\bu_{\Gamma(\bu^\text{ideal}(t))},\text{span}(F_{\Gamma(\bu^\text{ideal}(t))\backslash\{i\}})\right)\geq|\bz_i|\cdot\text{dist}\left(F_i,\text{span}(F_{\Gamma(\bu^\text{ideal}(t))\backslash\{i\}})\right)\geq\tau\cdot\gamma(F)^2,
			\end{equation*}
			where the last inequality follows~\autoref{def:nice}. Namely, if we pick $0<\alpha<\tau^2\cdot\gamma(F)^3$, then
			\begin{equation*}
			\text{dist}\left(\bu_{\Gamma(\bu^\text{ideal}(t))},L_{\bu^\text{ideal}(t),i,\alpha}\right)>0
			\end{equation*}
			and $\bu_{\Gamma(\bu^\text{ideal}(t))}\in\text{Cone}(F_i,L_{\bu^\text{ideal}(t),i,\alpha})$ because $\bz_i\geq\gamma(F)$. Finally, observe that for any $\bu\in W_i\cap L_{\bu^\text{ideal}(t),i,\alpha}$, we have $\bu_{\Gamma(\bu^\text{ideal}(t))}^\top\left(\bu-\bu_{\Gamma(\bu^\text{ideal}(t))}\right)>0$. This completes the proof of Claim~\ref{claim:ideal SNN intersection is far}.
		\end{proof}
		Combine Claim~\ref{claim:furthest point in safe region} and Claim~\ref{claim:ideal SNN intersection is far}, we know that when neuron $i$ fires, the corresponding coordinate $\bz_i$ will be at least $\alpha$. This completes the proof of Claim~\ref{claim:ideal SNN stay after fire coordinate}.
	\end{proof}
	Now,~\autoref{lem:ideal SNN unchaged} follows from Claim~\ref{claim:ideal SNN stay after fire coordinate} and equation~\autoref{eq:ideal SNN fire}.
\end{proof}

\subsection{Strict convergence of ideal SNN and auxiliary SNNs}\label{sec:strict}
In this subsection, the goal is to characterize the dynamics of both ideal and auxiliary SNN. Before defining auxiliary SNN, let us first see the following lemma about the dynamics of ideal SNN. 

\begin{lemma}[dynamics of ideal SNN]\label{lem:ideal SNN dynamics}
	If $F$ is nice, then for any $t\geq0$, we have
	$$\bu^\text{ideal}(t+dt)=\bu^\text{ideal}(t) + \left(\bx-\Pi_{F_{\Gamma(\bu^\text{ideal}(t))}}\bx\right)dt.$$
\end{lemma}
\begin{proof}[Proof of~\autoref{lem:ideal SNN dynamics}]
	We consider two cases: (i) there is no neuron fires any spike and (ii) there is a neuron fires a spike.
	
	\textbf{Case (i)}: By~\autoref{def:ideal coupling}, $\bu(t)=\bu^\text{ideal}+F_{\Gamma(\bu^\text{ideal}(t))}^\top\bz$ for some $\bz\geq0$. Also, rewrite the updates $\bx$ as 
	$$\bx=\left(\bx-\Pi_{F_{\Gamma(\bu^\text{ideal}(t))}}\bx\right)+\Pi_{F_{\Gamma(\bu^\text{ideal}(t))}}\bx.$$
	First, $F_i^\top\left(\bx-\Pi_{F_{\Gamma(\bu^\text{ideal}(t))}}\bx\right)=0$ for each $i\in\Gamma(\bu^\text{ideal}(t))$. Next, since there is no neuron fires at time $t$, observe that $\bu(t)+\Pi_{F_{\Gamma(\bu^\text{ideal}(t))}}\bx\in S_{\bu^\text{ideal}(t)}$. Finally, since $\bx-\Pi_{F_{\Gamma(\bu^\text{ideal}(t))}}\bx$ is orthogonal to the subspace spanned by the active neurons, we then have $\bu(t)+\bx dt\in S_{\bu^\text{ideal}(t)+(\bx-\Pi_{F_{\Gamma(\bu^\text{ideal}(t))}}\bx)dt}$.
	
	\textbf{Case (ii)}: To handle spikes, the idea is to focus on the spike term first, and once $\bu(t)$ goes back to the interior of the dual polytope, then it becomes case (i). Here, we use an assumption that if there are some neurons fire at time $t$ and they trigger consecutive firing, we add the external charging \textit{after} the consecutive firing. As a result, it suffices to show that $\bu(t)-\alpha F^\top\bs(t)\in S_{\bu^\text{ideal}(t)}$, which immediately follows from~\autoref{lem:ideal SNN unchaged}.
	
	We conclude that for any $t\geq0$, $\bu^\text{ideal}(t+dt)=\bu^\text{ideal}(t) + \left(\bx-\Pi_{F_{\Gamma(\bu^\text{ideal}(t))}}\bx\right)dt$.
\end{proof}

From~\autoref{lem:ideal SNN dynamics}, one can see that the improvement of ideal SNN is not proportional to the residual error when the $\Pi_{F_{\Gamma(\bu^\text{ideal}(t))}}\bx\neq F^\top\br^\text{ideal}(t)$. As a result, we have to design a bunch of \textit{auxiliary SNN} to make sure that at least one of them has improvement proportional to the residual error. The auxiliary SNNs $\{\bu^\text{auxiliary}_d(t)\}_{d\in[m-1]}$ is defined as follows.
\begin{definition}[auxiliary SNNs]\label{def:auxiliary}
	For each $t\geq0$, and $d\in[m-1]$, define $\bu^\text{auxiliary}(0)=\mathbf{0}$ and
	\[
	\bu^\text{auxiliary}_d(t+dt) = \left\{
	\begin{array}{ll}
	\bu^\text{auxiliary}_d(t)+F^\top\left(\br-\br^\text{ideal}(t)\right)dt&\text{, if }|\Gamma^*(\bu^\text{ideal}(t+dt))|=d\\&\text{ and } |\Gamma^*(\bu^\text{ideal}(t))|=d,\\
	\bu^\text{super}(t+dt)&\text{, if }|\Gamma^*(\bu^\text{ideal}(t+dt))|=d\\&\text{ and }|\Gamma^*(\bu^\text{ideal}(t))|\neq d,\\
	\bu_d^\text{auxiliary}(t)&\text{, else}.
	\end{array}
	\right.
	\]
\end{definition}

The auxiliary SNNs have the following important property that is crucial in the proof of the~\autoref{lem:strict improvement} which gives the strict improvement guarantee.

\begin{lemma}[Auxiliary SNNs jump]\label{lem:auxiliary SNN jump}
	Suppose $F$ is nice and $\tau=O(\frac{\gamma(F)}{n^2\cdot\lambda_{\max}^2})$. For any $t>0$ and $d\in[m-1]$, if $|\Gamma^*(\bu^\text{ideal}(t))|\neq|\Gamma^*(\bu^\text{ideal}(t+dt))|=d$, then $\bx^\top\left(\bu^\text{auxiliary}_d(t+dt)-\bu^\text{auxiliary}_d(t)\right)>0$.
\end{lemma}
\begin{proof}[Proof of~\autoref{lem:auxiliary SNN jump}]
	By the definition of auxiliary SNNs, we have three observations. First, $\|\bu^\text{auxiliary}_d(t+dt)-\bu^\text{ideal}(t)\|_2=\|\bu^\text{super}(t+dt)-\bu^\text{ideal}(t)\|_2=O(\tau\cdot n\cdot\lambda_{\max})$. Second, there exists $0\leq t'<t$ such that $\bu^\text{auxiliary}_d(t)=\bu^\text{super}(t')$ and $\Gamma(\bu^\text{ideal}(t'))\neq\Gamma(\bu^\text{ideal}(t))$. That is, we also have $\|\bu^\text{auxiliary}_d(t)-\bu^\text{ideal}(t')\|_2=\|\bu^\text{super}(t')-\bu^\text{ideal}(t)\|_2=O(\tau\cdot n\cdot\lambda_{\max})$. Finally, since $\Gamma(\bu^\text{ideal}(t'))\neq\Gamma(\bu^\text{ideal}(t))$, by~\autoref{lem:ideal SNN dynamics}, we have $\bx^\top\left(\bu^\text{ideal}(t)-\bu^\text{ideal}(t')\right)=\Omega(\|\bx\|_2\cdot\frac{\gamma(F)}{n\cdot\lambda_{\max}})$. Combine the three we have 
	\begin{align*}
	\bx^\top\left(\bu^\text{auxiliary}_d(t+dt)-\bu^\text{auxiliary}_d(t)\right)&\geq\bx^\top\left(\bu^\text{ideal}(t)-\bu^\text{ideal}(t')\right) - O(\|\bx\|_2\cdot\tau\cdot n\cdot\lambda_{\max})\\
	&\geq\Omega(\|\bx\|_2\cdot\frac{\lambda(F)}{n\cdot\lambda_{\max}}) - O(\|\bx\|_2\cdot\tau\cdot n\cdot\lambda_{\max})>0,
	\end{align*}
	where the last inequality holds when we pick $\tau=O(\frac{\gamma(F)}{n^2\lambda_{\max}^2})$.
\end{proof}

Now, we are able to prove the main lemma about identifying a potential function that is strictly improving as long as $\br^\text{ideal}(t)$ is not the optimal solution for $\ell_1$ minimization problem.

\begin{lemma}[strict improvement]\label{lem:strict improvement}
	For any $t>0$, we have
	\begin{equation*}
	\frac{d}{dt}\bx^\top\left(\bu^\text{ideal}(t)+\sum_{d\in[m-1]}\bu^\text{auxiliary}_d(t)\right)\geq\bx^\top F^\top\br^\text{ideal}(t).
	\end{equation*}
\end{lemma}
\begin{proof}[Proof of~\autoref{lem:strict improvement}]
	The proof is based on case analysis on the size of $|\Gamma^*(\bu^\text{ideal}(t))|$. We consider three cases:
	\begin{enumerate}[label=(\roman*)]
		\item $\Gamma^*(\bu^\text{ideal}(t))=\Gamma(\bu^\text{ideal}(t))$, 
		\item $\Gamma^*(\bu^\text{ideal}(t))\subsetneq\Gamma(\bu^\text{ideal}(t))$ and $|\Gamma^*(\bu^\text{ideal}(t))|=|\Gamma^*(\bu^\text{ideal}(t+dt))|$, and 
		\item $\Gamma^*(\bu^\text{ideal}(t))\subsetneq\Gamma(\bu^\text{ideal}(t))$ and $|\Gamma^*(\bu^\text{ideal}(t))|\neq|\Gamma^*(\bu^\text{ideal}(t+dt))|$.
	\end{enumerate}
	In each case, we are going to show that at least one of $\bu^\text{ideal}(t)$ or $\bu^\text{auxiliary}_d(t)$ for some $d\in[m-1]$ has the desired improvement. Also, we need to show that all of them would not get worse. Formally, we state it as the following claim.
	\begin{claim}\label{claim:ideal auxiliary does not get worse}
		For every $t>0$ and $d\in[m-1]$, we have $\frac{d}{dt}\bx^\top\bu^\text{ideal}(t), \frac{d}{dt}\bx^\top\bu^\text{auxiliary}_d(t)\geq0$.
	\end{claim}
	\begin{proof}[Proof of Claim~\ref{claim:ideal auxiliary does not get worse}]
		From~\autoref{lem:ideal SNN dynamics}, we already have $\bx^\top\bu^\text{ideal}(t)\geq0$. For any $d\in[m-1]$, consider three cases as in~\autoref{def:auxiliary}.
		
		If $|\Gamma^*(\bu^\text{ideal}(t))|=|\Gamma^*(\bu^\text{ideal}(t+dt))|=d$, then $\frac{d}{dt}\bx^\top\bu^\text{auxiliary}_d(t)=\bx^\top F^\top(\br-\br^\text{ideal}(t))\geq0$.
  
		If $|\Gamma^*(\bu^\text{ideal}(t))|\neq|\Gamma^*(\bu^\text{ideal}(t+dt))|=d$, then by~\autoref{lem:auxiliary SNN jump} we have $$\bx^\top\left(\bu^\text{auxiliary}_d(t+dt)-\bu^\text{auxiliary}_d(t)\right)>0$$ 
        and thus  $\frac{d}{dt}\bx^\top\bu^\text{auxiliary}_d(t)\geq0$.
		
		Finally, when none of the above happens, we simply have $\frac{d}{dt}\bx^\top\bu^\text{auxiliary}_d(t)=0$.
	\end{proof}
	With Claim~\ref{claim:ideal auxiliary does not get worse}, it suffices to show that at least one of $\bu^\text{ideal}(t)$ or $\bu^\text{auxiliary}_d(t)$ for some $d\in[m-1]$ has the desired improvement in all of the above three cases.
	
	\textbf{Case (i)}: In this case, $F^\top\br^\text{ideal}(t)=\Pi_{F_{\Gamma(\bu^\text{ideal}(t))}\bx}$. Thus, by~\autoref{lem:ideal SNN dynamics}, we have $\frac{d}{dt}\bx^\top\bu^\text{ideal}(t)=\bx^\top\left(F^\top\br-\br^\text{ideal}(t)\right)$.
	
	\textbf{Case (ii)}: In this case, let $d=|\Gamma^*(\bu^\text{ideal}(t+dt))|$. By~\autoref{def:auxiliary}, we have $\frac{d}{dt}\bx^\top\bu^\text{auxiliary}(t)=\bx^\top\left(F^\top\br-\br^\text{ideal}(t)\right)$.
	
	\textbf{Case (iii)}: In this case, let $d=|\Gamma^*(\bu^\text{ideal}(t+dt))|$. By~\autoref{lem:auxiliary SNN jump}, we have $$\bx^\top\left(\bu^\text{auxiliary}_d(t+dt)-\bu^\text{auxiliary}_d(t)\right)>0$$ and thus $\frac{d}{dt}\bx^\top\bu^\text{auxiliary}_d(t)\geq\bx^\top\left(F^\top\br-\br^\text{ideal}(t)\right)$.
	
	This completes the proof of~\autoref{lem:strict improvement}.
\end{proof}

Finally, before we go into the final proof for~\autoref{thm:l1}, we need the following lemma about some properties about the ideal solution defined in~\autoref{def:ideal solution}.

\begin{lemma}[properties of ideal solution]\label{lem:ideal solution properties}
	For any $t\geq0$, we have the following.
	\begin{enumerate}
		\item $\bx^\top F^\top\br^{\text{ideal}}(t)=\|F^\top\br^{\text{ideal}}(t)\|_2^2$,
		\item $\|\bx-F^\top\br^{\text{ideal}}(t)\|_2^2=\|\bx\|_2^2-\|F^\top\br^{\text{ideal}}(t)\|_2^2$, and
	\end{enumerate}
\end{lemma}
\begin{proof}[Proof of~\autoref{lem:ideal solution properties}]
	The lemma is directly followed by the following property of conic projection. For any $F\in\Real^{n\times m}$, $\bx\in\Real^m$, and $\Gamma\subseteq[\pm n]$ be a valid set, we have $\bx^\top F\bx_{F,\bx,\Gamma}=\|F\bx_{F,\bx,\Gamma}\|_2^2$. In the following, we are going to first prove this property of conic projection and then use it to prove the lemma.
	
	Let us rewrite the definition of conic projection as an optimization program.
	\begin{equation}\label{op:conicproj}
	\begin{aligned}
	& \underset{\br\in\Real^n}{\text{minimize}}
	& & \frac{1}{2}\|\bx-F^\top\br\|_2^2 \\
	& \text{subject to}
	& & \br_j\geq0,\ j\in\Gamma,\\
	& & & \br_i=0,\ i,-i\notin\Gamma.
	\end{aligned}
	\end{equation}
	Let $\by$ be the dual variable of~\autoref{op:conicproj} and $\by^*$ be the optimal dual value, the Lagrangian of~\autoref{op:conicproj} is
	\begin{equation*}
	\mathcal{L}(\br) = \frac{1}{2}\|\bx-F^\top\br\|_2^2-\by^\top\br,
	\end{equation*}
	and its gradient is
	\begin{equation*}
	\nabla_{\br}\mathcal{L}(\br) = FF^\top\br - F\bx - \by.
	\end{equation*}
	By the KKT condition, we know that the optimal primal solution $\br_{F,\bx,\Gamma}$ and the optimal dual solution $\by^*$ make the gradient of the Lagrangian diminish.
	\begin{equation}\label{eq:idealalgorithm-conicproj-KKT-gradient}
	\nabla_{\br}\mathcal{L}(\br_{F,\bx,\Gamma}) = FF^\top\br_{F,\bx,\Gamma} - F\bx - \by^*=0,
	\end{equation}
	and the complementary slackness
	\begin{equation}\label{eq:idealalgorithm-conicproj-KKT-cs}
	\br_{F,\bx,\Gamma}^\top\by^*=0.
	\end{equation}
	By~\autoref{eq:idealalgorithm-conicproj-KKT-gradient} and~\autoref{eq:idealalgorithm-conicproj-KKT-cs}, we have
	\begin{equation*}\label{eq:idealalgorithm-conicproj-orthogonal}
	(F^\top\br_{F,\bx,\Gamma})^\top(F^\top\br_{F,\bx,\Gamma}-\bx)=0.
	\end{equation*}
	As a result, $\bx^\top F^\top\br_{F,\bx,\Gamma}=\|F^\top\br_{F,\bx,\Gamma}\|_2^2$.
	
	This completes the proof of~\autoref{lem:ideal solution properties}.
\end{proof}

\subsection{The convergence of dual SNN}\label{sec:l1 convergence}
In this subsection, we are going to prove the main convergence theorem of the dual SNN using ideal and auxiliary SNN. The following lemma says that at least one of ideal SNN or auxiliary SNN improves at each step.

The following lemma shows the monotonicity of the residual error $\|\bx^{\text{ideal}}-F^\top\br\|_2$.
\begin{lemma}[monotonicity of residual error]\label{lemma:idealalgorithm-primal-nondecreasing}
	There exists a polynomial $\alpha(\cdot)$ such that when $0<\alpha\leq \alpha(\frac{\gamma(F)}{n\cdot\lambda_{\max}})$, we have $\|\bx-F^\top\br^{\text{ideal}}(t)\|_2$ is non-increasing and $\|F^\top\br^{\text{ideal}}(t)\|_2$ is non-decreasing in $t$.
\end{lemma}
\begin{proof}[Proof of~\autoref{lemma:idealalgorithm-primal-nondecreasing}]
	Consider two cases.
	\begin{enumerate}[label=(\arabic*)]
		\item When there is a new index joins the active set. Clearly that $\|F^\top\br^{\text{ideal}}(t)\|_2$ won't decrease since the new cone contains the old one. By~\autoref{lem:ideal solution properties}, we know that $\|\bx-F^\top\br^{\text{ideal}}(t)\|_2$ is non-increasing.
		\item When there is an index leaves the the active set. Without loss of generality, assume $j\in[\pm n]$ leaves the active set. In the following, we want to show that $\bx^{\text{ideal}}_{|j|}(t)=0$. As the direction of $\bu^{\text{ideal}}(t)$ is $\bx-F^\top\br^{\text{ideal}}(t)$, it means that $F_j^{\top}(\bx-F^\top\br^{\text{ideal}}(t))<0$. Suppose $\bx^{\text{ideal}}_{|j|}(t)\neq0$ for contradiction. Since $j$ was in the active set, it is the case that $\bx^{\text{ideal}}_j(t)>0$. Take $0<\epsilon<\min\{\bx^{\text{ideal}}_j(t)/2,-(\bx-F^\top\br^{\text{ideal}}(t))^{\top}F_j/\|F_j\|_2\}$ and define $\br'=\br^{\text{ideal}}(t) - \epsilon\cdot F_j/\|F_j\|_2$. Note that $\br'$ lies in the original active cone. Observe that
		\begin{align*}
		\|\bx-F^\top\br'\|_2^2 &= \|\bx - F^\top\br^{\text{ideal}}(t) - \epsilon\cdot F_j/\|F_j\|_2\|_2^2\\
		&=\|\bx-F^\top\br^{\text{ideal}}(t)\|_2^2 + \|\epsilon\cdot F_j/\|F_j\|_2\|_2^2+2\epsilon\cdot(\bx-F^\top\br^{\text{ideal}}(t))^{\top}F_j/\|F_j\|_2\\
		&\leq\|\bx-F^\top\br^{\text{ideal}}(t)\|_2^2 + \epsilon^2 - 2\epsilon^2\\
		&<\|\bx-F^\top\br^{\text{ideal}}(t)\|_2^2
		\end{align*}
		which contradicts to the optimality of $\br^{\text{ideal}}(t)$ since $\br'$ is also a feasible solution. We conclude that $\br^{\text{ideal}}_j(t)=0$. As a result, $F^\top\br^{\text{ideal}}(t)$ remains the same and $\|F^\top\br^{\text{ideal}}(t)\|_2$ won't decrease.
	\end{enumerate}
\end{proof}

The next lemma upper bounds the $\ell_2$ residual error of $\br^\text{ideal}(t)$.
\begin{lemma}[convergence of residual error]\label{lemma:idealalgorithm-l2bound}
	There exists a polynomial $\alpha(\cdot)$ such that when $0<\alpha\leq \alpha(\frac{\gamma(F)}{n\cdot\lambda_{\max}})$, we have  for any $\epsilon>0$, when $t\geq\frac{m\cdot\OPT^{\ell_1}}{\epsilon\cdot\|\bx\|_2}$, $\|\bx-F^\top\br^{\text{ideal}}(t)\|_2\leq\epsilon\cdot\|\bx\|_2$.
\end{lemma}
\begin{proof}[Proof of~\autoref{lemma:idealalgorithm-l2bound}]
	Assume the statement is wrong, i.e., $\|\bx-F^\top\br^{\text{ideal}}(t)\|_2>\epsilon\cdot\|\bx\|_2$. Then by~\autoref{lemma:idealalgorithm-primal-nondecreasing}, for any $0\leq s\leq t$,
	\begin{align*}
	\|\bx-F^\top\br^{\text{ideal}}(s)\|_2^2 &= \|\bx\|_2^2-\|F\br^{\text{ideal}}(s)\|_2^2\\
	&\geq\|\bx\|_2^2-\|F^\top\br^{\text{ideal}}(t)\|_2^2\\
	&=\|\bx-F^\top\br^{\text{ideal}}(t)\|_2^2>\epsilon^2\cdot\|\bx\|_2^2.
	\end{align*}
	Since $t\geq\frac{\OPT^{\ell_1}}{\epsilon\cdot\|\bx\|_2}$, by~\autoref{lem:strict improvement},
	\begin{align*}
	\bx^\top\left(\bu^\text{ideal}(t)+\sum_{d\in[m-1]}\bu^\text{auxiliary}_d(t)\right) &= \int_0^t\bx^\top d\bu^\text{ideal}(t)+\sum_{d\in[m-1]}\int_0^t\bx^\top d\bu^\text{auxiliary}_d(t)\\
	&> t\cdot\epsilon\cdot\|\bx\|_2\geq m\cdot\OPT^{\ell_1},
	\end{align*}
	which is a contradiction to the optimality of $\OPT^{\ell_1}$ since $\bx^\top\bu^{\text{ideal}}(t),\bx^\top\bu^\text{auxiliary}_d(t)\leq\OPT^{\ell_1}$ for all $d\in[m-1]$. As a result, we conclude that $\|\bx-F^\top\br^{\text{ideal}}(t)\|_2\leq\epsilon\cdot\|\bx\|_2$.
\end{proof}

Finally, the following lemma shows that the $\ell_1$ error of $\br^\text{ideal}(t)$ can be upper bounded by the $\ell_2$ error via the strong duality of $\ell_1$ minimization problem and perturbation trick.
\begin{lemma}[convergence of $\ell_1$ error]\label{lemma:idealalgorithm-OPTbounds}
	For any $t\geq0$,
	\begin{equation}
	\left|\|\br^{\text{ideal}}(t)\|_1-\OPT^{\ell_1}\right|\leq\sqrt{\frac{n}{\lambda_{\min}}}\cdot\|\bx-F^\top\br^{\text{ideal}}(t)\|_2
	\end{equation}
\end{lemma}
\begin{proof}[Proof sketch]
	The proof of~\autoref{lemma:idealalgorithm-OPTbounds} consists of two steps. First, we show that the primal and the dual solution pair of ideal SNN at time $t$ is the optimal solution pair of a \textit{perturbed $\ell_1$ minimization problem} defined as shifting the $\bx$ in the constraint $F\bx=\bx$ to $F^\top\br^{\text{ideal}}(t)$. See~\autoref{op:basispursuit-ideal-perturbed} for the definition of the perturbed program. Next, by the standard perturbation theorem from optimization, we can upper bound $\|\br^{\text{ideal}}(t)\|_1$ with the distance between the original program and the perturbed program. Specifically, the difference induced by the perturbation is related to the $\ell_2$ norm of the differnce between $\bx$ and $F^\top\br^{\text{ideal}}(t)$, which is exactly the residual error. As a result, we know that the difference between the optimal value of the original $\ell_1$ minimization program and that of the perturbed program will converge to 0. Namely, we yield a convergence of $\|\br^{\text{ideal}}(t)\|_1$ to $\OPT^{\ell_1}$. See Section~\ref{sec:missing proofs convergent analysis l1} for more details.
\end{proof}

Finally, we can prove the main theorem in this section as follows.
\begin{proof}[Proof of~\autoref{thm:l1}]
	Pick $t_0=\Theta(\frac{m\cdot\sqrt{n}\cdot\|\bx\|_2}{\epsilon\cdot\sqrt{\lambda_{\min}\cdot\OPT^{\ell_1}}})$. By~\autoref{lemma:idealalgorithm-l2bound}, for any $t\geq t_0$, we can upper bound the $\ell_2$ residual error by
	$$\|\bx-F^\top\br^\text{ideal}(t)\|_2\leq\sqrt{\frac{\lambda_{\min}}{n}}\cdot\frac{\epsilon}{10}\cdot\OPT^{\ell_1}.$$
	Next, by~\autoref{lemma:idealalgorithm-OPTbounds}, we can then upper bound the $\ell_1$ error by
	\[
    \left|\|\br^\text{ideal}(t)\|_1-\OPT^{\ell_1}\right|\leq\sqrt{\frac{n}{\lambda_{\min}}}\cdot\|\bx-F^\top\br^\text{ideal}(t)\|_2\leq\frac{\epsilon}{10}\cdot\OPT^{\ell_1} \, .
    \]
	Now, the only thing left is connecting the ideal solution $\br^\text{ideal}(t)$ to the firing rate $\br(t)$. First, divide $\br(t)$ into two parts: the firing rate $\br^{[0,t_0]}$ before time $t_0$ and the firing rate $\bx^{(t_0,t]}$ from time $t_0$ to $t$. That is, $\br(t) = \frac{t_0}{t}\cdot\br^{[0,t_0]}+\frac{t-t_0}{t}\cdot\br^{(t_0,t]}$.
	
	Note that after $t_0\geq\Omega(\frac{m\cdot\sqrt{n}\cdot\|\bx\|_2}{\epsilon\cdot\sqrt{\lambda_{\min}\cdot\OPT^{\ell_1}}})$, the ideal solution has $\ell_1$ norm at most $(1+\epsilon)\cdot\OPT^{\ell_1}$. Thus, $\|\br^{(t_0,t]}\|_1\leq(1+(1+O(\frac{1}{t}))\cdot\frac{\epsilon}{10})\cdot\OPT^{\ell_1}\leq(1+\frac{\epsilon}{5})\cdot\OPT^{\ell_1}$.
	As for $\br^{[0,t_0]}$, from~\autoref{lemma:idealalgorithm-OPTbounds}, we have $\|\br^{[0,t_0]}\|_1\leq\OPT^{\ell_1}+\sqrt{\frac{n}{\lambda_{\min}}}\cdot\|\bx\|_2$.
	Combine the two, we have
	\[
	\left|\|\br(t)\|_1-\OPT^{\ell_1}\right|\leq\frac{\epsilon\cdot\OPT^{\ell_1}}{5} + \frac{t_0\cdot\left(\OPT^{\ell_1}+\sqrt{\frac{n}{\lambda_{\min}}}\cdot\|\bx\|_2\right)}{t}\leq\epsilon\cdot\OPT^{\ell_1},
	\]
	where the last inequality holds since $t\geq\Omega(\frac{m^2\cdot n\cdot\|\bx\|_2^2}{\epsilon^2\cdot\lambda_{\min}\cdot\OPT^{\ell_1}})$. This completes the proof for~\autoref{thm:l1}.
\end{proof}

\subsection{Proofs for the properties of ideal and auxiliary SNN}\label{sec:missing proofs ideal auxiliary SNN}
\begin{proof}[Proof of~\autoref{lem:ideal partition}]
	Let us start with an observation on~\autoref{def:ideal partition} about the points on the boundary of  the ideal polytope $\mathcal{P}_{F,1-\tau}$.
	\begin{claim}\label{claim:ideal active PD}
		If $F$ is non-degenerate, then for any $\bu^\text{ideal}\in\partial\mathcal{P}_{F,1-\tau}$, $\text{rank}(F_{\Gamma(\bu^\text{ideal})})=|\Gamma(\bu^\text{ideal})|$. Thus, $F_{\Gamma(\bu^\text{ideal})}^\top F_{\Gamma(\bu^\text{ideal})}$ is positive definite.
	\end{claim}
	
	Next, let us show that for $\bu^\text{ideal}_1\neq\bu^\text{ideal}_2\in\mathcal{P}_{F,1-\tau}$, $S_{\bu^\text{ideal}_1}\cap S_{\bu^\text{ideal}_2}=\emptyset$. It is trivially true when at least one of them does not lie on the boundary\footnote{Note that $\bu^\text{ideal}$ does not lie on the boundary of $\mathcal{P}_{F,1-\tau}$ if and only if $\Gamma(\bu^\text{ideal})=\emptyset$.} of $\mathcal{P}_{F,1-\tau}$. Now, consider the case where both of them lie on the boundary of $\mathcal{P}_{F,1-\tau}$ and denote their active set as $\Gamma_1=\Gamma(\bu^\textit{ideal}_1)$ and $\Gamma_2=\Gamma(\bu^\textit{ideal}_2)$. To prove from contradiction, suppose there exists $\bu\in S_{\bu^\text{ideal}_1}\cap S_{\bu^\text{ideal}_2}$. By definition, we have
	\begin{align*}
	\bu &= \bu^\text{ideal}_1 + F_{\Gamma_1}^\top\bz_1\\
	&=\bu^\text{ideal}_2 + F_{\Gamma_2}^\top\bz_2,
	\end{align*}
	where $\bz_1,\bz_2\geq0$. Let $\Gamma=\Gamma_1\cap\Gamma_2$. Consider the following cases.
	\begin{itemize}
		\item ($\Gamma=\Gamma_1=\Gamma_2$) By~\autoref{def:ideal partition}, we have $F_\Gamma^\top\bu^\text{ideal}_1=F_\Gamma^\top\bu^\text{ideal}_2=\mathbf{1}$ and thus
		\begin{equation*}
		(\bz_2-\bz_1)^\top F_\Gamma^\top F_\Gamma(\bz_2-\bz_1)=(\bz_2-\bz_1)^\top F_\Gamma^\top(\bu^\text{ideal}_1-\bu^\text{ideal}_2)=0.
		\end{equation*}
		As $F_\Gamma^\top F_\Gamma$ is positive definite by Claim~\ref{claim:ideal active PD}, we have $\bz_1=\bz_2$ and $\bu^\text{ideal}_1=\bu^\text{ideal}_2$, which is a contradiction.
		
		\item ($\Gamma_1\neq\Gamma_2$) Without loss of generality, assume $\Gamma_1\backslash\Gamma\neq\emptyset$ and $\bz_1\neq\mathbf{0}$. By~\autoref{def:ideal partition}, we have
		\begin{align*}
		F_{\Gamma_1\backslash\Gamma_2}^\top\left(\bu^\text{ideal}_1-\bu^\text{ideal}_2\right)&>\mathbf{0},\\
		F_{\Gamma_2\backslash\Gamma_1}^\top\left(\bu^\text{ideal}_1-\bu^\text{ideal}_2\right)&\leq\mathbf{0},\\
		F_{\Gamma}^\top\left(\bu^\text{ideal}_1-\bu^\text{ideal}_2\right)&=\mathbf{0}.
		\end{align*}
		As $\bz_1\neq\mathbf{0}$, we then have
		\begin{align*}
		\|F_{\Gamma_2}\bz_2-F_{\Gamma_1}\bz_1\|_2^2 &= \left(F_{\Gamma_2}\bz_2-F_{\Gamma_1}\bz_1\right)^\top\left(\bu^\text{ideal}_1-\bu^\text{ideal}_2\right)\\
		&=\left(-F_{\Gamma_1\backslash\Gamma}\bz_1|_{\Gamma_1\backslash\Gamma}\right)^\top\left(\bu^\text{ideal}_1-\bu^\text{ideal}_2\right)\\
		&+\left(F_{\Gamma_2\backslash\Gamma}\bz_2|_{\Gamma_2\backslash\Gamma}\right)^\top\left(\bu^\text{ideal}_1-\bu^\text{ideal}_2\right)\\
		&+\left(F_{\Gamma}\bz_2|_{\Gamma}-F_{\Gamma}\bz_1|_{\Gamma}\right)^\top\left(\bu^\text{ideal}_1-\bu^\text{ideal}_2\right)\\
		&<0.
		\end{align*}
		Note that the reason why the last inequality holds is because $\left(-F_{\Gamma_1\backslash\Gamma}\bz_1|_{\Gamma_1\backslash\Gamma}\right)^\top\left(\bu^\text{ideal}_1-\bu^\text{ideal}_2\right)<0$.
	\end{itemize}
	Finally, it is easy to see that $\{S_{\bu^\text{ideal}}\}_{\bu^\text{ideal}\in\mathcal{P}_{F,1-\tau}}$ covers $\mathcal{P}_{F,1}$ and thus we conclude that it is indeed a partition for $\mathcal{P}_{F,1}$.
\end{proof}

\subsection{Proofs for the convergent analysis of solving $\ell_1$ minimization}\label{sec:missing proofs convergent analysis l1}

\begin{proof}[Proof of~\autoref{lemma:idealalgorithm-OPTbounds}]
	For any $t\geq0$, define the following perturbed program of (\ref{op:basispursuit}) and its dual.
	\vspace{3mm}
	\begin{minipage}{\linewidth}
		\begin{minipage}{0.45\linewidth}
			\begin{equation}\label{op:basispursuit-ideal-perturbed}
			\begin{aligned}
			& \underset{\br}{\text{minimize}}
			& & \|\br\|_1 \\
			& \text{subject to}
			& & F^\top\br-F^\top\br^{\text{ideal}}(t)=0
			\end{aligned}
			\end{equation}
		\end{minipage}
		\begin{minipage}{0.45\linewidth}
			\begin{equation}\label{op:basispursuit-ideal-perturbed-dual}
			\begin{aligned}
			& \underset{\bu\in\Real^m}{\text{maximize}}
			& & (F^\top\br^{\text{ideal}}(t))^\top\bu \\
			& \text{subject to}
			& & \|F\bu\|_{\infty}\leq1.
			\end{aligned}
			\end{equation}
		\end{minipage}
	\end{minipage}
	
	Note that $\br^{\text{ideal}}(t)$ is treated as a given constant to the optimization program. It turns out that the ideal algorithm optimizes this primal-dual perturbed program at time $t$ with the following parameters.
	\begin{lemma}\label{lemma:proofofidealalgorithm-KKT}
		For any $t\geq0$, $(\br^*,\bu^*) = (\br^{\text{ideal}}(t),\bu^{\text{ideal}}(t))$ is the optimal solutions of (\ref{op:basispursuit-ideal-perturbed}).
	\end{lemma}
	\begin{proof}[Proof of~\autoref{lemma:proofofidealalgorithm-KKT}]
		We simply check the KKT condition. Since the program can be rewritten as a linear program, it satisfies the regularity condition of the KKT condition.
		
		First, the primal and the dual feasibility can be verified by the dynamics of ideal algorithm. That is, $F^\top\br^*-F^\top\br^{\text{ideal}}(t)=0$ and $\|F\bu^*\|_{\infty}\leq1$. Next, consider the Lagrangian of~\autoref{op:basispursuit-ideal-perturbed} as follows.
		\begin{align*}
		\mathcal{L}(\br,\bu) &= \|\br\|_1 - \bu^\top (F^\top\br-F^\top\br^{\text{ideal}}(t)),\\
		\nabla_{\br}\mathcal{L}(\br,\bu) &= \nabla_{\br}\|\br\|_1-F \bu.
		\end{align*}
		Now, let's verify that the gradient of the Lagrangian over $\bx$ is vanishing at $(\br^*,\bu^*) = (\br^{\text{ideal}}(t),\bu^{\text{ideal}}(t))$. That is,  $0\in\nabla_{\br}\mathcal{L}(\br^*,\bu^*)=\nabla_{\br}\|\br\|_1-F\bu$. Consider two cases as follows. For any $i\in[n]$,
		\begin{enumerate}[label=(\arabic*)]
			\item When $i,-i\notin\Gamma^{\text{ideal}}(t)$. We have $\Big(\br^{\text{ideal}}(t) \Big)_i=0$, i.e., the sub-gradient of the $i$th coordinate of $\|\br^{\text{ideal}}(t)\|_1$ lies in $[-1,1]$. As $F_i^\top\bu^{\text{ideal}}(t)\in[-1,1]$, we have $F_i^\top\bu^{\text{ideal}}(t)\in\partial_{\br_i}\|\br^{\text{ideal}}(t)\|_1$.
			\item When $i\in\Gamma^{\text{ideal}}(t)$ (or $-i\in\Gamma^{\text{ideal}}(t)$). We have $F_i^\top\bu^{\text{ideal}}(t)=1$ (or $F_i^\top\bu^{\text{ideal}}(t)=-1$). As $\text{sgn}\Big(\br^{\text{ideal}}(t) \Big)_i=1$ (or $\text{sgn}\Big(\br^{\text{ideal}}(t) \Big)_i=-1$), we have $F_i^\top\bu^{\text{ideal}}(t)=\text{sgn}\Big(\br^{\text{ideal}}(t) \Big)_i=\partial_{\br_i}\|\br^{\text{ideal}}(t)\|_1$.
		\end{enumerate}
		Finally, the complementary slackness is satisfied because $F^\top\br^*-F^\top\br^{\text{ideal}}(t)=0$.
		As a result, we conclude that $(\br^*,\bu^*)=(\br^\text{ideal}(t),\bu^\text{ideal}(t))$ is the optimal solution of (\ref{op:basispursuit-ideal-perturbed}).
	\end{proof}
	
	Next, we are going to use the perturbation lemma in the Chapter 5.6 of~\cite{boyd2004convex} stated as follows.
	
	\begin{lemma}[perturbation lemma]
		Given the following two optimization programs.\\
		\begin{minipage}{\linewidth}
			\begin{minipage}{0.45\linewidth}
				\begin{equation}\label{op:perturbation-original}
				\begin{aligned}
				& \underset{\br}{\text{minimize}}
				& & f(\br) \\
				& \text{subject to}
				& & h(\br)=\mathbf{0}.
				\end{aligned}
				\end{equation}
			\end{minipage}
			\begin{minipage}{0.45\linewidth}
				\begin{equation}\label{op:perturbation-perturbed}
				\begin{aligned}
				& \underset{\br}{\text{minimize}}
				& & f(\br) \\
				& \text{subject to}
				& & h(\br)=\by.
				\end{aligned}
				\end{equation}
			\end{minipage}
		\end{minipage}
		Let $\OPT^{\text{original}}$ be the optimal value of the original program~\autoref{op:perturbation-original} and $\OPT^{\text{perturbed}}$ be the optimal value of the perturbed program~\autoref{op:perturbation-perturbed}. Let $\bu^*$ be the optimal dual value of the perturbed program~\autoref{op:perturbation-perturbed}. We have
		\begin{equation}
		\OPT^{\text{original}}\geq\OPT^{\text{perturbed}}+\by^\top\bu^*.
		\end{equation}
	\end{lemma}
	
	Now, think of~\autoref{op:l1 min} as the original program and~\autoref{op:basispursuit-ideal-perturbed} as the perturbed program. Namely, $f(\br)=\|\br\|_1$, $h(\br)=F^\top\br-\bx$, and $\by=F^\top\br^{\text{ideal}}(t)-\bx$. By the perturbation lemma, we have
	\begin{align*}
	\OPT^{\ell_1}&\geq \|\br^{\text{ideal}}(t)\|_1 +  \Big(F^\top\br^{\text{ideal}}(t)-\bx\Big)^\top\bu^{\text{ideal}}(t).
	\end{align*}
	As a result, the following upper bound holds.
	\begin{equation}\label{eq:proofofidealalgorithm-1}
	\|\br^{\text{ideal}}(t)\|_1\leq\OPT^{\ell_1} + \|\bu^{\text{ideal}}(t)\|_2\cdot\|\bx-F^\top\br^{\text{ideal}}(t)\|_2.
	\end{equation}
	Finally, as $\bu^{\text{ideal}}(t)$ lies in the feasible region $\{\bu:\ F\bu\|_{\infty}\leq1 \}$ and the range space of $F$, we can upper bound the $\|\bu^{\text{ideal}}(t)\|_2$ term in~\autoref{eq:proofofidealalgorithm-1} as follows.
	\begin{lemma}\label{lemma:proofofidealalgorithm-vbound}
		For every $\bu$ in the range space of $F$ and $\|F \bu\|_{\infty}\leq1$, we have $\|\bu\|_2\leq\sqrt{\frac{n}{\lambda_{\min}}}$.
	\end{lemma}
	\begin{proof}[Proof of~\autoref{lemma:proofofidealalgorithm-vbound}]
		As $\bu$ lies in the range space of $F$, we have $\|F\bu\|_2\geq\sqrt{\lambda_{\min}}\|\bu\|_2$. Also, because $\|F\bu\|_{\infty}\leq1$, we have $\|F\bu\|_2\leq\sqrt{n}$. As a result, 
		\begin{equation*}
		\|\bu\|_2\leq\frac{\|F\bu\|_2}{\sqrt{\lambda_{\min}}}\leq\sqrt{\frac{n}{\lambda_{\min}}}.
		\end{equation*}
	\end{proof}
	By~\autoref{eq:proofofidealalgorithm-1} and~\autoref{lemma:proofofidealalgorithm-vbound},~\autoref{lemma:idealalgorithm-OPTbounds} holds.		
\end{proof}

    \chapter{Publication list}\label{app:papers}

The authorship order of a paper is by contribution if it is marked with `*'; otherwise the authorship follows alphabetical order.

\section{Conference publications (peer-reviewed)}

\noindent Chi-Ning Chou, Alexander Golovnev, Amirbehshad Shahrasbi, Madhu Sudan, and Santhoshini Velusamy. Sketching approximability of (weak) monarchy predicates. In Approximation,Randomization, and Combinatorial Optimization. Algorithms and Techniques (APPROX/RANDOM 2022), 2022

\vspace{5mm}
\noindent J.Peter Love, Chi-Ning Chou, Juspreet Singh Sandhu, and Jonathan Shi. Limitations of local quantum algorithms on maximum cuts of sparse hypergraphs and beyond. In 49th EATCS International Colloquium on Automata, Languages and Programming (ICALP 2022), 2022

\vspace{5mm}
\noindent Alexander Chou, Chi-Ning Golovnev, Madhu Sudan, Ameya Velingker, and Santhoshini Velusamy. Linear space streaming lower bounds for approximating csps. In 54th Annual ACM Symposium on Theory of Computing (STOC 2022), 2022

\vspace{5mm}
\noindent Nai-Hui Chia, Chi-Ning Chou, Jiayu Zhang, and Ruizhe Zhang. Quantum meets the minimum circuit size problem. In 13th Innovations in Theoretical Computer Science Conference (ITCS 2022), 2022

\vspace{5mm}
\noindent Chi-Ning Chou, Alexander Golovnev, Madhu Sudan, and Santhoshini Velusamy. Approximability of all finite csps with linear sketches. In 2021 IEEE 62nd Annual Symposium on Foundations of Computer Science (FOCS 2021), 2021

\vspace{5mm}
\noindent Boaz Barak, Chi-Ning Chou, and Xun Gao. Spoofing linear cross-entropy benchmarking in shallow quantum circuits. In 12th Innovations in Theoretical Computer Science Conference (ITCS 2021), 2021

\vspace{5mm}
\noindent Chi-Ning Chou, Alexander Golovnev, and Santhoshini Velusamy. Optimal streaming approximations for all boolean max-2csps and max-ksat. In 2020 IEEE 61st Annual Symposium on Foundations of Computer Science (FOCS 2020), 2020

\vspace{5mm}
\noindent Chi-Ning Chou and Mien Brabeeba Wang. Ode-inspired analysis for the biological version of oja’s rule in solving streaming pca. In Proceedings of Thirty Third Conference on Learning Theory (COLT 2020), Proceedings of Machine Learning Research. PMLR, 2020

\vspace{5mm}
\noindent Boaz Barak, Chi-Ning Chou, Zhixian Lei, Tselil Schramm, and Yueqi Sheng. (Nearly) efficient algorithms for the graph matching problem on correlated random graphs. In Advances in Neural Information Processing Systems 32 (NeurIPS 2019), 2019

\vspace{5mm}
\noindent Chi-Ning Chou, Zhixian Lei, and Preetum Nakkiran. Tracking the l2 Norm with Constant Up- date Time. In Approximation, Randomization, and Combinatorial Optimization. Algorithms and Techniques (APPROX/RANDOM 2019), 2019

\vspace{5mm}
\noindent Chi-Ning Chou,Kai-Min Chung, and Chi-Jen Lu. On the Algorithmic Power of Spiking Neural Networks. In 10th Innovations in Theoretical Computer Science Conference (ITCS 2019), 2019

\vspace{5mm}
\noindent Chi-Ning Chou, Yu-Jing Lin, Ren Chen, Hsiu-Yao Chang, I-Ping Tu, and Shih-wei Liao. Personalized Difficulty Adjustment for Countering the Double-Spending Attack in Proof-of-Work Consensus Protocols. In 2018 IEEE Conference on Blockchain (Blockchain 2018), 2018

\vspace{5mm}
\noindent Chi-Ning Chou, Mrinal Kumar, and Noam Solomon. Hardness vs Randomness for Bounded Depth Arithmetic Circuits. In 33rd Computational Complexity Conference (CCC 2018), 2018

\section{Journal publications (peer-reviewed)}

\vspace{5mm}
\noindent *Andres E. Lombo, Jesus E. Lares, Matteo Castellani, Chi-Ning Chou, Nancy Lynch, and Karl K. Berggren. A superconducting nanowire-based architecture for neuromorphic computing. Neuromorphic Computing and Engineering, 2022

\vspace{5mm}
\noindent Chi-Ning Chou, Mrinal Kumar, and Noam Solomon. Closure results for polynomial factorization. Theory of Computing, 15(13):1–34, 2019

\section{Manuscripts}

\vspace{5mm}
\noindent *Michael Kuoch, Nikhil Parthasarathy, Chi-Ning Chou, Joel Dapello, Haim Sompolinsky, and SueYeon Chung. Probing Biological and Artificial Neural Networks with Task-Dependent Neural
Manifolds. 2023.

\vspace{5mm}
\noindent *Nihaad Paraouty, Justin D. Yao, Léo Varnet, Chi-Ning Chou, SueYeon Chung, and Dan H. Sanes. Sensory cortex plasticity supports auditory social learning, 2022

\vspace{5mm}
\noindent *Yao-Yuan Yang,Chi-Ning Chou, and Kamalika Chaudhuri. Understanding rare spurious correlations in neural networks. arXiv preprint: 2202.05189, 2022

\vspace{5mm}
\noindent Chi-Ning Chou, Juspreet Singh Sandhu, Mien Brabeeba Wang, and Tiancheng Yu. A general framework for analyzing stochastic dynamics in learning algorithms. arXiv preprint: 2006.06171,  2021

\vspace{5mm}
\noindent *Xun Gao, Marcin Kalinowski, Chi-Ning Chou, Mikhail Lukin, Boaz Barak, and Soonwon Choi. Limitations of linear cross-entropy as a measure for quantum advantage. arXiv preprint: 2112.01657, 2021

\vspace{5mm}
\noindent Chi-Ning Chou, Mrinal Kumar, and Noam Solomon. Closure of VP under taking factors: a short and simple proof. arXiv preprint: 1903.02366, 2019

\end{appendices}

\setstretch{\dnormalspacing}

% the back matter
\backmatter
\clearpage
\begin{spacing}{\dcompressedspacing}
\addcontentsline{toc}{chapter}{References}
\bibliographystyle{unsrtnat}
\end{spacing}

\end{document}